\pdfoutput=1  

\documentclass[bpam]{ipart}
\arxiv{2011.06533}

\usepackage{amsthm}
\usepackage{mathtools}
\usepackage{color}
\usepackage{xcolor}
\usepackage{stmaryrd}
\usepackage{rotating}
\usepackage{adjustbox}

\usepackage{amsmath}
\usepackage{graphicx}

\usepackage{tikz}
\usetikzlibrary{cd}

\usepackage{hyperref}
\hypersetup{
  colorlinks = true,
  allcolors=darkblue,
  urlcolor=black
}

\usepackage[all]{xy}

\usepackage{xfrac}

\usepackage{tikz}

\DeclareMathAlphabet{\mathpzc}{OT1}{pzc}{m}{it} 

\newcommand\mathscr[1]{\scalebox{1.1}{$\mathpzc{#1}$}}

\usepackage[titletoc]{appendix}

\usepackage[safe]{tipa} 

\definecolor{darkblue}{rgb}{0.05,0.25,0.65}
\definecolor{darkgreen}{RGB}{20,140,10}
\definecolor{greenii}{RGB}{20,140,10}
\definecolor{lightgray}{rgb}{0.9,0.9,0.9}
\definecolor{orangeii}{RGB}{200,100,5}
\definecolor{darkyellow}{rgb}{.91,.91,0}

\usepackage{multirow}

\usepackage{stmaryrd}    

\usepackage{enumitem}
\setlist{
  labelsep=6pt,
  leftmargin = *,
  itemsep=-2pt,
  topsep=1pt
}

\usepackage{enumerate} 

\usepackage{amssymb,amsmath}

\newcounter{sqindex}

\newcommand{\dgcAlgebras}{
  \mathrm{DiffGrCAlg}
    _{\scalebox{.6}{$\mathbb{R}$}}
    ^{\scalebox{.6}{$\geq 0$}}
}

\newcommand{\EquivariantdgcAlgebras}{
  \orbisingularGLarge\mathrm{DiffGrCAlg}
    _{\scalebox{.6}{$\mathbb{R}$}}
    ^{\scalebox{.6}{$\geq 0$}}
}

\newcommand{\EquivariantdgcAlgebrasOp}{
  \big(
    \EquivariantdgcAlgebras
  \big)^{\mathrm{op}}
}

\newcommand{\Gade}{
  G^{\mathrm{ADE}}
}

\newcommand{\ZTwo}{\mathbb{Z}_{2}}

\newcommand{\Grefl}{
  \ZTwo
}

\newcommand{\SpLR}{
  \mathrm{Sp}(1)
}

\newcommand{\EquivariantdgcAlgebrasProj}{
  \big(
    \EquivariantdgcAlgebras
  \big)_{\mathrm{proj}}
}

\newcommand{\EquivariantdgcAlgebrasProjOp}{
  \big(
    \EquivariantdgcAlgebras
  \big)^{\mathrm{op}}_{\mathrm{proj}}
}

\newcommand{\RepresentationRing}[1]{
  \mathrm{Rep}_{\scalebox{.6}{$#1$}}
}

\newcommand{\EquivariantRepresentationRing}[1]{
  \underline{\RepresentationRing{#1}}
}

\newcommand{\EquivariantSSet}{
  \orbisingularGLarge\mathrm{SSet}
}

\newcommand{\ZTwoEquivariantSSet}{
  \orbisingularZTwoLarge\mathrm{SSet}
}

\newcommand{\EquivariantSSetProj}{
  \EquivariantSSet_{\mathrm{proj}}
}

\newcommand{\ReducedSSet}{
  \mathrm{RedSSet}
}

\newcommand{\HomotopyTypes}{
  \mathrm{HoTypes}
}

\newcommand{\GadeEquivariantHomotopyTypes}{
  \orbisingularGadeLarge\HomotopyTypes
}

\newcommand{\EquivariantHomotopyTypes}{
  \orbisingularGLarge\HomotopyTypes
}

\newcommand{\EquivariantHomotopyTypesConnected}{
  \orbisingularGLarge\HomotopyTypes_{\geq 1}
}

\newcommand{\EquivariantSimplyConnectedHomotopyTypes}{
  \orbisingularGLarge\HomotopyTypes_{\geq 2}
}

\newcommand{\EquivariantHomotopyTypesSimplyConnected}{
  \EquivariantSimplyConnectedHomotopyTypes
}

\newcommand{\ZTwoActionsOnTopSp}{
  \ZTwo\mathrm{Act}
  \big(
    \mathrm{TopSp}
  \big)
}

\newcommand{\GOrbifolds}{
  G\mathrm{Orbifolds}
}

\newcommand{\TActionsOnTopSp}{
  T\mathrm{Act}
  \big(
    \mathrm{TopSp}
  \big)
}

\newcommand{\TGActionsOnTopSp}{
  \big(
    T \!\times G
  \big)
  \mathrm{Act}
  \big(
    \mathrm{TopSp}
  \big)
}

\newcommand{\GActionsOnTopSp}{
  G\mathrm{Act}
  \big(
    \mathrm{TopSp}
  \big)
}

\newcommand{\TGActionsOnSmoothManifolds}{
  \big(
    T \!\times G
  \big)\mathrm{Act}
  \big(
    \mathrm{SmthMfd}
  \big)
}

\newcommand{\TActionsOnSmoothManifolds}{
  G\mathrm{Act}
  \big(
    \mathrm{SmthMfd}
  \big)
}

\newcommand{\GActionsOnSmoothManifolds}{
  G\mathrm{Act}
  \big(
    \mathrm{SmthMfd}
  \big)
}

\newcommand{\GadeActionsOnSmoothManifolds}{
  \Gade\mathrm{Act}
  \big(
    \mathrm{SmthMfd}
  \big)
}

\newcommand{\ZTwoActionsOnSmoothManifolds}{
  \ZTwo\mathrm{Act}
  \big(
    \mathrm{SmthMfd}
  \big)
}

\newcommand{\GActionsOnSSet}{
  G\mathrm{Act}
  \big(
    \mathrm{SSet}
  \big)
}

\newcommand{\GActionsOnSets}{
  G\mathrm{Act}
  \big(
    \mathrm{Set}
  \big)
}

\newcommand{\EquivariantSimplyConnectedRFiniteHomotopyTypes}{
  \EquivariantHomotopyTypes
    _{\geq 2}^{\mathrm{fin}_{\mathbb{R}}}
}

\newcommand{\LieAlgebras}{
   \mathrm{LieAlg}
     _{\scalebox{.6}{$\mathbb{R}, \mathrm{fin}$}}
}

\newcommand{\LInfinityAlgebras}{
   L_\infty\mathrm{Alg}
     ^{\scalebox{.6}{$\geq 0$}}
     _{\scalebox{.6}{$\mathbb{R}, \mathrm{fin}$}}
}

\newcommand{\EquivariantLInfinityAlgebras}{
   \orbisingularGLarge L_\infty\mathrm{Alg}
     ^{\scalebox{.6}{$\geq 0$}}
     _{\scalebox{.6}{$\mathbb{R}, \mathrm{fin}$}}
}

\newcommand{\dgcAlgebrasFin}{
  \mathrm{DiffGrCAlg}
    ^{\scalebox{.6}{$\geq 0, \mathrm{fin}$}}
    _{\scalebox{.6}{$\mathbb{R}$}}
}

\newcommand{\CochainComplexes}{
    \mathrm{CoCmplx}
      ^{\scalebox{.6}{$\geq 0$}}
      _{\scalebox{.6}{$\mathbb{R}$}}
}

\newcommand{\EquivariantCochainComplexes}{
    \orbisingularGLarge\mathrm{CoComplx}
      ^{\scalebox{.6}{$\geq 0$}}
      _{\scalebox{.6}{$\mathbb{R}$}}
}

\newcommand{\CochainComplexesFin}{
    \mathrm{CoCmplx}
      ^{\scalebox{.6}{$\geq 0, \mathrm{fin}$}}
      _{\scalebox{.6}{$\mathbb{R}$}}
}

\newcommand{\VectorSpaces}{
    \mathrm{VecSp}
      _{\scalebox{.6}{$\mathbb{R}$}}
}

\newcommand{\VectorSpacesFin}{
    \mathrm{VecSp}
      ^{\scalebox{.6}{$\mathrm{fin}$}}
      _{\scalebox{.6}{$\mathbb{R}$}}
}

\newcommand{\EquivariantGroups}{
    \orbisingularGLarge\mathrm{Grp}
}

\newcommand{\EquivariantAbelianGroups}{
    \orbisingularGLarge\mathrm{AbelianGroups}
}

\newcommand{\EquivariantVectorSpaces}{
    \orbisingularGLarge\mathrm{VecSp}
      _{\scalebox{.6}{$\mathbb{R}$}}
}

\newcommand{\EquivariantDualVectorSpaces}{
    \orbisingularGLarge\mathrm{VecSp}
      ^{\scalebox{.6}{$\!\vee$}}
      _{\scalebox{.6}{$\mathbb{R}$}}
}

\newcommand{\EquivariantDualVectorSpacesInjective}{
    \orbisingularGLarge\mathrm{VecSp}
      ^{\scalebox{.6}{$\!\vee, \mathrm{inj}$}}
      _{\scalebox{.6}{$\mathbb{R}$}}
}

\newcommand{\EquivariantVectorSpacesFin}{
    \orbisingularGLarge\mathrm{VecSp}
      ^{\scalebox{.6}{$\mathrm{fin}$}}
      _{\scalebox{.6}{$\mathbb{R}$}}
}

\newcommand{\EquivariantDualVectorSpacesFin}{
    \orbisingularGLarge\mathrm{VecSp}
      ^{\!\scalebox{.6}{$\vee$},\scalebox{.5}{$\mathrm{fin}$}}
      _{\scalebox{.6}{$\mathbb{R}$}}
}

\newcommand{\Representations}{
    \mathrm{Rep}
      _{\scalebox{.6}{$\mathbb{R}$}}
}

\newcommand{\RepresentationsLeft}{
    \mathrm{Rep}
      ^{\scalebox{.7}{$l$}}
      _{\scalebox{.6}{$\mathbb{R}$}}
}

\newcommand{\RepresentationsRight}{
    \mathrm{Rep}
      ^{\scalebox{.7}{$r$}}
      _{\scalebox{.6}{$\mathbb{R}$}}
}

\newcommand{\RepresentationsFin}{
    \mathrm{Rep}
      ^{\scalebox{.6}{$\mathrm{fin}$}}
      _{\scalebox{.6}{$\mathbb{R}$}}
}

\newcommand{\WeylGroup}{
    \mathrm{W}
      _{\scalebox{.6}{$\!\!G$}}
}

\newcommand{\NormalizerGroup}{
    \mathrm{N}
      _{\scalebox{.6}{$\!G$}}
}

\newcommand{\GradedVectorSpaces}{
    \mathrm{GrVecSp}
      ^{\scalebox{.6}{$\geq 0$}}
      _{\scalebox{.6}{$\mathbb{R}$}}
}

\newcommand{\EquivariantGradedVectorSpaces}{
    \orbisingularGLarge\mathrm{GrVecSp}
      ^{\scalebox{.6}{$\geq 0$}}
      _{\scalebox{.6}{$\mathbb{R}$}}
}

\newcommand{\ZTwoEquivariantGradedVectorSpaces}{
    \ZTwo\mathrm{GrVecSp}
      ^{\scalebox{.6}{$\geq 0$}}
      _{\scalebox{.6}{$\mathbb{R}$}}
}

\newcommand{\mapsdown}{\rotatebox[origin=c]{-90}{$\mapsto$}}

\usepackage[new]{old-arrows}   

\newdir{> }{{}*!/10pt/@{>}}

\usepackage{amssymb}

\def\acts{\raisebox{1.4pt}{\;\rotatebox[origin=c]{90}{$\curvearrowright$}}\hspace{.5pt}}

\DeclareRobustCommand{\rchi}{{\mathpalette\irchi\relax}}
\newcommand{\irchi}[2]{\raisebox{\depth}{$#1\chi$}} 

\makeatletter
\newif\if@sup
\newtoks\@sups
\def\append@sup#1{\edef\act{\noexpand\@sups={\the\@sups #1}}\act}%
\def\reset@sup{\@supfalse\@sups={}}%
\def\mk@scripts#1#2{\if #2/ \if@sup ^{\the\@sups}\fi \else%
  \ifx #1_ \if@sup ^{\the\@sups}\reset@sup \fi {}_{#2}%
  \else \append@sup#2 \@suptrue \fi%
  \expandafter\mk@scripts\fi}
\def\tensor#1#2{\reset@sup#1\mk@scripts#2_/}
\def\multiscripts#1#2#3{\reset@sup{}\mk@scripts#1_/#2%
  \reset@sup\mk@scripts#3_/}
\makeatother

\makeatletter
\newbox\slashbox \setbox\slashbox=\hbox{$/$}
\def\itex@pslash#1{\setbox\@tempboxa=\hbox{$#1$}
  \@tempdima=0.5\wd\slashbox \advance\@tempdima 0.5\wd\@tempboxa
  \copy\slashbox \kern-\@tempdima \box\@tempboxa}
\def\slash{\protect\itex@pslash}
\makeatother

\def\clap#1{\hbox to 0pt{\hss#1\hss}}
\def\mathllap{\mathpalette\mathllapinternal}
\def\mathrlap{\mathpalette\mathrlapinternal}
\def\mathclap{\mathpalette\mathclapinternal}
\def\mathllapinternal#1#2{\llap{$\mathsurround=0pt#1{#2}$}}
\def\mathrlapinternal#1#2{\rlap{$\mathsurround=0pt#1{#2}$}}
\def\mathclapinternal#1#2{\clap{$\mathsurround=0pt#1{#2}$}}

\let\oldroot\root
\def\root#1#2{\oldroot #1 \of{#2}}
\renewcommand{\sqrt}[2][]{\oldroot #1 \of{#2}}

\DeclareSymbolFont{symbolsC}{U}{txsyc}{m}{n}
\SetSymbolFont{symbolsC}{bold}{U}{txsyc}{bx}{n}
\DeclareFontSubstitution{U}{txsyc}{m}{n}

\DeclareSymbolFont{stmry}{U}{stmry}{m}{n}
\SetSymbolFont{stmry}{bold}{U}{stmry}{b}{n}

\DeclareFontFamily{OMX}{MnSymbolE}{}
\DeclareSymbolFont{mnomx}{OMX}{MnSymbolE}{m}{n}
\SetSymbolFont{mnomx}{bold}{OMX}{MnSymbolE}{b}{n}
\DeclareFontShape{OMX}{MnSymbolE}{m}{n}{
    <-6>  MnSymbolE5
   <6-7>  MnSymbolE6
   <7-8>  MnSymbolE7
   <8-9>  MnSymbolE8
   <9-10> MnSymbolE9
  <10-12> MnSymbolE10
  <12->   MnSymbolE12}{}

\usepackage{cleveref}

\crefformat{section}{\S#2#1#3} 
\crefformat{subsection}{\S#2#1#3}
\crefformat{subsubsection}{\S#2#1#3}

\theoremstyle{italics}
\newtheorem{theorem}{Theorem}[section]
\newtheorem{lemma}[theorem]{Lemma}
\newtheorem{prop}[theorem]{Proposition}

\theoremstyle{definition}
\newtheorem{defn}[theorem]{Definition}
\newtheorem{notation}[theorem]{Notation}
\newtheorem{example}[theorem]{Example}

\newtheorem{remark}[theorem]{Remark}

\usepackage{amsfonts}
\usepackage{colortbl}

\renewcommand{\emph}{\textit}

\def\orbisingular{\rotatebox[origin=c]{70}{$\prec$}}

\def\orbisingularGLarge{ \raisebox{-4.2pt}{$\scalebox{.76}{$\orbisingular$}^{\hspace{-4.7pt}\raisebox{1pt}{\scalebox{.95}{$G$}}}$}
}
\def\orbisingularGadeLarge{ \raisebox{-4.2pt}{$\scalebox{.76}{$\orbisingular$}^{\hspace{-4.7pt}\raisebox{1pt}{\scalebox{.95}{$\Gade$}}}$}
}
\def\orbisingularZTwoLarge{\raisebox{-4.2pt}{$\scalebox{.76}{$\orbisingular$}^{\hspace{-4pt}\raisebox{1.2pt}{\scalebox{.88}{$\mathbb{Z}_{\scalebox{.6}{$2$}}$}}}$}}

\newcommand{\defneq}{\equiv}

\newcommand{\ten}{1\!0}

\begin{document}

\setlength{\abovedisplayskip}{3pt}
\setlength{\belowdisplayskip}{3pt}
\setlength{\abovedisplayshortskip}{-3pt}
\setlength{\belowdisplayshortskip}{3pt}

\begin{frontmatter}

\title
  [Character Map in Twisted Equivariant Nonabelian Cohomology] 
  { The Character Map in  Twisted Equivariant Nonabelian Cohomology} 

    \begin{aug}
        \author{
          \fnms{Hisham} \snm{Sati}
          \thanksref{CQTS}
          \thanksref{Courant}
          \ead[label=email1]{hsati@nyu.edu}
         }
          \address{
          Center for Quantum and Topological Systems (CQTS), 
          \\
          Research Institute, 
             New York University Abu Dhabi, 
             \\
             Saadiyat Island, Abu Dhabi, 
             \\
             UAE             
             \\
             \printead{email1}
           }
        \and
        \author{
          \fnms{Urs} \snm{Schreiber}
          \thanksref{CQTS}
          \ead[label=email2]{us13@nyu.edu}
        }
        \address{
        Center for Quantum and Topological Systems (CQTS), 
        \\
        Research Institute, 
           New York University Abu Dhabi, 
           \\
           Saadiyat Island, Abu Dhabi, 
           \\
           UAE             
           \\
           \printead{email2}
       }
        \thankstext{CQTS}{
           Mathematics, Division of Science; and
           Center for Quantum and Topological Systems,
           NYUAD Research Institute,
           New York University Abu Dhabi, UAE.   
         }
      \thankstext{Courant}{
        The Courant Institute for Mathematical Sciences, NYU, NY.
      }
    \end{aug}

\begin{abstract}
  The fundamental notion of {\it non-abelian} generalized cohomology gained recognition 
  in algebraic topology as the non-abelian Poincar{\'e}-dual to ``factorization homology'', and in theoretical physics as providing flux-quantization for non-linear Gau{\ss} laws. 
  However, already the archetypical example --- unstable Cohomotopy, first studied almost a century ago by Pontrjagin --- has remained underappreciated as a cohomology theory and has only recently received attention as a flux-quantizaton law (``Hypothesis H'').
  
  \smallskip 
  Here we lay out a general construction of the analogue of the Chern character map on twisted equivariant non-abelian cohomology theories (with equivariantly simply-connected classifying spaces) and illustrate the construction by spelling out a twisted equivariant form of Cohomotopy as an archetypical and intriguing running example, essentially by computing its equivariant Sullivan model.

  \smallskip 
  We close with an outlook on the application of this result to the rigorous deduction of anyonic quantum states on M5-branes wrapped over Seifert 3-orbifolds.
\end{abstract}

\begin{keyword}[class=AMS]  
  \kwd[Primary ]
      {55N25} 
  \kwd{55N20} 
  \kwd{18G50} 
  \kwd{57R20} 
  \kwd[; secondary ]
      {57R18} 
  \kwd{55Q55} 
  \kwd{81T30} 
\end{keyword}

\begin{keyword}
    \kwd{algebraic topology}
    \kwd{generalized cohomology}
    \kwd{nonabelian cohomology} 
    \kwd{equivariant cohomology} 
    \kwd{twisted cohomology}
    \kwd{character map}
    \kwd{M-branes}
\end{keyword}

\tableofcontents

\end{frontmatter}



\section{Introduction \& Overview}

Algebraic topology, of course, is the study of spaces via systems of (co)homologi{-}cal invariants (cf. \cite{Munkres84}\cite{Hatcher02} \cite{tD08}). Ever more generalized versions of cohomology are routinely discussed these days, but an ancient and archetypical example --- namely unstable Cohomotopy \eqref{Cohomotopy}, which we will refer to as just {\it Cohomotopy}, cf. \cite{Spanier49} \cite[\S VII]{STHu59} (going back to \cite{Borsuk36}\cite{Pontrjagin38})
--- has received little attention as a cohomology theory, since as such it falls outside the scope even of the generalized cohomology theories as commonly understood today: it is a {\it non-abelian} generalized cohomology theory, as we recall in a moment.

\smallskip 
Motivated by recent application \cite{SS24-Flux} in theoretical physics (exposition in \S\ref{ApplicationToFLuxQuantization}) of non-abelian generalized cohomology in general and of Cohomotopy 
in particular, we develop here the {\it equivariant} enhancement of the {\it character map} on twisted non-abelian cohomology theories due to \cite{FSS23-Char} and illustrate it by presenting a case study of constructions on and phenomena exhibited by unstable Cohomotopy when regarding it as a twisted equivariant cohomology theory. 
Concretely, besides the development of the general theory of the equivariant non-abelian character map, our main result
(Thm. \ref{FluxQuantizationInEquivariantTwistorialCohomotopy})
is the construction and analysis of
the nonabelian character map \cite{FSS23-Char}
on the ``twistorial'' variant of low-degree Cohomotopy from \cite{FSS20c}, now generalized to $\ZTwo$-equivariant form. 

We may motivate this example by its application in high energy physics indicated in the outlook section \S\ref{ApplicationToFLuxQuantization}, but at the same time ---
due to the low-dimensional spheres and projective spaces it involves, these being being among the simplest cell complexes --- it is also one of the most basic examples of the theory generally, as such of interest in its own right, and shall serve as our running example illustrating all our constructions.

\smallskip 
At the heart of this computation is, for reasons explained in a moment, the computation of minimal Sullivan models of fibrations of some basic cell complexes (like $S^7$ and $\mathbb{C}P^3$) but in the generality of {\it equivariant} homotopy theory (recalled in \S\ref{GEquivariantHomotopyTheory}). Since in this context even such basic examples of equivariant Sullivan models have not been discussed in print before -- to the best of our knowledge -- the reader may take \S\ref{EquivariantNonAbelianDeRhamCohomology} as an exposition of the notoriously more intricate equivariant version of dg-algebraic rational homotopy theory (which has seen little application in the past) along some illustrative examples and under the perspective of the equivariant generalization
(in \S\ref{EquivariantNonAbelianDeRhamTheorem}) of the non-abelian de Rham theorem from \cite[\S 6]{FSS23-Char}.

After the proof of the main theorem is thereby completed, for the inclined reader we end in \S\ref{ApplicationToFLuxQuantization} with a brief outlook on the somewhat remarkable implications of our computations to recent questions in theoretical physics, specifically to the rigorous derivation of anyonic quantum states on M5-branes.

\smallskip

But first, to set the scene, it is worthwhile to briefly take a step back and reconsider the notion of cohomology as such:

\medskip

\noindent
{\bf Cohomology via classifying spaces.}
It is a classical and yet possibly undervalued fact that {\it reasonable cohomology theories have classifying spaces} (and more generally {\it classifying stacks}). To quickly recall (more details and pointers in \cite[\S 2]{FSS23-Char}):

\smallskip

\noindent {\bf -- Ordinary cohomology.}
This begins with the observation that (reduced) ordinary singular cohomology, with coefficients in a discrete abelian group $A$, is classified in degree $n$ by Eilenberg-MacLane spaces $K(A,n)$ -- in that on well-behaved topological spaces $X$, notably on smooth manifolds, there are natural isomorphisms between the ordinary cohomology groups and the connected components of the respective (pointed) mapping spaces:
\begin{equation}
  \label{OrdinaryCohomologyRepresented}
  \def\arraystretch{1.5}
  \begin{array}{l}
  H^n(X;\, A)
  \;\simeq\;
  \pi_0
  \,
  \mathrm{Maps}
  \big(
    X,\,
    K(A,n)
  \big)
  \,,
  \\
  \widetilde H^n(X;\, A)
  \;\simeq\;
  \pi_0
  \,
  \mathrm{Maps}^{\ast}
  \big(
    X,\,
    K(A,n)
  \big)
  \,.
  \end{array}
\end{equation}
This equivalence makes manifest the characteristic properties of cohomology: homotopy invariance, exactness and wedge property, since these are now immediately implied by general abstract properties of mapping spaces.

Moreover, these EM-spaces are in fact loop spaces of each other, via weak homotopy equivalences
\begin{equation}
  \label{LoopingEquivOnEmSpaces}
  \sigma_n
  \;:\;
  \begin{tikzcd}
    K(A,n)
    \ar[r, "{\sim}"]
    &
    \Omega K(A,n+1)
  \end{tikzcd}
\end{equation}
that thereby represent the {\it suspension isomorphisms} between ordinary cohomology groups, as follows: 
$$
  \def\arraystretch{1.7}
  \begin{array}{rcl}
    \widetilde H^{n}(X;A)
    &
    \xrightarrow
      [
        \scalebox{.7}{
          \color{gray}\eqref{OrdinaryCohomologyRepresented}
        }
      ]
      { \phantom{-}\sim\phantom{-} }
    &
    \mathrm{Maps}^{\ast/\!}\big(
      X
      ,\,
      K(A,n)
    \big)
    \xrightarrow
      [
        \scalebox{.7}{
          \color{gray}
          \eqref{LoopingEquivOnEmSpaces}
        }
      ]
      {  
        (\sigma_n)_\ast
      }
    \mathrm{Maps}^{\ast/\!}\big(
      X
      ,\,
      \Omega
      K(A,n+1)
    \big)
    \\
    &
    \xrightarrow
      [
        \mathclap{
        \scalebox{.7}{
          \color{gray}
          adjunction
        }
        }
      ]
      { \phantom{-}\sim\phantom{-} }
    &
    \mathrm{Maps}^{\ast/\!}\big(
      \Sigma
      X
      ,\,
      K(A,n+1)
    \big)
    \xrightarrow
      [
        \scalebox{.7}{
          \color{gray}\eqref{OrdinaryCohomologyRepresented}
        }
      ]
      { \phantom{-}\sim\phantom{-} }
    \widetilde H^{n+1}\big(
      \Sigma X
      ;\,
      A
    \big)
    \,.
  \end{array}
$$

\noindent {\bf -- Ordinary non-abelian cohomology.} Note here that it is the loop space property \eqref{LoopingEquivOnEmSpaces}, and hence the corresponding suspension isomorphism, which reflect the fact that the coefficient $A$ has been assumed to be an {\it abelian} group:
For a non-abelian group $G$, an Eilenberg-MacLane space $K(G, 1) \,\simeq\, B G$ still exists, but is {\it not a loop space}.

While the suspension isomorphism is thus lost for non-abelian coefficients, the assignment
\begin{equation}
  \label{OrdinaryNonAbelianCohomology}
  X
  \;\;
   \longmapsto
  \;\;
  H^1(
    X
    ;\, 
    G
  )
  \;\;
  :=
  \;\;
  \pi_0
  \,
  \mathrm{Maps}(
    X
    ,\,
    B G
  )
  \;\;\;
  \in\;
  \mathrm{Set}^{\ast/}
\end{equation}
still satisfies homotopy invariance, exactness and wedge property, just by the general properties of mapping spaces, and hence has all the characteristic properties of ordinary cohomology -- except for its abelian-ness.
Accordingly, \eqref{OrdinaryNonAbelianCohomology} is known as {\it non-abelian cohomology}, famous from early applications in Chern-Weil theory.

\medskip

\noindent {\bf -- Whitehead-generalized cohomology theory.} But if or as long as we do insist on abelian cohomology groups related by suspension isomorphisms, we may still immediately generalize ordinary cohomology in the form \eqref{OrdinaryCohomologyRepresented}, simply by using any other sequence of classifying spaces $(E_n)_{n=0}^\infty$, being successive loop spaces of each other as in \eqref{LoopingEquivOnEmSpaces},
$$
  \sigma_n \,:\,
  \begin{tikzcd}
    E_n
    \ar[r, "{ \sim }"]
    &
    \Omega
    E_{n+1}
    \,,
  \end{tikzcd}
$$
as such called a {\it sequential $\Omega$-spectrum of spaces}, or just a {\it spectrum}, for short. The Brown representability theorem says that the resulting assignments
$$
  X
  \;\;
  \mapsto
  \;\;
  E^n(X)
  \,:=\,
  \pi_0
  \,
  \mathrm{Maps}(
    X
    ;\,
    E_n
  )
$$
are equivalently the {\it generalized cohomology theories} as introduced by Whitehead, including examples such as K-theory, elliptic cohomology and cobordism cohomology.

\medskip

\noindent {\bf -- Non-abelian generalized cohomology.} But as we just saw, suspension isomorphisms are to be regarded as {\it extra} structure on cohomology. Not necessarily requiring them leads to consider {\it any pointed space} $\mathscr{A}$ (which we may as well assume to be connected) as the classifying space of a non-abelian generalized cohomology theory, defined in evident generalization of \eqref{OrdinaryNonAbelianCohomology} simply by
\begin{equation}
  \label{NonabelianCohomology}
  H^1(
    X
    ;\,
    \Omega\mathscr{A}
  )
  \;\;
  :=
  \;\;
  \pi_0
  \,
  \mathrm{Maps}(
    X
    ,\,
    \mathscr{A}
  )
  \,.
\end{equation}
Here the notation on the left is suggestive of the fact that any loop space $\Omega \mathscr{A}$ canonically carries the structure of a higher homotopy-coherent group -- a groupal $A_\infty$-space or $\infty$-group, for short --  whose de-looping is equivalent to the connected component of the original space (cf. \cite[Prop. 2.2]{FSS23-Char}):
\begin{equation}
  \label{DeloopingEquivalence}
  \mathscr{A}
  \;\;
  \simeq
  \;\;
  B \, \Omega \mathscr{A}
  \,.
\end{equation}

For instance, in the archetypical case where
$\mathscr{A} \,\defneq\, S^n$ is the $n$-sphere, then the non-abelian generalized cohomology theory that it classifies is known as (unstable) {\it Cohomotopy} $\pi^n$ (cf. \cite{Spanier49}\cite[\S VII]{STHu59}\cite[Ex. 2.7]{FSS23-Char})
\begin{equation}
  \label{Cohomotopy}
  \widetilde H^1\big(
    X
    ;\, 
    \Omega S^n
  \big)
  \;\;
  \defneq
  \;\;
  \pi_0
  \,
  \mathrm{Maps}^{\ast/\!}\big(
    X
    ,\,
    S^n
  \big)
  \;\;
  \defneq
  \;\;
  \pi^n(X)
  \,,
\end{equation}
in dual reference to the familar {\it homotopy} groups
$$
  \pi_n(X)
  \;\;
  \simeq
  \;\;
  \pi_0
  \,
  \mathrm{Maps}^{\ast/\!}\big(
    S^n
    ,\,
    X
  \big)
  \,.
$$

Another example of non-abelian generalized cohomology is unstable topological K-theory \cite{HamanakaKono04}, whose classifying spaces are taken to be finite stages $\mathrm{U}(n)$ of the sequential colimits which construct the classifying spaces of topological K-theory.

\medskip

\noindent
{\bf Developing non-abelian cohomology.} Fundamental, elementary, and compelling as the notion of non-abelian generalized cohomology in \eqref{NonabelianCohomology} is, it has long remained underappreciated. For example, none of the original authors 
\cite{Borsuk36}\cite{Pontrjagin38}\cite{Spanier49}
on Cohomotopy \eqref{Cohomotopy} address their subject as a cohomology theory, instead the early development revolves around partial fixes for the perceived defect of co-homotopy sets to not in general carry group structure.
The situation does not improve with the early development of ``non-abelian gerbes'', whose original description \cite{Giraud71} appears unwieldy. 

\smallskip 
Explicit acknowledgment of (stacky) non-abelian generalized cohomology 
in the transparent guise \eqref{NonabelianCohomology}
appears only in a lecture  \cite{Toen02} (possibly following \cite{Simpson02}). 
Two independent developments in 2009 finally put non-abelian generalized cohomology into fruitful context: 
\begin{itemize}
\item The discovery of non-abelian Poincar{\'e} duality \cite[\S 3.8]{Lurie09-DAGVI}, relating non-abelian cohomology (later made explicit in \cite[Def. 6]{Lurie14}) of manifolds 
to ``non-abelian homology'' in the guise of ``factorization homology'' (which, in contrast to non-abelian cohomology, takes work to define); 

\item The observation in theoretical physics 
\cite{SS08-nactwist}\cite{Schreiber09}\cite{SSS12}
that charge/flux-quanti{-}zation laws  \cite{SS24-Flux} for higher gauge fields
are generally in non-abelian cohomology.

\end{itemize}

\medskip

With non-abelian generalized cohomology thus recognized as a worthwhile subject, we are led to generalize familiar constructions in abelian cohomology, as far as possible, and to explore the consequences.

\smallskip

First, we may straightforwardly equip non-abelian cohomology with further attributes:
Considering the right-hand side of \eqref{NonabelianCohomology} not just for plain spaces but for sheaves of spaces (higher stacks) leads to non-abelian generalized sheaf cohomology, including, in particular, non-abelian generalized versions of twisted cohomology and of equivariant cohomology (also of differential cohomology, but this shall not concern as here):

\medskip

\noindent {\bf -- Equivariant non-abelian cohomology.} Via the above identification of cohomology sets with
homotopy classes of maps to a classifying space, every flavor of homotopy theory comes with its corresponding flavor of cohomology theories. 
In {\it equivariant homotopy theory} one considers (cf. \cite{SS25-EBun}) topological spaces $\mathscr{A}$ equipped with the action $G \acts \mathscr{A}$ of a (finite, for our purposes) group $G$ and with $G$-equivariant maps between them -- and the corresponding flavor of cohomology is {\it equivariant cohomology} (which we also call {\it proper equivariant} cohomology in order to distinguish it from the coarser form of Borel-equivariance):
\begin{equation}
  \label{EquivariantCohomology}
  H^1_G\big(
    X;\,
    \Omega\mathscr{A}
  \big)
  \;\;
  =
  \;\;
  \pi_0
  \, 
  \mathrm{Maps}\big(
    G \acts \, X
    ,\,
    G \acts \, \mathscr{A}
  \big)^G
  \,.
\end{equation}

Here the notion of $G$-homotopy equivalence of maps is straightforward but, at face value, technically cumbersome
to reason about. However, Elmendorf's theorem (recalled as Prop. \ref{ElmendorfTheorem} below) reveals that $G$-homotopy equivalences (between $G$-cell complexes) are nothing but systems of ordinary weak homotopy equivalences between the $H$-fixed spaces $\mathscr{A}^H$ for all subgroups $H \subset G$. These systems of fixed spaces are conveniently re-packaged as presheaves on a small category called the {\it orbit category} $\mathrm{Orb}(G)$ of $G$, whence $G$-equivariant homotopy theory is equivalently the homotopy theory of presheaves of spaces on $\mathrm{Orb}(G)$.

\smallskip

\noindent {\bf -- Twisted non-abelian cohomology.} 
Somewhat similarly,
given any space $\mathscr{B}$ in any homotopy theory, the {\it $\mathscr{B}$-slice} is the homotopy theory whose objects are spaces fibered over $\mathscr{B}$ with maps between them respecting the fibration up to specified homotopy. If we assume, without essential restriction, that the base space is connected, then we may identify it as $\mathscr{B} \,\simeq\, B \mathscr{G}$, as in \eqref{DeloopingEquivalence}, which exhibits any fibration over it as the Borel construction $\mathscr{A} \sslash \mathscr{G}$ of the homotopy-quotient of a homotopy-coherent action $\mathscr{G} \acts \mathscr{A}$.

If we now think of a domain object $X \xrightarrow{\tau} B\mathscr{G}$ in this $B \mathscr{G}$-slice as a {\it twist} and of a codomain object $\mathscr{A} \sslash \mathscr{G} \xrightarrow{p}  B \mathscr{G}$ as a {\it local coefficient bundle}, then the corresponding non-abelian cohomology is just the homotopy classes of sections of the $\tau$-associated $\mathscr{A}$-fiber bundle, and as such is $\tau$-twisted $\mathscr{A}$-cohomology \cite[\S 3]{FSS23-Char}:
\vspace{1mm} 
\begin{equation}
  \label{TwistedCohomology}
  H^{1+\tau}(
    X
    ,\,
    \Omega \mathscr{A}
  )
  :=
  \pi_0
  \,
  \mathrm{Maps}\big(
    X 
    ,\,
    \mathscr{A}\!\sslash\!\mathscr{G}
  \big)_{
    \scalebox{.7}{$
      /B \mathscr{G}
    $}
  }
  \;
  =
  \;
  \left\{\!\!\!\!\!
  \adjustbox{raise=2pt}{
  \begin{tikzcd}[row sep=15pt, column sep=30pt]
    &
    \mathscr{A}
    \!\sslash\!
    \mathscr{G}
    \ar[
      d,
      "{ p }"
    ]
    \\
    X
    \ar[
      r, 
      "{ \tau }"
    ]
    \ar[
      ur,
      dashed
    ]
    &
    B\mathscr{G}
  \end{tikzcd}
  }
 \!\!\!\!\! \right\}_{\!\!\!\!\!\Big/\!\!\!\!\!\!\!\!\!
  \scalebox{.7}{
    \def\arraystretch{.9}
    \def\tabcolsep{0pt}
    \begin{tabular}{c}
      relative
      \\
      homotopy
    \end{tabular}
  }}
\end{equation}

\vspace{1mm} 
\noindent This works generally: If all spaces here are {\it in addition} equipped with $G$-actions as in \eqref{EquivariantCohomology}, hence if we are looking at a slice of equivariant 
homotopy, then the above is automatically {\it twisted \& equivariant} non-abelian cohomology. This is what we shall be concerned with here, concretely with the character map in this generality:

\medskip

\noindent {\bf -- The non-abelian character.}
A famous construction on abelian cohomology is the {\it Chern-Dold character map} to de Rham cohomology, which in the case of K-cohomology becomes the familiar Chern character (and which on ordinary cohomology is essentially just the de Rham theorem). 
One may think of the Chern-Dold character as universally extracting the non-torsion data in the cohomology groups. 
Its generalization to non-abelian cohomology was developed in \cite{FSS23-Char}: 

Observe that the Chern-Dold character is essentially just the cohomology operation induced by  {\it rationalization} of the classifying space, 
$$
  \begin{tikzcd}[column sep=large]
    \mathscr{A}
    \ar[
      rr,
      "{
        \eta^{\mathbb{Q}}
      }",
      "{
        \scalebox{.7}{
          \color{darkgreen}
          \bf
          rationalization
        }
      }"{swap}
    ]
    &&
    L^{\!\mathbb{Q}}
    \mathscr{A}\;.
  \end{tikzcd}
$$
As such, it makes sense in the generality of non-abelian classifying spaces (immediately so under mild technical assumptions, such as nilpotency, but with more work also more generally). In view of this, the fundamental theorem of dg-algebraic rational homotopy theory may be re-cast as a {\it non-abelian de Rham theorem} which identifies, over smooth manifolds $X$, the resulting non-abelian rational cohomology with the concordance classes of flat differential forms having coefficients in the real Whitehead-bracket $L_\infty$-algebra $\mathfrak{l}\mathscr{A}$ of the classifying space:
\begin{equation}
  \label{CharacterInIntroduction}
  \hspace{-.5cm}
  \adjustbox{
    scale=.92
  }{
  \begin{tikzcd}[
    row sep=0pt, 
    column sep=25pt
  ]
    \overset
    {
      \scalebox{1}{$
      H^1\big(
        X
        ;\,
        \mathscr{A}
      \big)    
      \, \defneq
      $}
    }
    {
    \pi_0\,
    \mathrm{Maps}\big(
      X
      ,\,
      \mathscr{A}
    \big)
    }
    \ar[
      rr,
      "{
        (\eta^{\mathbb{Q}})_\ast
      }"
    ]
    \ar[
      rrrr,
      rounded corners,
      to path={
            ([yshift=+00pt]\tikztostart.north)
         -- ([yshift=+5pt]\tikztostart.north)
         -- node[
               yshift=5pt
         ]{
           \scalebox{.7}{
             \color{darkgreen}
             \bf
             character map
           }
         }
            ([yshift=+18pt]\tikztotarget.north)
         -- ([yshift=+00pt]\tikztotarget.north)
      }
    ]
    &&
    \pi_0\,
    \mathrm{Maps}\big(
      X
      ,\,
      L^{\mathbb{Q}}\mathscr{A}
    \big)    
    \ar[
      rr
    ]
    &&
    H^1_{\mathrm{dR}}\big(
      X
      ;\,
      \mathfrak{l}\mathscr{A}
    \big)
    \\
    \scalebox{.7}{
      \color{darkblue}
      \bf
      \def\arraystretch{.9}
      \begin{tabular}{c}
        non-abelian 
        \\
        cohomology
      \end{tabular}
    }
    \ar[
      rr,
      phantom,
      "{
        \scalebox{.7}{
          \color{darkgreen}
          \bf
          \def\arraystretch{.9}
          \begin{tabular}{c}
            rationalization of
            \\
            classifying space
          \end{tabular}
        }
      }"
    ]
    &&
    \scalebox{.7}{
      \color{darkblue}
      \bf
      \def\arraystretch{.9}
      \begin{tabular}{c}
        non-abelian
        \\
        rational cohomology
      \end{tabular}
    }
    \ar[
      rr,
      phantom,
      "{
        \scalebox{.7}{
          \color{darkgreen}
          \bf
          \def\arraystretch{.9}
          \begin{tabular}{c}
            non-abelian 
            \\
            de Rham theorem
          \end{tabular}
        }
      }"
    ]
    &&
    \scalebox{.7}{
      \color{darkblue}
      \bf
      \def\arraystretch{.9}
      \begin{tabular}{c}
        non-abelian
        \\
        de Rham cohomology
      \end{tabular}
    }
    \hspace{-1cm}
  \end{tikzcd}
  }
\end{equation}

Since generalized cohomology theories are typically hard to analyze, in particular non-abelian ones, this character map may be regarded as extracting the first non-trivial stage of more tractable invariants. For instance, the character of a non-abelian class is the first obstruction to a trivialization of that class.
\footnote{
In the mentioned application to physics, the flux densities of a higher gauge field are sourced by charges that appear as classes in non-abelian de Rham cohomology on the right, and the completion of the higher gauge theory by a flux-quantization law means to lift these charges through the character map to classes in a chosen non-abelian cohomology theory on the left.}

\smallskip

It is fairly straightforward to generalize the non-abelian character \eqref{CharacterInIntroduction} to twisted non-abelian cohomology \eqref{TwistedCohomology}, now using {\it relative} minimal Sullivan models.

\smallskip

The following \hyperlink{TableCharacters}{\it Table 1.} shows some examples of the resulting form of twisted non-abelian character maps that we have computed elsewhere before -- the first few examples are for general illustration and orientation, the last one is the one of concern here: Our goal here is to equivariantize it.

\smallskip

This identification of the character map on non-abelian cohomology with the passage of classifying spaces to their minimal dgc-algebraic models in rational homotopy theory yields a new perspective on both subjects:

\begin{itemize}[
  leftmargin=.5cm,
  topsep=-1pt,
  itemsep=-1pt
]
\item On the one hand, it becomes clear at once how to make sense of the twisted equivariant non-abelian character, namely by construction of equivariant relative Sullivan models using the theory of \cite{Tri82}\cite[\S 11]{Scull02}\cite{Scull08};

\item and conversely it provides a sudden wealth of motivation and applications of the latter (which arguably has led a niche existence in the literature). 
\end{itemize}

\newpage

\hspace{-8mm}
\adjustbox{
 scale=.65
}{
\hypertarget{TableCharacters}{}
\fbox{
  \hspace{.3cm}
  \raisebox{-10pt}{
    \xymatrix@R=-21pt{
    \overset{
      \mathclap{
      \raisebox{3pt}{
        \tiny
        \color{darkblue}
        \bf
        \begin{tabular}{c}
          local
          coefficient
          \\
          bundle
        \end{tabular}
      }
      }
    }{
    {\begin{array}{c}
      A \!\sslash\! G
      \\
      \downarrow
      \\
      B G
    \end{array}}
    }
    \ar@{}[r]|-{
      \mbox{
        \tiny
        \begin{tabular}{c}
          \cite{FSS23-Char}
          \\
          Def. 5.4
        \end{tabular}
      }
    }
    &
    \underset{
      \mathclap{
      \mbox{
        \tiny
        \color{darkblue}
        \bf
        \begin{tabular}{c}
          twisted
          \\
          non-abelian cohomology
        \end{tabular}
      }
      }
    }{
    H^{
      \overset{
        \mathclap{
        \raisebox{-3pt}{
          $
          \!\!\!\!
          \mathrlap{
          \rotatebox[origin=l]{33}{
            \tiny
            {\color{darkblue}
            \bf
            twist}
            in $H(X,BG)$
          }
          }$
        }
        }
      }{
        \tau
      }
    }
    \big(
      \overset{
        \mathclap{
        \raisebox{-3pt}{
          $
          \!\!\!\!
          \mathrlap{
          \rotatebox[origin=l]{30}{
            \tiny
            \color{darkblue}
            \bf
            spacetime manifold
          }
          }$
        }
        }
      }{
        X
      }
      ;
      \,
      \overset{
        \mathclap{
        \raisebox{-3pt}{
          $
          \!\!\!\!
          \mathrlap{
          \rotatebox[origin=l]{30}{
            \tiny
            \color{darkblue}
            \bf
            classifying space
          }
          }$
        }
        }
      }{
        A
      }
    \big)
    }
    \ar[rr]^-{
      \mathrm{ch}_{\scalebox{.6}{$A$}}
    }_-{
      \mbox{
        \tiny
        \color{greenii}
        \bf
        \begin{tabular}{c}
          twisted
          \\
          non-abelian character map
        \end{tabular}
      }
    }
    &&
    \underset{
      \mathclap{
      \raisebox{-3pt}{
        \tiny
        \color{darkblue}
        \bf
        \begin{tabular}{c}
          twisted
          \\
          non-abelian de Rham cohomology
        \end{tabular}
      }
      }
    }{
    H^{
      \overset{
        \mathclap{
        \raisebox{-3pt}{
          $
          \!\!\!\!
          \mathrlap{
          \rotatebox[origin=l]{31.4}{
            \tiny
            {\color{darkblue}
            \bf
            twist }
            in
            $H_{\mathrm{dR}}(X, \mathfrak{l}BG)$
          }
          }$
        }
        }
      }{
        \tau_{\mathrm{dR}}
      }
    }_{\mathrm{dR}}
    \big(
      \overset{
        \mathclap{
        \raisebox{-3pt}{
          $
          \!\!\!\!\!\!\!
          \mathrlap{
          \rotatebox[origin=l]{30}{
            \tiny
            \color{darkblue}
            \bf
            spacetime manifold
          }
          }$
        }
        }
      }{
        X
      }
      \!;
      \,
      \overset{
        \mathclap{
        \raisebox{-3pt}{
          $
          \!\!\!\!
          \mathrlap{
          \rotatebox[origin=l]{30}{
            \tiny
            \color{darkblue}
            \bf
            Whitehead $L_\infty$-algebra
          }
          }$
        }
        }
      }{
        \mathfrak{l}A
      }
    \big)
    }
    \\
    &
    \underset{
      \mathclap{
      \raisebox{-3pt}{
        \tiny
        \color{darkblue}
        \bf
        \begin{tabular}{c}
          cocycle in
          \\
          twisted $A$-cohomology
        \end{tabular}
      }
      }
    }{
      [c]_{\tau}
    }
    \ar@{}[rr]|-{ \longmapsto }
    &&
    \underset{
      \mathclap{
      \raisebox{-3pt}{
        \tiny
        \color{darkblue}
        \bf
        \begin{tabular}{c}
          flux densities
          \\
          satisfying Bianchi identities
        \end{tabular}
      }
      }
    }{
    \mathrm{ch}_{\scalebox{.6}{${A}$}}
    \big(
      [c]
    \big)
    }
    \\
    \phantom{\vert^{\vert^{\vert^{\vert^{\vert^{\vert^{\vert^{\vert^{\vert^{\vert^{\vert^{\vert^{\vert^{\vert^{\vert^{\vert^{\vert}}}}}}}}}}}}}}}}}
    \\
    \raisebox{-0pt}{$
    {\begin{array}{c}
      B^{n} \mathbb{Z}
      \\
      \downarrow
      \\
      \;\;\ast\;\;
    \end{array}}
    $}
    \ar@{}[r]|-{
      \mbox{
        \tiny
        \begin{tabular}{c}
          \cite{FSS23-Char}
          \\
          Ex. 4.9
        \end{tabular}
      }
    }
    &
    \underset{
      \raisebox{-3pt}{
        \tiny
        \color{darkblue}
        \bf
        \begin{tabular}{c}
          ordinary cohomology
        \end{tabular}
      }
    }{
      H^n(X;\, \mathbb{Z})
    }
    \ar[rr]^-{ \mathrm{dR} }_-{
      \mbox{
        \tiny
        \color{greenii}
        \bf
        \begin{tabular}{c}
          de Rham
          \\
          homomorphism
        \end{tabular}
      }
    }
    &&
    \left\{
      \!\!
      {\begin{array}{c}
        F_{n}
      \end{array}}
      \in
      \Omega^n_{\mathrm{dR}}(X)
      \,\left\vert\;\,
      {\begin{aligned}
        d\, F_{n} & = 0
      \end{aligned}}
      \right.
      \!\!\!\!\!\!\!
    \right\}_{\!\!\!/ \sim}
    \\
\phantom{\vert^{\vert^{\vert^{\vert^{\vert^{\vert^{\vert^{\vert^{\vert^{\vert^{\vert^{\vert^{\vert^{\vert^{\vert^{\vert^{\vert}}}}}}}}}}}}}}}}}
    \\
    \raisebox{-0pt}{$
    {\begin{array}{c}
      B \mathrm{U}(n)
      \\
      \downarrow
      \\
      \;\;\ast\;\;
    \end{array}}
    $}
    \ar@{}[r]|-{
      \mbox{
        \tiny
        \begin{tabular}{c}
          \cite{FSS23-Char}
          \\
          Thm. 4.26
        \end{tabular}
      }
    }
    &
    \underset{
      \raisebox{-3pt}{
        \tiny
        \color{darkblue}
        \bf
        \def\arraystretch{.9}
        \begin{tabular}{c}
          ordinary
          \\
          non-abelian cohomology
        \end{tabular}
      }
    }{
      H^1\big(X;\, \mathrm{U}(n) \big)
    }
    \ar[rr]^-{ \mathrm{cw} }_-{
      \mbox{
        \tiny
        \color{greenii}
        \bf
        \begin{tabular}{c}
          Chern-Weil
          \\
          homomorphism
        \end{tabular}
      }
    }
    &&
    \left\{
      \!\!
      {\begin{array}{c}
        \vdots,
        \\
        c_{2}(A),
        \\
        c_{1}(A)
      \end{array}}
      \in
      \Omega^{2\bullet}_{\mathrm{dR}}(X)
      \,\left\vert\;\,
      {\begin{aligned}
        \vdots
        \\
        d\, c_2(A) & = 0
        \\[-3pt]
        d\, c_1(A) & = 0
      \end{aligned}}
      \right.
      \!\!\!\!\!\!\!
    \right\}_{\!\!\!\big/ \sim}
    \\
\phantom{\vert^{\vert^{\vert^{\vert^{\vert^{\vert^{\vert^{\vert^{\vert^{\vert^{\vert^{\vert^{\vert^{\vert^{\vert^{\vert^{\vert}}}}}}}}}}}}}}}}}
    \\
      \!\!\!\!\!\!\!\!\!\!  \scalebox{.85}{$
    {\begin{array}{c}
      \mathclap{
      \big(
        \mathbb{Z} \!\times\! B \mathrm{U}
      \big)
      \!\sslash\! B \mathrm{U}(1)
      }
      \\
      \downarrow
      \\
      B^2 \mathrm{U}(1)
    \end{array}}
    $}
    \ar@{}[r]|-{\;
      \raisebox{-2pt}{
        \tiny
        \begin{tabular}{c}
          \cite{FSS23-Char}
          \\
          Prop. 5.5
        \end{tabular}
      }
    }
    &
    \underset{
      \raisebox{-3pt}{
        \tiny
        \color{darkblue}
        \bf
        \begin{tabular}{c}
          twisted
          \\
          complex K-theory
        \end{tabular}
      }
    }{
      \mathrm{KU}^\tau(X)
    }
    \ar[rr]^-{ \mathrm{ch}^\tau }_-{
      \mbox{
        \tiny
        \color{greenii}
        \bf
        \begin{tabular}{c}
          twisted
          \\
          Chern character
        \end{tabular}
      }
    }
    &&
    \left\{
      \!\!\!
      {\begin{array}{c}
        F_{2\bullet},
        \\
        H_3\;
      \end{array}}
      \in
      \Omega^\bullet_{\mathrm{dR}}(X)
      \,\left\vert\;\,
      {\begin{aligned}
        d\, F_{2\bullet + 2} & = H_3 \wedge F_{2\bullet}
        \\[-3pt]
        d\, H_3\;\;\;\;\; & = 0
      \end{aligned}}
      \right.
    \right\}_{\!\!\!\big/ \sim}
    \\
\phantom{\vert^{\vert^{\vert^{\vert^{\vert^{\vert^{\vert^{\vert^{\vert^{\vert^{\vert^{\vert^{\vert^{\vert^{\vert^{\vert^{\vert}}}}}}}}}}}}}}}}}
    \\
    \scalebox{0.85}{$
    {\begin{array}{c}
      \mathclap{
      S^4
      \!\sslash\! B \widehat{\mathrm{Sp}(2)}
      }
      \\
      \downarrow
      \\
      B \widehat{\mathrm{Sp}(2)}
    \end{array}}
    $}
    \ar@{}[r]|-{
      \mbox{
        \tiny
        \begin{tabular}{c}
          \cite{FSS23-Char}
          \\
          Ex. 5.23a
        \end{tabular}
      }
    }
    &
    \underset{
      \raisebox{-3pt}{
        \tiny
        \color{darkblue}
        \bf
        \begin{tabular}{c}
          J-twisted
          \\
          4-Cohomotopy
        \end{tabular}
      }
    }{
      \pi^\tau(X)
    }
    \ar[rr]^-{ \mathrm{ch}_{\pi}^\tau }_-{
      \mbox{
        \tiny
        \color{greenii}
        \bf
        \begin{tabular}{c}
          twisted
          \\
          FSS-character
        \end{tabular}
      }
    }
    &&
    \left\{
      \!\!\!
      {\begin{array}{c}
        2G_7,
        \\
        \phantom{2}G_4\;
      \end{array}}
      \in
      \Omega^\bullet_{\mathrm{dR}}(X)
      \,\left\vert\;\,
      {\begin{aligned}
        d\, 2 G_7 & = - G_4 \wedge G_4 + \big(\tfrac{1}{4}p_1(\omega)\big)^2
        \\[-3pt]
        d\, \phantom{2} G_4 & = 0
      \end{aligned}}
      \right.
    \right\}_{\!\!\!\big/ \sim}
    \\
\phantom{\vert^{\vert^{\vert^{\vert^{\vert^{\vert^{\vert^{\vert^{\vert^{\vert^{\vert^{\vert^{\vert^{\vert^{\vert^{\vert^{\vert}}}}}}}}}}}}}}}}}
    \\
    \scalebox{.85}{$
    {\begin{array}{c}
      \mathclap{
      \mathbb{C}P^3
      \!\sslash\! B \widehat{\mathrm{Sp}(2)}
      }
      \\
      \downarrow
      \\
      B \widehat{\mathrm{Sp}(2)}
    \end{array}}
    $}
    \ar@{}[r]|-{
      \mbox{
        \tiny
        \begin{tabular}{c}
          \cite{FSS23-Char}
          \\
          Ex. 5.23b
        \end{tabular}
      }
    }
    &
    \underset{
      \raisebox{-3pt}{
        \tiny
        \color{darkblue}
        \bf
        \begin{tabular}{c}
          twistorial
          \\
          Cohomotopy
        \end{tabular}
      }
    }{
      \mathcal{T}^\tau(X)
    }
    \ar[rr]^-{ \mathrm{ch}_{\mathcal{T}}^\tau }_-{
      \mbox{
        \tiny
        \color{greenii}
        \bf
        \begin{tabular}{c}
          twisted
          \\
          FSS-character
        \end{tabular}
      }
    }
    &&
    \left\{
      \!\!\!
      {\begin{array}{c}
        \phantom{2}H_3
        \\
        \phantom{2}F_2
        \\
        2G_7,
        \\
        \phantom{2}G_4\;
      \end{array}}
      \in
      \Omega^\bullet_{\mathrm{dR}}(X)
      \,\left\vert\;\,
      {\begin{aligned}
        d\, \phantom{2} H_3 & = G_4 - \tfrac{1}{4}p_1(\omega) - F_2 \wedge F_2
        \\[-3pt]
        d\, \phantom{2} F_2 & = 0
        \\[-3pt]
        d\, 2 G_7 & = - G_4 \wedge G_4 + \big(\tfrac{1}{4}p_1(\omega)\big)^2
        \\[-3pt]
        d\, \phantom{2} G_4 & = 0
      \end{aligned}}
      \right.
    \right\}_{\!\!\!\big/ \sim}
  }
  }
}
}

\smallskip

{ 
  \footnotesize
  \noindent
  {\bf Table 1  -- Character maps.}
   \footnotesize The generalized {\it character maps} that we are concerned with here (on twisted non-abelian generalized cohomology \cite{FSS23-Char}, here to be further {\it equivariantly} enhanced) are the universal approximations of generalized cohomology by {\it rational} cohomology, which here we take to be $\mathbb{R}$-rational over smooth manifolds and hence represented by differential forms in a {\it de Rham complex} $\Omega^\bullet_{\mathrm{dR}}(-)$ with de Rham differential ``$d$'' (cf. \cite{BottTu82}), specifically by {\it non-abelian de Rham cohomology}, see \cite[\S 33]{FSS23-Char} and \S\ref{GEquivariantHomotopyTheory}.

  The table above indicates that special cases of the generalized character are celebrated classical constructions such as the de Rham map from integral to de Rham cohomology \cite[Ex. 7.1]{FSS23-Char}, the Chern-Weil homomorphism from ordinary non-abelian cohomology to characteristic forms \cite[\S 8]{FSS23-Char} and the (twisted) Chern character on (twisted) topological K-theory \cite[Ex. 7.2, Prop. 10.1]{FSS23-Char}. Their unified understanding via dg-algebraic rational homotopy theory of their classifying spaces
  shows how to construct novel non-classical character maps analogously, notably on flavors of Cohomotopy theory \cite[\S 12]{FSS23-Char}, indicated at the bottom of the table.

  In the present article we generalize the construction of this generalized character map further to {\it equivariant} (twisted non-abelian generalized) cohomology (for the case of equivariantly simply connected classifying spaces, for simplicity), with special attention to the equivariantization of the last two examples above.
}

\medskip
\noindent

\newpage

\noindent
{\bf Main result.}
The main result presented below is the general construction of the character map on twisted equivariant non-abelian cohomology 
\footnote{
  More specifically, here we develop the equivariant non-abelian character for the case of equivariant classifying spaces that are equivariantly simply-connected (namely, fixed locus-wise). If one drops this assumption, then the discussion becomes much more involved, as one needs to rationalize the fixed locus-wise covering spaces while retaining the respective actions of the fundamental groups by Deck transformations over each fixed locus --- all this on top of the action of the equivariance group $G$ and of the twisting group $\mathscr{G}$.
}
which culminates in \S\ref{TheEquivariantTwistedNonAbelianCharacterMap}.

The simplest applications (in numbers of cells) are the cases of twisted and ``twistorial'' equivariant Cohomotopy, whose equivariant classifying spaces are spheres and projective spaces. Among these, the possibly simplest (but already quite non-trivial) example is $\mathbb{Z}_2$-equivariant twistorial Cohomotopy in degree 7. 
This is our running example along which we develop and illustrate all the ingredients of the construction. The analysis of this example culminates in Rem. \ref{ProofOfMainTheorem} below, with a proof the following statement:

\vspace{0cm}

\begin{theorem}
  \label{FluxQuantizationInEquivariantTwistorialCohomotopy}
 {\bf (i)} The character map {\rm (Def. \ref{TwistedEquivariantNonabelianCharacterMap})}
  in $\Grefl$-equivariant
  twistorial Cohomotopy {\rm (Def. \ref{EquivariantTwistorialCohomotopyTheory})},
  on $\Grefl$-orbifolds {\rm (Def. \ref{GOrbifolds})}
  with $\SpLR$-structure
  $\tau$ and -connection $\omega$
  {\rm (Ex. \ref{TangentialDeRhamTwistsOnGOrbifoldsWithTStructure})},
  is of the form 
  shown in Table 2 on the following page.

\noindent {\bf (ii)} Moreover, a necessary condition for differential forms to be in the image of this character map is their (shifted) integrality, as follows:
\begin{equation}
  \label{IntegralityConditions}
  \def\arraystretch{1.5}
  \begin{array}{l}
  \big[
    \widetilde G_4
  \big]
    \;:=\;
  \big[
    G_4 + \tfrac{1}{4}p_1(\omega)
  \big]
  \;
  \in
  \;
  \xymatrix@C=12pt{
    H^4\big(X;\, \mathbb{Z} \big)
    \ar[r]
    &
    H^4\big(X;\, \mathbb{R} \big)\,,
  }
  \\
  \big[
    F_2
  \big]
  \;
  \;
  \in
  \;
  \xymatrix@C=12pt{
    H^2\big(X;\, \mathbb{Z} \big)
    \ar[r]
    &
    H^2\big(X;\, \mathbb{R} \big)\,.
  }
  \end{array}
\end{equation}
\end{theorem}
\begin{proof}
This derivation occupies the bulk of the article; it is wrapped up below in Rem. \ref{ProofOfMainTheorem}.
\end{proof}

Here this analysis serves to showcase the rich structure reflected in character maps on twisted equivariant non-abelian cohomology. At the same time, this example has a rather curious application to physics \cite{SS24-AbAnyonsOnSeifert}, following \cite{FSS20c}\cite{SS24-Flux}, which we briefly indicate in the closing \S\ref{ApplicationToFLuxQuantization}.

\newpage

\hspace{-8mm}
\adjustbox{
  scale=.71,
  fbox
}{
$
\!\!
\hspace{-3mm}
  \xymatrix@C=2pt@R=2pt{
    \overset{
      \mathclap{
      \raisebox{3pt}{
        \tiny
        \color{darkblue}
        \bf
        \begin{tabular}{c}
          equivariant
          \\
          Local coefficient
          \\
          bundle
        \end{tabular}
      }
      }
    }{
    \begin{array}{c}
      \mathscr{A} \!\sslash\! \mathscr{G}
      \\
      \downarrow
      \\
      B \mathscr{G}
    \end{array}
    }
    &\hspace{-.9cm}:&
    \underset{
      \mathclap{
      \mbox{
        \tiny
        \color{darkblue}
        \bf
        \begin{tabular}{c}
          equivariant twisted
          \\
          non-abelian cohomology
        \end{tabular}
      }
      }
    }{
    H^{
      \overset{
        \mathclap{
        \raisebox{-3pt}{
          $
          \!\!\!\!
          \mathrlap{
          \rotatebox[origin=l]{33}{
            \tiny
            {\color{darkblue}
            \bf
            twist} in $H\big(\mathcal{X}\!;\, B \mathscr{G}\big)$
          }
          }$
        }
        }
      }{
        \tau
      }
    }
    \big(
      \overset{
        \mathclap{
        \raisebox{-3pt}{
          $
          \!\!\!\!
          \mathrlap{
          \rotatebox[origin=l]{30}{
            \tiny
            \color{darkblue}
            \bf
            spacetime $G$-orbifold
          }
          }$
        }
        }
      }{
        \mathcal{X}
      }
      \!;
      \,
      \overset{
        \mathclap{
        \raisebox{-3pt}{
          $
          \!\!\!\!
          \mathrlap{
          \rotatebox[origin=l]{30}{
            \tiny
            \color{darkblue}
            \bf
            classifying $G$-space
          }
          }$
        }
        }
      }{
        \mathscr{A}
      }
    \big)
    }
    \ar[rrrr]^-{
      \mathrm{ch}_{\scalebox{.6}{$\mathscr{A}$}}
      (
        \mathcal{X}
      )
    }_-{
      \mbox{
        \tiny
        \color{greenii}
        \bf
        \begin{tabular}{c}
          equivariant twisted
          \\
          non-abelian character map
        \end{tabular}
      }
    }
    &&{\phantom{AAA}}&&
    \underset{
      \mathclap{
      \raisebox{-3pt}{
        \tiny
        \color{darkblue}
        \bf
        \begin{tabular}{c}
          equivariant twisted
          \\
          non-abelian de Rham cohomology
        \end{tabular}
      }
      }
    }{
    H^{
      \overset{
        \mathclap{
        \raisebox{-3pt}{
          $
          \!\!\!\!
          \mathrlap{
          \rotatebox[origin=l]{33}{
            \tiny
            {\color{darkblue}
            \bf
            twist} in
            $H_{\mathrm{dR}}\big(\mathcal{X}\!;\, \mathfrak{l}B \mathscr{G}\big)$
          }
          }$
        }
        }
      }{
        \tau_{\mathrm{dR}}
      }
    }_{\mathrm{dR}}
    \big(
      \overset{
        \mathclap{
        \raisebox{-3pt}{
          $
          \!\!\!\!
          \mathrlap{
          \rotatebox[origin=l]{30}{
            \tiny
            \color{darkblue}
            \bf
            spacetime $G$-orbifold
          }
          }$
        }
        }
      }{
        \mathcal{X}
      }
      \!;
      \,
      \overset{
        \mathclap{
        \raisebox{-3pt}{
          $
          \!\!\!\!
          \mathrlap{
          \rotatebox[origin=l]{30}{
            \tiny
            \color{darkblue}
            \bf
            Whitehead $G$-$L_\infty$-algebra
          }
          }$
        }
        }
      }{
        \mathfrak{l}\mathscr{A}
      }
    \big)
    }
    \\
    \raisebox{15pt}
    {$\begin{array}{c}
      \raisebox{1pt}{\rm\textesh}
      \big(\!\!
        \orbisingular
        \big(
          \overset{
        \mathclap{
        \raisebox{-3pt}{
          $
          \!\!\!\!
          \mathrlap{
          \rotatebox[origin=l]{33}{
            \tiny
            \color{darkblue}
            \bf
            twistor space
          }
          }$
        }
        }
          }{
            \mathbb{C}P^3
          }
            \!\sslash\!
          \overset{
        \mathclap{
        \raisebox{-3pt}{
          $
          \!\!\!\!
          \mathrlap{
          \rotatebox[origin=l]{30}{
            \tiny
            \color{darkblue}
            \bf
            $\Grefl$-equivariant
          }
          }$
        }
        }
          }{
            \Grefl
          }
        )
      \big)
      \!\sslash\!
      \overset{
        \mathclap{
        \raisebox{-3pt}{
          $
          \!\!\!\!
          \mathrlap{
          \rotatebox[origin=l]{30}{
            \tiny
            \color{darkblue}
            \bf
            $\SpLR$-parametrized
          }
          }$
        }
        }
      }{
        \SpLR\;
      }
      \\
      \downarrow
      \mathrlap{
      \mbox{
       \tiny
       \rm
       (Ex. \ref{GhetEquivariantParametrizedTwistorSpace})
      }
      }
      \\
      B \SpLR
    \end{array}$}
    &\hspace{-.9cm}:&
    \underset{
      \mathclap{
      \raisebox{-3pt}{
        \tiny
        \color{darkblue}
        \bf
        \begin{tabular}{c}
          $\Grefl$-equivariant
          twistorial Cohomotopy
        \end{tabular}
      }
      }
    }{
    \mathcal{T}^{
      \overset{
        \mathclap{
        \raisebox{-3pt}{
          $
          \!\!\!\!
          \mathrlap{
          \rotatebox[origin=l]{33}{
            \tiny
            \color{darkblue}
            \bf
            \hspace{-.1cm}
            tangential twist
          }
          }$
        }
        }
      }{
        \tau
      }
    }_{\scalebox{.7}{$\Grefl$}}
    \big(\!
      \overset{
        \mathclap{
          \!\!\!\!\!\!\!\!\!
          \mathrlap{
          \rotatebox[origin=c]{30}{
            \tiny
            \color{darkblue}
            \bf
            \hspace{-.44cm}
            \begin{tabular}{c}
              spacetime orbifold
              \\
              with $A_1$-singularity
            \end{tabular}
          }
        }
        }
      }{
      \orbisingular
      (
        X
        \!\sslash\!
        \Grefl
      )
      }
    \big)
    }
    \ar[rrrr]^-{
      \mbox{
        \tiny
        \color{greenii}
        \bf
        \begin{tabular}{c}
          equivariant
          \\
          twistorial
          \\
          character
        \end{tabular}
      }
    }_-{
      \;
      \mathrm{ch}_{\mathcal{T}}
      \;
    }
    \ar[dd]|-{
      \mathclap{
        \mbox{
          \tiny
          \color{greenii}
          \bf
          \begin{tabular}{c}
            $\mathclap{\phantom{\vert^{\vert}}}$
            push-forward along
            \\
            $\SpLR$-parametrized
            \\
            twistor fibration
            $\mathclap{\phantom{\vert_{\vert}}}$
          \end{tabular}
        }
      }
    }
    &&&&
    \left\{
      \!\!\!
      \overset{
        \raisebox{3pt}{
          \tiny
          \color{orangeii}
          \bf
          fluxes
        }
      }
      {\begin{array}{c}
        \mathclap{\phantom{\vert^{\vert^{\vert^{\vert^{\vert^{\vert}}}}}}}
        \\
        \phantom{2} H_3,
        \\
        \phantom{2} F_2,
        \\
        2 G_7,
        \\
        \phantom{2} \widetilde G_4
        \\
        \rotatebox[origin=l]{-90}{
        $
          \in
          \,
          \Omega_{\mathrm{dR}}^\bullet(X)
        $}
      \end{array}}
    \,\left\vert \!
    {\begin{array}{c}
      \overset{
       \raisebox{3pt}{
          \tiny
          \color{orangeii}
          \bf
          twisted Bianchi identities
        }
      }
      {\begin{aligned}
        d\, \phantom{2} H_3
          & =
          \widetilde G_4 - \tfrac{1}{2} p_1(\omega) -  F_2 \wedge F_2
        \\
        d\; \phantom{2} F_2\, & = 0
        \\
        d\, 2G_7 &
          = - \widetilde G_4 \wedge
          \big(
            \widetilde G_4 -
            \tfrac{1}{2}p_1(\omega)
          \big)
        \\
        d\, \phantom{2} \widetilde G_4 & = 0,
        \\
        \\
        d H_3\vert_{X^{\scalebox{.5}{$\Grefl$}}}
        & =
        - \tfrac{1}{2}p_1
        \big(
          \omega\vert_{X^{\scalebox{.5}{$\Grefl$}}}
        \big)
        - F_2 \wedge F_2\vert_{X^{\scalebox{.5}{$\Grefl$}}}
      \\
      d F_2 \vert_{X^{\scalebox{.5}{$\Grefl$}}} & = 0
      \\
      G_7\vert_{X^{\scalebox{.5}{$\Grefl$}}}
      &
      = 0
      \\
      \widetilde G_4\vert_{X^{\scalebox{.5}{$\Grefl$}}}
      &
      = 0
    \end{aligned}}
    \end{array}}
    \!\!\!
    \rotatebox[origin=c]{-90}{
      \tiny
      \color{orangeii}
      \bf
      $\phantom{AA}$
      bulk
      $\phantom{AAAAAAAAAAAAAAAAA}$
      fixed locus
    }
    \!\!\!
    \right.
    \;\;
    \right\}_{\!\!\!\!\big/ \sim}
    \ar@{->>}@<-96pt>[dd]^-{
      \scalebox{.6}{$
        \begin{array}{c}
          \widetilde G_4
          \\
          \mapsdown
          \\
          G_4 + \tfrac{1}{4}p_1(\omega)
        \end{array}
      $}
    }
    \\
    {\phantom{
      {
      {A \atop A}
      \atop
      {A \atop A}
      }
      \atop
      {
      {A \atop A}
      \atop
      {A \atop A}
      }
    }}
    \\
    \raisebox{15pt}
    {$\begin{array}{c}
      \raisebox{1pt}{\rm\textesh}
      \overset{
        \mathclap{
        \raisebox{-3pt}{
          $
          \!\!\!\!
          \mathrlap{
          \rotatebox[origin=l]{33}{
            \tiny
            \color{darkblue}
            \bf
            4-sphere
          }
          }$
        }
        }
          }{
            S^4
          }
      \!\sslash\!
      \overset{
        \mathclap{
        \raisebox{-3pt}{
          $
          \!\!\!\!
          \mathrlap{
          \rotatebox[origin=l]{30}{
            \tiny
            \color{darkblue}
            \bf
            $\SpLR$-parametrized
          }
          }$
        }
        }
      }{
        \SpLR
      }
      \\
      \downarrow
      \\
      B \SpLR
    \end{array}$}
    &\hspace{-.69cm}:&
    \underset{
      \mathclap{
      \raisebox{-3pt}{
        \tiny
        \color{darkblue}
        \bf
        \begin{tabular}{c}
          J-twisted Cohomotopy
        \end{tabular}
      }
      }
    }{
    \pi^{
      \overset{
        \mathclap{
        \raisebox{-3pt}{
          $
          \!\!\!\!
          \mathrlap{
          \rotatebox[origin=l]{33}{
            \tiny
            \color{darkblue}
            \bf
            \hspace{-.03cm}
            J-twist
          }
          }$
        }
        }
      }{
        \tau
      }
    }_{\scalebox{.7}{$\Grefl$}}
    (
      \overset{
        \mathclap{
          \!\!\!\!\!\!\!\!\!
          \mathrlap{
          \rotatebox[origin=c]{30}{
            \tiny
            \color{darkblue}
            \bf
            \hspace{-.2cm}
            \begin{tabular}{c}
              spacetime
            \end{tabular}
          }
        }
        }
      }{
        X
      }
    )
    }
    \ar[rrrr]_-{
      \;
      \mathrm{ch}_{\pi}
      \;
    }^-{
      \mbox{
        \tiny
        \color{greenii}
        \bf
        \begin{tabular}{c}
          twisted
          \\
          cohomotopical
          \\
          character
        \end{tabular}
      }
    }
    &&&&
    \left\{
      \!\!\!
      \left.
      \raisebox{-6pt}{$
      \overset{
        \raisebox{3pt}{
          \tiny
          \color{orangeii}
          \bf
          fluxes
        }
      }
      {\begin{array}{c}
        \mathclap{\phantom{\vert^{\vert^{\vert}}}}
        2 G_7,
        \\
        \phantom{2} G_4
      \end{array}}
      $}
    \,\right\vert \!
    {\begin{array}{c}
      \overset{
       \raisebox{3pt}{
          \tiny
          \color{orangeii}
          \bf
          twisted Bianchi identities
        }
      }
      {\begin{aligned}
        d\, 2G_7 &
          = - G_4 \wedge G_4 + \big( \tfrac{1}{4}p_1(\omega)\big)^2
          \;\;\;\;\;\;\;\;
        \\
        d\, \phantom{2} G_4 & = 0,
    \end{aligned}}
    \end{array}}
    \!\!\!
    \;\;
    \right\}_{\!\!\!\!\big/ \sim}
  }
  \!\!\!\!\!
$
}

\smallskip

{\footnotesize 

\noindent
{\bf Table 2 -- Theorem \ref{FluxQuantizationInEquivariantTwistorialCohomotopy}. } 
Shown summarized is the result of our running example of the image of the character map on the nonabelian twisted equivariant cohomology theory classified by the equivariant twistor fibration, according to Thm. \ref{FluxQuantizationInEquivariantTwistorialCohomotopy}.

At the heart of the proof of Theorem
\ref{FluxQuantizationInEquivariantTwistorialCohomotopy} is the computation
(Prop. \ref{Z2EquivariantRelativeMinimalModelOfSpin3ParametrizedTwistorSpace} below)
of the equivariant relative minimal model
(\cite[\S 5]{Tri82}\cite[\S 11]{Scull02}\cite[\S 4]{Scull08},
recalled as Def. \ref{MinimalEquivariantdgcAlgebras} below)
of the $\ZTwo$-equivariant $\SpLR$-parametrized twistor fibration
in equivariant rational homotopy theory.
}

\medskip

\noindent {\bf The equivariant twistor fibration.}
The \emph{twistor fibration} $t_{\mathbb{H}}$
(\cite[\S III.1]{Atiyah79}\cite{Bryant82}, see \cite[\S 2]{FSS20c})
is the map from $\mathbb{C}P^3$ (``twistor space'')
to $\mathbb{H}P^1 \simeq S^4$ which sends complex
lines to the right quaternionic lines that they span:
\vspace{-1mm}
\begin{equation}
  \label{TwistorFibration}
  \hspace{-6mm}
  \raisebox{38pt}{
  \xymatrix@R=6pt@C=22pt{
    \;\;
    S^2
    \;\;
    \ar@{}[r]|-{ \simeq }
    \ar[drr]_-{
      \mathrm{fib}(t_{\mathbb{H}})
    }
    &
    \;
    \mathbb{H}^\times / \mathbb{C}^\times
    \!\!
    \ar[drr]
    \\
    &&
    \;
    \mathbb{C}P^3
    \;
    \ar[dd]^-{ t_{\mathbb{H}} }_-{
        \tiny
        \color{darkblue}
        \bf
        \begin{tabular}{c}
     \bf     twistor
          \\
      \bf    fibration
        \end{tabular}
      }
    \ar@{}[r]|-{\simeq}
    &
    \big(
      \mathbb{C}^4 \setminus \{0\}
    \big)/ \mathbb{C}^\times
    \ar[dd]
    &
\hspace{-8mm} 
   \ni \big\{
      v \cdot z \,\left\vert\, z \in \mathbb{C}^\times\right. \!\!
    \big\}
    \\
    \\
    &&
    \mathbb{H}P^1
    \ar@{}[r]|-{\simeq}
    &
    \big(
      \mathbb{H}^2 \setminus \{0\}
    \big)/ \mathbb{H}^\times
    &
    \hspace{-8mm} 
  \ni  \big\{
      v \cdot q \,\left\vert\, q \in \mathbb{H}^\times\right. \!\!
    \big\}
  }
  }
\end{equation}

  \noindent
The fiber of the twistor fibration is hence
$\mathbb{H}^\times/ \mathbb{C}^\times \,\simeq\, \mathbb{C}P^1
\,\simeq\, S^2$.

\noindent {\bf (i)} There is the evident action of
$\mathrm{Sp}(2)$,
on both $\mathbb{C}P^3$ and $\mathbb{H}P^1$,
by left multiplication of homogeneous representatives
with unitary quaternion $2 \times 2$
matrices \eqref{QuaternionUnitaryGroup}:
\vspace{-1mm} 
\begin{equation}
  \label{Sp2ActionOnTwistorSpace}
  \xymatrix@C=3pt@R=-4pt{
    \mathrm{Sp}(2)
    \ar@{}[r]|-{\times}
    &
    \mathbb{C}P^3
    \ar[rr]
    &{\phantom{AAAA}}&
  \;  \mathbb{C}P^3
    \mathrlap{\,,}
    \\
    (A \!\!\! \ar@{}[r]|-{,}
    & \!\!\!\! [v])
    \ar@{}[rr]|-{\longmapsto}
    &&
    [A \cdot v]
  }
  \phantom{AAAA}
  \xymatrix@C=3pt@R=-3pt{
    \mathrm{Sp}(2)
    \ar@{}[r]|-{\times}
    &
    \mathbb{H}P^1
    \ar[rr]
    &{\phantom{AAAA}}&
   \; \mathbb{H}P^1 \;,
    \\
    (A \!\!\! \ar@{}[r]|-{,}
    & \!\!\!\! [v])
    \ar@{}[rr]|-{\longmapsto}
    &&
    [A \cdot v]
  }
\end{equation}
and the twistor fibration (being given by quotienting on the right)
is manifestly equivariant under this left action.

\noindent {\bf (ii)} Consider the following subgroups:

\vspace{-.4cm}
\def\arraystretch{1.4}
\begin{align}
  \label{Ghet}
  \Grefl
  & :=\;
  \left\{
     1
     \,:=\,
     \big(
       \begin{smallmatrix}
         1 & 0
         \\
         0 & 1
       \end{smallmatrix}
     \big)
   \,,\,
    \sigma
    \,:=\,
   \big(
     \begin{smallmatrix}
       0 & 1
       \\
       1 & 0
      \end{smallmatrix}
     \big)
   \right\}
  \;\;\subset\;
  \mathrm{Sp}(2)
  \,,
  \\
  \sigma
  & :\;
  [
    z_1 : z_2 : z_3 : z_4
  ]
  \;\;\mapsto\;\;
  [
    z_3 : z_4 : z_1 : z_2
  ]\;,
  \\
  \label{TheSp1Subgroup}
  \SpLR
  & :=
  \;
  \left\{
  q \cdot
  \,:=\,
   \big(
     \left.
     \begin{smallmatrix}
       \raisebox{1.6pt}{\scalebox{.7}{$q$}}
       &
       0
       \\
       0
       &
       \raisebox{1.6pt}{\scalebox{.7}{$q$}}
      \end{smallmatrix}
     \big)
     \,\right\vert\,
     q \,\in\, S(\mathbb{H})
   \right\}
  \;\;\,\subset\;
  \mathrm{Sp}(2)\;.
\end{align}
\vspace{-.4cm}

\noindent
Since these manifestly commute with each other,
the homotopy quotient $\mathbb{C}P^3 \!\sslash\! \SpLR $
of twistor space \eqref{TwistorFibration}
by $\SpLR$ still admits the structure of
a $G$-space (as in \cite[\S 8]{tomDieck79}\cite{May96}\cite{Blu17})
for $G = \Grefl$, fibered over $B \SpLR$
(see Ex. \ref{GhetEquivariantParametrizedTwistorSpace} below for details).

\medskip

\noindent {\bf The equivariant minimal relative dgc-algebra model of twistor space.}
Our Prop. \ref{Z2EquivariantRelativeMinimalModelOfSpin3ParametrizedTwistorSpace}
gives its equivariant minimal model:

\vspace{-.4cm}
  \begin{equation}
    \label{MinimalModelOfSpinParametrizedTwistorSpaceModZ2InIntroduction}
    \hspace{-1mm}
    \adjustbox{scale=0.9}{$
    \underset{
      \mathclap{
      \raisebox{0pt}{
        \tiny
        \color{darkblue}
        \bf
        {\begin{tabular}{c}
          twistor space
          \\
          homotopy-quotiented 
          \\
          by $\SpLR$ with
          \\
         residual $\Grefl$-action
        \end{tabular}}
      }
      }
    }{
    \xymatrix@C=-5pt{
      \mathbb{C}P^3
      \ar@(ul,ur)^-{ \Grefl }
      &
      \!\sslash \SpLR
    }
    }
      :
  \raisebox{40pt}{
  \xymatrix@C=38pt@R=3.5em{
    \ZTwo/1
    \ar@{->}[d]_-{
      \rotatebox[origin=c]{90}{
        \scalebox{1}{
          \tiny
          \bf
          $\ZTwo$-orbit category
        }
      }
    }
    \ar@(ul,ur)|-{\;\ZTwo}
    \ar@{|->}[r]^-{\mbox{
        \color{greenii}
        \bf \tiny
        bulk
      }}
    &
    \mathbb{R}
    \big[
      \tfrac{1}{4}p_1
    \big]
    \!\!
    \left[
      \!\!
      \def\arraystretch{1}
      {\begin{array}{c}
        h_3,
        \\
        f_2
        \\
        \omega_7,
        \\
        \widetilde \omega_4
      \end{array}}
      \!\!
    \right]
    \!\big/\!
    \left(
      {\begin{aligned}
        d\, h_3  & = \widetilde \omega_4 - \tfrac{1}{2}p_1 -  f_2 \wedge f_2
        \\[-4pt]
        d\, f_2 & = 0
        \\[-4pt]
        d\, \omega_7
          & =
          -
          \widetilde \omega_4 \wedge
          \big(
            \widetilde \omega_4 - \tfrac{1}{2}p_1
          \big)
        \\[-4pt]
        d\, \widetilde \omega_4 & = 0
      \end{aligned}}
    \right)
    \ar@<-54pt>@{->>}[d]^-{
      \mbox{
        \tiny
        \color{darkblue}
        \bf
        {\begin{tabular}{c}
          minimal $\Grefl$-equivariant model
          \\
          relative to $B \SpLR$
        \end{tabular}}
      }
    }
    \\
    \ZTwo/\ZTwo
    \ar@{|->}[r]^-{\mbox{
        \color{greenii}
        \bf
      \tiny  singularity
      }}
    &
    \mathbb{R}
    \big[
      \tfrac{1}{4}p_1
    \big]
    \!\!
    \left[
      \!\!
      \def\arraystretch{1}
      {\begin{array}{c}
        h_3,
        \\
        f_2
      \end{array}}
      \!\!
    \right]
    \!\big/\!
    \left(
      {\begin{aligned}
        d\, h_3  & = \phantom{\omega_4}\; - \tfrac{1}{2}p_1 - f_2 \wedge f_2
        \\[-3pt]
        d\, f_2 & = 0
      \end{aligned}}
    \right)
  }
  }
  $}
\end{equation}

\noindent
normalized (as in \cite{FSS19b}\cite{FSS19c}\cite{FSS20c}) such that:

\noindent
{\bf (a)} all closed generators shown are
rational images of {\it integral} and {\it integrally in-divisble}
cohomology classes;

\noindent
{\bf (b)}
$\omega \,:=\, \widetilde \omega - \tfrac{1}{4}p_1$
is fiberwise the volume form on
$\mathbb{H}P^1 \simeq S^4$, and $f_2$ is fiberwise the volume form on
$\mathbb{C}P^1 \simeq S^2$.

\vspace{.2cm}

As a non-trivial example of a (relative) minimal model in rational
equivariant homotopy theory, this may be of interest in its
own right. Such examples computed in the literature are rare
(we have not come across any). Here we are concerned
with a most curious aspect of this novel example:
Under substituting
the algebra generators in \eqref{MinimalModelOfSpinParametrizedTwistorSpaceModZ2InIntroduction}
with differential forms on a $\Grefl$-orbifold
(essentially the non-abelian character map,
Def. \ref{TwistedEquivariantNonabelianCharacterMap}),
the relations in \eqref{MinimalModelOfSpinParametrizedTwistorSpaceModZ2InIntroduction}
are those expected for flux densities in ``M-theory'', as briefly explained in \S\ref{ApplicationToFLuxQuantization}:
$$
\hspace{2mm} 
  \underset{
    \mathclap{
    \raisebox{-3pt}{
      \tiny
      \color{darkblue}
      \bf
      \begin{tabular}{c}
        dgc-algebra generators of
        \\[-2pt]
        equiv. relative minimal model
      \end{tabular}
    }
    }
  }{
  \big(
    \tfrac{1}{4}p_1
    ,
    \;
    \widetilde \omega_4
    ,
    \;
    \omega_7
    ,
    \;
    f_2
    ,
    \;
    h_3
  \big)
  }
\quad
    \longleftrightarrow
  \quad
  \big(\,
    \underset{
      \mathclap{
      \raisebox{-3pt}{
        \tiny
        \color{darkblue}
        \bf
        \begin{tabular}{c}
          Pontrjagin form
          \\[-2pt]
          (gravit. flux density)
        \end{tabular}
      }
      }
    }{
      \tfrac{1}{4}p_1(\omega)
    }
    \;
    ,
    \;\;
    \overset{
      \mathclap{
      \raisebox{3pt}{
        \tiny
        \color{darkblue}
        \bf
        \begin{tabular}{c}
          shifted
          C-field 
          \\[-2pt]
          flux density
        \end{tabular}
      }
      }
    }{
      G_4 + \tfrac{1}{4}p_1(\omega)
    }
    \;,
    \;\;
    \underset{
      \mathclap{
      \raisebox{-3pt}{
        \tiny
        \color{darkblue}
        \bf
        \begin{tabular}{c}
          dual
          C-field 
          \\[-2pt]
          flux density
        \end{tabular}
      }
      }
    }{
      2G_7
    }
    \;,
    \;\;
    \overset{
      \mathclap{
      \raisebox{+3pt}{
        \tiny
        \color{darkblue}
        \bf
        \begin{tabular}{c}
          gauge 
          \\[-2pt] 
          flux
        \end{tabular}
      }
      }
    }{
      F_2
    }
    \;,
    \;\;
    \underset{
      \mathclap{
      \raisebox{-3pt}{
        \tiny
        \color{darkblue}
        \bf
        \begin{tabular}{c}
          B-field 
          \\[-2pt]
          flux
        \end{tabular}
      }
      }
    }{
      H_3
    }
  \big).
$$

\newpage

\noindent
{\bf Outline.}

\noindent
In \cref{GEquivariantHomotopyTheory} we introduce
equivariant non-abelian cohomology theory
(in equivariant generalization of \cite[\S 2]{FSS23-Char})
and the example of
equivariant twistorial Cohomotopy theory
$\mathcal{T}^\tau_{\Grefl}(-)$
(Def. \ref{EquivariantTwistorialCohomotopyTheory}).

\vspace{.1cm}

\noindent
In \cref{EquivariantNonAbelianDeRhamCohomology} we introduce
equivariant non-abelian de Rham cohomology theory
and the equivariant non-abelian character map
(in equivariant generalization of \cite[\S 3-5]{FSS23-Char})
and compute the $\Grefl$-equivariant relative minimal model of
$\SpLR$-parametrized twistor space (Prop. \ref{Z2EquivariantRelativeMinimalModelOfSpin3ParametrizedTwistorSpace}).

\vspace{.1cm}

\noindent
In \S\ref{ApplicationToFLuxQuantization} we briefly indicate the application and impact of our result on the problem of flux-quantization of higher gauge fields arising in super-gravity.

\medskip
\medskip

\noindent
{\bf Notation.} For various types of symmetry groups
and their quotients, we use the following notation:

\vspace{1mm} 
\hspace{-8.5mm}
\scalebox{.64}{
\def\tabcolsep{4pt}
\begin{tabular}{|c|l|l||l|l|l|}
\hline
$
\mathclap{\phantom{\vert^{\vert^{\vert}}}}
T
\mathclap{\phantom{\vert_{\vert_{\vert}}}}
$
&
Compact Borel equivariance group
&
\multirow{4}{*}{
  Def. \ref{GActionsOnTopSp}
}
&
$
\mathclap{\phantom{\vert^{\vert^{\vert}}}}
\raisebox{1pt}{\textesh}\, \phantom{\orbisingular} \big(  X \!\sslash\! T \big) \mathclap{\phantom{\vert_{\vert_{\vert}}}}
$
&
Borel equivariant homotopy type
&
Ex. \ref{HomotopyTypeOfBorelConstruction}
\\
\cline{1-2}
\cline{4-6}
\multirow{2}{*}{
$
  \mathclap{\phantom{\vert^{\vert^{\vert}}}}
  G
  \mathclap{\phantom{\vert_{\vert_{\vert}}}}
$
}
&
\multirow{2}{*}{
Finite proper equivariance group
}
& &
$
 \mathclap{\phantom{\vert^{\vert^{\vert}}}}
 \phantom{\raisebox{1pt}{\textesh}}\,\orbisingular \big( X \!\sslash\! G \big)
 \mathclap{\phantom{\vert_{\vert_{\vert}}}}
$
&
Orbifold
&
Ex. \ref{SystemsOfFixedPointSpaces}
\\
&
&
&
$
  \mathclap{\phantom{\vert^{\vert^{\vert}}}}
  \raisebox{1pt}{\textesh} \, \orbisingular \big( X \!\sslash\! G \big)
  \mathclap{\phantom{\vert^{\vert^{\vert}}}}
$
&
Proper equivariant homotopy type
&
Def. \ref{EquivariantShape}
\\
\cline{1-2}
\cline{4-6}
$
\mathclap{\phantom{\vert^{\vert^{\vert}}}}
\! T \! \times \!G \!
\mathclap{\phantom{\vert_{\vert_{\vert}}}}
$
&
\!\!\! Borel \& proper equivariance group \!
 &
&
$
  \mathclap{\phantom{\vert^{\vert^{\vert^{\vert^{\vert}}}}}}
  \raisebox{1pt}{\textesh}\,
  \big(\!\orbisingular ( X \!\sslash\! G ) \!\big) \!\!\sslash\! T
  \mathclap{\phantom{\vert_{\vert_{\vert_{\vert_{\vert}}}}}}
$
&
  \hspace{-11pt}
  \def\arraystretch{.8}
  \begin{tabular}{l}
    Proper $G$-equivariant \&
    \\
    Borel $T$-equivariant
    homotopy type \hspace{-3mm} 
  \end{tabular}
&
Ex. \ref{EquivariantParametrizedHomotopyTypes}
\\
\hline
$
  \mathclap{\phantom{\vert^{\vert^{\vert}}}}
  \mathcal{G}
  \mathclap{\phantom{\vert_{\vert_{\vert}}}}
$
&
Simplicial group/$\infty$-group
&
Not. \ref{ModelCategoryOfSmplGrp}
&
$
  \mathclap{\phantom{\vert^{\vert^{\vert}}}}
  \phantom{\raisebox{1pt}{\textesh}\, \orbisingular}
  \mathcal{A} \!\sslash\! \mathcal{G}
  \mathclap{\phantom{\vert_{\vert_{\vert}}}}
$
&
  Homotopy quotient
&
\!\!\!  Prop. \ref{InfinityActionsEquivalentToFibrationsOverClassifyingSpace} 
\\
\hline
$
  \mathclap{\phantom{\vert^{\vert^{\vert}}}}
  \mathscr{G}
  \mathclap{\phantom{\vert_{\vert_{\vert}}}}
$
&
  $G$-equivariant $\infty$-group
&
  Rem. \ref{EquivariantInfinityAction}
&
$
  \mathclap{\phantom{\vert^{\vert^{\vert}}}}
  \phantom{\raisebox{1pt}{\textesh}\, \orbisingular}
  \mathscr{A} \!\sslash\! \mathscr{G}
  \mathclap{\phantom{\vert_{\vert_{\vert}}}}
$
&
$G$-equivariant homotopy quotient
&
\eqref{EquivariantUniversalFibrationOverEquivariantClassifyingSpace}
\\
\hline
\end{tabular}
}

\vspace{.2cm}

\noindent Our notation for equivariant homotopy theory
follows \cite{SS20b}.
The symbol ``$\orbisingular$'' refers to
proper equivariant
objects (``orbi-singular objects''),
parametrized over the orbit category (Def. \ref{OrdinaryOrbitCategory})
of the equivariance group \eqref{EquivarianceGroup}:

\smallskip 
\hspace{-.8cm}
\scalebox{.75}{
\def\tabcolsep{4pt}
\begin{tabular}{|c|c|c|p{.12\linewidth}|}
  \hline
  \multicolumn{2}{|c|}{
    $\mathclap{\phantom{\vert^{\vert^{\vert}}}}$
    \bf Symbol
    $\mathclap{\phantom{\vert_{\vert_{\vert}}}}$
  }
  &
  {
    $\mathclap{\phantom{\vert^{\vert^{\vert}}}}$
       \bf Meaning
    $\mathclap{\phantom{\vert_{\vert_{\vert}}}}$
  }
  &
  {
    $\mathclap{\phantom{\vert^{\vert^{\vert}}}}$
    \bf Details
    $\mathclap{\phantom{\vert_{\vert_{\vert}}}}$
  }
  \\
  \hline
  \hline
  $\GActionsOnTopSp$
  &
  \begin{tabular}{c}
    $G$-actions on
    \\[-5pt]
    topological spaces
  \end{tabular}
  &
  \begin{minipage}[left]{8.3cm}
    $\mathclap{\phantom{\vert^{\vert^{\vert}}}}$Category of topological spaces
    equipped with continuous action of the equivariance group $G$ $\mathclap{\phantom{\vert_{\vert_{\vert}}}}$
  \end{minipage}
  &
  Def. \ref{GActionsOnTopSp}
  \\
  \hline
  $G\mathrm{Orb}$
  &
  \begin{tabular}{c}
    $G$-orbits
  \end{tabular}
  &
  \begin{minipage}[left]{8.3cm}
    $\mathclap{\phantom{\vert^{\vert^{\vert}}}}$Category of canonical orbits
    $G/H$ of the equivariance group,
    with equivariant maps between them 
  \end{minipage}
  &
  Def. \ref{OrdinaryOrbitCategory}
  \\
  \hline
  $\EquivariantSSet$
  &
  \begin{tabular}{c}
    $G$-equivariant
    \\[-5pt]
    simplicial sets
  \end{tabular}
  &
  \begin{minipage}[left]{8.3cm}
    $\mathclap{\phantom{\vert^{\vert^{\vert}}}}$Category of contra-variant functors from
    $G$-orbits to simplicial sets $\mathclap{\phantom{\vert_{\vert_{\vert}}}}$
  \end{minipage}
  &
  Def. \ref{EquivariantSSetCategory}
  \\
  \hline
  $\EquivariantDualVectorSpaces$
  &
  \begin{tabular}{c}
    $G$-equivariant
    \\[-5pt]
    dual vector spaces
  \end{tabular}
  &
  \begin{minipage}[left]{8.3cm}
    $\mathclap{\phantom{\vert^{\vert^{\vert}}}}$Category of co-variant functors from
    $G$-orbits to vector spaces $\mathclap{\phantom{\vert_{\vert_{\vert}}}}$
  \end{minipage}
  &
  Def. \ref{EquivariantVectorSpaces}
  \\
  \hline
  $\EquivariantdgcAlgebras$
  &
  \begin{tabular}{c}
    $G$-equivariant
    \\[-5pt]
    dgc-algebras
  \end{tabular}
  &
  \begin{minipage}[left]{8.3cm}
    $\mathclap{\phantom{\vert^{\vert^{\vert}}}}$Category of co-variant functors from
    $G$-orbits to connective differential graded-commutative algebras $\mathclap{\phantom{\vert_{\vert_{\vert}}}}$
  \end{minipage}
  &
  Def. \ref{EquivariantdgcAlgebras}
  \\
  \hline
  $\EquivariantHomotopyTypes$
  &
  \begin{tabular}{c}
    $G$-equivariant
    \\[-5pt]
    homotopy types
  \end{tabular}
  &
  \begin{minipage}[left]{8.3cm}
    $\mathclap{\phantom{\vert^{\vert^{\vert}}}}$Homotopy category of
    projective model category of contra-variant functors from
    $G$-orbits to simplicial sets$\mathclap{\phantom{\vert_{\vert_{\vert}}}}$
  \end{minipage}
  &
  Def. \ref{EquivariantHomotopyTypes}
  \\
  \hline
\end{tabular}
}

\section{Equivariant non-abelian cohomology}
 \label{GEquivariantHomotopyTheory}

In \cref{HomotopyTheoryOfInfinitGroupActions}
we recall basics of $\infty$-groups and their
$\infty$-actions
and establish some technical Lemmas.

\noindent
In \cref{EquivariantHomotopyTheory} we recall basics of
proper equivariant homotopy theory and introduce our
running Example \ref{GhetEquivariantParametrizedTwistorSpace}.

\noindent
in \cref{EquivariantNonAbelianCohomologyTheories} we
 introduce equivariant non-abelian cohomology theory.

\noindent
in \cref{TwistedEquivariantNonAbelianCohomologyTheories}
 we introduce twisted equivariant non-abelian cohomology theory.

\medskip

Throughout, we illustrate all concepts in the running example
of the $\Grefl$-equiva{-}riant and $\SpLR$-parametrized twistor fibration
(Example \ref{GhetEquivariantParametrizedTwistorSpace}),
the induced equivariant twistorial Cohomotopy theory
(Def. \ref{EquivariantTwistorialCohomotopyTheory})
and its character image in equivariant de Rham cohomology
(Example \ref{FlatEquivariantTwistorialDifferentialForms}).
We highlight that here both flavors of equivariance are involved:
\begin{center}
{\small
\def\arraystretch{1.2}
\begin{tabular}{|c||c|c|}
  \hline
  &
  \hspace{-.1cm}
  {\bf Borel equivariance}
  \hspace{-.1cm}
  &
  \hspace{-.1cm}
  {\bf Proper equivariance}
  \hspace{-.1cm}
  \\
  \hline
  \hline
  \def\arraystretch{.9}
  \begin{tabular}{c}
    $\mathclap{\phantom{\vert^{\vert}}}$
    Equivariance
    \\
    group (\cref{GEquivariantHomotopyTheory})
    $\mathclap{\phantom{\vert_{\vert}}}$
  \end{tabular}
  &
  $T  \,=\, \mathrm{Sp}(1)$ & $G \,=\, \ZTwo$
  \\
  \hline
  \hspace{-.34cm}
  \def\arraystretch{.9}
  \begin{tabular}{c}
    $\mathclap{\phantom{\vert^{\vert}}}$
    Equivariant
    \\
    dR-cohomology (\cref{EquivariantNonAbelianDeRhamCohomology})
    $\mathclap{\phantom{\vert_{\vert}}}$
  \end{tabular}
  \hspace{-.34cm}
  &
  \def\arraystretch{.9}
  \begin{tabular}{c}
    Borel-Weil-Cartan
    \\
    model
  \end{tabular}
  &
  \def\arraystretch{.9}
  \begin{tabular}{c}
    Bredon-type
    \\
    theory
  \end{tabular}
  \\
  \hline
  \def\arraystretch{.9}
  \begin{tabular}{c}
    $\mathclap{\phantom{\vert^{\vert}}}$
    Physical
    \\
    effect
    (\S\ref{ApplicationToFLuxQuantization})
    $\mathclap{\phantom{\vert_{\vert}}}$
  \end{tabular}
  &
  \def\arraystretch{.9}
  \begin{tabular}{c}
    Flux quantization:
    \\
    shift of $G_4$ by $\scalebox{.7}{$\tfrac{1}{4}$}p_1$
  \end{tabular}
  &
  \def\arraystretch{.9}
  \begin{tabular}{c}
    orbifolding of
    \\
    M5-brane
  \end{tabular}
  \\
  \hline
\end{tabular}
}
\end{center}

\medskip

We make free use of basic concepts from
category theory and homotopy theory
(for joint introduction see \cite{Riehl14}\cite{Richter20}),
in particular of model category theory
(\cite{Quillen67},
review in \cite{Hovey99}\cite{Hirschhorn02}\cite[A.2]{Lurie09a}).
Relevant concepts and facts are recalled in
\cite[\S A]{FSS23-Char}.

\medskip
For $\mathcal{C}$ a category, and $X,\, A \, \in \, \mathcal{C}$
a pair of objects, we write
\begin{equation}
  \label{HomSets}
  \mathcal{C}(X,A) \;\in\; \mathrm{Sets}
\end{equation}
for its set of morphisms from $X$ to $A$. This assignment is, of course,
a contravariant functor in its first argument, to be denoted:
\begin{equation}
  \label{ContravariantHomFunctor}
  \mathcal{C}(-;\, A)
  \;:\;
  \xymatrix{
    \mathcal{C}^{\mathrm{op}}
    \ar[r]
    &
    \mathrm{Sets}
  }.
\end{equation}

\noindent
Elementary as it is, this is of profound interest
whenever $\mathcal{C}$ is the {\it homotopy category} of
a homotopy topos \cite{ToenVezzosi05} \cite{Lurie09a}\cite{Rezk10},
in which case the contravariant hom-functors
\eqref{ContravariantHomFunctor} are {\it non-abelian cohomology theories}
\cite{Toen02} \cite{dcct}\cite{SS20b}\cite{FSS23-Char}. These
subsume generalized and ordinary cohomology theories
(\cite[\S 2]{FSS23-Char}), as well as their equivariant
enhancements, which we consider below.

\newpage

\subsection{Homotopy theory of $\infty$-group actions}
\label{HomotopyTheoryOfInfinitGroupActions}

\noindent {\bf Plain homotopy theory.}

\vspace{-1mm} 
\begin{notation}[Classical homotopy category]
  \label{ClassicalHomotopyCategory}
  {\bf (i)} We write
  \begin{equation}
    \label{ClassicalModelStructureOnSSet}
    \mathrm{TopSp}_{\mathrm{Qu}}
    \,,
    \,
    \mathrm{SSet}_{\mathrm{Qu}}
    \;\in\;
    \mathrm{ModCat}
  \end{equation}
  \vspace{-.55cm}

  \noindent
  for the classical model category structures on topological spaces
  and on simplicial sets, respectively
  (\cite[\S II.3]{Quillen67}, review in \cite{Hi15}\cite{GoerssJardine99}).

  \noindent
  {\bf (ii)} The classical Quillen equivalence
  \vspace{-2mm} 
  \begin{equation}
    \label{QuillenEquivalenceBetweenTopSpAndSSet}
    \xymatrix{
      \mathrm{TopSp}_{\mathrm{Qu}}
      \ar@{<-}@<+6pt>[rr]^-{ \left\vert - \right\vert }
      \ar@<-6pt>[rr]_-{ \mathrm{Sing} }^-{ \simeq_{\mathrlap{\mathrm{Qu}}} }
      &&
      \mathrm{SSet}_{\mathrm{Qu}}
    }
  \end{equation}

  \noindent
  induces an equivalence between the corresponding homotopy
  categories, which we denote:
  \begin{equation}
    \label{TheClassicalHomotopyCategory}
    \xymatrix{
      \mathrm{SSet}
      \ar[rr]^-{ \gamma }_-{
        \mbox{
          \tiny
          \color{greenii}
          \bf
          localization
        }
      }
      &&
      \HomotopyTypes
    }
    \;:=\;
    \mathrm{Ho}
    \big(
      \mathrm{SSet}_{\mathrm{Qu}}
    \big).
  \end{equation}

\noindent   {\bf (iii)}
  We denote the localization functor from topological spaces
  to this classical homotopy category  by
  ``$\raisebox{1pt}{\textesh}$'': \footnote{
The ``esh''-symbol ``\raisebox{1pt}{\textesh}'' stands for \emph{shape}
\cite[3.4.5]{dcct}\cite[9.7]{Shulman15}\cite[\S 3.1.1]{SS20b},
following \cite{Borsuk75},
which for the well-behaved topological spaces of interest here is another term for
their \emph{homotopy type} \cite[7.1.6]{Lurie09a}\cite[4.6]{Wang17}.}
  \begin{equation}
    \label{TopologicalShapeAsLocalization}
    \xymatrix@R=12pt@C=100pt{
      \mathrm{TopSp}
      \ar[rr]^-{
        \mathllap{
          \mbox{
            \tiny
            \color{greenii}
            \bf
            shape
          }
        }
        \scalebox{.9}{\textesh}
      }_-{
        \mbox{
          \tiny
          \color{greenii}
          \bf
          localization at
          weak homotopy equivalences
        }
      }
      \ar[dr]_-{
        \mbox{
          \tiny
          \def\arraystretch{.9}
          \begin{tabular}{c}
            \color{greenii}
            \bf
            form singular
            \\
            \color{greenii}
            \bf
            simplicial set
            \\
            \eqref{QuillenEquivalenceBetweenTopSpAndSSet}
          \end{tabular}
        }
      }
      &&
      \HomotopyTypes
      \\
      &
      \mathrm{SSet}
      \ar[ur]_-{
        \gamma
        \mathrlap{
          \mbox{
          \tiny
          {
          \color{greenii}
          \bf
          localization
          }
          \eqref{TheClassicalHomotopyCategory}
          }
        }
      }
    }.
  \end{equation}

\end{notation}

\noindent {\bf Borel-equivariant homotopy theory.}
We recall basics of Borel-equivariant homotopy theory,
but in the generality of equivariance for
$\infty$-group actions (for the broader
picture see \cite{NSS12a}\cite[\S 2.2]{SS20b}).

\begin{notation}[Model category of simplicial groups]
  \label{ModelCategoryOfSmplGrp}

 {\bf (i)} We write
  \begin{equation}
    \label{CategoryOfSmplGrp}
    \mathrm{SmplGrp}
     \;:=\;
    \mathrm{Grp}
    (
      \mathrm{SSet}
    )
  \end{equation}
  \vspace{-.6cm}

  \noindent
  for the category of simplicial groups.

  \noindent {\bf (ii)} This
  becomes (\cite[\S II.3.7]{Quillen67}) a model category
  $$
    \mathrm{SmplGrp}_{\mathrm{proj}}
    \;\in\;
    \mathrm{ModCat}
  $$
  \vspace{-.55cm}

  \noindent
  by taking the
  weak equivalences and fibrations to be those of
  $\mathrm{SSet}_{\mathrm{Qu}}$
  (Notation \ref{ClassicalHomotopyCategory}).

  \noindent {\bf (iii)} We denote the homotopy category of this model structure by
  \begin{equation}
    \label{LocalizationOfCategoryOfSmplGrp}
    \xymatrix@C=25pt{
    \mathrm{SmplGrp}_{\mathrm{proj}}
    \ar[rrr]^-{ \gamma }_-{
      \mbox{
        \tiny
        \color{greenii}
        \bf
        \begin{tabular}{c}
          localization at
          \\[-3pt]
          weak homotopy equivalences
        \end{tabular}
      }
    }
    &&&
    \mathrm{Grp}_\infty
    \;:=\;
    \mathrm{Ho}
    \big(
      \mathrm{SmplGrp}_{\mathrm{proj}}
    \big)
    }
    \,
  \end{equation}
  
\vspace{-1mm} 
  \noindent
  and denote the generic object here by 
  \vspace{-1mm} 
  $$
    \mathcal{G}
    \;\in\;
    \xymatrix{
      \mathrm{SmplGrp}
      \ar[r]^-{ \gamma }
      &
      \mathrm{Grp}_\infty
    }.
  $$

\end{notation}
\begin{example}[Shapes of topological groups are $\infty$-groups]
  For $T \,\in \, \mathrm{TopGrp}$, its singular
  simplicial set
  \eqref{QuillenEquivalenceBetweenTopSpAndSSet}
  is canonically a simplicial group \eqref{CategoryOfSmplGrp}
  \begin{equation}
    \label{SingularSimplicialGroup}
    \mathrm{Sing}(T)
    \;\in\;
    \mathrm{SmplGrp}
    \,,
  \end{equation}

  \noindent
  and,
  since the weak equivalence of simplicial groups are those of
  the underlying simplicial sets, its image in the homotopy category
  is the shape $\raisebox{1pt}{\textesh}\, T$ \eqref{TopologicalShapeAsLocalization},
  now equipped with induced
  $\infty$-group structure (Notation \ref{ModelCategoryOfSmplGrp}):

  \begin{equation}
    \label{TopGrphapeAsLocalization}
    \xymatrix@R=13pt@C=90pt{
      \mathrm{TopGrp}
      \ar[rr]^-{
        \mathllap{
          \mbox{
            \tiny
            \color{greenii}
            \bf
            $\infty$-group shape
          }
        }
        \scalebox{.9}{\textesh}
      }_-{
        \mbox{
          \tiny
          \color{greenii}
          \bf
          localization at
          weak homotopy equivalences
        }
      }
      \ar[dr]_-{
        \mbox{
          \tiny
          \def\arraystretch{.9}
          \begin{tabular}{c}
            \color{greenii}
            \bf
            form singular
            \\
            \color{greenii}
            \bf
            simplicial group
            \\
            \eqref{QuillenEquivalenceBetweenTopSpAndSSet},
            \eqref{SingularSimplicialGroup}
          \end{tabular}
        }
      }
      &&
      \mathrm{Grp}_\infty
      \\
      &
      \mathrm{SmplGrp}
      \ar[ur]_-{ \!\!\!\!\!\!\!\!
        \gamma
        \mathrlap{
          \mbox{
          \tiny
          {
          \color{greenii}
          \bf
          localization
          }
          \eqref{LocalizationOfCategoryOfSmplGrp}
          }
        }
      }
    }.
  \end{equation}

\end{example}

\begin{notation}[Model category of reduced simplicial sets]
  \label{ModelCategoryOfReducedSSet}
 {\bf (i)}  We write

  $$
    \xymatrix{
      \ReducedSSet
      \;
      \ar@{^{(}->}[r]
      &
      \mathrm{SSet}
    }
  $$

  \noindent
  for the full subcategory on those $S \,\in\, \mathrm{SSet}$
  that have a single 0-cell, $S_0 \,=\, \ast$.

\noindent {\bf (ii)}  This becomes (\cite[\S V, Prop. 6.2]{GoerssJardine99}) a model category
  with weak equivalences and cofibrations those of
  $\mathrm{SSet}_{\mathrm{Qu}}$
  (Notation \ref{ClassicalHomotopyCategory}):
  $$
    \ReducedSSet_{\mathrm{GJ}}
    \;\in\;
    \mathrm{ModCat}
    \,.
  $$

 \noindent {\bf (iii)} Since reduced simplicial sets model those
  homotopy types \eqref{TheClassicalHomotopyCategory}
  which are {\it pointed and connected}
  (e.g. \cite[Prop. 3.16]{NSS12b}),
  we denote the corresponding
  homotopy category by
  \begin{equation}
    \xymatrix@C=15pt{
    \ReducedSSet_{\mathrm{GJ}}
    \ar[rr]^-{ \gamma }
    &&
    \mathrm{HomotopyTypes}^{\scalebox{.5}{$\ast$}}_{\scalebox{.5}{$\geq 1$}}
    \;:=\;
    \mathrm{Ho}
    \big(
      \ReducedSSet_{\mathrm{GJ}}
    \big).
    }
      \end{equation}

\end{notation}

\begin{prop}[Classifying space/loop space construction
{\cite[\S V, Prop. 6.3]{GoerssJardine99} \cite{Stevenson12}\cite[\S 3.5]{NSS12b}}]
  \label{ClassifyingSpaceLoopSpaceConstruction}
  There exists
  a Quillen equivalence
  between the model categories of
  reduced simplicial sets
  (Notation \ref{ModelCategoryOfReducedSSet})
  and that of
  simplicial groups (Notation \ref{ModelCategoryOfSmplGrp})
  \begin{equation}
    \label{SimplicialLoopGroupAdjunction}
    \xymatrix{
      \mathrm{SmplGrp}_{\mathrm{proj}}
   \;\;   \ar@{<-}@<+6pt>[rr]^{  }
      \ar@<-6pt>[rr]_{ \overline{W} }^-{ \simeq_{\mathrlap{\mathrm{Qu}}} }
      &&
  \;\;    \mathrm{ReducedSSet}
    }
  \end{equation}

  \noindent
  whose derived adjunction is given by forming homotopy types
  of {\it based loop spaces} and of {\it classifying spaces}:
  \begin{equation}
    \label{ClassifyingSpaceEquivalence}
    \xymatrix@C=3em{
      \mathllap{
        \mbox{
          \tiny
          \color{darkblue}
          \bf
          $\infty$-groups
        }
        \;\;\;
      }
      \mathrm{Grp}_\infty
      \ar@{<-}@<+5pt>[rrr]^-{
        \overset{
          \mathclap{
          \raisebox{3pt}{
            \tiny
            \color{greenii}
            \bf
            based loop $\infty$-group
          }
          }
        }{
          \Omega(-)
        }
      }
      \ar@<-5pt>[rrr]_-{
        \underset{
          \mathclap{
          \raisebox{-3pt}{
            \tiny
            \color{greenii}
            \bf
            classifying space
          }
          }
        }{
          B(-) \,\,:=\,\, \mathbb{R} \overline{W}(-)
        }
      }^-{\simeq}
      &&&
      \overset{
        \mbox{
          \tiny
          \color{darkblue}
          \bf
          \def\arraystretch{1}
          \begin{tabular}{c}
            pointed \& connected
            \\
            homotopy types
          \end{tabular}
        }      
      }{
        \mathrm{HomotopyTypes}^{\ast/}_{\geq 1}
      }
    }
  \end{equation}

\end{prop}

\begin{notation}[Homotopy theory of simplicial group actions]
  \label{ModelCategoryOfSimplicialGroupActions}
$\,$

\noindent
For $\mathcal{G} \,\in\, \mathrm{SmplGrp}$
(Notation \ref{ModelCategoryOfSmplGrp})

\noindent {\bf (i)} we write
$$
  \mathcal{G} \mathrm{Actions}
  \;:=\;
  \mathrm{SimplicialFunctors}
  \big(
    B \mathcal{G}, \, \mathrm{SSet}
  \big)
$$
\vspace{-.6cm}

\noindent
for the category of simplicial functors from
the simplicial groupoid with a single object and $\mathcal{G}$ as
its hom-object to the simplicial category of simplicial sets.

\noindent {\bf (ii)} This becomes a model category by taking the
weak equivalences and fibrations to be those of underling
simplicial sets (evaluating at the single vertex of $B \mathcal{G}$):
$$
  \mathcal{G} \mathrm{Actions}_{\mathrm{proj}}
  \;\in\;
  \mathrm{ModCat}
$$

\noindent
and we denote its homotopy category by:
$$
  \xymatrix{
  \mathcal{G} \mathrm{Actions}_{\mathrm{proj}}
  \ar[rr]^-{ \gamma }
  &&
  \mathrm{Ho}
  \big(
    \mathcal{G} \mathrm{Actions}_{\mathrm{proj}}
  \big)
  \;=:\;
  \mathcal{G} \mathrm{Actions}_{\infty}\;.
  }
$$

\end{notation}

\smallskip 
The following, Prop. \ref{InfinityActionsEquivalentToFibrationsOverClassifyingSpace},
is pivotal for the discussion of twisted non-abelian cohomology,
notably for the notion of equivariant local coefficient bundles
below in Def. \ref{EquivariantTwistedNonAbelianCohomology};
for more background and context, see
\cite[\S 4]{NSS12a}\cite[\S 2.2]{SS20b}\cite[Prop. 2.28]{FSS23-Char}.

\begin{prop}[$\infty$-Group actions equivalent to fibrations over classifying space {\cite[Prop. 2.3]{DDK80}\cite{Sharma15}}]
  \label{InfinityActionsEquivalentToFibrationsOverClassifyingSpace}
  $\,$

  \noindent
\noindent {\bf (i)}   For $\mathcal{G} \,\in\, \mathrm{SmplGrp}$ (Notation \ref{ModelCategoryOfSmplGrp}),
  the simplicial Borel construction
  (e.g. \cite[Prop. 3.37]{NSS12b})
  is the
  right adjoint of a Quillen equivalence
  \begin{equation}
    \label{SimplicialBorelConstructionEquivalence}
    \xymatrix@C=3em{
      \mathcal{G} \mathrm{Actions}_{\mathrm{proj}}
      \ar@{<-}@<+6pt>[rrr]^-{  }
      \ar@<-6pt>[rrr]_-{
        \underset{
          \mbox{
            \tiny
            \color{greenii}
            \bf
            \begin{tabular}{c}
              simplicial
              Borel construction
            \end{tabular}
          }
        }{
        \scalebox{.7}{$
          \mathcal{G} \acts \; X
          \;\mapsto\;
          \frac{X \times W \mathcal{G}}{\mathcal{G}}
        $}
        }
      }^-{ \simeq_{\mathrlap{\mathrm{Qu}}} }
      &&&
    \;\;  \mathrm{SSet}^{/\overline{W}\mathcal{G}}_{\mathrm{Qu}}
    }
  \end{equation}
  
\vspace{-1mm} 
  \noindent
  between the projective model structure on simplicial
  $\mathcal{G}$-actions (Notation \ref{ModelCategoryOfSimplicialGroupActions})
  and the slice model structure
  (\cite[\S 7.6.4]{Hirschhorn02}) of the classical model structure
  on simplicial sets \eqref{ClassicalModelStructureOnSSet}
  over $\overline{W}\mathcal{G}$ \eqref{SimplicialLoopGroupAdjunction}.
  
\noindent {\bf (ii)}  Its derived equivalence of homotopy categories
  \begin{equation}
    \label{HomotopyFiberHomotopyQuotientAdjunction}
    \hspace{-7mm}
    \xymatrix@C=35pt{
      {
        \mbox{
          \tiny
          \color{darkblue}
          \bf
          \def\arraystretch{1}
          {\begin{tabular}{c}
            $\infty$-actions of
            \\
            $\infty$-group $\mathcal{G}$
          \end{tabular}}
        }
        \!\!
      }
      \mathcal{G} \mathrm{Actions}_{\infty}
      \ar@{<-}@<+6pt>[rrr]^-{
        \overset{
          \raisebox{3pt}{
            \tiny
            \color{greenii}
            \bf
            homotopy fiber
          }
        }{
          \mathrm{hofib}_\ast(p)
          \;\mapsfrom\;
          (E \overset{p}{\to} B G)
        }
      }
      \ar@<-6pt>[rrr]_-{
      \underset{
         \mathclap{
          \raisebox{-3pt}{
            \tiny
            \color{greenii}
            \bf
            homotopy quotient
          }
          }
        }{
          \scalebox{.7}{$
            \mathcal{G} \acts \; A \;\;\mapsto\;\; A \!\sslash\! \mathcal{G}
          $}
        }
      }^-{ \simeq }
      &&&
      \!\!\!\!
      \overset{
        \mbox{
          \tiny
          \color{darkblue}
          \bf
          \def\arraystretch{.9}
          {\begin{tabular}{c}
            homotopy types fibered
            \\
            over classifying space $B \mathcal{G}$
          \end{tabular}}
        }      
      }{
      \mathrm{Ho}
      \Big(
        \mathrm{SSet}^{/\overline{W}\mathcal{G}}_{\mathrm{Qu}}
      \Big)
      }
    }
  \end{equation}

  \noindent
  is given
  in one direction by forming homotopy fibers of fibrations over $B G$
  and in the other by forming homotopy quotients of $\infty$-actions
  {\rm (\cite[Prop. 3.73]{NSS12b})}:

  \vspace{-.2cm}
  \begin{equation}
    \label{HomotopyFibrationsCorrespondingToInfinityActions}
    \overset{
      \mathclap{
      \raisebox{3pt}{
        \tiny
        \color{darkblue}
        \bf
        {\begin{tabular}{c}
          $\mathcal{G}$ $\infty$-action
          on $A$
        \end{tabular}}
      }
      }
    }{
      \mathcal{G} \acts \; A
    }
    \qquad
      \longleftrightarrow
   \qquad
    \raisebox{20pt}{
    \xymatrix@R=1.5em{
      A \ar[rr]^-{ \mathrm{hofib}_\ast(\rho_A) }
      &
      \ar@{}[d]|-{
        \mbox{
          \tiny
          \color{darkblue}
          \bf
          \def\arraystretch{1}
          {\begin{tabular}{c}
            $A$-fibration over
            \\
            $\mathcal{G}$-classifying space
          \end{tabular}}
        }
      }
      &
      A \!\sslash\! \mathcal{G}
      \ar[d]^-{ \rho_A }
      \\
      &&
      B \,\mathcal{G}
      \,.
    }
    }
  \end{equation}

\end{prop}

\begin{example}[Homotopy type of Borel construction]
  \label{HomotopyTypeOfBorelConstruction}
  $\,$

  \noindent
  For $T \,\in\, \mathrm{TopGrp}$
  and $T \acts \; X \,\in\, \TActionsOnTopSp$
  (Def. \ref{GActionsOnTopSp}),
  passage to singular simplicial sets \eqref{QuillenEquivalenceBetweenTopSpAndSSet}
  yields a simplicial action (Notation \ref{ModelCategoryOfSimplicialGroupActions}).
  The corresponding fibration (Prop. \ref{InfinityActionsEquivalentToFibrationsOverClassifyingSpace})
  is given by the topological shape \eqref{TopologicalShapeAsLocalization}
  of the Borel construction:
  $$
    \xymatrix@R=1.5em{
      \raisebox{1pt}{\textesh} \, X
      \ar[rr]^-{
        \mathrm{hofib}(\rho_X)
      }
      &&
      \raisebox{1pt}{\textesh}
      \,
      \left(
      \tfrac{
        X \times E T
      }{
        T
      }
      \right)
      \mathrlap{
        =:
        \;
        \raisebox{1pt}{\textesh}
        \,
        \big(
          X \!\sslash\! T
        \big).
      }
      \ar[d]^-{ \rho_X }
      \\
      &&
      \raisebox{1pt}{\textesh}\, B T
    }
  $$

\end{example}

\begin{lemma}[Pasting law {\cite[Lem 4.4.2.1]{Lurie09a}}]
  \label{PastingLaw}
  For $\mathcal{C}$ a model category, and given
  a pasting composite of two commuting squares

  \vspace{-.3cm}
  $$
    \xymatrix@R=1em@C=3em{
      A
      \ar[r]
      \ar[d]
      &
      B
      \ar[r]
      \ar[d]
      \ar@{}[dr]|-{
        \mbox{
          \tiny
          \rm
          (hpb)
        }
      }
      &
      C
      \ar[d]
      \\
      D \ar[r] & E \ar[r] & F
    }
  $$

  \noindent
  such that the right square is homotopy Cartesian,
  then the left square is homotopy Cartesian if and only if
  the total rectangle is.
\end{lemma}

\begin{lemma}[Homotopy fibers of homotopy-quotiented morphisms]
  \label{HomotopyFibersOfHomotopyQuotientedMorphisms}
  $\,$

  \noindent
  Let
  $\mathcal{G} \in \mathrm{Grp}_{\infty}$
  (Notation \ref{ModelCategoryOfSmplGrp})
  and
  $
      (A,\rho_A)
      \xrightarrow{ (f, \,\rho_f) }
      (A', \rho_{A'})
    $
    $
   \in
    \mathcal{G} \mathrm{Actions}^{\ast/}_{\infty}
  $
  a morphism of $\infty$-actions
  (Notation \ref{ModelCategoryOfSimplicialGroupActions})
  preserving
  an $\mathcal{G}$-fixed point
  $\mathrm{pt}: \ast \to A \overset{f}{\longrightarrow} A'$
  (see also \cite[Def. 2.97]{SS20b}). Then:

  \noindent
  {\bf (i)}
  The homotopy fiber of the homotopy-quotiented morphism
  $f \!\sslash\! \mathcal{G}$ \eqref{HomotopyFiberHomotopyQuotientAdjunction}
  coincides with the homotopy fiber of $f$
  \begin{equation}
    \label{TheHomotopyFiberOfHomotopyQuotientedMorphism}
    \mathrm{hofib}_\ast
    \big(
      f \!\sslash\! \mathcal{G}
    \big)
    \;\simeq\;
    \mathrm{hofib}_\ast( f )\;.
  \end{equation}

  \noindent
  {\bf (ii)}
  The homotopy fiber of $f$ is canonically
  equipped with an $\infty$-action by $\mathcal{G}$:  
  $$
    \big(
    \mathrm{hofib}_\ast
    (
      f
    ),
    \,
    \rho_h
    \big)
    \;\in\;
    \mathcal{G} \mathrm{Actions}_\infty \;.
  $$

  \noindent
  {\bf (iii)} The corresponding homotopy quotient is
  equivalent to the homotopy fiber of the homotopy-quotiented
  morphism
  parametrized over $B G$,
  namely the following homotopy pullback:
  \begin{equation}
    \label{SquareForBGParametrizedHomotopyFiber}
\hspace{1cm} 
    \xymatrix@R=1.8em@C=5em{
      \mathllap{
        \mathrm{hofib}_\ast
        (
          f
        )
        \sslash \mathcal{G}
        \;\;
        \simeq
        \;\;\,
      }
      \mathrm{hofib}_{B \, \mathcal{G}}
      \big(
        f \sslash \mathcal{G}
      \big)
      \ar[d]
      \ar[r]
      \ar@{}[dr]|-{
        \mbox{
          \tiny
          \rm
          (hpb)
        }
      }
      &
      A \!\sslash\! \mathcal{G}
      \ar[d]^-{ f \!\sslash\!  \, \mathcal{G} }
      \\
      B \, \mathcal{G}
      \ar[r]_-{ \mathrm{pt}' \sslash \mathcal{G} }
      &
      A' \!\sslash\! \mathcal{G}
      \,.
    }
  \end{equation}

\end{lemma}
\begin{proof}
Consider the following pasting diagrams:
\vspace{0mm}
\begin{equation}
  \label{TowardsParametrizedHomotopyFibersOfHomotopyQuotientedMorphisms}
  \def\arraystretch{1}
  \begin{array}{c}
  \raisebox{20pt}{
  \xymatrix@C=3em{
    \mathrm{hofib}_\ast
    \big(
      f \!\sslash\! \mathcal{G}
    \big)
    \ar[d]
    \ar[r]
    \ar@{}[dr]|-{
      \mbox{
        \tiny
        \rm
        (hpb)
      }
    }
    &
    \mathrm{hofib}_{\scalebox{.6}{$B \,\mathcal{G}$}}
    \big(
      f \!\sslash\! \mathcal{G}
    \big)
    \ar@{}[dr]|-{
      \mbox{
        \tiny
        \rm
        (hpb)
      }
    }
    \ar[r]
    \ar[d]
    &
    A \!\sslash\! \mathcal{G}
    \ar[d]_<<<<{
      \mathclap{\phantom{\vert^{\vert}}}
      f \sslash \mathcal{G}
      \mathclap{\phantom{\vert_{\vert}}}
    }
    \\
    \ast
    \ar[r]_-{
    }
    &
    B \mathcal{G}
    \ar[r]_-{
      \;\mathrm{pt}' \sslash \mathcal{G}\;
    }
    &
    A' \!\sslash\! \mathcal{G}
    \\
    & &
    {\phantom{B \, \mathcal{G}}}
  }
  }
  \\[-25pt]
  \simeq
  \;\;
  \raisebox{20pt}{
  \xymatrix@C=45pt{
    \mathrm{hofib}_\ast(f)
    \ar@{}[dr]|-{
      \mbox{
        \tiny
        \rm
        (hpb)
      }
    }
    \ar[d]
    \ar[r]
    &
    A
    \ar@{}[drr]|-{
      \mbox{
        \tiny
        \rm
        (hpb)
      }
    }
    \ar[rr]
    \ar[d]^<<<<<{
      \mathclap{\phantom{\vert^{\vert}}}
      f
      \mathclap{\phantom{\vert_{\vert}}}
    }
    &&
    A \!\sslash\! \mathcal{G}
    \ar[d]_<<<<{
      \mathclap{\phantom{\vert^{\vert}}}
      f \sslash \mathcal{G}
      \mathclap{\phantom{\vert_{\vert}}}
    }
    \ar@[gray]@/^1.5pc/[dd]^-{
      \mathclap{\phantom{\vert^{\vert}}}
      \color{gray} \rho_A
      \mathclap{\phantom{\vert_{\vert}}}
    }
    \\
    \ast
    \ar[r]_-{
      \;\mathrm{pt}'\;
    }
    &
    A'
    \ar[rr]
    \ar@[gray][d]
    \ar@{}[drr]|-{
     \mbox{
       \tiny
       \color{gray}
       \rm
       (pb)
     }
    }
    &&
    A' \!\sslash\! \mathcal{G}
    \ar@[gray][d]_<<<<{
      \mathclap{\phantom{\vert^{\vert}}}
      \color{gray} \rho_{A'}
      \mathclap{\phantom{\vert_{\vert}}}
    }
    \\
    &
    {\color{gray}
    \ast
    }
    \ar@[gray][rr]
    &&
    {\color{gray}
    B \, \mathcal{G}
    }
  }
  }
  \end{array}
\end{equation}

\noindent
With the right Cartesian square \eqref{SquareForBGParametrizedHomotopyFiber}
given, the pasting law (Lem. \ref{PastingLaw})
identifies the top left objects on both sides
as shown; in particular, the left square on the right
gives \eqref{HomotopyFiberOfHomotopyQuotientedMorphism}.
But, since the composite bottom morphism
is the same basepoint inclusion on both sides, this implies:
  \begin{equation}
    \label{HomotopyFiberOfHomotopyQuotientedMorphism}
    \mathrm{hofib}_{\ast}
    \big(
      f \!\sslash\! \mathcal{G}
    \big)
    \;\;
    \simeq
    \;\;
    \mathrm{hofib}_{\ast}
    (
      f
    )\;.
  \end{equation}

\noindent
Moreover,
the left Cartesian square
on the left of \eqref{TowardsParametrizedHomotopyFibersOfHomotopyQuotientedMorphisms}
exhibits,
by Prop. \ref{InfinityActionsEquivalentToFibrationsOverClassifyingSpace},
a $\mathcal{G}$-action on
$
    \mathrm{hofib}_{\ast}
    \big(
      f \!\sslash\! \mathcal{G}
    \big)
$
with homotopy quotient given by
\begin{equation}
  \label{FurtherTowardsParametrizedHomotopyFibersOfHomotopyQuotientedMorphisms}
    \mathrm{hofib}_{\ast}
    \big(
      f \!\sslash\! \mathcal{G}
    \big)
    \!\sslash\!
    \mathcal{G}
    \;\;
    \simeq
    \;\;
    \mathrm{hofib}_{B\, \mathcal{G}}
    \big(
      f \!\sslash\! \mathcal{G}
    \big).
\end{equation}

\vspace{0mm}
\noindent
The combination of the equivalences
\eqref{TheHomotopyFiberOfHomotopyQuotientedMorphism} and
\eqref{FurtherTowardsParametrizedHomotopyFibersOfHomotopyQuotientedMorphisms}
yields the claimed equivalence in \eqref{SquareForBGParametrizedHomotopyFiber}.
\hfill
\end{proof}

\subsection{Proper equivariant homotopy theory}
 \label{EquivariantHomotopyTheory}

We now recall relevant basics
of proper\footnote{
  Here by ``proper equivariance''
  we refer to the fine notion of equivariant homotopy/cohomology
  in the sense of Bredon, as opposed to the coarse notion
  in the sense of Borel. For in-depth conceptual discussion
  of this distinction see \cite{SS20b}.
  Besides the colloquial meaning of ``proper'',
  the action of our finite equivariance groups
  is necessarily
  {\it proper}
  in the technical sense of general topology
  (see Lemma \ref{FixedLociOfProperSmoothActionsAreSmoothManifolds} below),
  whence this terminology nicely matches that recently advocated in \cite{DHLPS19}.
  }
  equivariant homotopy theory
\cite[\S 8]{tomDieck79} \cite{May96}\cite{Blu17}
and introduce the examples of interest here.

\medskip

\noindent {\bf $G$-Actions.}

\begin{defn}[Group actions on topological spaces]
  \label{GActionsOnTopSp}

\noindent {\bf (i)} For a given compact topological group,
which serves as the symmetry group of {\it Borel equivariance}
in the following, generically to be denoted
\begin{equation}
  \label{CompactTopologicalGroup}
    \mathllap{
    \raisebox{2pt}{
      \tiny
      \color{darkblue}
      \bf
      Borel equivariance group
    }
    \;\;
    }
    T \;\in\; \mathrm{CompactTopGrp}
    \,,
\end{equation}

\noindent
we write
  \begin{equation}
    \label{TActionsOnTopSp}
    \TActionsOnTopSp
    \;\in\;
    \mathrm{Categories}
  \end{equation}
  \vspace{-.55cm}

  \noindent
  for the category whose objects
  are topological spaces $X$ equipped with a
  continuous $T$-action
  \begin{equation}
    \label{GroupAction}
    \begin{array}{l}
    T \acts \; X
    \;:\;
    \xymatrix@R=-3pt@C=3em{
      T \times X
      \ar[r]^-{
        \mbox{
          \tiny
          continuous
        }
      }
      &
      X
      \\
      (t\,,\,x)
      \ar@{}[r]|-{\longmapsto}
      &
      t \cdot x
    }
    \\
      \mbox{such that:}
      \;\;\;
      \underset{x \in X}{\forall}
      \,
      \mathrm{e}\cdot x=x
      \;\;\;
      \mbox{and}
      \;\;\;
      \underset{
        {x \in X}
        \atop
        {t_1, \, t_2 \in G}
      }{\forall}
      \big(
        t_1 \cdot
        (t_2 \cdot x)
      \big)
      \,=\,
      (t_1 \cdot t_2) \cdot x
    \end{array}
  \end{equation}

  \noindent
  and whose morphisms are $T$-equivariant continuous functions, which we
  denote as follows:

  \vspace{-1mm}
  \begin{equation}
    \label{EquivatiantContinuousFunction}
    \xymatrix{
      X_1
      \ar@(ul,ur)|-{\;T\;}
      \ar[rr]^-{f}
      &&
      X_2
      \ar@(ul,ur)|-{\;T\;}
    }
    \phantom{AAAA}
    \Leftrightarrow
    \phantom{AAAA}
    \raisebox{8pt}{$\underset{
      {x \in X}
      \atop
      {t \in T}
    }{\forall}  \;
    f ( t \cdot x)
    \,=\,
    t \cdot f(x)\;.
    $}
  \end{equation}

\noindent
{\bf (ii)}
Throughout,
our {\it proper} equivariance group is a finite group, to be denoted:

\begin{equation}
  \label{EquivarianceGroup}
    \mathllap{
      \raisebox{2pt}{
        \tiny
        \color{darkblue}
        \bf
        proper equivariance group
      }
      \;\;
    }
    G
    \;\in\;
    \mathrm{FiniteGroups}
    \,.
\end{equation}
\vspace{-.6cm}

\noindent
This finite group can be viewed as a topologically discrete
topological group and we have the corresponding category
\eqref{TActionsOnTopSp} of continuous actions: 
  \begin{equation}
    \label{CategoryOfGActionsOnTopSp}
    \GActionsOnTopSp
    \;\in\;
    \mathrm{Categories}\;.
  \end{equation}
  \vspace{-.55cm}

\noindent {\bf (iii)}
The full subcategory of the latter category
on those objects, where also the topological space being acted on
is discrete, is that of $G$-actions on sets:
\begin{equation}
  \label{GActionsOnSets}
  \xymatrix{
    \GActionsOnSets
    \;
    \ar@{^{(}->}[r]
    &
    \;
    \GActionsOnTopSp
  }.
\end{equation}

\noindent {\bf (iv)}
Regarding the direct product group of
the Borel equivariance group \eqref{CompactTopologicalGroup}
with the proper equivariance group \eqref{EquivarianceGroup}
as a compact topological group
$$
  {
    \raisebox{2pt}{
      \tiny
      \color{darkblue} \bf
      Borel \& proper equivariance group
    }
    \;\;
  }
  \;\;\;
  T \times G \;\in\;
  \mathrm{CompactTopGrp}\;,
$$

\noindent
we have the category of topological actions of this
product group. This contains the previous
categories, \eqref{TActionsOnTopSp}
and \eqref{CategoryOfGActionsOnTopSp},
as full subcategories (via equipping a space with trivial action)
\begin{equation}
  \label{CategoryOfTGActionsOnTopSp}
  {\small
  \xymatrix@C=18pt{
    \TActionsOnTopSp
    \;
    \ar@{^{(}->}[r]
    &
    \;
    \TGActionsOnTopSp
    \;
    \ar@{<-^{)}}[r]
    &
    \;\;
    \GActionsOnTopSp   \;.
  }
  }
\end{equation}

\end{defn}

\begin{example}[Representation spheres]
  \label{RepresentationSpheres}
  Let $V \,\in\, T\RepresentationsFin$ be a
  finite-dimensional linear representation
  of a compact topological group \eqref{CompactTopologicalGroup}.
  Then the one-point compactification of $V$
  (the topological sphere of the same dimension, e.g. \cite[p. 150]{Kelly55})
  inherits a topological $T$-action (Def. \ref{GActionsOnTopSp})
  via stereographic projection,
  denoted
  $$
    S^V
    \;\;
    \in
    \;
    \TActionsOnTopSp
  $$

  \noindent
  and called the {\it representation sphere} of $V$
  (e.g. \cite[\S 1.1.5]{Blu17}\cite[\S 3]{SS19a}).
\end{example}

\begin{defn}[Orbit category]
  \label{OrdinaryOrbitCategory}
  The \emph{category of $G$-orbits}
  or \emph{orbit category} of the equivariance group
  $G$ \eqref{EquivarianceGroup}
  $$
    G \mathrm{Orb}
    \longhookrightarrow
    \GActionsOnSets
    \;\in\;
    \mathrm{Categories}
  $$

  \noindent
  is (up to equivalence of categories)
  the full subcategory of discrete $G$-actions \eqref{GActionsOnSets}
  on the coset spaces $G/H$
  (which are discrete spaces,
  since $G$ is assumed to be finite)
  for all subgroup inclusions
  $H \overset{\iota}{\hookrightarrow} G$.
\end{defn}
\begin{example}[Explicit parameterization of morphisms of $G\mathrm{Orb}$]
  \label{MorphismSetsInGOrbitCategory}
  The hom-sets \eqref{HomSets} in the $G$-orbit category (Def. \ref{OrdinaryOrbitCategory}) from any $G/H_1$ to any $G/H_2$
  are in bijection with sets of conjugations, inside $G$, of $H_1$ into subgroups of $H_2$, modulo conjugations
  in $H_2$:
  \begin{equation}
    \label{AnalyzingHomsInAnOrbitCategory}
    \hspace{-5mm} 
    G \mathrm{Orb}\big(
      G/H_1, \,G/H_2
    \big)
    \simeq
    \frac{
      \big\{
        \phi : H_1 \overset{}{\hookrightarrow} H_2
        ,\; g \in G
        \;\vert\;
        \mathrm{Ad}_{g^{-1}} \circ \iota_1 = \iota_2 \circ \phi
      \big\}
    }
    {
      \big( \, (\phi,g) \sim ( \mathrm{Ad}_{h_2^{-1}} \circ \phi, g h_2) \,\vert\, h_2 \in H_2  \, \big)
    }
    \,.
  \end{equation}
  \vspace{-.1cm}

\noindent (Here ``$\mathrm{Ad}$'' denotes the adjoint action
of the group on itself, and
$\xymatrix@C=12pt{ H_i\; \ar@{^{(}->}[r]^{\iota_i} & G  }$
are the two subgroup inclusions.)

\end{example}

\begin{example}[Orbit category of $\ZTwo$]
  \label{OrbitCategoryOfZ2}
  The orbit category (Def. \ref{OrdinaryOrbitCategory})
  of the cyclic group $\ZTwo \;:=\; \{\mathrm{e}, \sigma \,\vert\, \sigma \circ \sigma = \mathrm{e}\}$
  is
  $$
    \ZTwo
    \mathrm{Orb}
    \;\;
    \simeq
    \;\;
       \left\{\!\!\!\!
       \raisebox{-4pt}{
    \xymatrix{
      \ar@(ul,ur)|-{\; \ZTwo \,}
      \ZTwo/1
      \ar[r]^-{ \exists ! }
      &
      \ZTwo/\ZTwo
      \ar@(ul,ur)|-{\, 1 \,}
    }
    }
    \!\! \right\}.
  $$

  \noindent
  Hence its hom-sets \eqref{HomSets} are:
  
  \begin{equation}
    \label{HomSetsInOrbitCategoryOfZ2}
    \begin{array}{lcl}
      \ZTwo\mathrm{Orb}
      \big(
        \ZTwo/1
        \,,\,
        \ZTwo/1
        \;\,
      \big)
      \;\simeq\;
      \ZTwo\;,
      &&
      \ZTwo\mathrm{Orb}
      \big(
        \ZTwo/\ZTwo
        \,,\,
        \ZTwo/\ZTwo
      \big)
        \;\simeq\;
      1\;,
      \\[3pt]
      \ZTwo\mathrm{Orb}
      \big(
        \ZTwo/1
        \,,\,
        \ZTwo/\ZTwo
      \big)
      \;\simeq\;
      \ast\,,
      &&
      \ZTwo\mathrm{Orb}
      \big(
        \ZTwo/\ZTwo
        \,,\,
        \ZTwo/1\;\;\,
      \big)
      \;\simeq\;
      \varnothing
      \,.
    \end{array}
  \end{equation}

\end{example}

\begin{example}[Automorphism groups in orbit category]
  \label{WeylGroup}
  For $G$ a finite group and $H \subset G$ a subgroup,
  the endomorphisms of $G/H \,\in\, G\mathrm{Orb}$
  (Def. \ref{OrdinaryOrbitCategory}) form the
  {\it Weyl group} $\WeylGroup(H)$ (e.g. \cite[p. 13]{May96})
  of $H$ in $G$,
  \begin{equation}
    \label{WeylGroups}
    \mathrm{End}_{G \mathrm{Orb}}
    (
     G/H
    )
    \;\simeq\;
    \mathrm{Aut}_{G \mathrm{Orb}}
    (
     G/H
    )
    \;=\;
    \WeylGroup(H)
    \;:=\;
    \NormalizerGroup(H)/H
    \,,
  \end{equation}

  \noindent
  namely the quotient group by $H$ of the normalizer
  $\NormalizerGroup(H)$ of  $H$ in $G$.
  For instance:
  $$
    \WeylGroup(1)
    \;=\;
    G
    \,,
    \phantom{AA}
    \WeylGroup(G)
    \;=\;
    1
    \,;
    \phantom{AA}
    \mbox{generally:}
    \phantom{A}
    H
    \underset{
      \mathclap{
      \mbox{
        \tiny
        normal
      }
      }
    }{\subset}
    G
    \;\;\;
    \Rightarrow
    \;\;\;
    \WeylGroup(H) \;=\; G/H
    \,.
  $$

\end{example}

Generally:

\begin{example}[Hom-sets in orbit category via Weyl groups]
  \label{HomSetsInOrbitCategoryViaWeylGroups}
  For any two subgroups $K, H \,\subset\, G$,
  the hom-set \eqref{HomSets} in the $G$-orbit category
  (Def. \ref{OrdinaryOrbitCategory}) between their
  corresponding coset spaces
  is, as a right $\WeylGroup(H)$-set via Example
  \ref{WeylGroup}, a disjoint union of copies of
  $\WeylGroup(H)$, one for each way of conjugating
  $K$ into a subgroup of $H$:  
\vspace{2mm} 
  \begin{equation}
    \label{HomSetsInOrbitCategoryInTermsOfWeylGroups}
    G\mathrm{Orb}
    \big(
      G/K\,,\,G/H
    \big)
    \;\simeq
    \underset{
      \scalebox{.7}{$
        {g \;\in\; G/\NormalizerGroup(K)}
        \atop
        {
        \scalebox{.7}{\rm s.t.}
        \;\,
        g^{-1} K g \;\subset\; H
        }
      $}
    }{\bigsqcup}
    g
    \WeylGroup(H)
    \;\;\;\;\;
    \in
    \;
    \WeylGroup(H)\mathrm{Actions}
    \big(
      \mathrm{Sets}
    \big)
    \,.
  \end{equation}

\end{example}

\begin{example}[More examples of orbit categories]
  \label{MoreExamplesOfOrbitCategories}
$\,$
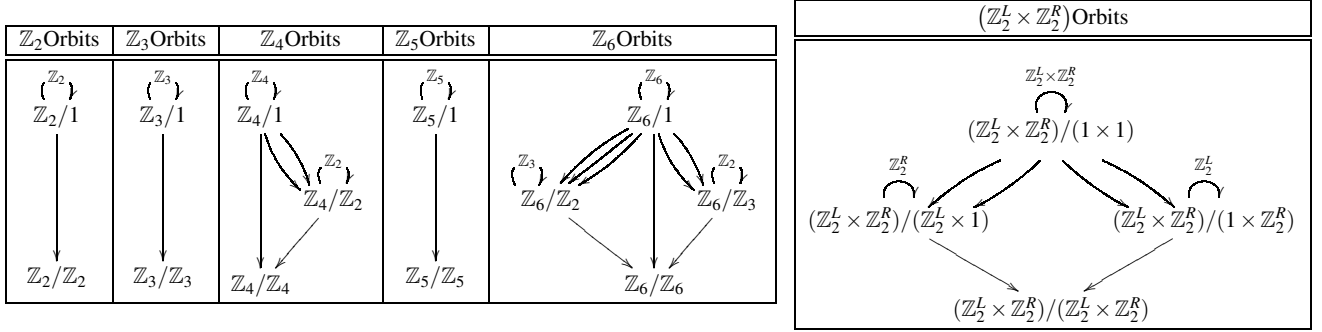
\begin{center}
\scalebox{0.78}{
\begin{tabular}{|c|c|c|c|c|}
  \hline
  $\ZTwo\mathrm{Orb}$
  &
  $\mathbb{Z}_3\mathrm{Orb}$
  &
  $\mathbb{Z}_4\mathrm{Orb}$
  &
  $\mathbb{Z}_5\mathrm{Orb}$
  &
  $\mathbb{Z}_6\mathrm{Orb}$
  \\
  \hline
  \hline
  \vspace{-9pt}
  &
  &
  &
  &
  \\
  \!\!\!
  \xymatrix{
    \ZTwo/1
    \ar@(ul,ur)|-{ \;\ZTwo\; }
    \ar[dd]
    \\
    {\phantom{A}}
    \\
    \ZTwo/\ZTwo
  }
  & \!\!\!
  \xymatrix{
    \mathbb{Z}_3/1
    \ar@(ul,ur)|-{ \;\mathbb{Z}_3\; }
    \ar[dd]
    \\
    {\phantom{A}}
    \\
    \mathbb{Z}_3/\mathbb{Z}_3
  }
  & \!\!\!\!\!\!\!\!\!\!\!\!
  \xymatrix@C=-14pt{
    \mathbb{Z}_4/1
    \ar@(ul,ur)|-{ \;\mathbb{Z}_4\; }
    \ar[dd]
    \ar@/^.4pc/[dr]
    \ar@/_.4pc/[dr]
    \\
    {\phantom{AAAAAAAAA}}
    &
    {\phantom{A}}
    &
    \mathbb{Z}_4/\ZTwo
    \ar@(ul, ur)|-{ \;\ZTwo\; }
    \ar[dll]
    \\
    \mathbb{Z}_4/\mathbb{Z}_4
  }
  &
  \xymatrix{
    \mathbb{Z}_5/1
    \ar@(ul,ur)|-{ \;\mathbb{Z}_5\; }
    \ar[dd]
    \\
    {\phantom{A}}
    \\
    \mathbb{Z}_5/\mathbb{Z}_5
  }
  &
  \hspace{-14pt}
  \xymatrix@C=-14pt{
    {\phantom{A}}
    &&&
    {\phantom{AAAAAAAAA}}
    &
  \!\!\!  \mathbb{Z}_6/1
    \ar@(ul,ur)|-{ \;\mathbb{Z}_6\; }
    \ar[dd]
    \ar@/^.4pc/[dl]
    \ar@/0pc/[dl]
    \ar@/_.4pc/[dl]
    \ar@/^.4pc/[dr]
    \ar@/_.4pc/[dr]
    \\
    {\phantom{A}}
    \ar@(ul, ur)|-{ \;\mathbb{Z}_3\; }
    &&&
    \mathbb{Z}_6/\mathbb{Z}_{\mathrlap{2}}
    \ar[dr]
    &
    {\phantom{AAAAAAAAA}}
    &
    {\phantom{A}}
    &
    \mathbb{Z}_6/\mathbb{Z}_3
    \ar@(ul, ur)|-{ \;\ZTwo\; }
    \ar[dll]
    \\
    {\phantom{A}}
    &&&
    {\phantom{AAAAAAAAA}}
    &
    \mathbb{Z}_6/\mathbb{Z}_6
  }
  \\
  \hline
\end{tabular}
}

\vspace{.1cm}

\scalebox{.75}{
\begin{tabular}{|c|}
  \hline
  $
    \mathclap{\phantom{\vert^{\vert^{\vert}}}}
    \big(\ZTwo^L \times \ZTwo^R\big) \mathrm{Orb}
    \mathclap{\phantom{\vert_{\vert_{\vert}}}}
  $
  \\
  \hline
  \hline
  \vspace{-4pt}
  \\
  \!\!\!  
  \xymatrix@C=-24pt{
    &
    (
      \mathbb{Z}^L_2
      \times
      \mathbb{Z}^R_2
    )
      /
    ( 1 \times 1 )
    \ar@(ul,ur)^-{
      \;
      \mathbb{Z}^L_2
      \times
      \mathbb{Z}^R_2
      \;
    }
    \ar@/^.6pc/@<+6pt>[dl]
    \ar@/_.6pc/@<+6pt>[dl]
    \ar@/^.6pc/@<-6pt>[dr]
    \ar@/_.6pc/@<-6pt>[dr]
    \\
    \mathclap{\phantom{\vert^{\vert^{\vert}}}}
(
      \mathbb{Z}^L_2
      \times
      \mathbb{Z}^R_2
    )
      /
    (
      \mathbb{Z}^L_2 \times 1
    )
 \ar@(ul,ur)^-{ \;\ZTwo^R\; }
      \ar[dr]
    &&
    \mathclap{\phantom{\vert^{\vert^{\vert}}}}
    (
      \mathbb{Z}^L_2
      \times
      \mathbb{Z}^R_2
    )
      /
    (
      1
      \times
      \mathbb{Z}^R_2
    )
    \ar@(ul,ur)^-{ \;\ZTwo^L\; }
    \ar[dl]
    \\
    &
    (
      \mathbb{Z}^L_2
      \times
      \mathbb{Z}^R_2
    )
      /
    (
      \mathbb{Z}^L_2 \times \ZTwo^R
    )
  }
  \!
  \\
  \hline
\end{tabular}
}
\end{center}

\end{example}

\noindent {\bf Equivariant homotopy types.}

\vspace{-1mm} 
\begin{defn}[Equivariant simplicial sets]
  \label{EquivariantSSetCategory}
  We write
  $$
    \EquivariantSSet
    \;:=\;
    \mathrm{Functors}
    \big(
      G \mathrm{Orb}^{\mathrm{op}}
      \,,\,
      \mathrm{SSet}
    \big)
  $$

  \noindent
  for the category of functors
  from the opposite of
  $G$-orbits (Def. \ref{OrdinaryOrbitCategory})
  to simplicial sets.
\end{defn}

\begin{example}[Systems of fixed loci of topological $G$-actions]
  \label{SystemsOfFixedPointSpaces}
  Let
  $
    G \acts \; X
    \;\in\;
    \GActionsOnTopSp
  $ (Def. \ref{GActionsOnTopSp}).
  For $H \subset G$ any subgroup,
  a $G$-equivariant function \eqref{EquivatiantContinuousFunction}
  \begin{equation}
    \label{FixedLoci}
    \begin{array}{l}
    \xymatrix{
      G/H
      \ar@(ul,ur)|-{ \;G\; }
      \ar[r]^f
      &
      X
      \ar@(ul,ur)|-{ \;G\; }
    }
    \\    
    \Leftrightarrow
    \phantom{AAAAA}
    \raisebox{5pt}{$
    f([\mathrm{e}])
    \;\in\;\;\;\;
    \overset{
      \mathclap{
      \raisebox{3pt}{
        \tiny
        \color{darkblue}
        \bf
        $H$-fixed locus
      }
      }
    }{
      X^H
    }
      \;\;\;\;
      :=
      \;\;
      \Big\{
        x \in X
        \,\Big\vert\;
        \underset{h \in H}{\forall}
        \big(
          h \cdot x \,=\, x
        \big)
      \Big\}
    \;\subset X\;
    $}
    \end{array}
  \end{equation}

  \noindent
  from the corresponding $G$-orbit (Def. \ref{OrdinaryOrbitCategory}) is
  determined by its image $f([\mathrm{e}]) \in X$ of the class of the neutral element,
  and that image has to be fixed by the action of $H \subset G$ of $X$.
  Therefore, the corresponding $G$-equivariant mapping spaces

  \vspace{-.3cm}
  $$
    \mathrm{Maps}\big( G/H,\, X\big)^G
    \;\simeq\;\,
    X^H
  $$

  \noindent
  are the topological subspaces of $H$-fixed points inside $X$,
  the {\it $H$-fixed loci in $G \acts \; H$.}
  By functoriality of the mapping-space
  construction,
  these fixed point loci are exhibited as arranging into a contravariant functor on the $G$-orbit category (Def. \ref{OrdinaryOrbitCategory}):
   \vspace{-3mm}
  \begin{equation}
    \label{SystemOfFixedLoci}
    \hspace{-2cm} 
    \xymatrix@R=6pt@C=20pt{
      \orbisingular (X \!\sslash\! G)
      :
      &
      G \mathrm{Orb}^{\mathrm{op}}
      \ar[rrr]^-{\mathrm{Maps}(-,\, X)^G}
      &&&
      \mathrm{TopSp}
      \\
      &
      G/H_1
      \ar[dd]_-{
        [(\mathrm{id}, g)]
      }
      \ar@{|->}[rrr]
      &&&
      X^{H_1}
      \mathrlap{
        \;
        \;\;
        \mbox{
          \tiny
          \color{darkblue}
          \bf
          \begin{tabular}{c}
            $H_1$-fixed locus
          \end{tabular}
        }
      }
      \\
      \\
      &
      G/H_1
      \ar@{|->}[rrr]
      \ar[dd]_-{
        [(\phi,e)]
      }
      &&&
      X^{H_1}
      \ar[uu]^{\simeq}_{
        g\cdot(-)
        \mathrlap{
          \!\!\!\!\!\!\!
          \mbox{
            \tiny
            \color{greenii}
            \bf
            \def\arraystretch{1}
            \begin{tabular}{c}
              residual action on
              \\
              $H_2$-fixed locus
            \end{tabular}
          }
        }
      }
      \\
      \\
      &
      G/H_2
      \ar@{|->}[rrr]
      &&&
      \mathclap{\phantom{\vert^{\vert^\vert}}}
      X^{H_2}
       \mathrlap{
        \;
        \;\;
        \mbox{
          \tiny
          \color{darkblue}
          \bf
          \begin{tabular}{c}
            $H_2$-fixed locus
          \end{tabular}
        }
      }
      \ar@{^{(}->}[uu]_{
        \phi^\ast
        \mathrlap{
          \;\;\;
          \mbox{
            \tiny
            \color{greenii}
            \bf
            \def\arraystretch{1}
            \begin{tabular}{c}
              inclusion of
              \\
              $H_2$-fixed locus
            \end{tabular}
          }
        }
      }
    }
  \end{equation}

 \def\arraystretch{1}
   
  \noindent
  Here we used Example \ref{MorphismSetsInGOrbitCategory} to make
  explicit the nature of the continuous functions between
  fixed point spaces that this functor assigns to morphisms of
  $G \mathrm{Orb}$. In particular, we see from Example
  \ref{WeylGroup} that the residual action on the $H$-fixed
  locus $X^H$ is by the Weyl group $\WeylGroup(H)$ \eqref{WeylGroups}.
  Postcomposing \eqref{SystemOfFixedLoci} with the singular simplicial set
  functor \eqref{QuillenEquivalenceBetweenTopSpAndSSet}
  yields an equivariant simplicial set (Def. \ref{EquivariantSSetCategory}), to be denoted
  (the notation follows \cite[\S 3.2, 5.1]{SS20b}):

  \vspace{-1mm}
  \begin{equation}
    \label{SingularEquivariantSimplicialSet}
    G \acts \; X
    \;\;
      \longmapsto
    \;\;
    \mathrm{Sing}
    \big(
    \orbisingular
    \big(
      X \!\sslash\! G
    \big)
    \big)
      :=
    \mathrm{Sing}
    \Big(
    \mathrm{Maps}
    \big(
      -
      \,,\,
      X
    \big)^G
    \Big)
    \;
    \in
    \;
    \EquivariantSSet
    \,.
  \end{equation}

 \end{example}

\begin{prop}[Model category of equivariant simplicial sets
{\cite[Thm. 11.6.1]{Hirschhorn02} \cite[Thm. 3.3]{Guillou06}\cite[\S 2.2]{Stephan16}}]
  \label{ModelCategoryOnEquivariantSSet}
  The category of equivariant simplicial sets
  (Def. \ref{EquivariantSSetCategory})
  carries a model category structure whose

  \begin{itemize}
  \item[{\bf (a)}] {$\mathrm{W}$} --
    {\it weak equivalences} are the weak equivalences
    of $\mathrm{SSet}_{\mathrm{Qu}}$
    over each $G/H \in G\ \mathrm{Orb}$;

   \item[{\bf (b)}]  {$\mathrm{Fib}$} --
    {\it fibrations} are the weak equivalences
    of $\mathrm{SSet}_{\mathrm{Qu}}$
    over each $G/H \in G\ \mathrm{Orb}$.
  \end{itemize}
\end{prop}
\noindent
  We denote this model category by
  $
    \EquivariantSSetProj
    \;\in\;
    \mathrm{ModCat}
    \,.
  $

\begin{defn}[Equivariant homotopy types]
  \label{EquivariantHomotopyTypes}
  We denote the homotopy category
  of the projective model structure on equivariant
  simplicial sets (Prop. \ref{ModelCategoryOnEquivariantSSet})
  by
  \begin{equation}
    \label{GHomotopyTypes}
    \xymatrix{
    \EquivariantSSetProj
    \ar[rr]^{\gamma}_{
      \mbox{
        \tiny
        {\color{greenii}
        localization}
      }
    }
    &&
    \EquivariantHomotopyTypes
    }
    :=
    \;\;
    \mathrm{Ho}
    \big(
      \EquivariantSSet_{\mathrm{proj}}
    \big)
    \,.
  \end{equation}

\end{defn}
The key source of equivariant homotopy types is
the shapes of orbi-singularized homotopy quotients of
topological spaces by continuous group actions
(we follow \cite[\S 3.2]{SS20b} in terminology and notation):
\begin{defn}[Equivariant shape]
  \label{EquivariantShape}
  The composite of forming
  systems of fixed loci (Example \ref{SystemsOfFixedPointSpaces})
  with localization to equivariant homotopy types (Def. \ref{EquivariantHomotopyTypes})
  is the \emph{equivariant shape} operation,
  generalizing the plain shape \eqref{TopologicalShapeAsLocalization}:
  \begin{equation}
    \label{EquivariantShapeAsLocalization}
    \hspace{-7mm}
    \xymatrix@R=15pt@C=75pt{
      \GActionsOnTopSp
      \ar[rr]^-{
        \scalebox{.7}{$
          G \acts \; X
        $}
        \;\;\;\longmapsto\;\;\;
        \overset{
          \mathclap{
          \raisebox{3pt}{
            \mbox{
              \tiny
              \color{greenii}
              \bf
              equivariant shape
            }
          }
          }
        }{
          \raisebox{1pt}{\textesh}
          \orbisingular
          (
            X \!\sslash\! G
          )
        }
      }_-{
        \mbox{
          \tiny
          \color{greenii}
          \bf
          \def\arraystretch{.9}
          \begin{tabular}{c}
          localization at fixed locus-wise
          \\
          weak homotopy equivalences
          \end{tabular}
        }
      }
      \ar[dr]_-{
        \mbox{
          \tiny
          \def\arraystretch{1}
          \begin{tabular}{c}
            \color{greenii}
            \bf
            form singular
            \\
            \color{greenii}
            \bf
            equivariant simplicial set
            \\
            \eqref{SingularEquivariantSimplicialSet}
          \end{tabular}
        }
      }
      &&
      \EquivariantHomotopyTypes
      \\
      &
      \EquivariantSSet
      \ar[ur]_-{
        \gamma
        \mathrlap{
          \mbox{
            \tiny
            {\color{greenii}
            \bf
            localization}
            \eqref{GHomotopyTypes}
          }
        }
      }
    }.
  \end{equation}
  
\end{defn}

\begin{example}[Smooth equivariant homotopy types]
  \label{SmoothSingularHomotopyTypes}
  A topological space $X$ equipped with trivial $G$-action
  has equivariant shape (Def. \ref{EquivariantShape})
  given by the functor
  on the orbit category which is constant on its ordinary shape
  \eqref{TopologicalShapeAsLocalization}

  \vspace{-3mm}
  \begin{equation}
    \xymatrix@C=35pt{
      \mathrm{TopSp}
      \ar[rr]_-{
          \raisebox{3pt}{
            \tiny
            \color{greenii}
            shape
          }
        }^-{
        \scalebox{.7}{$
          \raisebox{1pt}{\rm\textesh}
        $}
        }
      \ar[d]_-{
        \mbox{
          \tiny
          \color{greenii}
          \bf
          \def\arraystretch{1}
          \begin{tabular}{c}
            equip with
            \\
            trivial action
          \end{tabular}
        }
      }
      &&
      \HomotopyTypes
      \ar[d]_-{
               \mathrm{Smth}
           }^-{
          \mbox{
            \tiny
            \color{greenii}
            \bf
            \def\arraystretch{1}
            \begin{tabular}{c}
              form constant functor
              \\
              on orbit category
            \end{tabular}
        }
      }
      \\
      \GActionsOnTopSp
      \ar[rr]_-{
             \raisebox{-3pt}{
            \tiny
            \color{greenii}
            \bf
            equivariant shape
          }
          }^-{
        \scalebox{.7}{$
          \raisebox{1pt}{\rm\textesh}
          \big(
            - \sslash G
          \big)
        $}
        }
      &&
      \EquivariantHomotopyTypes \;.
    }
  \end{equation}

  \noindent
  For brevity, we will mostly leave this embedding
  notationally implicit and write
  \begin{equation}
    \label{TopologicalSpaceAsSmoothEquivariantHomotopyType}
    X
    \;\;
      :=
    \;\;
    \mathrm{Smth}
    \,
    \raisebox{1pt}{\textesh}
    \,
    X
    \;\;
    \in
    \;
    \EquivariantHomotopyTypes \;.
  \end{equation}

\end{example}

\medskip

\noindent {\bf Elmendorf's theorem.}
In fact, every equivariant homotopy type
(Def. \ref{EquivariantHomotopyTypes})
is the equivariant shape (Def. \ref{EquivariantShape})
of some topological space with $G$-action
(Def. \ref{GActionsOnTopSp}). This is the
content of Elmendorf's theorem
(\cite{Elmendorf83}, see Prop. \ref{ElmendorfTheorem} below).
Due to this fact,
topological $G$-actions
in equivariant homotopy theory
are often conflated
with their $G$-equivariant shape, and jointly referred to
as {\it $G$-spaces} (e.g., \cite[\S 8]{tomDieck79}\cite[\S 1]{Blu17}).

\begin{prop}[Model category of simplicial $G$-actions and fixed
 loci {\cite[Thm. 3.12]{Guillou06}\cite[Prop. 2.6]{Stephan16}}]
  \label{ModelCategoryOfSimplicialGActionsAndFixedLoci}
  The category
  $\GActionsOnSSet$
  of $G$-actions $G \acts \; S$ on simplicial sets
  (analogous to Def. \ref{GActionsOnTopSp})
  carries a model category structure
  whose weak equivalences and fibrations are those that
  become so in the classical model structure on simplicial sets
  \eqref{ClassicalModelStructureOnSSet}
  under the functor (analogous to Example \ref{SystemsOfFixedPointSpaces})
  \begin{equation}
    \label{SystemOfSimplicialFixedLoci}
    \xymatrix@R=-2pt@C=3em{
     \GActionsOnSSet
      \ar[rr]^-{
        \mathrm{Maps}
        (
          -
          \,,\,
          -
        )^G
      }
      &&
      \EquivariantSSet
      \\
      G
        \raisebox{1pt}{\scalebox{.8}{$\acts$}}
        \;
      S
      \ar@{}[rr]|-{\qquad \qquad \longmapsto }
      &&
      \big(
        G/H \;\mapsto\; S^H
      \big)
    }
  \end{equation}

 \noindent
  which sends a $G$-action
  $G \acts \; S$ to its system of $H$-fixed loci
  parametrized over $G/H \,\in\, G \mathrm{Orb}$.
\end{prop}
\noindent
We denote this model category by
$$
  \GActionsOnSSet_{\mathrm{fine}}
  \;\;
  \in
  \;\;
  \mathrm{ModCat}
  \,.
$$

\begin{prop}[Elmendorf's theorem via model categories {\cite[Thm. 3.17]{Stephan16}\cite[Prop. 3.15]{Guillou06}}]
  \label{ElmendorfTheorem}
  The functor assigning systems of simplicial fixed loci
  \eqref{SystemOfSimplicialFixedLoci}
  is the right adjoint in a Quillen equivalence    
  \begin{equation}
    \label{ElmendorfQuillenEquivalence}
    \xymatrix@R=10pt@C=4em{
      \GActionsOnSSet_{\mathrm{fine}}
      \ar@{<-}@<+6pt>[rr]^-{ (-)(G/1) }
      \ar@<-6pt>[rr]_-{
          \mathrm{Maps}
          (
            -
            \,,\,
            -
          )^G
      }^-{ \simeq_{\mathrlap{\mathrm{Qu}}} }
      &&
     \;
     \EquivariantSSet_{\mathrm{proj}}
    }
  \end{equation}

 \noindent
  between the fine model structure on simplicial $G$-actions
  (Prop. \ref{ModelCategoryOfSimplicialGActionsAndFixedLoci})
  and the model category of equivariant simplicial sets
  (Prop. \ref{ModelCategoryOnEquivariantSSet}).
\end{prop}

\medskip

\noindent {\bf Examples of equivariant homotopy types.}

\begin{example}[$\Gade$-equivariant 4-sphere]
  \label{ADEEquivariant4Sphere}
Let
$$
  G \,:=\, \Gade \;\subset\; \mathrm{Spin}(3) \;\simeq\; \mathrm{Sp}(1)
$$

\noindent
be a finite subgroup of the Spin group in dimension 3;
these are famously classified along an ADE-pattern
(reviewed in \cite[Rem. A.9]{HSS18}). Via the
exceptional isomorphism with the quaternionic unitary group,
this induces a canonical smooth action
(Def. \ref{ProperGActionsOnSmoothManifolds})
on
the Euclidean 4-space underlying the space of quaternions
(reviewed as \cite[Prop. A.8]{HSS18})
and hence also on the corresponding representation 4-sphere
(Example \ref{RepresentationSpheres}):
$$
  \xymatrix{
    \mathbb{R}^4
    \ar@(ul,ur)^-{
      \Gade\;
    }
  }
  ,\;\;
  \xymatrix{
    S^4
    \ar@(ul,ur)^-{
      \Gade\;
    }
  }
  \;\in\;
  \GadeActionsOnSmoothManifolds
  \,.
$$
The corresponding $\Gade$-equivariant homotopy types
(Def. \ref{EquivariantHomotopyTypes})
(their equivariant shape, Def. \ref{EquivariantShape})
$$
  \overset{
    \mathclap{
    \raisebox{3pt}{
      \tiny
      \color{darkblue}
      \bf
      \begin{tabular}{c}
        $\Gade$-equivariant shape
        \\
        of 4-sphere
      \end{tabular}
    }
    }
  }{
  \raisebox{1pt}{\textesh}
  \,
  \orbisingular
  \big(
    S^4 \!\sslash\! \Gade
  \big)
  }
  \;\;
  \in
  \;
  \GadeEquivariantHomotopyTypes
$$

\noindent
are the coefficients of ADE-equivariant Cohomotopy theory
\cite[\S 5.2]{HSS18}\cite[\S 3]{SS19a}
(lifted to equivariant twistorial Cohomotopy theory below in
Def. \ref{EquivariantTwistorialCohomotopyTheory}).
\end{example}

\begin{example}[$\Grefl$-equivariant twistor space]
 \label{GhetEquivariantTwistorSpace}
  Consider the quaternion unitary group
  (e.g. \cite[\S A]{FSS20c} )
  with its two commuting subgroups from
  \eqref{Ghet} and \eqref{TheSp1Subgroup}:  
  \begin{equation}
    \label{QuaternionUnitaryGroup}
    \Grefl,
    \,
    \SpLR
    \;\;
      \subset
    \;\;
    \mathrm{Sp}(2)
    \;:=\;
    \big\{
    g \in \mathrm{Mat}_{2 \times 2}(\mathbb{H})
    \,\left\vert\,
      g \cdot g^\dagger \,=\, 1
    \right.
   \! \big\}
   \,.
  \end{equation}

  \noindent
  Their canonical action
  on
  $\mathbb{H}^2 \simeq_{\scalebox{.5}{$\mathbb{R}$}} \mathbb{R}^8$
  by left matrix multiplication
  induces an action \eqref{Sp2ActionOnTwistorSpace} on
  $\mathbb{C}P^3$
  (``twistor space'').
 The fixed locus \eqref{FixedLoci}
 of the subgroup $\Grefl$ \eqref{Ghet}
 under this action
 is evidently given by those
$[z_1 : z_2 : z_3 : z_4] \,\in\, \mathbb{C}P^3$ such that
$z_1 + \mathrm{j} \cdot z_2   \,=\, z_3 + \mathrm{j} \cdot z_4   \,\in\, \mathbb{H}$.
Since these are exactly the elements that are sent
by the twistor fibration $t_{\mathbb{H}}$ \eqref{TwistorFibration}
to the base point
$[1 : 1] \,\in\, \mathbb{H}P^1$, the
$\Grefl$-fixed locus in twistor space $\mathbb{C}P^3$
coincides with the $S^2$-fiber of
the twistor fibration $t_{\mathbb{H}}$ \eqref{TwistorFibration}:

\vspace{-5mm}
\begin{equation}
  \label{Z2FixedPointsInTwistorSpaceFormS2Fiber}
  \big(
    \mathbb{C}P^3
  \big)^{\Grefl}
  \;\simeq\;
  \xymatrix{
    S^2 \;
    \ar@{^{(}->}[rr]^-{ \mathrm{fib}(t_{\mathbb{H}}) }
    &&
    \mathbb{C}P^3.
  }
\end{equation}

\noindent
Hence the $\ZTwo$-equivariant homotopy type
\eqref{TopologicalShapeAsLocalization}
of twistor space with its $\Grefl$ action \eqref{Sp2ActionOnTwistorSpace}
is given by the following
functor on the $\ZTwo$-orbit category \eqref{OrbitCategoryOfZ2}:
\begin{equation}
  \label{Z2EquivariantTwistorSpaceAsPresheafOnOrbitCategory}
  \hspace{-2.3cm}
  \overset{
    \mathclap{
    \raisebox{3pt}{
      \tiny
      \color{darkblue}
      \bf
      \begin{tabular}{c}
        $\Grefl$-equivariant shape
        \\
        of twistor space
      \end{tabular}
    }
    }
  }{
  \raisebox{1pt}{\textesh}
  \orbisingular
  \big(
    \mathbb{C}P^3 \!\sslash\! \Grefl
  \big)
  }
  \;\;\;
   :
  \;\;\;
  \raisebox{23pt}{
  \xymatrix@R=1.5em{
    \ZTwo/1
    \ar@(ul,ur)|-{\;\ZTwo}
    \ar[d]
    \ar@{}[rr]|-{ \longmapsto }
    &&
    \raisebox{1pt}{\textesh}\,
    \mathbb{C}P^3
    \ar@(ul,ur)|-{\, \Grefl }
    \\
    \mathclap{\phantom{\vert^{\vert}}}
    \ZTwo/\ZTwo
    \ar@{}[rr]|-{ \longmapsto }
    &&
    {\raisebox{1pt}{\textesh}\,
    S^2}^{\phantom{A}}
    \ar@{^{(}->}[u]_-{
      \mathrm{fib}(t_{\mathbb{H}})
      \mathrlap{
        \!\!\!\!
        \mbox{
          \tiny
          \color{greenii}
          \bf
          \begin{tabular}{c}
            fiber inclusion of
            \\[-2pt]
            twistor fibration
          \end{tabular}
        }
      }
    }
  }
  }
\end{equation}

\end{example}

\medskip

\noindent {\bf Equivariant homotopy groups.}

\begin{defn}[Equivariant groups]
  \label{EquivariantGroups}
 {\bf (i)}  We write
  $$
    \EquivariantGroups
    \;:=\;
    \mathrm{Functors}
    \big(
      G \mathrm{Orb}^{\mathrm{op}}
      \,,\,
      \mathrm{Grp}
    \big)
  $$

  \noindent
  for the category of contravariant functors on the
  $G$-orbit category (Def. \ref{OrdinaryOrbitCategory})
  with values in groups.

  \noindent  {\bf (ii)} We write
  $$
    \EquivariantAbelianGroups
    \;:=\;
    \mathrm{Functors}
    \big(
      G \mathrm{Orb}
      \,,\,
      \mathrm{AbelianGroups}
    \big)
  $$
  \vspace{-.55cm}

  \noindent
  for the sub-category of contravariant functors
  with values in abelian groups.
\end{defn}

\begin{example}[Equivariant singular homology groups]
  \label{EquivariantSingularHomologyGroup}
  For $\mathscr{X} \,\in\, \EquivariantHomotopyTypes$
  (Def. \ref{EquivariantHomotopyTypes}),
  $A \,\in\, \mathrm{AbelianGroups}$,
  the ordinary $A$-homology groups in degree $n\in \mathbb{N}$
  of the stages of $\mathscr{X}$
  form an equivariant abelian group
  in the sense of Def. \ref{EquivariantGroups}, to be denoted:
  $$
    \underline{H}_{\, n}
    \big(
      \mathscr{X}
      ;
      \,
      A
    \big)
    \;\;
      :
    \;\;
    G/H
      \;\longmapsto\;
    H_n
    \big(
      \mathscr{X}(G/H);
      \,
      A
    \big)
    \,.
  $$
\end{example}

\begin{defn}[Equivariant homotopy groups]
  \label{EquivariantHomotopyGroups}
  $\,$

  \noindent
  {\bf (i)} For $\mathscr{X} \,\in\, \EquivariantHomotopyTypes$
  (Def. \ref{EquivariantHomotopyTypes}),
  $\orbisingular(\ast \!\sslash\! G)
   \overset{x}{\longrightarrow}
  \mathscr{X}$
  a base-point, and $n \in \mathbb{N}$,
  we say that the $n$th {\it equivariant homotopy group}
  of $\mathscr{X}$ at $x$ is the equivariant group
  (Def. \ref{EquivariantGroups}) which is stage-wise
  the ordinary $n$th homotopy group, to be denoted:
  \begin{equation}
    \label{EquivariantHomotopyGroupsOfEquivariantHomotopyType}
    \underline{\pi}_{\, n}
    (
      \mathscr{X}
      ,
      x
    )
    \;:=\;
    \Big(
      G/H
      \;\mapsto\;
      \pi_n
      \big(
        X(G/H)
        ,
        x(G/H)
      \big)
   \! \Big).
  \end{equation}

  \noindent
  {\bf (ii)}
  Similarly,
  for $G \acts \; X \;\in\; \GActionsOnTopSp$
  (Def. \ref{GActionsOnTopSp}),
  $G \acts \; \ast \overset{x}{\longrightarrow} G \acts \, X$
  a fixed base point, and $n \in \mathbb{N}$,
  we say that the $n$th {\it equivariant homotopy group}
  of $G \acts \; X$ is that \eqref{EquivariantHomotopyGroupsOfEquivariantHomotopyType}
  of its equivariant shape \eqref{TopologicalShapeAsLocalization}:
  \begin{equation}
    \underline{\pi}_{\, n}
    (
      X
      ,
      x
    )
    \;:=\;
    \underline{\pi}_{\, n}
    \Big(
      \raisebox{1pt}{\textesh}
      \orbisingular
      \big(
        X \!\sslash\! G
      \big)
      \, ,\,
      \raisebox{1pt}{\textesh}
      \orbisingular
      \big(
        x \!\sslash\! G
      \big)
    \!\Big)
    \;=\;
    \Big(
      G/H
      \;\mapsto\;
      \pi_n
      \big(
        X^H,
        \,
        x
      \big)
   \! \Big).
  \end{equation}
 \end{defn}

\begin{defn}[Equivariant connected homotopy types]
  \label{EquivariantConnectedHomotopyTypes}
  We write
  \begin{equation}
    \label{InclusionOfEquivariantConnectedHomotopyTypes}
    \xymatrix{
      \EquivariantHomotopyTypesConnected
      \;
      \ar@{^{(}->}[r]
      &
      \; \EquivariantHomotopyTypes
    }
  \end{equation}
  \vspace{-.55cm}

  \noindent
  for the full subcategory on those
  equivariant homotopy types $\mathscr{X}$ (Def. \ref{EquivariantHomotopyTypes})
  which
  \begin{itemize}
  \item[{\bf (a)}] are equivariantly connected, in that
  $\mathscr{X}(G/H) \,\in\, \HomotopyTypes$ is connected for all $H \subset G$;

  \item[{\bf (b)}] admit an equivariant base point
  $
    \orbisingular\big( \ast \sslash G \big)
  \to
    \mathscr{X}$.
  \end{itemize}

\end{defn}

\begin{defn}[Equivariant 1-connected homotopy types]
  \label{SubcategoryOfEquivariantSimplyConnectedRFiniteHomotopyTypes}
  $\,$
  
  \noindent   
  {\bf (i)} We write
  \begin{equation}
    \label{InclusionOfEquivariantSimplyConnectedHomotopyTypes}
    \xymatrix{
      \EquivariantHomotopyTypesSimplyConnected
      \;
      \ar@{^{(}->}[r]
      &
      \EquivariantHomotopyTypesConnected
      \;
      \ar@{^{(}->}[r]
      &
      \; \EquivariantHomotopyTypes
    }
  \end{equation}

  \noindent
  for the further full subcategory on those
  equivariant homotopy types $\mathscr{X}$ (Def. \ref{EquivariantHomotopyTypes})
  which
  \begin{itemize}
  \item[{\bf (a)}] are equivariantly connected and
  admit an equivariant base point (Def. \ref{EquivariantConnectedHomotopyTypes});

  \item[{\bf (b)}] have trivial
  first equivariant homotopy group (Def. \ref{EquivariantHomotopyGroups})
  at that base point:
  $$
    \underline{\pi}_1
    (
      \mathscr{X}
      ,\,
      x
    )
    \;\;
    =
    \;\;
    \underline{1}
    \,.
  $$
  \end{itemize}

  \noindent
  {\bf (ii)}
  By the Hurewicz theorem, this implies that
  the equivariant real cohomology groups
  (Example \ref{EquivariantSingularHomologyGroup}) of these
  objects are trivial in degrees $\leq 1$
  $$
    X
    \;\in\;
    \EquivariantSimplyConnectedHomotopyTypes
    \;\;\;\;\;\;\;\;
    \Rightarrow
    \;\;\;\;\;\;\;\;
    \big(
      \underline{H}^0(X)
      \;\simeq\;
      \underline{\mathbb{R}}
      \;\;\;
      \mbox{and}
      \;\;\;
      \underline{H}^1(X)
      \;\simeq\;
      0
    \big)
    \,.
  $$

  \noindent
  {\bf (iii)}
  We write
  $$
    \xymatrix{
      \EquivariantSimplyConnectedRFiniteHomotopyTypes
      \;
      \ar@{^{(}->}[r]
      &
      \;
      \EquivariantSimplyConnectedHomotopyTypes
      \;
      \ar@{^{(}->}[r]
      &
      \;
      \EquivariantHomotopyTypes
    }
  $$

  \noindent
  for the further full subcategory
  of those equivariant 1-connected homotopy types
  \eqref{InclusionOfEquivariantSimplyConnectedHomotopyTypes}
  which are of {\it finite type} over $\mathbb{R}$,
  in that all their equivariant real homology groups
  (Example \ref{EquivariantSingularHomologyGroup})
  are finite-dimensional:
  $$
    \underset{
      {H \subset G}
      \atop
      {n \in \mathbb{N}}
    }{\forall}
    \;
    \mathrm{dim}_{\mathbb{R}}
    \Big(
    H_n
    \big(
      \mathscr{X}(G/H);
      \,
      \mathbb{R}
    \big)
    \!\Big)
    \;<\;
    \infty
    \,.
  $$
  \end{defn}

\smallskip

\noindent {\bf $G$-Orbifolds.}
Given a smooth manifold $X$ equipped with a smooth group
action $G \acts \; X$, there are several somewhat different
mathematical notions of what exactly counts as the corresponding
{\it quotient orbifold}
(review in \cite{MoerdijkMrcun03}\cite[\S 6]{Kapovich08}\cite{IKZ10}).
\begin{itemize}

\item First, there is the singular quotient space $X/G$
that dominates the early literature on orbifolds
\cite{Satake56}\cite{Satake57} \cite{Thurston80}\cite{Haefliger84}
as well as the contemporary physics literature \cite[\S 1.3]{BailinLove99}.

\item  Second, there is the smooth stacky homotopy quotient $X \!\sslash\! G$ that has become
the popular model for orbifolds among Lie theorists
\cite{MoerdijkPronk97}\cite{Moerdijk02}\cite{Lerman08}\cite{Amenta12}.

\item Third,
there is the fine incarnation of orbifolds
{\it orbisingular homotopy quotients}
$\orbisingular\big( X \!\sslash\! G\big)$
in singular cohesive homotopy theory \cite{SS20b},
which unifies the above two perspectives and lifts them to
make orbifolds carry proper equivariant
differential cohomology theories.
\end{itemize}

\noindent Here we extract from \cite{SS20b}
the essence of this latter fine perspective that is
necessary and convenient for the present purpose,
as Def. \ref{GOrbifolds} below.

\begin{lemma}[Fixed loci of finite smooth actions are smooth manifolds]
  \label{FixedLociOfProperSmoothActionsAreSmoothManifolds}
  If $G \acts  \; X \,\in\, \GActionsOnTopSp$
  (Def. \ref{GActionsOnTopSp})
  is such that
  $X$ admits the structure of a smooth manifold
  and such that the action \eqref{GroupAction} of $G$ is smooth,
  %
  then the fixed loci $X^H \hookrightarrow X$ \eqref{FixedLoci}
  are themselves smooth submanifolds.
\end{lemma}
\begin{proof}
  Since $G$ is assumed to be finite \eqref{EquivarianceGroup},
  its smooth action is proper
  (e.g. \cite[Cor. 21.6]{Lee12}).
  But in smooth manifolds with proper smooth $G$-action,
  every closed submanifold
  inside a fixed locus has a $G$-equivariant tubular neighborhood
  \cite[\S VI, Thm. 2.2]{Bredon72}\cite[Thm. 4.4]{Kankaanrinta07}.
  This applies, in particular, to individual fixed points, where it
  says that each such has a neighborhood in the fixed locus diffeomorphic
  to an open ball.
\end{proof}

\begin{defn}[Smooth group actions on smooth manifolds]
 \label{ProperGActionsOnSmoothManifolds}
{\bf (i)} We write
$$
  \xymatrix@R=8pt{
    \GActionsOnSmoothManifolds
    \ar[rr]
    &&
    \GActionsOnTopSp
  }
$$

\noindent
for the category of smooth manifolds equipped with
$G$-actions on the underlying topological spaces
(Def. \ref{GActionsOnTopSp}) which are smooth.

\noindent
{\bf (ii)} Similarly, if the compact Borel-equivariance group
\eqref{CompactTopologicalGroup} is equipped with
smooth structure making it a Lie group
$$
  T \;\in\;
  \xymatrix{
    \mathrm{CompactLieGroups}
    \ar[r]
    &
    \mathrm{CompactTopGrp}\;,
  }
$$

\noindent
we write
$$
  \xymatrix@R=8pt{
    \TGActionsOnSmoothManifolds
    \ar[rr]
    &&
    \TGActionsOnTopSp
  }
$$

\noindent
for the category of smooth manifolds equipped with
$T \times G$-actions on the underlying topological spaces
(Def. \ref{GActionsOnTopSp}) which are smooth.

\end{defn}

\begin{defn}[$G$-Orbifolds {\cite{SS20b}}]
 \label{GOrbifolds}
{\bf (i)} We write
  \begin{equation}
   \label{CategoryOfGOrbifolds}
   G\mathrm{Orbifolds}
   \;:=\;
   \mathrm{Functors}
   \big(
     G \mathrm{Orb}^{\mathrm{op}}
     ,\,
     \mathrm{SmoothManifolds}
   \big)
 \end{equation}

  \noindent
 for the category of contravariant functors
 from $G$-orbits (Def. \ref{OrdinaryOrbitCategory}) to smooth manifolds.

 \noindent {\bf (ii)} By Lemma \ref{FixedLociOfProperSmoothActionsAreSmoothManifolds},
 the system of fixed loci \eqref{SystemOfFixedLoci}
 of a smooth action
 $G \acts \; X$ (Def. \ref{ProperGActionsOnSmoothManifolds})
 takes values in smooth manifolds
  \begin{equation}
    \label{SystemOfSmoothSubmanifoldsofProperSmoothGAction}
    \hspace{-2mm}
    \mbox{
      \rm
      $G \acts \; X$
      smoothly
    }
    \phantom{A}
    \Rightarrow
    \phantom{A}
    \orbisingular
    \big(
      X \!\sslash\! G
    \big)
    \;:\;
    \xymatrix@C=15pt{
      G \mathrm{Orb}^{\mathrm{op}}
      \ar@{-->}[r]
      &
    \;  \mathrm{SmoothManifolds}
    \;  \ar[r]
      &
     \; \mathrm{TopSp}
    },
  \end{equation}

 \noindent and hence
 witnesses an object
 $\orbisingular \big( X \!\sslash\! G  \big) \in
 \GOrbifolds$ \eqref{GOrbifolds}
 which is a smooth geometric refinement of the
 underlying equivariant homotopy type (Def. \ref{EquivariantShape}),
 in that we have the following commuting diagram of functors:
 $$
 \hspace{2mm} 
   \xymatrix{
     \GActionsOnSmoothManifolds
     \ar[d]_-{
       \mathclap{\phantom{\vert^{\vert}}}
       \mbox{\hspace{-5mm} 
         \tiny
         \def\arraystretch{1}
         \begin{tabular}{c}
           \color{greenii}
           \bf
           forget smooth structure
           \\
           \eqref{SystemOfSmoothSubmanifoldsofProperSmoothGAction}
         \end{tabular}
       }
       \mathclap{\phantom{\vert_{\vert}}}
     }
     \ar[rrrr]^-{
       \scalebox{.7}{$
       G \acts \; X
       \;\;\longmapsto\;\;
       \orbisingular
       (
         X \!\sslash\! G
       )
       $}
     }
     &&
     &&
     \GOrbifolds
     \ar[d]_-{
       \,\raisebox{0pt}{\textesh}
     \;\;  }^-{\hspace{-3mm} 
             \mbox{
           \tiny
           \def\arraystretch{1}
           \begin{tabular}{c}
             \color{greenii}
             \bf
             equivariant shape
             \\
             (Def. \ref{EquivariantShape})
           \end{tabular}
         }
       }
     \\
     \GActionsOnTopSp
     \ar[rrrr]^-{
       \scalebox{.7}{$
       G \acts \; X
       \;\;\longmapsto\;\;
       \raisebox{1pt}{\textesh}
       \orbisingular
       (
         X \!\sslash\! G
       )
       $}
     }_-{
      \mbox{
        \tiny
        \bf
        \eqref{SingularEquivariantSimplicialSet}
      }
     }
     &&&&
     \EquivariantHomotopyTypes
     \,.
    }
 $$

\end{defn}

\subsection{Equivariant non-abelian cohomology theories}
 \label{EquivariantNonAbelianCohomologyTheories}

We introduce the general concept of
equivariant non-abelian cohomology theories,
in direct generalization of \cite[\S 2.1]{FSS23-Char},
and consider some examples. This is in preparation for the
twisted case in the next subsection.

\medskip

In equivariant generalization of \cite[\S 2.1]{FSS23-Char},
we set:

\begin{defn}[Equivariant non-abelian cohomology]
  \label{EquivariantNonAbelianCohomology}
  Let $ \mathscr{X}, \, \mathscr{A} \,\in\, \EquivariantHomotopyTypes$
  (Def. \ref{EquivariantHomotopyTypes}).

  \noindent
  {\bf (i)} The \emph{proper $G$-equivariant non-abelian cohomology}
  of $\mathscr{X}$ with coefficients in $\mathscr{A}$ is the hom-set
  \eqref{HomSets}
  $$
    \overset{
      \mathclap{
      \raisebox{3pt}{
        \tiny
        \color{darkblue}
        \bf
        \def\arraystretch{1}
        \begin{tabular}{c}
          equivariant
          \\
          non-abelian cohomology
        \end{tabular}
      }
      }
    }{
    H
    \big(
      \mathscr{X}
      ;
      \,
      \mathscr{A}
    \big)
    }
    \;\;
      :=
    \;\;
    \EquivariantHomotopyTypes
    \big(
      \mathscr{X}
      \,,\,
      \mathscr{A}
    \big)
    \,.
  $$

  \noindent
  {\bf (ii)}
  For $X \in \GActionsOnTopSp$
  (Def. \ref{GActionsOnTopSp}),
  with induced equivariant homotopy type
  $\raisebox{1pt}{\textesh} \orbisingular \big( X \!\sslash\! G\big)$
  \eqref{TopologicalShapeAsLocalization},
  we write

  \vspace{-.4cm}
  $$
    \overset{
      \mathclap{
      \raisebox{3pt}{
        \tiny
        \color{darkblue}
        \bf
        \def\arraystretch{1}
        \begin{tabular}{c}
          equivariant
          \\
          non-abelian cohomology
        \end{tabular}
      }
      }
    }{
      H_G
      \big(
        X;
        \,
        \mathscr{A}
      \big)
    }
    \;\;
    :=\;
    H\Big(
      \raisebox{1pt}{\textesh}
      \orbisingular
      \big(
        X \!\sslash\! G
      \big)
      ;
      \,
      \mathscr{A}
    \Big)
    \;:=\;
    \EquivariantHomotopyTypes
    \big(
      \raisebox{1pt}{\textesh}
      \orbisingular
      \big(
        X \!\sslash\! G
      \big)
      \,,\,
      \mathscr{A}
    \big)
    \,.
  $$
  \vspace{-.45cm}

  \noindent
  {\bf (iii)} We call the corresponding contravariant functor

  \vspace{-3.5mm}
  \begin{equation}
    \label{CohomologyTheory}
    \xymatrix@C=30pt@R=.5em{
      \GActionsOnTopSp^{\mathrm{op}}
      \ar@/_1.4pc/[rrrr]|-{
        \;
        H_G
        (
          -;
          \,
          \mathscr{A}
        )
        \;
      }
      \ar[rr]^-{
        \raisebox{0pt}{\textesh}
        \,
        \orbisingular
        (
          - \!\sslash\! G
        )
      }
      &&
      \EquivariantHomotopyTypes^{\mathrm{op}}
      \ar[rr]^-{
        H(
          -;
          \,
          \mathscr{A}
        )
      }
      &&
      \mathrm{Sets}
    }
  \end{equation}

  \vspace{-0.5mm}
  \noindent
 the \emph{equivariant non-abelian cohomology theory}
  with coefficients in $\mathscr{A}$.
\end{defn}

\noindent {\bf Equivariant ordinary cohomology.}

\vspace{-1mm} 
\begin{example}[Equivariant representation ring]
  \label{EquivariantRepresentationRing}
  For $H$ a finite group and $\mathbb{F}$ a field, write
  \begin{equation}
    \label{RepresentationRing}
    \RepresentationRing{\mathbb{F}}(X)
    \;\;
    \in
    \;\;
    \xymatrix{
      \mathrm{Rings}
      \ar[r]
      &
      \mathrm{AbelianGroups}
    }
  \end{equation}
  \vspace{-.6cm}

  \noindent
  for the additive abelian group underlying
  the representation ring of $H$ (i.e., the Grothendieck group of the
  semi-group of finite-dimensional $\mathbb{F}$-linear $H$-representations
  under tensor product of representations, review in \cite[\S 2.1]{BSS19}).
  Under the evident restriction of representations to subgroups
  and under conjugation action on representations,
  these groups arrange into a contravariant functor
  on the $G$-orbit category (Def. \ref{OrdinaryOrbitCategory})
  \begin{equation}
    \label{TheEquivariantRepresentationRing}
    \hspace{7mm} 
    \xymatrix@R=-4pt{
      \mathllap{
        \EquivariantRepresentationRing{\mathbb{F}}
        \;
        :
        \;
      }
      G \mathrm{Orb}^{\mathrm{op}}
      \ar[rr]
      &&
      \mathrm{AbelianGroups}
      \\
      G/H
      \ar@{}[rr]|-{ \longmapsto }
      &&
      \RepresentationRing{\mathbb{F}}(H)
    }
    \;\;\;
    \in
    \;
    \EquivariantAbelianGroups
  \end{equation}

  \noindent
  and hence constitute an equivariant abelian group
  (Def. \ref{EquivariantGroups}).
\end{example}

\begin{example}[Bredon cohomology {\cite[p. 3]{Bredon67a}\cite[Thm. 2.11 \& (6.1)]{Bredon67b}\cite[p. 10]{GreenleesMay95}}]
  \label{BredonCohomology}
  $\,$

  \noindent
  Given
  $
    \underline{A} \;\in\; \EquivariantAbelianGroups
  $
  (Def. \ref{EquivariantGroups})
  and $n \in \mathbb{N}$:

 \noindent
 {\bf (i)}   There is the \emph{Eilenberg-MacLane $G$-space}
  \begin{equation}
    \label{EquivariantEilenbergMacLaneSpace}
    \mathscr{K}(\underline{A},n)
    \;\in\;
    \EquivariantHomotopyTypes
  \end{equation}

  \noindent
  in equivariant connected homotopy types (Def. \ref{EquivariantHomotopyTypes}),
  characterized by the fact that it admits a fixed point
  with equivariant homotopy groups (Def. \ref{EquivariantHomotopyGroups})
  given by
  \vspace{-2mm} 
  $$
    \xymatrix{
      \underline{\pi}_{\; k}
      \big(
        \mathscr{K}(\underline{A},n)
      \big)
      \;\simeq\;
      \left\{
      \!\!
      {\begin{array}{lcl}
        \underline{A} &\vert& k = n,
        \\[-4pt]
        0 & \vert & \mbox{otherwise}.
      \end{array}}
      \!\!
      \right.
    }
  $$

  \vspace{-1mm} 
  \noindent
  {\bf (ii)}
  The \emph{ordinary equivariant cohomology}
  or \emph{Bredon cohomology} in degree $n$
  of $X \;\in\; \GActionsOnTopSp$
  (Def. \ref{GActionsOnTopSp})
  with coefficients in $\underline{A}$ is
  its equivariant non-abelian cohomology (Def. \ref{EquivariantNonAbelianCohomology})
  with coefficients in
  $\mathscr{K}(\underline{A},n)$ \eqref{EquivariantEilenbergMacLaneSpace}:
  $$
  \hspace{13mm} 
    \overset{
      \mathclap{
      \raisebox{3pt}{
        \tiny
        \color{darkblue}
        \bf
        \def\arraystretch{1}
        \begin{tabular}{c}
          Bredon cohomology
          \\
          (equivariant ordinary cohomology)
        \end{tabular}
      }
      }
    }{
    H^n_G
    \big(
      X;
      \,      \underline{A}
    \big)
    }
    \qquad 
    \qquad
    \simeq
    H_G
    \big(
      X;
      \,
      \mathscr{K}(\underline{A},n)
    \big)
      =
    H
    \big(
      \raisebox{1pt}\orbisingular
      (
        X \!\sslash\! G
      )
      \,
      ,\,
      \mathscr{K}(\underline{A}, n)
    \big).
  $$

\end{example}

%
%
%
%
%
%
%
%
%
%

\medskip

\noindent {\bf Equivariant Cohomotopy.}

\vspace{-1mm} 
\begin{example}[Equivariant non-abelian Cohomotopy {\cite[\S 8.4]{tomDieck79}\cite{Peschke94}\cite{Cruickshank03} \cite{SS19a}}]
  For $G \acts \; V$ a linear $G$-representation
  on a finite-dimensional real vector space $V$, the
  \emph{representation sphere} (e.g. \cite[Ex. 1.1.5]{Blu17})
  $$
    S^V
    \;:=\;
    V^{\mathrm{cpt}}
    \;\;
    \in
    \;\;
    \xymatrix{
      \GActionsOnTopSp
      \ar[rrr]^-{
        \scalebox{.7}{$
          \raisebox{1pt}{\textesh}
          \,
          \orbisingular
          \big(
            - \sslash G
          \big)
        $}
      }
      &&&
      \EquivariantHomotopyTypes
    }
  $$

  \noindent
  defines an equivariant homotopy type \eqref{TopologicalShapeAsLocalization}.
  This is the coefficient
  space for the equivariant non-abelian cohomology theory
  (Def. \ref{EquivariantNonAbelianCohomology})
  called (unstable) {\it equivariant Cohomotopy}
  in $\mathrm{RO}$-degree $V$:
  $$
    \overset{
      \mathclap{
      \raisebox{3pt}{
        \tiny
        \color{darkblue}
        \bf
        \def\arraystretch{1}
        \begin{tabular}{c}
          equivariant
          \\
          Cohomotopy
        \end{tabular}
      }
      }
    }{
      \pi^V_G(X)
    }
    \;\;
      :=
    \;\;
    H_G
    \Big(
      X;
      \,
      \raisebox{1pt}
      \orbisingular
      \big(
        S^V \!\sslash\! G
      \big)
    \!\Big)
    \;\;
      \simeq
    \;\;
    H
    \Big(\!\!
      \raisebox{1pt}
      \orbisingular
      \big(
        X \!\sslash\! G
      \big)
      ;
      \,
      \raisebox{1pt}
      \orbisingular
      \big(
        S^V \!\sslash\! G
      \big)
    \!\Big).
  $$
\end{example}

\medskip

\noindent {\bf Equivariant non-abelian cohomology operations.}

\begin{defn}[Equivariant non-abelian cohomology operations]
  \label{EquivariantNonabelianCohomologyOperation}
  For $\mathscr{A}, \, \mathscr{B} \,\in\, \EquivariantHomotopyTypes$
  (Def. \ref{EquivariantHomotopyTypes}),
  a {\it cohomology operation} from equivariant non-abelian
  $\mathscr{A}$-cohomology to $\mathscr{B}$-cohomology (Def. \ref{EquivariantNonAbelianCohomology})
  is a natural transformation
  \vspace{-1mm}
  $$
    \xymatrix{
      H
      (
        -;
        \,
        \mathscr{A}
      )
      \ar[r]^-{ \phi_\ast }
      &
      H
      (
        -;
        \,
        \mathscr{B}
      )
    }
  $$
  \vspace{-.5cm}

  \noindent
  of the corresponding equivariant non-abelian cohomology theories
  \eqref{CohomologyTheory}.
  By the Yoneda lemma, such operations are induced by post-composition
  with morphisms between equivariant coefficient spaces:
  \vspace{-1mm} 
  \begin{equation}
    \label{MorphismInducingCohomologyOperation}
    \xymatrix{ \mathscr{A} \ar[r]^-{\phi} & \mathscr{B} }
    \;\;\;\;\;
    \in
    \;
    \EquivariantHomotopyTypes
    \,.
  \end{equation}

\end{defn}

\subsection{Equivariant twisted non-abelian cohomology theories}
  \label{TwistedEquivariantNonAbelianCohomologyTheories}

We introduce equivariant twisted non-abelian cohomology,
in direct generalization of \cite[\S 2.2]{FSS23-Char},
and introduce the main example of interest here
(Def. \ref{EquivariantTwistorialCohomotopyTheory} below).

\medskip

\noindent {\bf Equivariant $\infty$-Actions.}

\vspace{-1mm} 
\begin{remark}[Equivariant $\infty$-actions]
  \label{EquivariantInfinityAction}
  {\bf (i)}
  In equivariant generalization of
  Prop. \ref{ClassifyingSpaceLoopSpaceConstruction}
  (and as a special case of \cite[Thm. 2.19]{NSS12a}\cite[Thm. 3.30, Cor. 3.34]{NSS12b}),
  every equivariantly pointed and equivariantly connected equivariant
  homotopy type
  (Def. \ref{EquivariantConnectedHomotopyTypes})
  is, equivalently, the equivariant classifying space $B \mathscr{G}$
  of an {\it equivariant $\infty$-group}

  \vspace{-4mm}
  $$
    \mathscr{G}
    \;\;
      \in
    \;\;
    \orbisingularGLarge\mathrm{EquivariantGroups}_{\infty}
    \;:=\;
    \mathrm{Ho}
    \Big(
      \mathrm{Functors}
      \big(
        G\mathrm{Orb}^{\mathrm{op}}
        \,,\,
        \mathrm{SmplGrp}
      \big)_{\mathrm{proj}}
    \Big)
    \,.
  $$

  \noindent
  {\bf (ii)}   In equivariant generalization of
  Prop. \ref{InfinityActionsEquivalentToFibrationsOverClassifyingSpace}
  (and as a special case of \cite[\S 4]{NSS12a}\cite[\S 2.2]{SS20b}),
  $\infty$-actions of such equivariant $\infty$-groups on
  equivariant homotopy types $\mathscr{A}$ are, equivalently,
  homotopy fibrations
  of equivariant homotopy types
  over $B\mathscr{G}$ with homotopy fiber
  $\mathscr{A}$, hence a system of
  non-equivariant homotopy fibration \eqref{HomotopyFibrationsCorrespondingToInfinityActions}
  parametrized by the $G/H \,\in\, G\mathrm{Orb}$ (Def. \ref{OrdinaryOrbitCategory}),
  denoted as follows \footnote{
    Here and in the following we indicate the ambient category
    of a given diagram. The notation ``$\mathrm{Diagram} \in \mathrm{Category}$''
    means that each vertex of the diagram is an object in that
    category, and each arrow is a morphism in that category.
  }
  \begin{equation}
    \label{EquivariantUniversalFibrationOverEquivariantClassifyingSpace}
    \begin{array}{lcl}
    &
    \raisebox{20pt}{
    \xymatrix@R=30pt@C=5.5em{
      \mathscr{A}
      \ar[rr]^-{
        \scalebox{.7}{$
        \mathrm{hofib}
        (
          \rho_{\mathscr{A}}
        )
        $}
      }
      &
      \ar@{}[d]|-{
        \mathclap{
        \scalebox{.65}{
          \color{darkblue}
          \bf
          \def\arraystretch{1}
          \begin{tabular}{c}
            equivariant homotopy fibration
            \\
            associated to
            $\infty$-action of $\mathscr{G}$ on $\mathscr{A}$
          \end{tabular}
        }
        \;\;\;\;\;\;\;\;
        }
      }
      &
      \mathscr{A}\!\sslash\!\!\mathscr{G}
      \ar[d]^-{
        \scalebox{.7}{$
        \rho_{\mathscr{A}}
        $}
      }
      \\
      &&
      B \mathscr{G}
    }
    }
    \\
    G/H
    \;\;
    \longmapsto
    &
    \raisebox{20pt}{
    \xymatrix@R=30pt@C=4.2em{
      \mathscr{A}(G/H)
      \ar[rr]^-{
        \scalebox{.7}{$
        \mathrm{hofib}
        (
          \rho_{\mathscr{A}}(G/H)
        )
        $}
      }
      &
      \ar@{}[d]|-{
        \mathclap{
        \mbox{
          \tiny
          \color{darkblue}
          \bf
          \def\arraystretch{1}
          \begin{tabular}{c}
            homotopy fibration
            \\
            associated to
            $\infty$-action 
            \\
            of $\mathscr{G}(G/H)$ on $\mathscr{A}(G/H)$
          \end{tabular}
        }
        \;\;\;\;\;\;\;\;
        }
      }
      &
      \mathscr{A}(G/H)
      \!  \sslash \!\!
      \mathscr{G}(G/H)
      \ar[d]^-{
        \scalebox{.7}{$
        \rho_{\mathscr{A}}(G/H)
        $}
      }
      \\
      &&
      B \mathscr{G}(G/H)
    }
    }
    \end{array}
  \end{equation}
\end{remark}

A key source of equivariant $\infty$-actions are
equivariant parametrized homotopy types, in the following
sense:

\begin{example}[Equivariant parametrized homotopy types]
  \label{EquivariantParametrizedHomotopyTypes}
  $\,$

  \noindent
  Consider
  $T \,\in\, \mathrm{CompactTopGrp}$
  \eqref{CompactTopologicalGroup},
  $G \in$ $\mathrm{FiniteGroups}$
  \eqref{EquivarianceGroup},
  \newline
  and
  $
    X \;\in\;
    \TGActionsOnTopSp
  $
  \eqref{CategoryOfTGActionsOnTopSp}.

  \noindent
  {\bf (i)}
  Since the two group actions separately commute with each
  other,
  we may consider forming the combined
  \begin{itemize}[
    leftmargin=.9cm
  ]
  \item[{\bf (a)}] 
  proper equivariant shape
  (Def. \ref{EquivariantShape}) with respect to the
  $G$-action;

  \item[{\bf (b)}] ordinary shape \eqref{TopologicalShapeAsLocalization}
  of the homotopy quotient
  (Borel construction, Ex. \ref{HomotopyTypeOfBorelConstruction})
  with respect to the $T$-action:
  \begin{equation}
    \label{EquivariantParametrizedSpaceHomotopyTypeFromGTAction}
    \EquivariantHomotopyTypes
    \;\;
      \ni
    \;\;
    \Big(\!
      \big(
        \orbisingular
        (
          X \!\sslash\! G
        )
     \! \big)
      \!\sslash\!
      T
    \Big)
    \;;\:\;\;
    G/H
    \;\;
    \longmapsto
    \;\;
    \raisebox{1pt}{\textesh}
    \,
    \big(
      X^H \!\sslash\! T
    \big).
  \end{equation}
  \end{itemize}

  \noindent
  This is the $G$-equivariant homotopy type
  (Def. \ref{EquivariantHomotopyTypes})
  given on $G/H \,\in\, G\mathrm{Orb}$ (Def. \ref{OrdinaryOrbitCategory})
  by the Borel homotopy quotient construction
  (Example \ref{HomotopyTypeOfBorelConstruction})
  of the $T$-action on the $G \,\supset\, H$-fixed locus
  (Example \ref{SystemsOfFixedPointSpaces}).

  \noindent
  {\bf (ii)}
  With the classifying space $B T$ regarded as
  a smooth $G$-equivariant homotopy type
  (i.e., with trivial $G$-action, Example \ref{SmoothSingularHomotopyTypes})
  the $G$-equivariant $T$-parametrized space
  \eqref{EquivariantParametrizedSpaceHomotopyTypeFromGTAction}
  sits in an equivariant fibration \eqref{EquivariantUniversalFibrationOverEquivariantClassifyingSpace}
  over $B T$ with homotopy fiber the $G$-equivariant shape of
  $X$ (Def. \ref{EquivariantShape}):
  $$
    \hspace{-2.7cm}
    \begin{array}{ll}
    &
    \raisebox{25pt}{
    \xymatrix@R=15pt@C=21pt{
      \raisebox{1pt}{\textesh}
      \,
      \orbisingular
      \big(
        X \!\sslash\! G
      \big)
      \ar[rrr]^-{
        \mathrm{hofib}
        \left(
          \rho_{
            \scalebox{.6}{$
              \orbisingular (X \sslash G)
            $}
          }
        \right)
      }
      &&&
      \raisebox{1pt}{\textesh}
      \,
      \Big(\!\!
      \big(
        \orbisingular
        (
          X \!\sslash\! G
        )
      \big)
      \!\sslash\!
      T
      \Big)
      \ar[d]^-{
        \rho_{
          \scalebox{.5}{$
            \orbisingular (X \sslash G)
          $}
        }
      }
      \\
      &&&
      B T
    }
    }
    \mathrlap{
      \;\;
        \in
      \;\;
      \EquivariantHomotopyTypes
    }
    \\
    G/H
    \; \longmapsto
    &
    \raisebox{25pt}{
    \xymatrix@C=31pt@R=1.5em{
      \raisebox{1pt}{\textesh}
      \,
      X^H
      \ar[rr]^-{
        \mathrm{hofib}
        (
          \rho_{
            \scalebox{.5}{$
              X^H
            $}
          }
        )
      }
      &&
      \raisebox{1pt}{\textesh}
      \,
      \big(
        X^H \!\sslash\! T
      \big)
      \ar[d]^-{
        \rho_{
          \scalebox{.5}{$
            X^H
          $}
        }
      }
      \\
      &&
      B T
    }
    }
    \mathrlap{
      \;\;
        \in
      \;\;
      \;\;\;\;\;
      \HomotopyTypes
    }
    \end{array}
  $$

  \noindent We may refer to these objects as
  {\it proper $G$-equivariant and Borel $T$-equivariant homotopy types },
  but for brevity and due to their above fibration over
  $B T$, we will say
  {\it $G$-equivariant $T$-parametrized homotopy types}.

\end{example}

\begin{example}[$\Grefl$-equivariant $\SpLR$-parametrized twistor fibration]
 \label{GhetEquivariantParametrizedTwistorSpace}

Recall the $\Grefl$-equivariant twistor fibration \eqref{TwistorFibration}
from Example \ref{GhetEquivariantTwistorSpace}.
Since the $\mathrm{Sp}(2)$-subgroups $\Grefl$ \eqref{Ghet}
and $\SpLR$ \eqref{TheSp1Subgroup}
commute with each other, the quotient by the action of $\SpLR$
of the
Cartesian product of the twistor fibration
\eqref{TwistorFibration} with (the identity map on)
the total space $E \mathrm{Sp}(2)$
of the universal principal  $\mathrm{Sp}(2)$-bundle
still has a residual equivariance under $\Grefl$:
\vspace{-.2cm}
\begin{equation}
  \label{TwistorFibrationTimesESp2ModSpLR}
  \begin{array}{c}
  \raisebox{30pt}{
  \xymatrix@C=35pt@R=1.5em{
    \frac{
      S^2 \times E \mathrm{Sp}(2)
    }{
      \SpLR
    }
    \ar[d]
    \ar[rr]^-{
      \;
      \frac{
        \mathclap{\phantom{\vert}}
        \mathrm{fib}(t_{\scalebox{.45}{$\mathbb{H}$}})
        \times
        \mathrm{id}
      }{
        \mathclap{\phantom{\vert}}
        \SpLR
      }
      \;
    }
   &&
    \frac{
      \mathclap{\phantom{\vert}}
      \mathbb{C}P^3 \times E \mathrm{Sp}(2)
    }{
      \mathclap{\phantom{\vert}}
      \SpLR
    }
    \ar@(ul,ur)@<-0pt>^-{ \Grefl\!\! }
    \ar[rr]^-{
      \;
      \frac{
        \mathclap{\phantom{\vert}}
        \overset{
        \mathclap{
        \raisebox{-3pt}{
          $
          \!\!\!\!
          \mathrlap{
          \rotatebox[origin=l]{23}{
            \tiny
            \color{greenii}
            \bf
            twistor fibration
          }
          }$
        }
        }
        }{
          t_{\scalebox{.45}{$\mathbb{H}$}}
        }
        \times
        \mathrm{id}
      }{
        \mathclap{\phantom{\vert}}
        \SpLR
      }
      \;
    }
    \ar[d]^<<<<<{
    }
    &&
    \frac{
      S^4 \times E \mathrm{Sp}(2)
    }{
      \SpLR
    }
    \ar@(ul,ur)^-{ \Grefl\!\! }
    \ar[d]
    \\
    \frac{
      E \mathrm{Sp}(2)
    }{
      \SpLR
    }
    \ar@{=}[rr]
    &&
    \frac{
      E \mathrm{Sp}(2)
    }{
      \SpLR
    }
    \ar@{=}[rr]
    &&
    \frac{
      E \mathrm{Sp}(2)
    }
    {
      \SpLR
    }
  }
  }
  \\
  \in
  \;
  \Grefl\mathrm{Actions}
  \big(
    \mathrm{TopSp}
  \big)^{\scalebox{.6}{$\left/\frac{E \mathrm{Sp}(2)}{\SpLR}\right.$}}
  \end{array}
\end{equation}

\noindent
Hence, using Example \ref{GhetEquivariantTwistorSpace}
and identifying the Borel construction of homotopy quotients
(e.g. \cite[Prop. 3.73]{NSS12b}, here for subgroups $H \subset G$):
\begin{equation}
  \label{BorelConstruction}
  \hspace{10mm}
  \underset{
    \mathclap{
    \raisebox{-3pt}{
      \tiny
      \color{darkblue}
      \bf
      \begin{tabular}{c}
        Borel
        construction
      \end{tabular}
    }
    }
  }{
    \frac{X \times E G}{H}
  }
  \;\;\;\;
    \simeq
  \;\;\;\;
  \underset{
    \mathclap{
    \raisebox{-3pt}{
      \tiny
      \color{darkblue}
      \bf
      \def\arraystretch{1}
      \begin{tabular}{c}
        homotopy
        \\
        quotient
      \end{tabular}
    }
    }
  }{
    X \!\sslash\! H
  }
  \;\;\;\;
  \in
  \;
  \HomotopyTypes
  \,,
\end{equation}

\noindent
the $\Grefl$-equivariant homotopy type (Def. \ref{EquivariantHomotopyTypes})
of the middle vertical
morphism in \eqref{TwistorFibrationTimesESp2ModSpLR}
exhibits a
$\Grefl$-equivariant $\SpLR$-parametrized homotopy type
(in the sense of Example \ref{EquivariantParametrizedHomotopyTypes})
of this form:
\begin{equation}
  \label{GHetEquivariantSpLRParametrizedTwistorSpace}
  \hspace{-12mm}
  \adjustbox{scale=0.88}{$
  \raisebox{49pt}{
  \xymatrix@R=1.1pt@C=-125pt{
    &
    &&
    &&
    \raisebox{1pt}{\textesh}\,
    \mathbb{C}P^3 \!\sslash\! \SpLR
    \ar@(ul,ur)@<+20pt>|-{\, \Grefl }
    \ar@{<-^{)}}[dddd]_-{
      \scalebox{.6}{$
        \mathrm{fib}(t_{\mathbb{H}})
        \sslash \SpLR
      $}
     }
    \ar[dddr]^-{
       \rho_{\scalebox{.53}{$\raisebox{1pt}{\textesh}\,\mathbb{C}P^3$}}
    }
    \\
    \raisebox{1pt}{\textesh}
    \big(
    \orbisingular
    \big(
      \mathbb{C}P^3
      \!\sslash\!
      \Grefl
    \big)
    \big)
    \!\sslash\!
    \SpLR
    \ar@{}[rr]|<<<<<{:}
    \ar[dddr]^<<<<<<<<<<<{
     \rho_{
           \scalebox{.65}{$
                \scalebox{.85}{\raisebox{1pt}{\textesh}}
                \orbisingular
                \big(
                  \mathbb{C}P^3 \sslash \Grefl
                \big)
                $}
              }
     }_<<<<{
       \mbox{\tiny
          \bf
          \color{darkblue}
               \begin{tabular}{c}
          $\Grefl$-equivariant \&
          \\
          $\SpLR$-parametrized
          \\
          twistor space
          \\
          $\phantom{A}$
                  \end{tabular}
       }
    }
    &
    &&
    \ZTwo/1
    \ar@(ul,ur)|-{\; \ZTwo\,}
    \ar[dddd]
    \ar@{}[rr]|{\hspace{-1cm}
      \longmapsto
    }
    &&
    \\
    &
    &&
    &&
    \\
    &
    &
    &
    &&
    &
    \raisebox{1pt}{\textesh}
    \,
    B \SpLR
    \ar@{=}[dddd]
    \\
    &
    \raisebox{1pt}{\textesh}
    \left(
    \orbisingular
    \big(
      \ast
      \!\sslash
      \Grefl
    \big)
    \right)
    \!\sslash\!
    \SpLR
    \ar@{}[rr]|<{:}
    &&
    &&
    \mathclap{\phantom{\vert^{\vert^{\vert}}}}
    \raisebox{1pt}{\textesh}
    \,
    S^2
    \!\sslash\!
    \SpLR
    \ar[dddr]_-{
          \rho_{\scalebox{.53}{$\raisebox{1pt}{\textesh}\,S^2$}}
    }
    \\
    &
    &&
    \ZTwo/\ZTwo
    \ar@{}[rr]|-{\hspace{-1.5cm}
      \longmapsto
    }
    &&
    \\
    \\
    &
    &&
    &&
    &
    \raisebox{1pt}{\textesh}
    \,
    B \SpLR\,.
    \\
    {\phantom{AAAAAAAAAAAAAAAAAAAAAA}}
    &
    {\phantom{AAAAAAAAAAAAAAAAAAAA}}
    &
    {\phantom{AAAAAAAAAAAAAAAAAAAA}}
    &
    {\phantom{AAAAAAAAAAAAAAAAAAAA}}
    &
    {\phantom{AAAAAAAAAAAAAAAAAAAA}}
    &
    {\phantom{AAAAAAAAAAAAAAAAAAAA}}
    &
    {\phantom{AAAAAAAAAAAAAAAAAAAA}}
  }
  }
$}
  \hspace{-1cm}
\end{equation}

\noindent
The analogous statement holds for the
vertical morphism on the right of \eqref{TwistorFibrationTimesESp2ModSpLR},
so that the full square on the right of
\eqref{TwistorFibrationTimesESp2ModSpLR}
exhibits a morphism in
$\Grefl$-equivariant $\SpLR$-parametrized homotopy types
(Example \ref{EquivariantParametrizedHomotopyTypes})
of this form:
\begin{equation}
\label{EquivariantParametrizedTwistorFibration}
\hspace{-8mm} 
  \begin{array}{c}
  \raisebox{10pt}{
  \xymatrix@R=8pt@C=28pt{
    \overset{
      \mathclap{
      \raisebox{3pt}{
        \tiny
        \color{darkblue}
        \bf
        \def\arraystretch{1}
        \begin{tabular}{c}
          $\Grefl$-equivariant
          \\
          $\SpLR$-parametrized
          \\
          twistor space
        \end{tabular}
      }
      }
    }{
    \raisebox{1pt}{\textesh}
    \big(
    \orbisingular
    \big(
      \mathbb{C}P^3 \!\sslash\! \Grefl
    \big)
    \big)
    \sslash
    \SpLR
    }
    \ar[rr]^-{
        \mathclap{
        \raisebox{0pt}{
          \tiny
          \color{greenii}
          \bf
          \def\arraystretch{1}
          \begin{tabular}{c}
            $\Grefl$-equivariant
            \\
            $\SpLR$-parametrized
            \\
            twistor fibration
          \end{tabular}
        }
        }}_-{
      \scalebox{.7}{$
        \raisebox{1pt}{\textesh}
        \,
        \orbisingular
        \big(
          t_{\mathbb{H}}
          \!\sslash\!
          \Grefl
        \big)
        \sslash
        \SpLR
      $}
      }
       \ar[dr]
    &&
    \overset{
      \raisebox{3pt}{
        \tiny
        \color{darkblue}
        \bf
        \def\arraystretch{1}
        \begin{tabular}{c}
          $\Grefl$-equivariant
          \\
          $\SpLR$-parametrized
          \\
          4-sphere
        \end{tabular}
      }
    }{
    \raisebox{1pt}{\textesh}
    \big(
    \orbisingular
    \big(
      S^4 \!\sslash\! \Grefl
    \big)
    \big)
    \sslash
    \SpLR
    }
    \ar[dl]
    \\
    &
    B \SpLR
  }
  }
  \\
  \in
  \mathrm{Ho}
  \Big(
    \ZTwoEquivariantSSet^{\scalebox{.6}{$/\raisebox{1pt}{\textesh}B \SpLR$}}_{\mathrm{proj}}
  \Big)
  \,,
  \end{array}
\end{equation}

\vspace{1mm} 
\noindent
where $B \SpLR \;:=\; \mathrm{Smth} \, \raisebox{0pt}{\textesh} \, B \SpLR $
(Example \ref{SmoothSingularHomotopyTypes}).

\end{example}

\medskip

\noindent{\bf Twisted equivariant non-abelian cohomology.}

In twisted generalization of Def. \ref{EquivariantNonAbelianCohomology}
and in equivariant generalization of \cite[\S 2.2]{FSS23-Char},
we set:

\begin{defn}[Twisted equivariant non-abelian cohomology]
  \label{EquivariantTwistedNonAbelianCohomology}
  Let
  \begin{equation}
    \label{EquivariantLocalCoefficientBundle}
    \raisebox{20pt}{
    \xymatrix@R=1.5em{
      \mathscr{A}
      \ar[rr]^-{
        \scalebox{.7}{$
        \mathrm{hofib}
        (
          \rho_{\mathscr{A}}
        )
        $}
      }
      &
      \ar@{}[d]|-{
        \mbox{
          \tiny
          \color{darkblue}
          \bf
          \def\arraystretch{1}
          \begin{tabular}{c}
            equivariant
            \\
            local coefficient
            \\
            bundle
          \end{tabular}
        }
      }
      &
      \mathscr{A}\!\sslash\!\!\mathscr{G}
      \ar[d]^-{
        \scalebox{.7}{$
        \rho_{\mathscr{A}}
        $}
      }
      \\
      &&
      B \mathscr{G}
    }
    }
    \;\;\;\;\;\;\;
    \in
    \EquivariantHomotopyTypes
  \end{equation}

  \noindent
  be an homotopy fibration as in Remark \ref{EquivariantInfinityAction},
  to be regarded now as an
  {\it equivariant local coefficient bundle},
  and let $\mathscr{X} \,\in\, \EquivariantHomotopyTypes$
  (Def. \ref{EquivariantHomotopyTypes})
  equipped with an {\it equivariant twist}
  \begin{equation}
    \label{ATwist}
    [
      \tau
    ]
    \;\;
    \in
    \;\;
    H\big(
      \mathscr{X}\!;
      \,
      B \mathscr{G}
    \big)
  \end{equation}

  \noindent
  in
  equivariant non-abelian cohomology (Def. \ref{EquivariantNonAbelianCohomology})
  with coefficients in $B \mathscr{G}$. We say that
  the {\it $\tau$-twisted equivariant non-abelian cohomology}
  of $\mathscr{X}$ with coefficients in $\mathscr{A}$ is
  the hom-set from $\tau$ to $\rho_A$ in the
  homotopy category of the
  slice model structure
  (see \cite[Ex. A.10]{FSS23-Char})
  over  $B \mathscr{G}$
  of the projective model structure on
  equivariant simplicial sets (Prop. \ref{ModelCategoryOnEquivariantSSet}):
  $$
    \overset{
      \mathclap{
      \raisebox{3pt}{
        \tiny
        \color{darkblue}
        \bf
        \def\arraystretch{1}
        \begin{tabular}{c}
          twisted equivariant
          \\
          non-abelian cohomology
        \end{tabular}
      }
      }
    }{
    H^\tau
    \big(
      \mathscr{X}
      \!
      ;
      \,
      \mathscr{A}
    \big)
    }
    \;\;
      :=
    \;\;
    \mathrm{Ho}
    \Big(
      \EquivariantSSetProj^{\scalebox{.7}{$/B\mathscr{G}$}}
    \Big)
    \big(
      \tau
      \,,\,
      \rho_{\mathscr{A}}
    \big)
    \,.
  $$

\end{defn}

\medskip

\noindent {\bf Twisted equivariant ordinary cohomology.}

\begin{example}[Twisted Bredon cohomology]
  Let $G \acts \; X \in \GActionsOnTopSp$
  (Def. \ref{GActionsOnTopSp})
  with a base point $G \acts \; \ast  \overset{x}{\longrightarrow} G \acts \; X$, let
  $\underline{A} \,\in\, \EquivariantAbelianGroups$ (Def. \ref{EquivariantGroups}),
  and let
  $$
    r \;:\;
    \xymatrix{
      \underline{\pi}_1(X)
      \times
      \underline{A}
      \ar[r]
      &
      \underline{A}
    }
  $$

\noindent
  be an action of the equivariant fundamental group
  (Def. \ref{EquivariantHomotopyGroups}) of $X$ on
  $\underline{A}$. For $n \in \mathbb{N}$,
  there is an equivariant local
  coefficient bundle \eqref{EquivariantLocalCoefficientBundle}
  $$
    \xymatrix@R=1.5em{
      \mathscr{K}
      (
        \underline{A},
        n
      )
      \ar[rr]
      &
      \ar@{}[d]|-{
       \mathclap{
        \mbox{
          \tiny
          \color{darkblue}
          \bf
          \def\arraystretch{1}
          \begin{tabular}{c}
            equivariant ordinary
            \\
            local coefficients
          \end{tabular}
        }
        }
      }
      &
      \mathscr{K}
      (
        \underline{A},
        n
      )
      \!\sslash\!
      \underline{\pi}_1(X)
      \ar[d]^-{ \rho }
      \\
      &&
      B \underline{\pi}_1(X)
    }
  $$

  \noindent
  with typical fiber the  equivariant Eilenberg-MacLane space
  \eqref{EquivariantEilenbergMacLaneSpace},
  such that the twisted equivariant non-abelian cohomology
  with local coefficients in $\rho$
  coincides (by \cite[Cor. 3.6]{Golasinski97a}\cite[Thm. 5.10]{MukherjeeSen10})
  with traditional
  $r$-twisted Bredon cohomology in degree $n$
  (\cite[Def. 2.1]{MoerdijkSvensson93}\cite[Def. 3.8]{MukMuk96}\cite{MukherjeePandey02}):
  $$
    \overset{
      \mathclap{
      \raisebox{3pt}{
        \tiny
        \color{darkblue}
        \bf
        \def\arraystretch{1}
        \begin{tabular}{c}
          twisted
          \\
          Bredon cohomology
        \end{tabular}
      }
      }
    }{
    H_G^{n + r}
    \big(
      X;
      \,
      \underline{A}
    \big)
    }
    \;\;
    \simeq
    \;\;
    H^\tau
    \big(
      X;
      \,
      \mathscr{K}(\underline{A}, n)
    \big)
    \,.
  $$

\end{example}

\medskip

\noindent {\bf Equivariant tangential structure.}
In equivariant generalization of \cite[Example 2.33]{FSS23-Char}, we have:

\begin{defn}[Equivariant tangential structure]
  \label{EquivariantTangentialStructures}
  Let $G \acts  \; X \,\in\, \GActionsOnSmoothManifolds$
  (Def. \ref{ProperGActionsOnSmoothManifolds})
  of dimension $n :=\mathrm{dim}(X)$,
  and let $\mathcal{G} \overset{\phi}{\longrightarrow} B \mathrm{GL}(n)$
  be a topological group homomorphism.
  An {\it equivariant tangential $(\mathcal{G},\phi)$-structure}
  (or just $\mathcal{G}$-structure, for short) on
  the orbifold $\orbisingular\big( X \!\sslash\! G \big)$
  (Def. \ref{GOrbifolds}) is a
  class in the equivariant twisted non-abelian
  cohomology (Def. \ref{EquivariantTwistedNonAbelianCohomology})
  of the equivariant shape (Def. \ref{EquivariantShape})
  of the orbifold with equivariant local coefficients
  \eqref{EquivariantLocalCoefficientBundle}
  in
  $$
    \xymatrix@R=1em{
      \mathrm{GL}(n)\sslash \mathcal{G}
      \ar[rr]
      &&
      B \,\mathcal{G}
      \ar[d]^-{ B \phi }
      \\
      && B \mathrm{GL}(n)
    }
  $$

  \noindent
  and with twist given by the classifying map $\tau_{\,\mathrm{Fr}}$
  of the
  frame bundle:
  $$
    (\mathcal{G},\phi)\mathrm{Structures}
    \big(
      \orbisingular
      \big(
        X \!\sslash\! G
      \big)
    \big)
    \;\;
    :=
    \;\;
    H^{\tau_{\, \mathrm{Fr}}}
    \Big(\!\!
      \orbisingular
      \big(
        X \!\sslash\! G
      \big)
      ;
      \,
      \mathrm{GL}(n) \!\sslash\! \mathcal{G}
    \Big).
  $$

\end{defn}

\medskip

\noindent {\bf Equivariant twistorial Cohomotopy.}
In equivariant generalization of \cite[Ex. 2.44]{FSS23-Char}
we have:

\begin{defn}[Equivariant twistorial Cohomotopy theory]
  \label{EquivariantTwistorialCohomotopyTheory}
  $\,$

  \noindent
  Let $X^8 \,\in\, \ZTwoActionsOnTopSp$
  (Def. \ref{GActionsOnTopSp})
  be a smooth spin 8-manifold equipped with tangential structure
  (see \cite[Ex. 2.33]{FSS19b})
  for the subgroup
  $\SpLR \subset \mathrm{Sp(2)} \subset \mathrm{Spin}(8)$
  (where the first inclusion is \eqref{Ghet}
  and the second is again given by left quaternion
  multiplication, e.g. \cite[Ex. 2.12]{FSS19b})
  $$
    [
      \tau
    ]
    \;\in\;
    H_{\Grefl}
    \big(
      X^8;
      \,
      B \SpLR
    \big)
    \,.
  $$

  \noindent  
  We say that:
  \begin{itemize}[
    leftmargin=.9cm
  ]
  \item[{\bf (a)}]  its {\it $\Grefl$-equivariant twistorial Cohomotopy}
  $\mathcal{T}^\tau_{\Grefl}(-)$
  is the $\tau$-twisted equivariant non-abelian cohomology theory
  (Def. \ref{EquivariantTwistedNonAbelianCohomology})
  with local coefficients in the
  $\Grefl$-equivariant $\SpLR$-parametrized twistor space;

  \item[{\bf (b)}]  its {\it $\Grefl$-equivariant J-twisted Cohomotopy}
  $\pi^\tau_{\Grefl}(-)$
  is the $\tau$-twisted equivariant non-abelian cohomology theory
  (Def. \ref{EquivariantTwistedNonAbelianCohomology})
  with local coefficients in the
  $\Grefl$-equivariant $\SpLR$-parametrized 4-sphere;

  \item[{\bf (c)}] the twisted equivariant cohomology operation
  $\mathcal{T}^\tau_{\Grefl}(-) \longrightarrow  \pi^\tau_{\Grefl}(-)$
  is that induced by the $\Grefl$-equivariant $\SpLR$-parametrized
  twistor fibration;

  \end{itemize}

  \noindent
  all as induced by the (morphism of)
  local coefficient bundles \eqref{EquivariantParametrizedTwistorFibration}
  in Example \ref{GhetEquivariantParametrizedTwistorSpace}:
  
  \vspace{-2mm}
  \begin{equation}
    \label{EquivariantTwistorialCohomotopyAndJTwistedCohomotopy}
    \hspace{-4mm} 
    \xymatrix@C=30pt@R=-5pt{
    H_{\Grefl}^\tau
    \Big(
      X;
      \,
      \raisebox{1pt}{\textesh}
      \orbisingular
      \big(
        \mathbb{C}P^3 \!\sslash\! \Grefl
      \big)
    \!\Big)
    \ar[rr]^-{
        \raisebox{3pt}{\bf
          \tiny
          \color{greenii}
          \def\arraystretch{1}
          \begin{tabular}{c}
            push-forward along
            \\
            equivariant parametrized
            \\
            twistor fibration
          \end{tabular}
        }
        }_-{
      \scalebox{.7}{$
        \Big(
        \raisebox{1pt}{\textesh}
        \,
        \orbisingular
        \big(
          t_{\mathbb{H}}
          \!\sslash\!
          \Grefl
        \big)
        \sslash
        \SpLR
        \Big)_\ast
      $}
      }
    &{\phantom{AAAA}}&
    H_{\Grefl}^\tau
    \Big(
      X;
      \,
      \raisebox{1pt}{\textesh}
      \orbisingular
      \big(
        S^4 \!\sslash\! \Grefl
      \big)
    \!\Big)
    \\
   \rotatebox[origin=c]{90}{$:=$}
   &&
   \rotatebox[origin=c]{90}{$:=$}
   \\
    \underset{
      \mathclap{
      \raisebox{3pt}{
        \tiny
        \color{darkblue}
        \bf
        \def\arraystretch{1}
        \begin{tabular}{c}
          equivariant
          \\
          twistorial Cohomotopy
        \end{tabular}
      }
      }
    }{
      \mathcal{T}^{\tau}_{\Grefl}(X)
    }
   &&
    \;
    \underset{
      \mathclap{
      \raisebox{3pt}{
        \tiny
        \color{darkblue}
        \bf
        \def\arraystretch{1}
        \begin{tabular}{c}
          equivariant
          \\
          J-twisted Cohomotopy
        \end{tabular}
      }
      }
    }{
      \pi^\tau_{\Grefl}\big(X\big)
    }.
    }
  \end{equation}

\end{defn}

\section{Equivariant non-abelian de Rham cohomology}
\label{EquivariantNonAbelianDeRhamCohomology}

We had shown in \cite[\S 3]{FSS23-Char}
how the fundamental theorem of dgc-algebraic rational homotopy theory
(\cite[\S 9.4, \S 11.2]{BousfieldGugenheim76}),
augmented by differential-geometric
observations \cite[\S 9]{GriffithMorgan13},
provides a non-abelian de Rham theorem for
$L_\infty$-algebra valued differential forms,
which serve as the recipient of non-abelian character maps.

\medskip
The equivariant generalization of this fundamental theorem
had been obtained in \cite{Scull08} (following \cite{Tri82})
without having found much attention yet.
Here we review, in streamlined form and highlighting examples and applications,
the underlying theory of injective equivariant
dgc-algebras/$L_\infty$-algebras in \cref{EquivariantDgcAlgebrasAndEquivariantLInfinityAlgebras}
and how these serve to model equivariant rational homotopy theory in
\cref{EquivariantRationalHomotopyTheory}.
Then
we use this in \cref{EquivariantNonAbelianDeRhamTheorem}
to prove the equivariant non-abelian de Rham theorem
(Prop. \ref{EquivariantNonabelianDeRhamTheorem})
including its twisted version (Prop. \ref{EquivariantTwistedNonabelianDeRhamTheorem});
which, in turn, we use in
\cref{TheEquivariantTwistedNonAbelianCharacterMap} to
construct the equivariant non-abelian character map
(Def. \ref{EquivariantNonAbelianCharacterMap})
and its twisted version
(Def. \ref{TwistedEquivariantNonabelianCharacterMap}).

\subsection{Equivariant dgc-algebras and equivariant $L_\infty$-algebras}
\label{EquivariantDgcAlgebrasAndEquivariantLInfinityAlgebras}

We discuss here the generalization of the homotopy theory of connective dgc-algebras and of connective $L_\infty$-algebras
(following \cite[\S 3.1]{FSS23-Char})
to $G$-equivariant homotopy theory, for any finite
equivariance group $G$ \eqref{EquivarianceGroup}.
While the homotopy theory of equivariant connective dgc-algebras
has been developed in \cite{Tri82}\cite{Scull02} \cite{Scull08},
previously little to no examples or applications have been
worked out. Here we develop equivariantized twistor space
as a running example (culminating in Prop. \ref{Z2EquivariantRelativeMinimalModelOfSpin3ParametrizedTwistorSpace} below).

\medskip
While the general form of the homotopy theory of
plain dgc-algebras generalizes to equivariant dgc-algebras,
the crucial new aspect is that equivariantly not every
connective cochain complex, and hence not every
connective dgc-algebra, is fibrant.
The fibrant equivariant cochain complexes must be
degreewise injective, which is now a non-trivial condition
(Prop. \ref{InjectiveEnvelopeOfEquivariantDualVectorSpaces} below).

\smallskip 
The key effect on the theory is that
equivariant minimal Sullivan models (Def. \ref{MinimalEquivariantdgcAlgebras})
-- which still exist and still have
the expected general properties -- are no longer given just by iterative
adjoining of (equivariant systems of) generators, but by adjoining of
injective resolutions (Example \ref{InjectiveResolutionOfEquivariantDualVectorSpaces})
of systems of generators. This has interesting effects,
as shown in Example \ref{CheckingTwistorSpaceModZ2MinimalModel},
which is at the heart of the proof of
Prop. \ref{Z2EquivariantRelativeMinimalModelOfSpin3ParametrizedTwistorSpace}
and thus of Theorem \ref{FluxQuantizationInEquivariantTwistorialCohomotopy}.

\medskip

\noindent {\bf Plain homological algebra.}
For plain (i.e., non-equivariant) dgc-algebra,
we follow the conventions of
\cite[\S 3.1]{FSS23-Char}. In particular, we make use of the
following notation:

\begin{notation}[Generators/relations presentation of cochain complexes]
  \label{GeneratorsAndRelationsPresentationOfCochainComplexes}
  $\,$

  \noindent
  We may denote any
  $V \in \CochainComplexesFin$
  by generators (a graded linear basis)
  and relations (the linear relations given by the differential).
  For instance:
  \vspace{1mm} 
  $$
  \qquad \;
    \mathbb{R}
    \langle c_2 \rangle
    \big/
    (
      d\, c_2 \,=\, 0
    )
    \;\;\;\;
      \simeq
    \;\;\;\;
    \big(
    \xymatrix{
      0
      \ar[r]
      &
      0
      \ar[r]
      &
      \mathbf{1}
      \ar[r]
      &
      0
      \ar[r]
      &
      0
      \ar[r]
      &
      \cdots
    }
    \big),
  $$
  \vspace{0mm}
  $$
    \mathbb{R}
    \left\langle
      \!\!\!
      \begin{array}{l}
        c'_3,
        \\[-2pt]
        c_3,
        \\[-2pt]
        b_2
      \end{array}
      \!\!\!
    \right\rangle
    \!\Big/\!
    \begin{pmatrix}
           d\, c'_3 = 0
        \\[-4pt]
        d\, c_3  = 0
        \\[-4pt]
        d\, b_2  = c_3
      \end{pmatrix}
    \;\;\;\;
      \simeq
    \;\;\;\;
    \big(
    \xymatrix{
      0
      \ar[r]
      &
      0
      \ar[r]
      &
      \mathbf{1}
      \;
      \ar@{^{(}->}[r]
      &
      \;
      \mathbf{2}
      \ar[r]
      &
      0
      \ar[r]
      &
      \cdots
    }
    \big).
  $$
\end{notation}

\begin{notation}[Generators/relations presentation of dgc-algebras]
  \label{GeneratorsAndRelationsPresentationOfdgcAlgebras}
  We may denote the Chevalley-Eilenberg algebra
  $\mathrm{CE}(\mathfrak{g}) \in \dgcAlgebrasFin$
  of any
  $\mathfrak{g} \in \LInfinityAlgebras$
  (\cite[Def. 3.25]{FSS23-Char})
  by generators (a graded linear basis)
  and relations (the polynomial relations given by the differential).
  For instance (see \cite[Ex. 3.67, 3.68]{FSS23-Char}):
  \vspace{-1mm}
  $$
    \begin{array}{l}
    \mathbb{R}
    [
      c_2
    ]
    \big/
    (
      d\, c_2 \,=\, 0
    )
    \;\;\;\;
    \simeq
    \;\;\;\;
    \mathrm{CE}
    (
      \mathfrak{b}\mathbb{R}
    )
    \\
  {\rm and}
  \;\;\;\;\;
    \mathbb{R}
    \!
    \bigg[
      \!\!\!
      \begin{array}{l}
        \omega_7,
        \\[-2pt]
        \omega_4
      \end{array}
      \!\!\!
    \bigg]
    \!\big/\!
    \left(
      \begin{aligned}
        d\, \omega_7 & = - \omega_4 \wedge \omega_4
        \\[-4pt]
        d\, \omega_4 & = 0
      \end{aligned}
    \right)
    \;\;\;\;
    \simeq
    \;\;\;\;
    \mathrm{CE}
    \big(
      \mathfrak{l}S^4
    \big)
    \,.
    \end{array}
  $$
    Similarly, for $T$
  a finite-dimensional compact and simply-connected Lie group
  with Lie algebra
  $$
  \mathfrak{t}
  \;\simeq\;
  \Big\{
    \langle t_a \rangle_{a = 1}^{\mathrm{dim}(T)}
    \,,\,
    [-,-]
  \Big\}
  \;\in\;
  \LieAlgebras\;,
  $$

  \noindent
  the abstract Chern-Weil isomorphism
  (e.g. \cite[\S 4.2]{FSS23-Char})
  reads:
  \begin{equation}
    \label{AbstractChernWeilHomomorphism}
    \Big(
    \mathbb{R}
    \big[
      \{
        r^{\, a}_2
      \}_{a = 1}^{\mathrm{dim}(T)}
    \big]
    \!\big/\!
    \big(
      d \, r^{\, a}_2 \,=\, 0
    \big)
    \! \Big)^T
    \;\;
      \simeq
    \;\;
    \mathrm{CE}
    (
      \mathfrak{l} BT
    )
    \,,
  \end{equation}

  \noindent
  where on the left $(-)^T$ denotes the $T$-invariant
  elements with respect to the coadjoint action on the
  dual vector space of the Lie algebra.

\end{notation}

\smallskip

\noindent {\bf Equivariant vector spaces.}

\vspace{-1mm} 
\begin{example}[Linear representations as functors]
  \label{LinearRepresentationsAsFunctors}
  For $G$ any finite group, write
  $BG$ for the category with a single object
  and with $G$ as its endomorphisms (hence its automorphisms).
  Then functors on $BG$ with values in vector spaces
  are, equivalently, linear $G$-representations
  with $G$ acting either from the left or from the right,
  depending on whether the functor is contravariant or covariant:
  \vspace{1mm} 
    \begin{equation}
      \begin{aligned}
        G \RepresentationsLeft
        & \simeq
        \mathrm{Functors}
        \big(
          B G^{\mathrm{op}}
          \,,\,
          \VectorSpaces
        \big),
        \\
        G \RepresentationsRight
        & \simeq
        \mathrm{Functors}
        \big(
          B G
          \,,\,
          \VectorSpaces
        \big).
      \end{aligned}
    \end{equation}
\end{example}

\begin{example}[Irreducible $\ZTwo$-representations]
  \label{Z2Irreps}
  We write
  $$
    \mathbf{1},
    \,
    \mathbf{1}_{\mathrm{sgn}}
    \;\;
    \in
    \;
    \ZTwo \RepresentationsRight
  $$

\noindent  for the two irreducible right representations
  (Example \ref{LinearRepresentationsAsFunctors})
  of
  $\ZTwo$, namely the trivial representation and the
  sign representation, respectively.
\end{example}

\begin{defn}[Equivariant vector spaces]
  \label{EquivariantVectorSpaces}
  We write
  \begin{equation}
    \label{EquivariantFiniteDimensionalVectorSpaces}
   \begin{aligned}
      \EquivariantVectorSpacesFin
      & := \;
      \mathrm{Functors}
      \big(
        G \mathrm{Orb}^{\mathrm{op}}
        \,,\,
        \VectorSpacesFin
      \big),
      \\[-2pt]
      \EquivariantDualVectorSpacesFin
      & :=\;
      \mathrm{Functors}
      \big(
        G \mathrm{Orb}
        \,,\,
        \VectorSpacesFin
      \big)
    \end{aligned}
  \end{equation}

    \vspace{-.5mm}
\noindent
  for the categories of
  contravariant or covariant functors, respectively,
  from the $G$-orbit category
  (Def. \ref{OrdinaryOrbitCategory}) to the
  category of finite-dimensional vector spaces
  over the real numbers.
\end{defn}

Notice that forming linear dual vector spaces
constitutes an equivalence of categories
$$
  \xymatrix{
    \VectorSpacesFin
    \ar[rr]^{ (-)^\vee }_-{ \simeq }
    &&
    \big(
      \VectorSpacesFin
    \big)^{\mathrm{op}}
  }
$$

\noindent
and hence induces an equivalence:
$$
  \begin{aligned}
    \big(
      \EquivariantVectorSpaces
    \big)^{\mathrm{op}}
    &
    =
    \Big(
      \mathrm{Functors}
      \big(
        G \mathrm{Orb}^{\mathrm{op}}
        \,,\,
        \VectorSpacesFin
      \big)
   \! \Big)^{\mathrm{op}}
    \\
    &
    \simeq
      \mathrm{Functors}
      \Big(
        G \mathrm{Orb}
        \,,\,
        \big(
          \VectorSpacesFin
        \big)^{\mathrm{op}}
      \Big)
    \\
    & \simeq
      \mathrm{Functors}
      \big(
        G \mathrm{Orb}
        \,,\,
        \VectorSpacesFin
      \big)
    \\
    & =
    \EquivariantDualVectorSpacesFin .
  \end{aligned}
$$

\noindent
This justifies extending the notation \eqref{EquivariantFiniteDimensionalVectorSpaces}
to  vector spaces which are not necessarily finite-dimensional
  $$
   \begin{aligned}
      \EquivariantVectorSpaces
      & :=\;
      \mathrm{Functors}
      \big(
        G \mathrm{Orb}^{\mathrm{op}}
        \,,\,
        \VectorSpaces
      \big)
      \\[-2pt]
      \EquivariantDualVectorSpaces
      & :=\;
      \mathrm{Functors}
      \big(
        G \mathrm{Orb}
        \,,\,
        \VectorSpaces
      \big)
    \end{aligned}
  $$

  \noindent
  and to speak of the latter as the category of
  \emph{equivariant dual vector spaces}
  (denoted $\mathrm{Vec}_G^\ast$ in \cite{Tri82}).

\begin{example}[Equivariant dual vector spaces of real cohomology groups]
  \label{EquivariantDualVectorSpacesOfRealCohomologyGroups}
  For $\mathscr{X} \,\in\, \EquivariantHomotopyTypes$
  (Def. \ref{EquivariantHomotopyTypes})
  and $n \in \mathbb{N}$,
  the stage-wise real cohomology groups in degree $n$ form
  an equivariant dual vector space (Def. \ref{EquivariantVectorSpaces})
  $$
    \underline{H}^n
    \big(
      \mathscr{X}
      ;
      \,
      \mathbb{R}
    \big)
    \;\;
      :
    \;\;
    G/H \;\longmapsto\;
    H^n
    \big(
      \mathscr{X}(G/H);
      \,
      \mathbb{R}
    \big)
    \,.
  $$

  If these are stage-wise finite-dimensional, then
  these are the linear dual equivariant vector spaces
  of the equivariant singular real homology groups
  $\underline{H}_{\, n}\big(\mathscr{X}; \mathbb{R} \big)$
  from Example \ref{EquivariantSingularHomologyGroup}.

\end{example}

\begin{example}[$\ZTwo$-equivariant dual vector spaces]
\label{Z2EquivariantVectorSpaces}
A (finite-dimensional) dual $\ZTwo$-equivariant
vector space (Def. \ref{EquivariantVectorSpaces})
is a diagram of (finite-dimensional) vector spaces indexed
by the $\ZTwo$-orbit category (Example \ref{OrbitCategoryOfZ2})
$$
 \left(\!\!\!
    \raisebox{10pt}{
    \xymatrix@C=4pt@R=1em{
      \ZTwo/1
      \ar@(ul,ur)|-{\; \ZTwo\, }
      \ar[d]
      &\mapsto&
      \mathbf{N}
      \ar@(ul,ur)|-{ \; \ZTwo\, }
      \ar[d]^-{\phi}
      \\
      \ZTwo/\ZTwo
      &\mapsto&
      V
    }
    }
     \right)
    \;\;\;
    \in
    \;\;
    \ZTwo \EquivariantDualVectorSpaces
$$

\noindent
hence constitutes:

-- a right $\ZTwo$-representation $\mathbf{N}$ (Example \ref{LinearRepresentationsAsFunctors}),

-- a vector space $V$ (finite-dimensional),

-- a linear map $\phi$ from the underlying vector space of $\mathbf{N}$ to $V$.
\end{example}

\begin{example}[Restriction of equivariant vector spaces to Weyl group linear representation]
  \label{RestrictionOfEquivariantVectorSpacesToLinearRepresentations}
  For $H \subset G$ a subgroup, with Weyl group
  $\WeylGroup(H) = \mathrm{Aut}_{G \mathrm{Orb}}(G/H)$
  (Example \ref{WeylGroup}),
  the canonical inclusion of categories
  \begin{equation}
    \label{InclusionOfDeloopedWeylGroupIntoOrbitCategory}
    \xymatrix{
      B \WeylGroup(H)
      \;
      \ar@{^{(}->}[r]^-{ i_H }
      &
      \;
      G \mathrm{Orb}
    }
  \end{equation}

  \noindent
  induces restriction functors of equivariant vector spaces
  (Def. \ref{EquivariantVectorSpaces}) to linear representations
  (Example \ref{LinearRepresentationsAsFunctors}):
  \begin{equation}
    \label{RestrictionFunctorFromEquivariantVectorSpacesToRepresentations}
    \xymatrix@R=-2pt{
      \WeylGroup(H)\RepresentationsLeft
    \;  \ar@{<-}[r]^-{ i_H^\ast }
      &
   \;   \EquivariantVectorSpaces \;,
      \\
      \WeylGroup(H)\RepresentationsRight
 \;     \ar@{<-}[r]^-{ i_H^\ast }
      &
    \;  \EquivariantDualVectorSpaces \;.
    }
  \end{equation}

\end{example}

\begin{example}[Regular equivariant vector space]
  \label{RegularEquivariantVectorSpace}
  For any subgroup $K \subset G$ we have
  an equivariant dual vector space (Def. \ref{EquivariantVectorSpaces})
  given by the $\mathbb{R}$-linear spans of the
  hom-sets \eqref{HomSets} out of $G/K$ in the
  orbit category (Def. \ref{OrdinaryOrbitCategory}):
  $$
    \mathbb{R}
    \big[
      G \mathrm{Orb}( G/K\,,\, - )
    \big]
    \;\in\;
    \EquivariantDualVectorSpaces
    \,.
  $$

  \noindent
  For any further subgroup $H \subset G$,
  its restriction
  (Example \ref{RestrictionOfEquivariantVectorSpacesToLinearRepresentations})
  to a linear representation
  from the right (Example \ref{LinearRepresentationsAsFunctors})
  of the Weyl group of $H$ (Def. \ref{WeylGroup})
  is
  $$
    i_H^\ast
    \big(
    \mathbb{R}
    \big[
      G \mathrm{Orb}( G/K\,,\, - )
    \big]
    \big)
    \;\;=\;\;
    \mathbb{R}
    \big[
      G \mathrm{Orb}( G/K\,,\, G/H )
    \big]
    \;\in\;
    \WeylGroup(H)
    \RepresentationsRight
    \,,
  $$

  \noindent
  where $\WeylGroup(H)$ acts in linear extension
  of its canonical right action on the hom-set
  of the orbit category (Example \ref{WeylGroup}).
\end{example}

\begin{lemma}[Extension of linear representations to equivariant vector spaces]
  \label{RightExtensionOfLinearRepresentationsToEquivariantVectorSpaces}
  For any $H \subset G$,
  the restriction of equivariant vector spaces to linear
  representations (Example \ref{RestrictionOfEquivariantVectorSpacesToLinearRepresentations})
  has a right adjoint
  $$
    \xymatrix{
     \WeylGroup(H)\RepresentationsRight
    \; \ar@{<-}@<+6pt>[rr]^-{ i_H }
     \ar@<-6pt>[rr]_-{ \mathrm{Inj}_H }^-{ \bot }
     &&
    \; \EquivariantDualVectorSpaces
    \,,
    }
  $$
  where
  $$
    \mathrm{Inj}_H(V^\ast)
    \;\in\;
    \EquivariantDualVectorSpaces
    \;=\;
    \mathrm{Functors}
    \big(
      G \mathrm{Orb}
      \,,\,
      \VectorSpaces
    \big)
  $$
  is given by
  \begin{align}
    \label{InjectiveAtomByRightKanExtension}
    \mathrm{Inj}_H
    (
      V^\ast
    )
    \;\;
    :
    \;\;
    G/K
    \;\longmapsto\;
    &
    \WeylGroup(H)\RepresentationsRight
    \Big(
      \mathbb{R}
      \big[
        G \mathrm{Orb}
        ( G/K\,,\, G/H )
      \big]
      \,,\,
      V^\ast
    \big)
    \\[-2pt]
    \label{InjectiveExtensionAsDisjointUnion}
    & \;\; =
    \underset{
      \scalebox{.75}{$
        {g \;\in\; G/\NormalizerGroup(K)}
        \atop
        {
        \scalebox{.75}{\rm s.t.} \;\;
        g^{-1} K g \;\subset\; H
        }
      $}
    }{\bigoplus}
    V^\ast
    .
  \end{align}

  \noindent
  Here the regular $\WeylGroup(H)$-representation
  in the first argument on the
  right of \eqref{InjectiveAtomByRightKanExtension}
  is from Example \ref{RegularEquivariantVectorSpace}.
\end{lemma}
\begin{proof}
  Formula \eqref{InjectiveAtomByRightKanExtension}
  is a special case of the general formula
  for right Kan extension
  \cite[(4.24)]{Kelly82}, here applied to
  the inclusion
  \eqref{InclusionOfDeloopedWeylGroupIntoOrbitCategory}
  regarded in $\VectorSpaces$-enriched category theory.
  Its equivalence to \eqref{InjectiveExtensionAsDisjointUnion}
  follows with Example \ref{HomSetsInOrbitCategoryViaWeylGroups}.
  See also \cite[(4.1)]{Tri82}\cite[Lemma 2.3]{Scull08}.
\end{proof}

\medskip

\noindent {\bf Injective equivariant dual vector spaces.}
Recall the general definition of injective objects
(e.g. \cite[p. 30]{HiltonStammbach71}),
applied to equivariant dual vector spaces:
\begin{defn}[Injective equivariant dual vector spaces]
  \label{InjectiveObjects}
  An object $I \in \EquivariantDualVectorSpaces$
  (Def. \ref{EquivariantVectorSpaces}) is called
  \emph{injective} if morphisms into it extend
  along all injections, hence if every solid diagram
  of the form
  \begin{equation}
    \label{ExtensionOfMapIntoInjectiveObject}
    \xymatrix@R=-4pt@C=5em{
      W
      \ar@{-->}[rr]^-{ \exists }
      &&
      I
      \mathrlap{
      \mbox{
        \tiny
        \color{darkblue}
        \bf
        \def\arraystretch{1}
        \begin{tabular}{c}
          injective
          \\
          object
        \end{tabular}
      }
      }
      \\
      &
      \; V
      \ar@{_{(}->}[ul]^-{
        \mbox{
          \tiny
          \color{greenii}
          \bf
          injection
        }
      }
      \ar[ur]
    }
  \end{equation}

  \noindent
  admits a dashed morphism that makes it commute, as shown.
  We write
  $$
    \xymatrix{
      \EquivariantDualVectorSpacesInjective
      \;
      \ar@{^{(}->}[r]
      &
      \EquivariantDualVectorSpaces
    }
  $$

  \noindent
  for the full sub-category on the injective objects.
\end{defn}

\begin{prop}[Injective envelope of equivariant dual vector spaces {\cite[p. 2]{Tri82}\cite[Prop. 7.34]{Scull02}\cite[Lem. 2.4, Prop. 2.5]{Scull08}}]
  \label{InjectiveEnvelopeOfEquivariantDualVectorSpaces}
  For $V \in \EquivariantDualVectorSpaces$
  (Def. \ref{EquivariantVectorSpaces}),
  the direct sum of extensions $\mathrm{Inj}_{(-)}$
  (Def. \ref{RightExtensionOfLinearRepresentationsToEquivariantVectorSpaces})
  \begin{equation}
    \label{InjectiveEnvelopeConstruction}
    \mathrm{Inj}(V)
    \;:=\;
    \underset{
      [H \subset G]
    }{\bigoplus}
    \,
    \mathrm{Inj}_H
    \big(
      V_H
    \big)
    \;\;\;\;
    \in
    \;
    \EquivariantDualVectorSpaces
    \,,
  \end{equation}

  \noindent
  of those components at stage $H$ which vanish on all deeper stages
  \begin{equation}
    \label{TheJointKernel}
    V_H
    :=
    \left\{\!\!\!
    \def\arraystretch{1.4}
    \begin{array}{lcl}
      \underset{ [K \supsetneq H] }{\bigcap}
      \mathrm{ker}
      \Big(\!\!
      \xymatrix@C=1.5em{
        V(G/H)
          \ar[rrr]^{
            V(
              G /(H \hookrightarrow K)
            )
          }
          &&&
        V(G/K)
      }
       \!\!\!\Big)
      &\vert&
      H \neq G
      \\
      V(G/G)
      &\vert&
      H = G
    \end{array}
    \right.
  \end{equation}

  \noindent
  receives an injection

  \vspace{-5mm}
  \begin{equation}
    \label{CanonicalInjectiveEnvelope}
    \xymatrix{
      V
      \;
      \ar@{^{(}->}[r]
      &
      \;
      \mathrm{Inj}(V)
    }
  \end{equation}

 \noindent
  that extends the canonical inclusion of the $V_H$,
  and which is an injective envelope (e.g. \cite[\S I.9]{HiltonStammbach71})
  of $V$ in $\EquivariantDualVectorSpaces$.
  In particular:

\noindent   {\bf (i)} the summands
  $\mathrm{Inj}_H(V)$ (Example \ref{RightExtensionOfLinearRepresentationsToEquivariantVectorSpaces})
  are injective objects (Def. \ref{InjectiveObjects});

\noindent   {\bf (ii)}
  $V$ is injective
  (Def. \ref{InjectiveObjects}) precisely if
  \eqref{CanonicalInjectiveEnvelope} is an isomorphism.
\end{prop}

\begin{example}[Ground field is injective as equivariant dual vector space]
  \label{GroundFieldIsInjectiveAsEquivariantDualVectorSpace}
  The equivariant dual vector space (Def. \ref{EquivariantVectorSpaces})
  which is constant on the ground field
  $$
    \underline{\mathbb{R}}
    \;:=\;
    \mathrm{const}_{G \mathrm{Orb}}(\mathbb{R})
    \;:\;
    G/H
    \;\longmapsto\;
    \mathbb{R}
  $$

\noindent
  is isomorphic to the right extension
  (Lemma \ref{RightExtensionOfLinearRepresentationsToEquivariantVectorSpaces})
  $
    \underline{\mathbb{R}}
    \;\simeq\;
    \mathrm{Inj}_G(\mathbf{1})
  $
  of $\mathbb{R} \simeq \mathbf{1} \in 1 \Representations$,
  and hence is injective, by Prop. \ref{InjectiveEnvelopeOfEquivariantDualVectorSpaces}.
\end{example}

\begin{example}[Injective $\ZTwo$-equivariant dual vector spaces, cf. {\cite[Prop. 4.1]{SanthanamThandar23}}]
  \label{InjectiveZ2EquivariantDualVectorSpaces}
  For $G = \ZTwo$ (Example \ref{OrbitCategoryOfZ2})
  the
  irreducible representations
  $$
    \mathbf{1},\, \mathbf{1}_{\mathrm{sgn}}
    \;\;
    \in
    \ZTwo \Representations
    \,,
    \;\;\;\;\;\;\;
    \mathbf{1}
    \;\;
    \in
    1 \Representations
    \;\simeq\;
    \VectorSpaces
  $$
  of the respective Weyl groups
  (Example \ref{WeylGroup},
  Example \ref{Z2Irreps})
  induce by right extension (Def. \ref{RightExtensionOfLinearRepresentationsToEquivariantVectorSpaces})
  the following three
  $\ZTwo$-equivariant vector spaces
  (Example \ref{Z2EquivariantVectorSpaces}),
  which, by Prop. \ref{InjectiveEnvelopeOfEquivariantDualVectorSpaces},
  are the direct summand building blocks of all injective
  $\ZTwo$-equivariant dual vector spaces:
  \begin{equation}
    \label{TheInjectiveZ2EquivariantDualVectorSpaceInBulk}
    \mathrm{Inj}_1(\mathbf{1})
    \;\;
    :
    \raisebox{23pt}{
    \xymatrix@C=4pt@R=1.5em{
      \ZTwo/1
      \ar@(ul,ur)|-{\; \ZTwo }
      \ar[d]
      &\longmapsto&
      \mathbf{1}
      \ar[d]^-{0}
      \\
      \ZTwo/\ZTwo
      &\longmapsto&
      0  \mathrlap{\,,}
    }
    }
    \phantom{AAA}
    \mathrm{Inj}_1(\mathbf{1}_{\mathrm{sgn}})
    \;\;
    :
    \raisebox{23pt}{
    \xymatrix@C=4pt@R=1.5em{
      \ZTwo/1
      \ar@(ul,ur)|-{\; \ZTwo }
      \ar[d]
      &\longmapsto&
      \mathbf{1}_{\mathrm{sgn}}
      \ar[d]^-{0}
      \\
      \ZTwo/\ZTwo
      &\longmapsto&
      0  \mathrlap{\,,}
    }
    }
      \end{equation}

\noindent
and
  \begin{equation}
    \label{TheInjectiveZ2EquivariantDualVectorSpaceConcentratedOnFixedLocus}
    \mathrm{Inj}_{\ZTwo}(\mathbf{1})
    \;\;
    \;:\;
    \raisebox{23pt}{
    \xymatrix@C=4pt@R=1.5em{
      \ZTwo/1
      \ar@(ul,ur)|-{\; \ZTwo }
      \ar[d]
      &\longmapsto&
      \mathbf{1}
      \ar[d]^-{\mathrm{id}}
      \\
      \ZTwo/\ZTwo
      &\longmapsto&
      \mathbf{1}
      \mathrlap{\,.}
    }
    }
  \end{equation}
  To see this, use \eqref{HomSetsInOrbitCategoryOfZ2} in
  \eqref{InjectiveAtomByRightKanExtension} to get, for two cases,
  $$
    \mathrm{Inj}_1(\mathbf{1})
    \;\;
    \;:\;
    \raisebox{24pt}{
    \xymatrix@R=10pt@C=4pt{
      \ZTwo/1
      \ar@(ul,ur)|-{\; \ZTwo }
      \ar[d]
      &\longmapsto&
      \ZTwo \Representations
      \Big(
        \underset{
          \simeq
          \,
          \mathbf{1}
          \oplus
          \mathbf{1}_{\mathrm{sgn}}
        }{
        \underbrace{
        \mathbb{R}
        \big[
          \ZTwo \mathrm{Orb}
          (
            \ZTwo/1
            \,,\,
            \ZTwo/1
          )
        \big]
        }
        }
        \,,\,
        \mathbf{1}
      \big)
      \ar@{}[r]|-{\simeq}
      &
      \mathbf{1}
      \ar[d]^-{0}
      \\
      \ZTwo/\ZTwo
      &\longmapsto&
      \ZTwo \Representations
      \Big(
        \underset{
          \simeq
          \,
          0
        }{
        \underbrace{
        \mathbb{R}
        \big[
          \ZTwo \mathrm{Orb}
          (
            \ZTwo/\ZTwo
            \,,\,
            \ZTwo/1
          )
        \big]
        }
        }
        \,,\,
        \mathbf{1}
      \big)
      \ar@{}[r]|-{\simeq}
      &
      0
    }
    }
  $$

\noindent
  and
  $$
    \mathrm{Inj}_{\ZTwo}(\mathbf{1})
    \;\;
    \;:\;
    \raisebox{24pt}{
    \xymatrix@R=10pt@C=4pt{
      \ZTwo/1
      \ar@(ul,ur)|-{\; \ZTwo }
      \ar[d]
      &\longmapsto&
      1 \Representations
      \Big(
        \underset{
          \simeq
          \,
          \mathbf{1}
        }{
        \underbrace{
        \mathbb{R}
        \big[
          \ZTwo \mathrm{Orb}
          (
            \ZTwo/1
            \,,\,
            \ZTwo/\ZTwo
          )
        \big]
        }
        }
        \,,\,
        \mathbf{1}
      \big)
      \ar@{}[r]|-{\simeq}
      &
      \mathbf{1}
      \ar[d]^-{\mathrm{id}}
      \\
      \ZTwo/\ZTwo
      &\longmapsto&
      1 \Representations
      \Big(
        \underset{
          \simeq
          \,
          \mathbf{1}
        }{
        \underbrace{
        \mathbb{R}
        \big[
          \ZTwo \mathrm{Orb}
          (
            \ZTwo/\ZTwo
            \,,\,
            \ZTwo/\ZTwo
          )
        \big]
        }
        }
        \,,\,
        \mathbf{1}
      \big)
      \ar@{}[r]|-{\simeq}
      &
      \mathbf{1}   \mathrlap{\,.}
    }
    }
  $$
\end{example}

\begin{lemma}[Tensor product preserves injectivity of finite-dim dual vector $G$-spaces
 {\cite[Lem. 3.6, Rem 1.2]{Golasinski97b} \cite[Prop. 7.36]{Scull02}}]
  \label{TensorProductPreservesInjectivityOfFiniteDimensionalDualVectorGSpaces}
  Let $V,W \in \EquivariantDualVectorSpacesFin$
  (Def. \ref{EquivariantVectorSpaces}).
  If $V$ and $W$ are both injective (Def. \ref{InjectiveObjects}),
  then so is their tensor product
  $V \otimes W : G/H \;\longmapsto\; V(G/H) \otimes W(G/H)$.
\end{lemma}

\medskip

\noindent {\bf Equivariant smooth differential forms.}
In preparation of discussing equivariant de Rham cohomology,
consider:

\begin{example}[Equivariant smooth differential forms]
  \label{EquivariantSmoothDifferentialForms}
  Let
  $G \acts \; X \,\in\,$ 
 \newline  
  $\GActionsOnSmoothManifolds$
  (Def. \ref{ProperGActionsOnSmoothManifolds})
  and $n \in \mathbb{N}$.
  Then there is the equivariant dual vector space
  (Def. \ref{EquivariantdgcAlgebras})
  $$
    \Omega^n_{\mathrm{dR}}
    \big(
      \orbisingular
      \big(
        X \!\sslash\! G
      \big)
    \big)
    \;\;
    \in
    \;\;
    \EquivariantDualVectorSpaces
  $$

  \noindent
  given by the system of vector spaces of
  smooth differential $n$-forms (e.g. \cite{BottTu82})
  of the fixed submanifolds \eqref{SystemOfSmoothSubmanifoldsofProperSmoothGAction},
  with pullback of differential forms
  along residual actions and along inclusions of fixed loci:
  $$
  \hspace{-2.5cm}
    \overset{
      \mathclap{
      \raisebox{3pt}{
        \tiny
        \color{darkblue}
        \bf
        \def\arraystretch{1}
        \begin{tabular}{c}
          Equivariant dual vector
          \\
          space of equivariant smooth
          \\
          differential $n$-forms
        \end{tabular}
      }
      }
    }{
    \Omega^n_{\mathrm{dR}}
    \big(\!
      \orbisingular
      \big(
        X \!\sslash\! G
      \big)
    \big)
    }
    \;\;
      :
    \;
    \raisebox{26pt}{
    \xymatrix@C=1.5em{
     G/H_1
     \ar@(ul,ur)^-{ g_1 \in \WeylGroup(H_1) }
     \ar[d]_-{ p }
     &\longmapsto&
     \Omega^n_{\mathrm{dR}}
     \big(
       X^{H_1}
     \big)
     \mathrlap{\hspace{-4mm} 
       \mbox{
         \tiny
         \color{darkblue}
         \bf
         \def\arraystretch{1}
         \begin{tabular}{c}
           ordinary differential forms
           \\
           on fixed submanifold
         \end{tabular}
       }
     }
     \ar[d]^-{
       X^{p^\ast}
       \mathrlap{
         \mbox{
           \tiny
           \color{greenii}
           \bf
           \def\arraystretch{1}
           \begin{tabular}{c}
             pullback along inclu-
             \\
             sion of fixed loci
           \end{tabular}
         }
       }
     }
     \ar@(ul,ur)^-{ \;X^{g_1^\ast}\; }
     \\
     G/H_2
     \ar@(dl,dr)_-{ g_2 \in \WeylGroup(H_2)  }
     &\longmapsto&
     \Omega^n_{\mathrm{dR}}
     \big(
       X^{H_2}
     \big)
     \ar@(dl,dr)_-{ \;X^{g_2^\ast}\; }
    }
    }
  $$
 \end{example}

\begin{remark}[Equivariant smooth differential forms are injective]
  \label{EquivariantSmoothDifferentialFormsAreInjective}
  The following Lemmas
  \ref{ZpEquivariantSmoothDifferentialFormsFormInjectiveDualVectorSpace},
  \ref{Z4EquivariantSmoothDifferentialFormsAreInjective},
  \ref{Z2TimesZ2EquivariantDifferentialFormsAreInjective} show
  that the equivariant dual vector spaces of smooth differential
  $n$-forms (Def. \ref{EquivariantSmoothDifferentialForms})
  are injective objects (Def. \ref{InjectiveObjects}), at least
  if the equivariance group is of order 4 or cyclic of prime order (in which case cf. \cite[Prop. 4.1]{SanthanamThandar23}):
  $$
    G
      \;\in\;
      \big\{\!
        \left. \mathbb{Z}_{p} \right\vert \mbox{$p$ prime}
      \big\}
      \;\cup\;
      \big\{
        \mathbb{Z}_4, \, \ZTwo \times \ZTwo
      \big\}
      \,.
  $$

\noindent  From the proofs of these lemmas, given below,
  it is fairly clear how to approach
  the proof of the general case. But since this is
  heavy on notation if done properly, and since we
  do not need further generality for our application here,
  we will not go into that.
\end{remark}

\begin{notation}[Extension of smooth differential forms away from fixed loci]
  \label{ExtensionOfSmoothDifferentialFormsAwayFromFixedLoci}
  $\,$

  \noindent
  For $G \acts \; X \,\in\, \GActionsOnSmoothManifolds$
  (Def. \ref{ProperGActionsOnSmoothManifolds})
  and $H \,\subset\, G$,
  choose a tubular neighborhood (e.g. \cite[\S 1.2]{Kochman96})
  $
    \mathcal{N}_X
    \big(
      X^{H}
    \big)
    \;\subset\;
    X
  $
  of the fixed locus
  (which exists by Lemma \ref{FixedLociOfProperSmoothActionsAreSmoothManifolds}).
  Then multiplication of smooth $n$-forms on $X^{H}$
  with a choice of bump function in the neighborhood coordinates
  induces a linear section, which we denote $\mathrm{ext}_{H}$,
  of the operation of restricting differential forms to the
  fixed locus:
  $$
    \xymatrix@C=30pt{
      \Omega^n_{\mathrm{dR}}
      \big(
        X^{H}
      \big)
      \ar@{-->}[rr]^-{
        \mathrm{ext}_{H}
      }
      \ar@/_1.6pc/[rrrrr]|-{
        \;
        \mathrm{id}
        \;
      }
      &
      {\phantom{AAA}}
      &
      \Omega^n_{\mathrm{dR}}
      (
        X
      )
      \ar[rrr]^-{
        \;
        (-)\vert_{ X^{H}}
        \;
      }
      &&&
      \Omega^n_{\mathrm{dR}}
      \big(
        X^{H}
      \big)
      \,.
    }
  $$
  
\end{notation}

\begin{lemma}[$\mathbb{Z}_p$-Equivariant smooth differential forms are injective]
  \label{ZpEquivariantSmoothDifferentialFormsFormInjectiveDualVectorSpace}
  Let the equivariance group
  $G = \mathbb{Z}_p$ be a cyclic group of prime order.
  Then, for
  $\mathbb{Z}_p \acts \; X
    \;\in\;
  \mathbb{Z}_p\mathrm{Actions}\big(\mathrm{SmoothManifolds}\big)$
  (Def. \ref{ProperGActionsOnSmoothManifolds}),
  the equivariant dual vector space of
  $\mathbb{Z}_p$-equivariant smooth differential $n$-forms
  (Def. \ref{EquivariantSmoothdeRhamComplex}) is
  injective (Def. \ref{InjectiveObjects}):
  \begin{equation}
    \label{ZpEquivariantSmoothnFormsInjectiveDualVectorSpace}
    \Omega^n_{\mathrm{dR}}\big( \orbisingular(X \!\sslash\! \mathbb{Z}_p) \big)
    \;\;
    \in
    \;\;
    \EquivariantDualVectorSpacesInjective\;.
  \end{equation}

\end{lemma}
\begin{proof}
  By extension of differential
  forms away from the fixed locus (Notation \ref{ExtensionOfSmoothDifferentialFormsAwayFromFixedLoci}),
  we obtain the following isomorphism
  of equivariant dual vector spaces
  to a direct sum of injective extensions (Lemma \ref{RightExtensionOfLinearRepresentationsToEquivariantVectorSpaces})
  $$
    \hspace{-1mm}
    \adjustbox{scale=0.9}{$
    \xymatrix@R=-3pt@C=3pt{
      &&
      \overset{
        \mathclap{
        \raisebox{6pt}{
          \tiny
          \color{darkblue}
          \bf
          \def\arraystretch{1}
          \begin{tabular}{c}
            equivariant smooth
            \\
            differential $n$-forms
          \end{tabular}
        }
        }
      }{
      \Omega^n_{\mathrm{dR}}
      \big(
        \orbisingular(X \!\sslash\! \mathbb{Z}_p)
      \big)
      }
      \ar[rr]^-{ \simeq }
      &&
      \mathrm{Inj}_{\mathbb{Z}_p}
      \Big(
        \overset{
          \mathclap{
          \raisebox{7pt}{
            \tiny
            \color{darkblue}
            \bf
            \begin{tabular}{c}
              differential $n$-forms
              \\[-3pt]
              on fixed locus
            \end{tabular}
          }
          }
        }{
        \Omega^n_{\mathrm{dR}}
        \big(
          X^{\mathbb{Z}_p}
        \big)
        }
      \Big)
      \ar@{}[r]|-{\mbox{$\oplus$}}
      &
      \mathrm{Inj}_{1}
      \Big(
        \overset{
          \mathclap{
          \raisebox{4pt}{
            \tiny
            \color{darkblue}
            \bf
            \begin{tabular}{c}
               differential $n$-forms whose
               \\[-3pt]
               restriction to the fixed locus vanishes
            \end{tabular}
          }
          }
        }{
        \Big\{
          \!
          \left.
          \omega
          \,\in\,
          \Omega^n_{\mathrm{dR}}
          \big(
            X
          \big)
          \right\vert
          \omega \vert_{ X^{\mathbb{Z}_p}}
          \,=\,
          0
        \Big\}
        }
      \! \Big)
      \\
      \mathbb{Z}_p/1
      \ar@(ul,ur)|-{
        \;\mathbb{Z}_p\;
      }
      \ar[dd]
      &&
      \alpha
      \ar@{|->}[rr]
      \ar@{|->}[dd]
      &&
      \Big(
      \alpha \vert_{ X^{\mathbb{Z}_p}}
      \ar@{}[r]|-{,}
      \ar@{|->}[dd]
      &
      \alpha
        \,-\,
      \mathrm{ext}_{\mathbb{Z}_p}
      \big(
        \alpha \vert_{ X^{\mathbb{Z}_p}}
      \big)
      \Big)
      \ar@{|->}[dd]
      \\
      {\phantom{A} \atop {\phantom{A}}}
      \\
      \mathbb{Z}_p/\mathbb{Z}_p
      &&
      \alpha \vert_{ X^{\mathbb{Z}_p}}
      \ar@{|->}[rr]
      &&
      \Big(
      \alpha \vert_{ X^{\mathbb{Z}_p}}
      \ar@{}[r]|-{,}
      &
      \;\;\;\;\;\;\;\;\;\;\;\;\;\;
      0
      \;\;\;\;\;\;\;\;\;\;\;\;\;\;
      \mathrlap{\Big)\,,}
    }
    $}
  $$

  \noindent
  where we used, since $p$ is assumed to be prime,
  that the only subgroups of $G$
  are $1$ and  $\mathbb{Z}_p$ itself (Example \ref{MoreExamplesOfOrbitCategories}).
  By Prop. \ref{InjectiveEnvelopeOfEquivariantDualVectorSpaces},
  this implies the claim \eqref{ZpEquivariantSmoothnFormsInjectiveDualVectorSpace}.
\end{proof}

\begin{lemma}[$\mathbb{Z}_4$-Equivariant smooth differential forms are injective]
  \label{Z4EquivariantSmoothDifferentialFormsAreInjective}
  Let the equivariance group
  $G = \mathbb{Z}_4$ be the cyclic group of order 4.
  Then, for
  $\mathbb{Z}_4 \acts \; X
    \;\in\;
  \mathbb{Z}_4\mathrm{Actions}\big(\mathrm{SmoothManifolds}\big)$
  (Def. \ref{ProperGActionsOnSmoothManifolds}),
  the equivariant dual vector space of
  $\mathbb{Z}_4$-equivariant smooth differential $n$-forms
  (Def. \ref{EquivariantSmoothdeRhamComplex}) is
  injective (Def. \ref{InjectiveObjects}):
  \vspace{0mm}
  \begin{equation}
    \label{Z4EquivariantSmoothnFormsInjectiveDualVectorSpace}
    \Omega^n_{\mathrm{dR}}\big( \orbisingular(X \!\sslash\! \mathbb{Z}_4) \big)
    \;\;
    \in
    \;\;
    \EquivariantDualVectorSpacesInjective\;.
  \end{equation}
  \vspace{-.6cm}

\end{lemma}
\begin{proof}
  Since the subgroups of $\mathbb{Z}_4$ are linearly ordered
  $1 \,\subset\, \ZTwo \,\subset\, \mathbb{Z}_4$
  (Example \ref{MoreExamplesOfOrbitCategories}),
  the proof of Lemma \ref{ZpEquivariantSmoothDifferentialFormsFormInjectiveDualVectorSpace}
  generalizes immediately.
    Using extensions of differential $n$-forms
  (Notation \ref{ExtensionOfSmoothDifferentialFormsAwayFromFixedLoci}),
  both from $X^{\mathbb{Z}_4}$ as well as from $X^{\ZTwo}$,
  we obtain the following isomorphism
  of equivariant dual vector spaces
  to a direct sum of injective extensions (Lemma \ref{RightExtensionOfLinearRepresentationsToEquivariantVectorSpaces})
  \vspace{-1mm}
  $$
    \hspace{-1mm}
    \scalebox{.65}{
    \hspace{-.3cm}
    \xymatrix@R=-3.5pt@C=2pt{
      &
      \overset{
        \mathclap{
        \raisebox{6pt}{
          \tiny
          \color{darkblue}
          \bf
          \def\arraystretch{1}
          \begin{tabular}{c}
            Equivariant smooth
            \\
            differential $n$-forms
          \end{tabular}
        }
        }
      }{
      \Omega^n_{\mathrm{dR}}
      \big(
        \orbisingular(X \!\sslash\! \mathbb{Z}_4)
      \big)
      }
      \ar[rr]^-{ \simeq }
      &&
      \mathrm{Inj}_{\mathbb{Z}_4}
      \!
      \Big(
        \!
        \overset{
          \mathclap{
          \raisebox{7pt}{
            \tiny
            \color{darkblue}
            \bf
            \def\arraystretch{1}
            \begin{tabular}{c}
              Differential $n$-forms
              \\
              on deep fixed locus
            \end{tabular}
          }
          }
        }{
        \Omega^n_{\mathrm{dR}}
        \big(
          X^{\mathbb{Z}_4}
        \big)
        }
        \!
      \Big)
      \ar@{}[r]|-{\mbox{$\oplus$}}
      &
      \mathrm{Inj}_{\ZTwo}
      \!
      \Big(
        \!
        \overset{
          \mathclap{
          \raisebox{4pt}{
            \tiny
            \color{darkblue}
            \bf
            \def\arraystretch{1}
            \begin{tabular}{c}
               Differential $n$-forms on shallow fixed locus whose
               \\
               restriction to the deep fixed locus vanishes
            \end{tabular}
          }
          }
        }{
        \Big\{
          \!
          \left.
          \omega
          \in
          \Omega^n_{\mathrm{dR}}
          \big(
            X^{\ZTwo}
          \big)
          \right\vert
          \omega \vert_{ X^{\mathbb{Z}_4}}
          =
          0
        \Big\}
        }
      \!
      \Big)
      \ar@{}[r]|-{\mbox{$\oplus$}}
      &
      \mathrm{Inj}_{1}
      \!
      \Big(
        \!
        \overset{
          \mathclap{
          \raisebox{4pt}{
            \tiny
            \color{darkblue}
            \bf
            \def\arraystretch{1}
            \begin{tabular}{c}
               Differential $n$-forms whose
               restriction 
               \\
               to the shallow fixed locus vanishes
            \end{tabular}
          }
          }
        }{
        \Big\{
          \!
          \left.
          \omega
          \in
          \Omega^n_{\mathrm{dR}}
          \big(
            X
          \big)
          \right\vert
          \omega \vert_{ X^{\ZTwo}}
          \,=\,
          0
        \Big\}
        }
        \!
      \Big)
      \\
      \mathbb{Z}_4/1
      \ar@(ul,ur)|-{
        \;\mathbb{Z}_4\;
      }
      \ar[dd]
      &
      \alpha
      \ar@{|->}[rr]
      \ar@{|->}[dd]
      &&
      \bigg(
      \alpha \vert_{ X^{\mathbb{Z}_4}}
      \ar@{}[r]|-{,}
      \ar@{|->}[dd]
      &
      \Big(
        \alpha
          \,-\,
        \mathrm{ext}_{\mathbb{Z}_4}
        \big(
          \alpha \vert_{ X^{\mathbb{Z}_4}}
        \big)
      \Big) \vert_{ X^{\ZTwo}}
      \ar@{|->}[dd]
      \ar@{}[r]|-{,}
      &
        \alpha
          \,-\,
        \mathrm{ext}_{\ZTwo}
        \big(
          \alpha \vert_{ X^{\ZTwo}}
        \big)
        \bigg)
      \ar@{|->}[dd]
      \\
      {\phantom{A} \atop {\phantom{A}}}
      \\
      \mathbb{Z}_4/\ZTwo
      \ar[dd]
      &
      \alpha \vert_{ X^{\ZTwo}}
      \ar@{|->}[rr]
      \ar@{|->}[dd]
      &&
      \bigg(
      \alpha \vert_{ X^{\mathbb{Z}_4}}
      \ar@{|->}[dd]
      \ar@{}[r]|-{,}
      &
        \alpha \vert_{ X^{\ZTwo}}
          \,-\,
      \Big(
        \mathrm{ext}_{\mathbb{Z}_4}
        \big(
          \alpha \vert_{ X^{\mathbb{Z}_4}}
        \big)
      \Big) \vert_{ X^{\ZTwo}}
      \ar@{|->}[dd]
      \ar@{}[r]|-{,}
      &
      \;\;\;\;\;\;\;\;\;\;\;\;\;
        0
      \;\;\;\;\;\;\;\;\;\;\;\;\;
      \mathrlap{\bigg)}
      \ar@{|->}[dd]
      \\
      {\phantom{A} \atop {\phantom{A}}}
      \\
      \mathbb{Z}_4/\mathbb{Z}_4
      &
      \alpha \vert_{ X^{\mathbb{Z}_4}}
      \ar@{|->}[rr]
      &&
      \bigg(
      \alpha \vert_{ X^{\mathbb{Z}_4}}
      \ar@{}[r]|-{,}
      &
      \;\;\;\;\;\;\;\;\;\;\;\;\;\;\;
      0
      \mathclap{\phantom{\bigg)}}
      \;\;\;\;\;\;\;\;\;\;\;\;\;\;\;
      \ar@{}[r]|-{,}
      &
      \;\;\;\;\;\;\;\;\;\;\;\;\;
        0
      \;\;\;\;\;\;\;\;\;\;\;\;\;
      \mathrlap{\bigg)}
    }
    }
  $$

\vspace{1mm} 
  \noindent
  By Prop. \ref{InjectiveEnvelopeOfEquivariantDualVectorSpaces},
  this implies the claim \eqref{Z4EquivariantSmoothnFormsInjectiveDualVectorSpace}.
\end{proof}

\begin{lemma}[$\ZTwo \times \ZTwo$-Equivariant smooth differential forms are injective]
  \label{Z2TimesZ2EquivariantDifferentialFormsAreInjective}
  Let the equivariance group
  $G = \mathbb{Z}^L_2 \times \mathbb{Z}^R_2$ be the Klein 4-group.
  Then, for
  $\mathbb{Z}^L_2 \times \ZTwo^R \acts \; X
    \;\in\;
  \mathbb{Z}^L_2 \times \mathbb{Z}^R_2\mathrm{Actions}\big(\mathrm{SmoothManifolds}\big)$
  (Def. \ref{ProperGActionsOnSmoothManifolds}),
  the equivariant dual vector space of
  equivariant smooth differential $n$-forms
  (Def. \ref{EquivariantSmoothdeRhamComplex}) is
  injective (Def. \ref{InjectiveObjects}):  
  \begin{equation}
    \label{Z2TimesZ2EquivariantSmoothnFormsInjectiveDualVectorSpace}
    \Omega^n_{\mathrm{dR}}\big(
      \orbisingular \big(X \!\sslash\! \mathbb{Z}^L_2 \times \mathbb{Z}^R_2\,\big)
    \big)
    \;\;
    \in
    \;\;
    \EquivariantDualVectorSpacesInjective \;.
  \end{equation}
  \vspace{-.6cm}

\end{lemma}
\begin{proof}
  We obtain an isomorphism to
  a direct sum of injective extensions
  (Lemma \ref{RightExtensionOfLinearRepresentationsToEquivariantVectorSpaces}),
  much as in the proofs of Lemmas
  \ref{ZpEquivariantSmoothDifferentialFormsFormInjectiveDualVectorSpace} and
  \ref{Z4EquivariantSmoothDifferentialFormsAreInjective},
  \vspace{-2mm}
  $$
    \scalebox{.6}{
    \hspace{-3mm}
    \xymatrix@R=-3.5pt@C=3pt{
      &&
      \overset{
        \mathclap{
        \raisebox{6pt}{
          \tiny
          \color{darkblue}
          \bf
          \def\arraystretch{1}
          \begin{tabular}{c}
            equivariant smooth
            \\
            differential $n$-forms
          \end{tabular}
        }
        }
      }{
      \Omega^n_{\mathrm{dR}}
      \big(
        \orbisingular(X \!\sslash\! \mathbb{Z}_4)
      \big)
      }
      \ar[rr]^-{ \simeq }
      &&
      \mathrm{Inj}_{\mathbb{Z}_4}
      \!
      \Big(
        \!
        \overset{
          \mathclap{
          \raisebox{6pt}{
            \tiny
            \color{darkblue}
            \bf
            \def\arraystretch{1}
            \begin{tabular}{c}
              differential $n$-forms
              \\
              on deep fixed locus
            \end{tabular}
          }
          }
        }{
        \Omega^n_{\mathrm{dR}}
        \big(
          X^{\mathbb{Z}_4}
        \big)
        }
        \!
      \Big)
      \ar@{}[r]|-{\mbox{$\oplus$}}
      &
      \raisebox{10pt}{$
      {\begin{array}{c}
      \mathrm{Inj}_{\mathbb{Z}^L_2}
      \!
      \bigg(
        \!\!
        \overset{
          \mathclap{
          \raisebox{3pt}{
            \tiny
            \color{darkblue}
            \bf
            \def\arraystretch{1}
            \begin{tabular}{c}
               differential $n$-forms on shallow fixed loci whose
               \\
               restriction to the deep fixed locus vanishes
            \end{tabular}
          }
          }
        }{
        \Big\{
          \!\!
          \left.
          \omega
          \in
          \Omega^n_{\mathrm{dR}}
          \big(
            X^{\mathbb{Z}^L_2}
          \big)
          \right\vert
          \omega \vert_{ X^{\mathbb{Z}_4}}
          =
          0
        \Big\}
        }
      \!\!
      \bigg)
      \\
      \oplus
      \\
      \mathrm{Inj}_{\mathbb{Z}^R_2}
      \!
      \bigg(
        \!\!
        \Big\{
          \!\!
          \left.
          \omega
          \in
          \Omega^n_{\mathrm{dR}}
          \big(
            X^{\mathbb{Z}^R_2}
          \big)
          \right\vert
          \omega \vert_{ X^{\mathbb{Z}_4}}
          =
          0
        \Big\}
      \!\!
      \bigg)
      \end{array}}
      $}
      \hspace{-4mm} 
      \ar@{}[r]|-{\mbox{$\oplus$}}
      &
      \mathrm{Inj}_{1}
      \!
             \overset{
          \mathclap{
          \raisebox{3pt}{
            \tiny
            \color{darkblue}
            \bf
            \def\arraystretch{1}
            \begin{tabular}{c}
               differential $n$-forms whose
               restriction 
               \\
               to the shallow fixed loci vanishes
            \end{tabular}
          }
          }
        }{
      \left(
        \!   \left\{
          \omega
          \in
          \Omega^n_{\mathrm{dR}}
          \big(
            X
          \big)
          \left\vert
            \def\arraystretch{.1}
            \begin{array}{l}
              \omega \vert_{ X^{\mathbb{Z}^L_2}}
              \,=\,
              0
              \\[2pt]
              \omega \vert_{ X^{\mathbb{Z}^K_2}}
              \,=\,
              0
            \end{array}
          \right.
           \!\!
        \right\}\!
          \right)
        } 
      \\
      G/1
      \ar@(ul,ur)|-{
        \;\mathbb{Z}_4\;
      }
      \ar[dr]
      \ar[dd]
      &&
      \alpha
      \ar@{|->}[rr]
      \ar@{|->}[dd]
      &&
      \Bigg(
      \alpha \vert_{ X^{\mathbb{Z}_4}}
      \ar@{}[r]|-{,}
      \ar@{|->}[dd]
      &
      \Big(
        \overset{
          =: \beta
        }{
        \overbrace{
        \alpha
          \,-\,
        \mathrm{ext}_{\ZTwo^L \times \ZTwo^R}
        \big(
          \alpha \vert_{ X^{\ZTwo^L \times \ZTwo^R}}
        \big)
        }
        }
      \Big) \vert_{ X^{\mathbb{Z}^L_2} + \vert X^{\mathbb{Z}^R_2}}
      \ar@{|->}[dd]
      \ar@{}[r]|-{,}
      &
        {
        \arraycolsep=1.4pt\def\arraystretch{.1}
        \begin{array}{rcl}
          \beta
          &
          -
          \mathrm{ext}_{\mathbb{Z}^R_2}
          \big(
            \beta \vert_{ X^{\mathbb{Z}^R_2}}
          \big)
          \\
          &
          -
          \mathrm{ext}_{\mathbb{Z}^L_2}
          \big(
            \beta \vert_{ X^{\mathbb{Z}^L_2}}
          \big)
        \end{array}}
      \Bigg)
      \ar@{|->}[dd]
      \\
      &
      \scalebox{.9}{$G/\mathbb{Z}^R_2$}
      \!\!\!\!\!\!\!\!\!\!\!\!\!\!\!\!
      \ar[dddl]
      \\
      G/\mathbb{Z}^L_2
      \ar[dd]
      &&
      \alpha \vert_{ X^{\ZTwo}}
      \ar@{|->}[rr]
      \ar@{|->}[dd]
      &&
      \bigg(
      \alpha \vert_{ X^{\mathbb{Z}_4}}
      \ar@{|->}[dd]
      \ar@{}[r]|-{,}
      &
      \Big(
        \alpha
          \,-\,
        \mathrm{ext}_{\ZTwo^L \times \ZTwo^R}
        \big(
          \alpha \vert_{ X^{\ZTwo^L \times \ZTwo^R}}
        \big)
      \Big) \vert_{ X^{\mathbb{Z}^L_2} + \vert X^{\mathbb{Z}^R_2}}
      \ar@{|->}[dd]
      \ar@{}[r]|-{,}
      &
      \;\;\;\;\;\;\;\;\;\;\;\;\;
        0
      \;\;\;\;\;\;\;\;\;\;\;\;\;
      \mathrlap{\bigg)}
      \ar@{|->}[dd]
      \\
      {\phantom{A} \atop {\phantom{A}}}
      \\
      G/\mathbb{Z}^L_2 \times \mathbb{Z}^R_2
      &&
      \alpha \vert_{ X^{\mathbb{Z}_4}}
      \ar@{|->}[rr]
      &&
      \bigg(
      \alpha \vert_{ X^{\mathbb{Z}_4}}
      \ar@{}[r]|-{,}
      &
      \;\;\;\;\;\;\;\;\;\;\;\;\;\;\;
      0
      \mathclap{\phantom{\bigg)}}
      \;\;\;\;\;\;\;\;\;\;\;\;\;\;\;
      \ar@{}[r]|-{,}
      &
      \;\;\;\;\;\;\;\;\;\;\;\;\;
        0
      \;\;\;\;\;\;\;\;\;\;\;\;\;
      \mathrlap{\bigg)}
    }
    }
  $$

  \noindent
  and hence conclude the result, again by Prop. \ref{InjectiveEnvelopeOfEquivariantDualVectorSpaces}.
  The only further subtlety to take care of here
  is that the two extensions $\mathrm{ext}_{\ZTwo^L}$
  and $\mathrm{ext}_{\ZTwo^R}$
  (Notation \ref{ExtensionOfSmoothDifferentialFormsAwayFromFixedLoci})
  need to be chosen compatibly, such as to ensure that
  each preserves the property of a form to vanish on the corresponding
  other fixed locus:
  $$
    \Big(
    \mathrm{ext}_{\ZTwo^L}
    \Big(
      \beta \vert_{ X^{\ZTwo^L}}
    \Big)
    \Big) \Big\vert_{ \ZTwo^R}
    \;=\;
    0
    \,,
    \phantom{AAAA}
    \Big(
    \mathrm{ext}_{\ZTwo^R}
    \Big(
      \beta \vert_{ X^{\ZTwo^R}}
    \Big)
    \Big) \Big\vert_{ \ZTwo^L}
    \;=\;
    0
    \,.
  $$

 \noindent  This is achieved by choosing an {\it equivariant}
  tubular neighborhood
  (by \cite[\S VI, Thm. 2.2]{Bredon72}\cite[Thm. 4.4]{Kankaanrinta07})
  around the intersection
  $X^{\ZTwo^R} \cap X^{\ZTwo^L}$
  and using this to choose the extension away from
  $X^{\ZTwo^L}$ to be orthogonal to that away from
  $X^{\ZTwo^R}$.
\end{proof}

\smallskip

\noindent {\bf Equivariant graded vector spaces.}

\vspace{-1mm} 
\begin{defn}[Equivariant graded vector spaces]
  \label{EquivariantGradedVectorSpaces}
  We write
  $$
    \EquivariantGradedVectorSpaces
    \;:=\;
    \EquivariantVectorSpaces^{\mathbb{N}}
    \;\simeq\;
    \mathrm{Functors}
    \big(
      G\mathrm{Orb}^{\mathrm{op}}
      \,,\,
      \GradedVectorSpaces
    \big)
  $$

\noindent
  for the category of $\mathbb{N}$-graded
  objects in equivariant vector spaces
  (Def. \ref{EquivariantVectorSpaces}).
\end{defn}

\begin{defn}[Equivariant rational homotopy groups]
  \label{EquivariantRationalHomotopyGroups}
  For $\mathscr{X} \,\in\, \EquivariantHomotopyTypesConnected$
  (Def. \ref{EquivariantConnectedHomotopyTypes})
  and $n \in \mathbb{N}$,
  the rationalized $n$th equivariant homotopy group
  (Def. \ref{EquivariantHomotopyGroups})
  hence the stage-wise rationalized simplicial homotopy group
  (Def. \ref{EquivariantHomotopyGroups})
  $$
    \underline{\pi}_{\, n}\big( \mathscr{X} \big)
    \otimes_{\scalebox{.5}{$\underline{\mathbb{Z}}$}}
    \underline{\mathbb{R}}
    \;\;
    :
    \;\;
    G/H
    \;\;
    \longmapsto
    \;\;
    \pi_n
    \big(
      \mathscr{X}(G/H)
    \big)
    \otimes_{\scalebox{.5}{$\mathbb{Z}$}}
    \mathbb{R}\;,
  $$

  \noindent
  form an equivariant graded vector space (Def. \ref{EquivariantGradedVectorSpaces}):
  $$
    \underline{\pi}_{\, \bullet+1}\big( \mathscr{X} \big)
    \otimes_{\scalebox{.5}{$\mathbb{Z}$}}
    \mathbb{R}
    \;\;
    \in
    \;\;
    \EquivariantVectorSpaces
    \,.
  $$

\end{defn}

\begin{example}[$\Grefl$-Equivariant rational homotopy groups of twistor space]
  \label{GhetEquivariantRationalHomotopyGroupsOfTwistorSpace}
  The $\ZTwo$-equivariant rational homotopy groups
  (Def. \ref{EquivariantRationalHomotopyGroups})
  of $\Grefl$-equivariant twistor space
  (Example \ref{GhetEquivariantTwistorSpace})
  are, by \eqref{Z2EquivariantTwistorSpaceAsPresheafOnOrbitCategory},
  given by the rational homotopy groups
  of $\mathbb{C}P^3$ and, on the fixed locus, of $S^2$.
  Hence these look as follows
  (using, e.g., \cite[Lemma 2.13]{FSS20c} with \cite[Prop. 3.65]{FSS23-Char}):
  \begin{equation}
    \label{Z2EquivariantHomotopyGroupsOfTwistorSpace}
    \hspace{-4mm}
    \begin{array}{l}
      \underline{\pi}^{\mathbb{Z}/2}_{\, \bullet}
      \big(
        \mathbb{C}P^3
      \big)
      \otimes_{\scalebox{.5}{$\mathbb{Z}$}}
      \mathbb{R}
      \simeq
      \\
  \scalebox{.76}{
  \def\arraystretch{1}
  \def\tabcolsep{3pt}
  \begin{tabular}{|c||c||c|c|c|c|c|c|c|c|c|}
    \hline
    $\ZTwo/H$
    &
    $
    \mathclap{\phantom{\vert^{\vert^{\vert^{\vert}}}}}
    \big(
      \mathbb{C}P^3
    \big)^H
    \mathclap{\phantom{\vert_{\vert_{\vert_{\vert}}}}}
    $
    &
    $\pi_2 \otimes \mathbb{R}$
    &
    $\pi_3 \otimes \mathbb{R}$
    &
    $\pi_4 \otimes \mathbb{R}$
    &
    $\pi_5 \otimes \mathbb{R}$
    &
    $\pi_6 \otimes \mathbb{R}$
    &
    $\pi_7 \otimes \mathbb{R}$
    &
    $\pi_8 \otimes \mathbb{R}$
    &
    $\pi_9 \otimes \mathbb{R}$
    &
    $\cdots$
    \\
    \hline
    \hline
    $\ZTwo/1$
    &
    $
      \mathclap{\phantom{\vert^{\vert^{\vert}}}}
      \mathbb{C}P^3
    $
    &
    $\mathbf{1}$
    &
    $0$
    &
    $0$
    &
    $0$
    &
    $0$
    &
    $\mathbf{1}$
    &
    $0$
    &
    $0$
    &
    $\cdots$
    \\
    \hline
    $\ZTwo/\ZTwo$
    &
    $
      \mathclap{\phantom{\vert^{\vert^{\vert}}}}
      S^2
    $
    &
    $\mathbf{1}$
    &
    $\mathbf{1}$
    &
    $0$
    &
    $0$
    &
    $0$
    &
    $0$
    &
    $0$
    &
    $0$
    &
    $\cdots$
    \\
    \hline
  \end{tabular}
  }
  \end{array}
  \end{equation}
\end{example}

\medskip

\noindent {\bf Equivariant cochain complexes.}

\vspace{-2mm} 
\begin{defn}[Equivariant cochain complexes]
  \label{EquivariantCochainComplexes}
  We write
  $$
    \EquivariantCochainComplexes
    \;;=\;
    \mathrm{Functors}
    \big(
      G \mathrm{Orb}
      \,,\,
      \CochainComplexes
    \big)
  $$

  \noindent
  for the category of functors from the $G$-orbit category
  (Def. \ref{OrdinaryOrbitCategory}) to the
  category of connective cochain complexes
  (i.e., in non-negative degrees with differential of degree +1)
  over the real numbers.
\end{defn}

\begin{defn}[Delooping of equivariant cochain complexes]
  \label{DeloopingOfEquivariantCochainComplex}
  For
  $V \in \EquivariantCochainComplexes$ (Def. \ref{EquivariantCochainComplexes}),
  we  denote its delooping as
  \vspace{-2.5mm} 
  $$
    \mathfrak{b}V
    :\;
    G/H
    \longmapsto 
    \Big(\!\!
      \xymatrix@C=24pt{
        0
        \ar[r]
        &
        V^0(G/H)
        \ar[r]^-{ d^0_V }
        &
        V^1(G/H)
        \ar[r]^-{ d^1_V }
        &
        V^2(G/H)
        \ar[r]
        &
        \cdots
      }
   \!\! \Big).
  $$
\end{defn}

As an instance of the general notion of mapping cones
(e.g. \cite[Def. 3.2.2]{Schapira11}), we get:
\begin{example}[Cone on an equivariant cochain complex]
  \label{ConeOnEquivariantCochainComplex}
  For
  $V \in \EquivariantCochainComplexes$ (Def. \ref{EquivariantCochainComplexes}),
  we say that the \emph{cone} on
  its delooping $\mathfrak{b}V$ (Def. \ref{DeloopingOfEquivariantCochainComplex})
  is the equivariant cochain complex
  $
    \mathfrak{e}V
    \;\in\;
    \EquivariantCochainComplexes
  $
  given by
  \vspace{-2mm}
  $$
    \begin{array}{l}
    \mathfrak{e}V
    \;:=\;
    \mathrm{Cone}(\mathfrak{b}V)
    \;:\;
    G/H
    \;\longmapsto\;
    \\
    \left(\!\!\!\!
    \raisebox{18pt}{
    \xymatrix@R=15pt@C=30pt{
      V^{0}(G/H)
      \ar[r]^-{ -d_V^{0} }
      \ar@{}[d]|-{\oplus}
      \ar[dr]|-{ \;\mathrm{id}\; }
      &
      V^{1}(G/H)
      \ar[r]^-{ -d_V^{1} }
      \ar@{}[d]|-{\oplus}
      \ar[dr]|-{ \;\mathrm{id}\; }
      &
      V^{2}(G/H)
      \ar[r]^-{ -d_V^{2} }
      \ar@{}[d]|-{\oplus}
      \ar[dr]|-{ \;\mathrm{id}\; }
      &
      V^{3}(G/H)
      \ar[r]^-{ -d_V^{3} }
      \ar@{}[d]|-{\oplus}
      \ar[dr]|-{ \;\mathrm{id}\; }
      &
      \cdots
      \\
      0
      \ar[r]_-{ \;0\; }
      &
      V^{0}(G/H)
      \ar[r]_-{ d_V^{0} }
      &
      V^{1}(G/H)
      \ar[r]_-{ d_V^{1} }
      &
      V^{2}(G/H)
      \ar[r]_-{ d_V^{2} }
      &
      \cdots
    }
    }
    \right)
    \,.
    \end{array}
  $$

  \noindent
  This sits in the evident cofiber sequence:
  \vspace{-1mm} 
  \begin{equation}
    \label{ConeFiberSequence}
    \raisebox{16pt}{
    \xymatrix@R=1.2em{
      V
      \ar@{<-}[rr]^-{ \mathrm{cofib}(i_{\mathfrak{b}V}) }
      &&
      \mathfrak{e}V
      \ar@{<-}[d]^-{ i_{\mathfrak{b}V} }
      \\
      &&
      \mathfrak{b}V
    }
    }
    \;\;\;\;\;\;\;\;\;
    \in
    \;\;
    \EquivariantCochainComplexes \;.
  \end{equation}
  
\end{example}

As an instance of the general notion of injective resolutions
(e.g. \cite[\S 4.5]{Schapira11}), we have:

\begin{example}[Injective resolution of equivariant dual vector spaces] 
  \label{InjectiveResolutionOfEquivariantDualVectorSpaces}
    Let $V \in \EquivariantDualVectorSpaces$ (Def. \ref{EquivariantVectorSpaces}).
  Then, by Prop. \ref{InjectiveEnvelopeOfEquivariantDualVectorSpaces},
  we obtain an \emph{injective resolution} (e.g. \cite[p. 129]{HiltonStammbach71})
  of $V$ given by the equivariant cochain complex (Def. \ref{EquivariantCochainComplexes})
  which in degree 0 is the injective envelope \eqref{InjectiveEnvelopeConstruction}
  of $V$, and whose differentials are, recursively,
  the injective envelope inclusions \eqref{CanonicalInjectiveEnvelope}
  of the quotients by the
  image of the previous degree:
 \vspace{-3mm}
  $$
  \raisebox{54pt}{
   \xymatrix@R=14pt{
     \vdots
      \ar@{<-}[d]
      &
      \vdots
      \ar@{<-}[d]
      \\
      0
      \ar@{<-}[d]
      \ar[r]
      &
      \mathrm{Inj}
      \big(
        \mathrm{coker}(d^1)
      \big)
      \ar@{<-}[d]^-{ d^2 }
      \\
      0
      \ar@{<-}[d]
      \ar[r]
      &
      \mathrm{Inj}
      \big(
        \mathrm{coker}(d^0)
      \big)
      \ar@{<-}[d]^-{ d^1 }
      \\
      0
      \ar[r]
      \ar@{<-}[d]
      &
      \mathrm{Inj}
      \big(
        \mathrm{Inj}(V)/V
      \big)
      \ar@{<-}[d]^-{ d^0 }
      \\
      V
      \;
      \ar@{^{(}->}[r]
      &
      \;
      \mathrm{Inj}(V)
    }
      }
    \;\;\;\;\;\;\;
    =:
    \mathrm{Inj}^\bullet(V)
    \;
    \in
    \;
    \EquivariantCochainComplexes
    \mathrlap{\,.}
  $$

  This is such that for any
  $A^\bullet \in \EquivariantCochainComplexes$
  which is degreewise injective
  (Def. \ref{InjectiveObjects})
  and any morphism of equivariant dual vector spaces
  \vspace{-3mm}
  $$
  \big\{\!  \xymatrix{
      V
      \ar[r]^-{ \phi }
      &
      A^n_{\mathrm{clsd}}
    }
    \!\big\}
    \;\;\;\;
    \in
    \;
    \EquivariantDualVectorSpaces
  $$
  from $V$ to the subspace of closed elements
  (cocycles) in $A^n$, there exists an extension to
  a morphism
  \begin{equation}
    \label{InjectiveExtension}
    \hspace{-3mm}
    \big\{\! \xymatrix{
      \mathfrak{b}^n
      \mathrm{Inj}^\bullet(V)
      \ar[r]^-{ \phi^\bullet }
      &
      A^\bullet
    }
    \!\big\}
    \;\;\;\;
    \in
    \;
    \EquivariantCochainComplexes
  \end{equation}
  of equivariant cochain complexes
  \eqref{InjectiveResolution}
  given recursively
  by using injectivity of $A^{n+i+1}$
  to obtain dashed extensions \eqref{ExtensionOfMapIntoInjectiveObject}
  $$
    \xymatrix@R=12pt@C=3em{
      \mathrm{Inj}^{i+1}(V)
      \ar@{-->}[rr]^{ \phi^{n + i + 1} }
      &&
      A^{n+i+1}\;.
      \\
{}^{\phantom{A}^{\phantom{A}}}\mathrm{Inj}^{i}(V)/\mathrm{im}(d^{i-1})
      \ar@{^{(}->}[u]
      \ar[urr]_-{ d_A \circ \phi^i }
    }
  $$
  
  \begin{equation}
    \label{InjectiveResolution}
    \raisebox{85pt}{
    \xymatrix@C=70pt@R=18pt{
      \vdots
      \ar@{<-}[d]
      &
      \vdots
      \\
      \mathrm{Inj}
      \big(
        \mathrm{coker}(d^1)
      \big)
      \ar@{<-}[d]^-{ d^2 }
      \ar@{-->}[r]^-{ \;\phi^{n+3}\; }
      &
      A^{n+3}
      \ar[u]_-{ d_A^{n+3} }
      \\
      \mathrm{Inj}
      \big(
        \mathrm{coker}(d^0)
      \big)
      \ar@{<-}[d]^-{ d^1 }
      \ar@{-->}[r]^-{ \;\phi^{n+2}\; }
      &
      A^{n+2}
      \ar[u]_-{ d_A^{n+2} }
      \\
      \mathrm{Inj}
      \big(
        \mathrm{Inv}(V)/V
      \big)
      \ar@{<-}[d]^-{ d^0 }
      \ar@{-->}[r]^-{ \;\phi^{n+1}\; }
      &
      A^{n+1}
      \ar[u]_-{ d_A^{n+1} }
      \\
      \mathrm{Inj}(V)
      \ar@{-->}[r]^{ \;\phi^n\; }
      &
      A^{n}
      \ar[u]_-{ d_A^n }
      \\
      \mathclap{\phantom{\vert^\vert}}
      V
      \ar@{^{(}->}[u]
      \ar[r]^{ \;\phi \;=:\; \phi^{n}_{\vert V} \; }
      &
      \mathclap{\phantom{\vert^\vert}}
      A^n_{\mathrm{clsd}}
      \ar@{^{(}->}[u]
    }
    }
  \end{equation}
\end{example}

\begin{example}[Injective resolution of $\ZTwo$-equivariant dual vector spaces]
  \label{InjectiveResolutionOfZ2equiavriantVectorSpaces}
  Consider the $\ZTwo$-equivariant dual vector space
  (Example \ref{Z2EquivariantVectorSpaces}) given by
  \begin{equation}
    \label{Z2equivariantVectorSpaceConcentratedOnFixedLocus}
    \left(\!\!   \raisebox{10pt}{
    \xymatrix@C=4pt@R=1em{
      \ZTwo/1
      \ar@(ul,ur)|-{\; \ZTwo\, }
      \ar[d]
      &\longmapsto&
      0
      \ar[d]^-{0}
      \\
      \ZTwo/\ZTwo
      &\longmapsto&
      \mathbf{1}
    }
    }
    \right)
    \;\;\;
    \in
    \;\;
    \ZTwo \EquivariantDualVectorSpaces\;.
  \end{equation}
  Recalling the three injective atoms of
  $\ZTwo$-equivariant dual vector spaces from Example
  \ref{InjectiveZ2EquivariantDualVectorSpaces},
  we find that the injective resolution
  (Example \ref{InjectiveResolutionOfEquivariantDualVectorSpaces})
  of \eqref{Z2equivariantVectorSpaceConcentratedOnFixedLocus} is
  the $\ZTwo$-equivariant cochain complex:.
  $$
    \xymatrix@C=5pt@R=-1.5pt{
      & &
      &&&&
      & & &
      \,\,\,\rotatebox[origin=c]{-27}{$\vdots$}
      \\
      \\
      \\
      & &
      &&&&
      & &
      0
      \ar[dddd]
      \ar[uuur]
      \\
      & &
      &&&&
      & & &
      \,\,\,\rotatebox[origin=c]{-27}{$\vdots$}
      \\
      \\
      & &
      &&&&
      &
      \mathbf{1}
      \ar[dddd]
      \ar[uuur]
      \\
      & &
      &&&&
      & &
      0
      \ar[uuur]
      \\
      \\
      \ZTwo/1
      \ar@(ul,ur)|-{\; \ZTwo\, }
      \ar[dddd]
      &\mapsto&
      0
      \;
      \ar[dddd]
      \ar@{^{(}->}[rrrr]
      &&&&
      \mathbf{1}
      \ar[dddd]|-{
        \phantom{\vert}
        \mathrm{id}
        \phantom{\vert}
      }
      \ar[uuur]^<<<<<<{
        \phantom{\vert}
        \mathrm{id}
        \phantom{\vert}
      }
      \\
      & &
      &&&&
      &
      0
      \ar[uuur]
      \\
      \\
      \\
      \ZTwo/\ZTwo
      &\mapsto&
      \mathbf{1}
      \;
      \ar@{^{(}->}[rrrr]
      &&&&
      \mathbf{1}
      \ar[uuur]
    }
  $$

  In terms of generators-and-relations
  (Notation \ref{GeneratorsAndRelationsPresentationOfCochainComplexes}), this says:
  \begin{equation}
    \label{InjectiveResolutionOfZ2EquivariantVectorSpace}
    \def\arraystretch{1}
    \begin{array}{l}
    \mathrm{Inj}^\bullet
    \left(\!\!
      \raisebox{19pt}{
      \xymatrix@C=8pt{
        \ZTwo/1
        \ar@(ul,ur)|-{\; \ZTwo}
        \ar[d]
        &\longmapsto&
        0
        \ar[d]
        \\
        \ZTwo/\ZTwo
        &\longmapsto&
        \mathbb{R}
        \langle
          c_0
        \rangle
        \big/
        (
          d\, c_0 \, = 0
        )
      }
      }
    \! \right)
    \\
    \\
    =
    \;\;\;\;
    \left(\!\!
      \raisebox{20pt}{
      \xymatrix@C=8pt@R=1.5em{
        \ZTwo/1
        \ar@(ul,ur)|-{\; \ZTwo}
        \ar[d]
           &\longmapsto&
        \mathbb{R}
        \left\langle
          \!\!\!
          {\begin{array}{l}
            c'_0, c_1
            \\
            c_0
          \end{array}}
          \!\!\!
        \right\rangle
        \big/\!\!
        \left(
          {\begin{aligned}
            d\, c'_0 & = c_1
            \\[-5pt]
            d \, c_1 & = 0
            \\[-5pt]
            d\, c_0 &  = 0
          \end{aligned}}
        \right)
        \ar@<-22pt>@{->>}[d]
        \\
        \ZTwo/\ZTwo
        &\longmapsto&
        \mathbb{R}
        \langle
          c_0
        \rangle
        \big/
        \big(
          d\, c_0 \, = 0
        \big)
      }
      }
    \!\!\!\!\! \right)
    \,.
    \end{array}
  \end{equation}
\end{example}

\medskip

\noindent {\bf Equivariant dgc-algebras.}

\begin{defn}[Equivariant dgc-Algebras]
  \label{EquivariantdgcAlgebras}
  We write
  $$
    \EquivariantdgcAlgebras
    \; :=\;
    \mathrm{Functors}
    \big(
      G \mathrm{Orb}
      \,,\,
      \dgcAlgebras
    \big)
  $$

  \noindent
  for the category of functors from the $G$-orbit category
  (Def. \ref{OrdinaryOrbitCategory}) to the
  category of connective dgc-algebras over the real numbers.
\end{defn}

\begin{defn}[Equivariant cochain cohomology groups]
  \label{EquivariantCochainCohomologyGroups}
  For $A \in \EquivariantdgcAlgebras$ (Def. \ref{EquivariantdgcAlgebras})
  and $n \in \mathbb{N}$, we write
    $$
    \underline{H}^n(A)
    \;\;
    \in
    \;
    \EquivariantDualVectorSpaces
  $$

  \noindent
  for the equivariant dual vector space (Def. \ref{EquivariantVectorSpaces})
  of cochain cohomology groups
  $$
    \underline{H}^n(A)
    \;\;
    :
    \;\;
    G/H
    \;\;
    \longmapsto
    \;\;
    H^n
    \big(
      A(G/H)
    \big)
    \,.
  $$
\end{defn}

\begin{example}[Equivariant base dgc-algebra]
  \label{EquivariantdgcBaseAlgebra}
  We write
  $
    \underline{\mathbb{R}}
    \,
    \in
    \,
       \EquivariantdgcAlgebras
  $
  for the equivariant dgc-algebra (Def. \ref{EquivariantdgcAlgebras})
  which is constant on the ground field $\mathbb{R}$:
  $$
    \underline{\mathbb{R}}
    \;\;
    :
    \;\;
    G/H
    \;\;
    \longmapsto
    \;\;
    \mathbb{R}
    \,.
  $$
  \vspace{-.6cm}

  For the case $G = \ZTwo$ (Example \ref{OrbitCategoryOfZ2}):
  $$
    \raisebox{12pt}{
    \xymatrix@R=12pt@C=7pt{
      \ZTwo/1
      \ar@(ul,ur)|-{\; \ZTwo}
      \ar[d]
      &\longmapsto&
      \mathbb{R}
      \ar[d]^-{ \mathrm{id} }
      \\
      \ZTwo/\ZTwo
      &\longmapsto&
      \mathbb{R}
    }
    }
  $$
\end{example}

\begin{example}[Equivariant smooth de Rham complex]
  \label{EquivariantSmoothdeRhamComplex}
  $\,$

  \noindent
  For
  $G \acts \; X \,\in\, \GActionsOnSmoothManifolds$
  (Def. \ref{ProperGActionsOnSmoothManifolds}).
  there is the equivariant dgc-algebra
  (Def. \ref{EquivariantdgcAlgebras})  
  $$
    \Omega^\bullet_{\mathrm{dR}}
    \big(
      \orbisingular( X \!\sslash\! G )
    \big)
    \;\;
    \in
    \;\;
    \EquivariantdgcAlgebras
  $$

  \noindent
  of equivariant smooth differential forms
  (Example \ref{EquivariantSmoothDifferentialForms})
  equipped with the wedge product and de Rham
  differential formed stage-wise, as in the
  ordinary smooth de Rham complex (e.g. \cite{BottTu82})
  of the fixed loci.
\end{example}

\begin{example}[Free equivariant dgc-algebra on equivariant cochain complex]
  \label{FreeEquivariantdgcAlgebraOnEquivariantCochainComplex}
  $\,$
  
\noindent   For $V^\bullet \,\in\, \EquivariantCochainComplexes$
  (Def. \ref{EquivariantCochainComplexes}):

  \noindent {\bf (i)} We obtain
  the free equivariant dgc-algebra (Def. \ref{EquivariantdgcAlgebras})
  $$
    \mathrm{Sym}(V^\bullet)
    \;\;
      \in
    \;\;
    \EquivariantdgcAlgebras \;,
  $$

\noindent   given over each $G/H \,\in\, G \mathrm{Orb}$,
  by the free dgc-algebra on the cochain complex at that stage:
  $$
    \mathrm{Sym}(V^\bullet)
    \;
      :
    \;
    G/H
    \;\mapsto\;
    \mathrm{Sym}
    \big(
      V^\bullet(G/H)
    \big),
  $$

  \noindent
  with all structure maps induced by the functoriality of
  the non-equivariant $\mathrm{Sym}$-construction.

 \noindent {\bf (iii)}  This extends to a functor
 \vspace{-2mm}
  \begin{equation}
    \label{AdjunctionBetweenEquivariantdgcAlgebrasAndEquivariantCochainComplexes}
    \xymatrix{
      \EquivariantdgcAlgebras \;
      \ar@{<-}@<+6pt>[rrr]^-{ \mathrm{Sym} }
      \ar@<-6pt>[rrr]_-{ \mathrm{CchnCmplx}  }^-{ \bot }
      &&&
  \;    \EquivariantCochainComplexes
    },
  \end{equation}

  \noindent
  which is left adjoint to the evident assignment of
  underlying equivariant cochain complexes.
\end{example}

In terms of generators and relations
(Notation \ref{GeneratorsAndRelationsPresentationOfCochainComplexes},
\ref{GeneratorsAndRelationsPresentationOfdgcAlgebras}),
passing to free dgc-algebras means to replace angular brackets by
square brackets:

\begin{example}[Free $\ZTwo$-equivariant dgc-algebra on injective resolution]
  \label{FreeZ2EquivariantdgcAlgebra}
  In the case $G = \ZTwo$ (Example \ref{OrbitCategoryOfZ2}),
  the free $\ZTwo$-equivariant dgc-algebra
  (Example \ref{FreeEquivariantdgcAlgebraOnEquivariantCochainComplex})
  on the $n$-fold delooping (Def. \ref{DeloopingOfEquivariantCochainComplex})
  of the injective resolution \eqref{InjectiveResolutionOfZ2EquivariantVectorSpace}
  from Example \ref{InjectiveResolutionOfZ2equiavriantVectorSpaces}
  is:
  \begin{equation}
    \label{ExampleOfInjectiveResolutionOfZ2EquivariantVectorSpace}
    \hspace{-5mm} 
    \def\arraystretch{2}
    \begin{array}{ll}
    \mathrm{Sym}
    \circ
    \mathfrak{b}^n
    \circ
    \mathrm{Inj}^\bullet
    \!
    &
    \hspace{-3mm} 
    \left(
      \!\!\!\!\!
      \raisebox{20pt}{
      \xymatrix@C=2pt{
        \ZTwo/1
        \ar@(ul,ur)|-{\; \ZTwo}
        \ar[d]
        &\longmapsto&
        0
        \ar[d]
        \\
        \ZTwo/\ZTwo
        &\longmapsto&
        \mathbb{R}
        \langle
          c_0
        \rangle
        \big/
        (
          d\, c_0 \, = 0
        )
      }
      }
      \!\!\!
    \right)
    \\
      &
      =
    \left(
      \!\!\!\!
      \raisebox{20pt}{
      \xymatrix@C=2pt@R=1.5em{
        \ZTwo/1
        \ar@(ul,ur)|-{\; \ZTwo}
        \ar[d]
        &\longmapsto&
        \mathbb{R}
        \!
        \left[
          \!\!\!
          {\begin{array}{l}
            c'_n, c_{n+1}
            \\
            c_n
          \end{array}}
          \!\!\!
        \right]
        \!\big/\!
        \left(
          {\begin{aligned}
            d\, c'_n & = c_{n+1}
            \\[-5pt]
            d \, c_{n+1} & = 0
            \\[-5pt]
            d\, c_n &  = 0
          \end{aligned}}
        \right)
        \ar@<-22pt>@{->>}[d]
        \\
        \ZTwo/\ZTwo
        &\longmapsto&
        \mathbb{R}
        [
          c_n
        ]
        \big/
        (
          d\, c_n \, = 0
        )
      }
      }
      \!\!\!\!
    \right).
    \end{array}
  \end{equation}
\end{example}

In equivariant generalization of \cite[Def. 3.25]{FSS23-Char}, we have:

\begin{defn}[Equivariant $L_\infty$-algebras]
\label{EquivariantLInfinityAlgebras}
We write
\begin{equation}
  \label{DefiningEmbeddingOfEquivariantLInfinityAlgebrasInEquivariantdgcAlgebras}
  \xymatrix@R=-5pt{
    \EquivariantLInfinityAlgebras
    \;
    \ar@{^{(}->}[rr]^-{ \mathrm{CE} }
    &&
    \EquivariantdgcAlgebrasOp
    \\
    \underline{\mathfrak{g}}
    \ar@{}[rr]|<<<<<<<<<<<<<<<<<<<<<<{ \longmapsto }
    &&
    \mathrm{CE}
    \big(\,
      \underline{\mathfrak{g}}
    \, \big)
  }
\end{equation}

\noindent
for the opposite of the full subcategory of
equivariant dgc-algebras (Def. \ref{EquivariantdgcAlgebras})
on those that are stage-wise Chevalley-Eilenberg algebras
of $L_\infty$-algebras
(connective and finite-type  over the real numbers, as in \cite[Def. 3.25]{FSS23-Char}).

\end{defn}

In generalization of Example \ref{EquivariantSmoothdeRhamComplex},
we have:

\begin{example}[Proper $G$-equivariant and Borel-Weil-Cartan $T$-equivariant smooth de Rham complex]
\label{ProperGEquivariantAndBorelTEquivariantSmoothDeRhamComplex}
$\,$

\noindent
Let $\big( T \!\times G \big) \acts \; X \,\in\, \TGActionsOnSmoothManifolds$
(Def \ref{ProperGActionsOnSmoothManifolds}),
where $T \,\in\, \mathrm{CompactLieGroups}$ is finite-dimensional
with Lie algebra denoted (as in Notation \ref{GeneratorsAndRelationsPresentationOfCochainComplexes})
\begin{equation}
  \label{LieAlgebraOfCompactEquivarianceGroup}
  \mathfrak{t}
  \;\simeq\;
  \Big\{
    \langle t_a \rangle_{a = 1}^{\mathrm{dim}(T)}
    \,,\,
    [-,-]
  \Big\}
  \;\in\;
  \LieAlgebras
  \,.
\end{equation}

\noindent
  Consider the equivariant dgc-algebra
  (Def. \ref{EquivariantdgcAlgebras})
  $$
    \Omega^\bullet_{\mathrm{dR}}
    \Big(\!
      \big(
        \orbisingular( X \!\sslash\! G )
      \big) \!\sslash\! T
    \Big)
    \;\;
    \in
    \;\;
    \EquivariantdgcAlgebras
  $$

  \noindent
  of $T$-invariants
  in the tensor product of
  proper $G$-equivariant smooth differential forms
  (Example \ref{EquivariantSmoothDifferentialForms})
  with
  the free symmetric graded algebra
  on
  $$
    \mathfrak{b}^2\mathfrak{t}^\vee
    \;\simeq\;
    \langle r^{\, a}_2\rangle_{a = 1}^{\mathrm{dim}(T)}
    \,,
  $$

  \noindent
  (the linear dual space
  of
  \eqref{LieAlgebraOfCompactEquivarianceGroup} in degree 2)
  and equipped  with the sum of the
  de Rham differential
  $$
    d_{\mathrm{dR}}
    \;:\;
    \omega \wedge r_2^{a_1} \wedge \cdots \wedge r_2^{a_p}
    \;\longmapsto\;
    \big(
      d_{\mathrm{dR}}
      \omega
    \big)
    \wedge r_2^{a_1} \wedge \cdots \wedge r_2^{a_p}
  $$
  \vspace{-.55cm}

  \noindent
  and the operator
  \vspace{0mm}
  $$
    r_2^a \wedge \iota_{\, t_a}
    \;:\;
    \omega \wedge r_2^{a_1} \wedge \cdots \wedge r_2^{a}
    \;\longmapsto\;
    \big(
      \iota_{\,t_a} \omega
    \big)
    \wedge
    r^{\, a}_2
    \,
    \wedge r_2^{a_1} \wedge \cdots \wedge r_2^{a}
    \,,
  $$
  \vspace{-.6cm}

  \noindent
  where
  \begin{itemize}[
    itemsep=2pt
  ]
  
\item  $\omega \in \Omega^\bullet_{\mathrm{dR}}(-)$,
  
\item  $\iota_{\,t_a}$ denotes the contraction of
  differential forms with the vector field that
  is the derivative of the action $T \times X \to X$
  along $t_a$,
    
  \item summation over the
  index $a \in \{1, \cdots, \mathrm{dim}(T)\}$ is understood,
  and
    
 \item  the $T$-action
  on
  $\mathfrak{t}^\vee$ is the coadjoint action
  and on that differential forms
  is by pullback along the
  given action on $X$:
  \end{itemize}

  \begin{equation}
    \label{CartanModelOnFixedLocus}
    \hspace{-18mm} 
    \def\arraystretch{2}
    \begin{array}{l}
    \overset{
      \mathclap{
      \raisebox{3pt}{
        \tiny
        \color{darkblue}
        \bf
        \def\arraystretch{1}
        \begin{tabular}{c}
          proper $G$-equivariant 
          \\
          \&
          Borel $T$-equivariant
          \\
          smooth de Rham complex
        \end{tabular}
      }
      }
    }{
    \Omega^\bullet_{\mathrm{dR}}
    \Big(\!
      \big(
        \orbisingular( X \!\sslash\! G )
      \big) \!\sslash\! T
    \Big)
    }
    \;\;
    :
    \\
    \;\;
    G/H
    \;\;
    \longmapsto
    \;\;
    \overset{
      \mathclap{
      \raisebox{+0pt}{
        \tiny
        \color{darkblue}
        \bf
        \begin{tabular}{c}
          Cartan model for
          $T$-equivariant Borel cohomology
          of $H$-fixed locus $X^H$
        \end{tabular}
      }
      }
    }{
    \Big(
    \Omega^\bullet_{\mathrm{dR}}
    \big(
      X^H
    \big)
    \Big[
      \{r^a_2\}_{a = 1}^{\mathrm{dim}(T)}
    \Big]
    \,,\,
      d_{\mathrm{dR}}
        +
      r_2^a \wedge  \iota_{t_a}
    \Big)^T.
    }
    \end{array}
  \end{equation}

  \noindent 
  This is, stage-wise over $G/H \in G\mathrm{Orb}$
  (Def. \ref{OrdinaryOrbitCategory}),
  the {\it Cartan model} dgc-algebra for Borel $T$-equivariant
  de Rham cohomology
  (\cite{AtiyahBott84}\cite[\S 5]{MathaiQuillen86}\cite{Kalkman93}\cite{GuilleminSternberg99},
  review in \cite{Meinrenken06} \cite{KubelThom15}\cite{Pestun17}),
  here formed for the fixed submanifolds
  (Lemma \ref{FixedLociOfProperSmoothActionsAreSmoothManifolds})
  of  all the subgroups of the $G$-action.

\end{example}

\noindent {\bf Homotopy theory of equivariant dgc-algebras.}

\begin{prop}[Projective model structure on connective equivariant dgc-algebras {\cite[Theorem 3.2]{Scull02}}]
  \label{ProjectiveModelStructureOnConnectiveEquivariantdgcAlgebras}
  There is the structure of a model category
  on $\EquivariantdgcAlgebras$ (Def. \ref{EquivariantdgcAlgebras})
  whose
  \begin{itemize}[
  leftmargin=.85cm,
  itemsep=1pt,
  topsep=2pt
] 
   \item[{$\mathrm{W}$}] -
    {\it weak equivalences} are the quasi-isomorphisms
    over each $G/H \in G\ \mathrm{Orb}$;

   \item[{$\mathrm{Fib}$}] -
    {\it fibrations} are the degreewise surjections
    whose degreewise kernels
    are injective (Def. \ref{InjectiveObjects}).
  \end{itemize}

\noindent  We denote this model category by
  $$
    \EquivariantdgcAlgebrasProj
    \;\in\;
    \mathrm{ModCat} \;.
  $$

\end{prop}

A key technical subtlety of the model structure on
equivariant dgc-algebras (Prop. \ref{ProjectiveModelStructureOnConnectiveEquivariantdgcAlgebras}),
compared to its non-equivariant version
(\cite[\S 4.3]{BousfieldGugenheim76}\cite[\S V.3.4]{GelfandManin96} \cite[Prop. 3.36]{FSS23-Char}),
is that not all objects are fibrant anymore,
since equivariantly the injectivity condition
(Def. \ref{InjectiveObjects})
is non-trivial (Prop. \ref{InjectiveEnvelopeOfEquivariantDualVectorSpaces}).
However, we have the following class of examples of fibrant objects:

\begin{prop}[Equivariant smooth de Rham complex is projectively fibrant]
  \label{EquivariantSmoothDeRhamComplexIsProjectivelyFibrant}
 For
  $G \acts \; X \,\in\, \GActionsOnSmoothManifolds$
  (Def. \ref{ProperGActionsOnSmoothManifolds}),
  the equivariant smooth de Rham complex (Example \ref{EquivariantSmoothdeRhamComplex})
  is a fibrant object in the projective model structure
  (Prop. \ref{ProjectiveModelStructureOnConnectiveEquivariantdgcAlgebras})
  $$
    \xymatrix{
      \Omega^\bullet_{\mathrm{dR}}
      \big(\!
        \orbisingular( X \!\sslash\! G )
      \big)
      \ar[rr]^-{ \in\; \mathrm{Fib} }
      &&
      0
    }
    \;\;\;\;
      \in
    \;\;\;
    \EquivariantdgcAlgebrasProj\;,
  $$

  \noindent at least if $G$ is of order 4 or cyclic of prime order.
\end{prop}
\begin{proof}
  By Prop. \ref{ProjectiveModelStructureOnConnectiveEquivariantdgcAlgebras},
  the statement is equivalent to the claim that the
  equivariant dual vector spaces of equivariant smooth differential
  $n$-forms are injective. This is indeed the case, by
  Lemmas
  \ref{ZpEquivariantSmoothDifferentialFormsFormInjectiveDualVectorSpace},
  \ref{Z4EquivariantSmoothDifferentialFormsAreInjective},
  \ref{Z2TimesZ2EquivariantDifferentialFormsAreInjective}
  (Remark \ref{EquivariantSmoothDifferentialFormsAreInjective}).
\end{proof}

Next we turn to discussion of
fibrant and cofibrant equivariant dgc-algebras.

\medskip

\noindent {\bf Minimal equivariant dgc-algebras.}

\vspace{-1mm} 
\begin{defn}[Minimal equivariant dgc-algebras {\cite[Construction 5.10]{Tri82}\cite[\S 11]{Scull02}\cite[\S 4]{Scull08}}]
  \label{MinimalEquivariantdgcAlgebras}

  $\,$

  \noindent
  Let $A \in \EquivariantdgcAlgebras$ (Def. \ref{EquivariantdgcAlgebras})
  be such that, for all $k \in \mathbb{N}$, the underlying
  $
    \mathrm{ChnCmplx}(A)^k
    \in
    \EquivariantCochainComplexes
  $
  is injective (Def. \ref{InjectiveObjects}).

  \noindent
  {\bf (i)} For $n \in \mathbb{N}$, an \emph{elementary extension}
  $
    A
    \longhookrightarrow
      A
      [
        \mathfrak{b}^n V
      ]_\phi
  $
  of $A$ in degree $n$ is a pushout of
  the image under $\mathrm{Sym}$ (Example \ref{FreeEquivariantdgcAlgebraOnEquivariantCochainComplex})
  of the cone inclusion
  (Example \ref{ConeOnEquivariantCochainComplex})
  of the $(n+1)$-fold delooping (Def. \ref{DeloopingOfEquivariantCochainComplex})
  of the injective resolution $\mathrm{Inj}^\bullet(V)$
  (Example \ref{InjectiveResolutionOfEquivariantDualVectorSpaces})
  \begin{equation}
    \label{ElementaryExtensionPushout}
    \hspace{-7mm} 
    \raisebox{28pt}{
    \xymatrix@R=1.2em@C=1em{
      A
      \big[
        \mathfrak{b}^n V_n
      \big]_{\phi_n}
      \ar@{<-^{)}}[dd]
      \ar@{<-}[rr]
      \ar@{}[ddrr]|-{ \mbox{\tiny(po)} }
      &&
      \mathrm{Sym}
      \big(
    \mathfrak{e}\mathfrak{b}^n \mathrm{Inj}^\bullet(V_n)
      \big)
      \ar@{<-^{)}}[dd]^-{
        \mathrm{Sym}
        \big(
\raisebox{1pt}{$i_{\mathfrak{b}^{n+1} \mathrm{Inj}^\bullet(V_n)}$}
        \big)
      }
      \\
      \\
      \mathclap{\phantom{\vert^{\vert^\vert}}}
      A
      \ar@{<-}[rr]^-{ \widetilde{ \phi_n^\bullet } }
      &&
      \mathclap{\phantom{\vert^{\vert^\vert}}}
      \mathrm{Sym}
      \big(
        \mathfrak{b}^{n+1}\mathrm{Inj}^\bullet(V_n)
      \big)
    }
    }
    \;
    \in
    \EquivariantdgcAlgebras
  \end{equation} 
  along the adjunct $\widetilde \phi^\bullet$
  \eqref{AdjunctionBetweenEquivariantdgcAlgebrasAndEquivariantCochainComplexes}
  of an injective extension \eqref{InjectiveExtension}

  \vspace{0mm}
\noindent
  \begin{equation}
    \label{ExtensionToInjectiveResolutionOfAttachingMap}
  \xymatrix{
      A^\bullet
      \ar@{<-}[r]^-{ \phi_n^\bullet }
      &
      \mathfrak{b}^{n+1}
      \mathrm{Inj}^\bullet(V_n)
      }
      \in
      \;
      \EquivariantCochainComplexes
      \end{equation}

\noindent
  of a given {\it attaching map} out of a given
  equivariant dual vector space
  $V_n$ (Def. \ref{EquivariantVectorSpaces}):
  \begin{equation}
    \label{AttachingMap}
    \xymatrix{
      A^{n+1}_{\mathrm{clsd}}
      \ar@{<-}[r]^-{ \phi_n }
      &
      V_n
      }
      \;
      \in
      \;
      \EquivariantDualVectorSpaces
      \,.
      \end{equation}

\vspace{-4mm}
  \noindent
  {\bf (ii)}
  An inclusion
  \begin{equation}
    \label{RelativeMinimalEquivariantdgcAlgebra}
    \xymatrix{
      B^\bullet
      \;
      \ar@{^{(}->}[r]^-{ \mathrm{min} }
      &
      A^\bullet \; \in \EquivariantdgcAlgebras
    }
  \end{equation}

  \noindent
  of degreewise injective (Def. \ref{InjectiveObjects})
  equivariant dgc-algebras (Def. \ref{EquivariantdgcAlgebras})
  which are equivariantly 1-connected
  $$
    B^0 \;\simeq\; \underline{\mathbb{R}}
    \,,\;\;\;\;
    B^1 \;\simeq\; \underline{\mathbb{R}}
  $$

  \noindent
  is called \emph{relative minimal} if it is isomorphic
  under $B^\bullet$
  to the result of a sequence of elementary extensions
  \eqref{ElementaryExtensionPushout}
  in strictly increasing degrees
  (noticing with Lemma \ref{TensorProductPreservesInjectivityOfFiniteDimensionalDualVectorGSpaces},
  that the result of an elementary extension
  \eqref{ElementaryExtensionPushout} is again degreewise injective).

  \noindent
  {\bf (iii)}  An equivariant dgc-algebra $A^\bullet$,
  such that the unique inclusion of the equivariant ground field
  $\underline{R}$
  (which is clearly 1-connected and injective, by
  Example \ref{GroundFieldIsInjectiveAsEquivariantDualVectorSpace})
  is a relative minimal dgc-algebra \eqref{RelativeMinimalEquivariantdgcAlgebra}
  \begin{equation}
    \label{GroundFieldInclusionIntoEquivariantMinimaldgcAlgebra}
    \xymatrix{
      \underline{\mathbb{R}}
      \;
      \ar@{^{(}->}[r]^-{ \mathrm{min} }
      &
      A^\bullet \in \EquivariantdgcAlgebras \;,
    }
  \end{equation}

\noindent
  is called a \emph{minimal equivariant dgc-algebra}.
\end{defn}

\begin{defn}[Minimal equivariant $L_\infty$-algebra]
  \label{MinimalEquivariantLInfinityAlgebra}
  Any minimal equivariant dgc-algebra $A$ (Def. \ref{MinimalEquivariantdgcAlgebras})
  is the equivariant Chevalley-Eilenberg algebra \eqref{DefiningEmbeddingOfEquivariantLInfinityAlgebrasInEquivariantdgcAlgebras}
  \vspace{-2mm}
  $$
    A
      \;\simeq\;
    \mathrm{CE}
    \big(\,
      \underline{\mathfrak{g}}^A
    \big)
  $$

  \vspace{0mm}
  \noindent
  of an equivariant $L_\infty$-algebra
  $\underline{\mathfrak{g}}^A \;\in\; \EquivariantLInfinityAlgebras$
  (Def. \ref{EquivariantLInfinityAlgebras}), defined uniquely
  up to isomorphism.
  We say that the
  \emph{underlying} graded equivariant vector space
  (Def. \ref{EquivariantGradedVectorSpaces})  
  $$
    \underline{\mathfrak{g}}^A_{\, \bullet}
    \;\;
    \in
    \;\;
    \EquivariantGradedVectorSpaces
  $$

\noindent
  of this equivariant $L_\infty$-algebra is the
  linear dual of the spaces of generators
  $V^A_n \;\in\; \EquivariantDualVectorSpaces $ \eqref{AttachingMap}
  of the elementary extensions \eqref{ElementaryExtensionPushout}
  that exhibit the minimality of $A$:
  $$
    \underline{\mathfrak{g}}^A_{\, n}
    \;\;
    :=
    \big(
      V^A_n
    \big)^\vee
    \;\;\;\;
    \in
    \EquivariantVectorSpaces
    \,.
  $$
\end{defn}

\begin{example}[A minimal $\ZTwo$-equivariant dgc-algebra]
  \label{CheckingTwistorSpaceModZ2MinimalModel}
  We spell out the construction of an equivariant minimal dgc-algebra
  (Def. \ref{MinimalEquivariantdgcAlgebras}),
  for $G = \ZTwo$  (Example \ref{OrbitCategoryOfZ2}),
  which involves three basic cases of the
  elementary extensions \eqref{ElementaryExtensionPushout}:

  \noindent
  {\bf (i)} In the first stage, begin with the
  equivariant base algebra $\underline{\mathbb{R}}$
  (Example \ref{EquivariantdgcBaseAlgebra})
  and consider the attaching map \eqref{AttachingMap}
  in degree 2 given by
  \begin{equation}
    \label{AttachingMapInDegree2ForZ2EquivariantTwistorSpace}
    \phi_2
    \;\;
    :
    \;\;
    \raisebox{20pt}{
    \xymatrix@R=1.5em@C=3em{
      \ZTwo/1
      \ar@(ul,ur)|-{\; \ZTwo}
      \ar[d]
      &\longmapsto&
      \mathbb{R}
      \ar[d]^-{
        \mathrm{id}
      }
      \ar@{<-}[rr]^-{
        0
        \;
        \mapsfrom
        \;
        c_3
      }
      &&
      \mathbb{R}
      \langle
        c_{3}
      \rangle
      \ar[d]^-{ \mathrm{id} }
      \\
      \ZTwo/\ZTwo
      &\longmapsto&
      \mathbb{R}
      \ar@{<-}[rr]^-{
        0
        \;
        \mapsfrom
        \;
        c_3
      }
      &&
      \mathbb{R}
      \langle
        c_{3}
      \rangle
    }
    }
  \end{equation}

\noindent  By Example \ref{InjectiveZ2EquivariantDualVectorSpaces},
  the equivariant dual vector space
  on the right
  is already injective \eqref{TheInjectiveZ2EquivariantDualVectorSpaceConcentratedOnFixedLocus},
  so that
  we may extend this attaching map immediately to an
  equivariant cochain map \eqref{ExtensionToInjectiveResolutionOfAttachingMap}
  $$
    \phi_2^\bullet
    \;\;
    :
    \;\;
    \raisebox{24pt}{
    \xymatrix@R=1.5em@C=3em{
      \ZTwo/1
      \ar@(ul,ur)|-{\; \ZTwo}
      \ar[d]
      &\longmapsto&
      \mathbb{R}
      \ar[d]^-{
        \mathrm{id}
        }
      \ar@{<-}[rr]^-{
        0
        \;
        \mapsfrom
        \;
        c_3
      }
      &&
      \mathbb{R}
      \langle
        c_{3}
      \rangle
      \big/
      (
        d\, c_3 \,=\, 0
      )
      \ar[d]^-{ \mathrm{id} }
      \\
      \ZTwo/\ZTwo
      &\longmapsto&
      \mathbb{R}
      \ar@{<-}[rr]^-{
        0
        \;
        \mapsfrom
        \;
        c_3
      }
      &&
      \mathbb{R}
      \langle
        c_{3}
      \rangle
      \big/
      (
        d\, c_3 \,=\, 0
      )
      \mathrlap{\,,}
    }
    }
  $$

\noindent  where on the right we are using the generators-and-relations
  Notation \ref{GeneratorsAndRelationsPresentationOfCochainComplexes}.
  By Example
  \ref{FreeZ2EquivariantdgcAlgebra},
  its adjunct morphism
  of equivariant dgc-algebras is
  $$
    \widetilde \phi_2^\bullet
    \;\;
    :
    \;\;
    \raisebox{24pt}{
    \xymatrix@R=1.5em@C=3em{
      \ZTwo/1
      \ar@(ul,ur)|-{\; \ZTwo}
      \ar[d]
      &\longmapsto&
      \mathbb{R}
      \ar[d]^-{
        \mathrm{id}
      }
      \ar@{<-}[rr]^-{
        0
        \;
        \mapsfrom
        \;
        c_3
      }
      &&
      \mathbb{R}
      [
        c_{3}
      ]
      \big/
      (
        d\, c_3 \,=\, 0
      )
      \ar[d]^-{ \mathrm{id} }
      \\
      \ZTwo/\ZTwo
      &\longmapsto&
      \mathbb{R}
      \ar@{<-}[rr]^-{
        0
        \;
        \mapsfrom
        \;
        c_3
      }
      &&
      \mathbb{R}
      [
        c_{3}
      ]
      \big/
      (
        d\, c_3 \,=\, 0
      )
      \mathrlap{\,.}
    }
    }
  $$

 \noindent Since all these diagrams so far are constant on the
  orbit category, the resulting
  pushout \eqref{ElementaryExtensionPushout}
  is computed over both objects
  $\ZTwo/H \in \ZTwo \mathrm{Orb}$ as in
  non-equivariant dgc-theory, and thus yields
  this minimal equivariant dgc-algebra:
  \begin{equation}
    \label{FirstStageInMinimnalModelForCP3ModZ2}
    \raisebox{24pt}{
    \xymatrix@R=1.5em@C=3em{
      \ZTwo/1
      \ar@(ul,ur)|-{\; \ZTwo}
      \ar[d]
      &\longmapsto&
      \mathbb{R}
      [
        f_2
      ]
      \big/\!
      (
        d\, f_2 \,=\, 0
      )
      \ar@<-22pt>[d]_-{
        \mathrm{id}
      }
      \\
      \ZTwo/\ZTwo
      &\longmapsto&
      \mathbb{R}
      [
        f_2
      ]
      \!\big/\!
      (
        d\, f_2 \,=\, 0
      ) \,.
    }
    }
  \end{equation}

  \noindent
  {\bf (ii)}
  Consider next the following attaching map
  \eqref{AttachingMap}
  in degree 3
  to the equivariant dgc-algebra \eqref{FirstStageInMinimnalModelForCP3ModZ2}:
  \begin{equation}
    \label{AttachingMapInDegree3ForZ2EquivariantTwistorSpace}
    \hspace{-4mm} 
    \phi_3
    \;\;
    :
    \;\;
    \raisebox{24pt}{
    \xymatrix@R=1.5em@C=2.3em{
      \ZTwo/1
      \ar@(ul,ur)|-{\; \ZTwo}
      \ar[d]
      &\longmapsto&
      \mathbb{R}
      [
        f_2
      ]
      \big/\!
      (
        d\, f_2 \,=\, 0
      )
      \ar@<-22pt>[d]_-{
          \mathrm{id}
      }
      \ar@{<-}[rr]
      &&
      0
      \ar[d]
      \\
      \ZTwo/\ZTwo
      &\longmapsto&
      \mathbb{R}
      [
        f_2
      ]
      \!\big/\!
      (
        d\, f_2 \,=\, 0
      )
      \ar@{<-}[rr]^-{
        f_2 \wedge f_2
        \;
        \mapsfrom
        \;
        c_4
      }
      &&
      \mathbb{R}
      \langle
        c_{4}
      \rangle \,.
    }
    }
  \end{equation}

\noindent
  Here the equivariant dual vector space on the right
  is \emph{not} injective: Its injective envelope is
  given in Example \ref{InjectiveResolutionOfZ2equiavriantVectorSpaces},
  and the free dgc-algebra on this is given in
  Example \ref{FreeZ2EquivariantdgcAlgebra}, which says that
  the required extension
  \eqref{ExtensionToInjectiveResolutionOfAttachingMap}
  of the attaching map $\phi$ is hence of this form:
  $$
    \widetilde \phi_3^\bullet
    \;\;
    :
    \;\;
    \raisebox{20pt}{
    \xymatrix@R=1.6em@C=17pt{
      \ZTwo/1
      \ar@(ul,ur)|-{\; \ZTwo}
      \ar[d]
      &\longmapsto&
      \mathbb{R}
      [
        f_2
      ]
      \big/\!
      (
        d\, f_2 \,=\, 0
      )
      \ar@<-22pt>[d]_-{
        \mathrm{id}
      }
      \ar@{<-}[rr]^-{
        \scalebox{.7}{$
          \arraycolsep=1.4pt
          \def\arraystretch{1}
          {\begin{array}{ccc}
            0 &\mapsfrom& c_5
            \\
            f_2 \wedge f_2 &\mapsfrom& c_4
          \end{array}}
        $}
      }
      &&
      \mathbb{R}
      \!
      \left[
        \!\!\!
        \def\arraystretch{1}
        {\begin{array}{c}
          c_{5}
          \\
          c_{4}
        \end{array}}
        \!\!\!
      \right]
      \!\big/\!
      \left(
        {\begin{aligned}
          d\, c_5 & = 0
          \\[-5pt]
          d\, c_4 & = c_5
        \end{aligned}}
      \right)
      \ar@<-26pt>@{->>}[d]
      \\
      \ZTwo/\ZTwo
      &\longmapsto&
      \mathbb{R}
      [
        f_2
      ]
      \!\big/\!
      (
        d\, f_2 \,=\, 0
      )
      \ar@{<-}[rr]^-{
        f_2 \wedge f_2
        \;
        \mapsfrom
        \;
        c_4
      }
      &&
      \mathbb{R}
      [
        c_{4}
      ]
      \big/
      (
        d\, c_4 \,=\, 0
      )
      \,.
    }
    }
  $$

  \newpage
  
  \noindent   
  The pushout \eqref{ElementaryExtensionPushout}
  along this map is the following, yielding the next stage
  of the minimal equivariant dgc-algebra on the rear left:
  $$
    \hspace{1mm}
    \adjustbox{
      scale=.72
    }{$
    \xymatrix@C=-10pt@R=9pt{
      \mathbb{R}
      \left[
        \!\!\!
        \def\arraystretch{1}
        {\begin{array}{l}
          h_3,
          \\
          \omega_4,
          \\
          f_2
        \end{array}}
        \!\!\!
      \right]
      \!\big/\!
      \left(
        {\begin{aligned}
          d\, h_3 & = \omega_4 - f_2 \wedge f_2
          \\[-4pt]
          d\, \omega_4 & = 0
          \\[-4pt]
          d\, f_2 & = 0
        \end{aligned}}
      \right)
      \ar@{<-}[rrrr]^-{\tiny
          \arraycolsep=1.4pt
          \def\arraystretch{1}
          \begin{array}{ccc}
            0\;\; & \mapsfrom & c_5
            \\
            f_2 \wedge f_2 & \mapsfrom & c_4
          \end{array}}_-{
            \tiny
            \arraycolsep=1.4pt
            \;\;\;
            \def\arraystretch{1}
            \begin{array}{ccc}
            \omega_4\;\;\;\;\; &\mapsfrom& b_4
            \\
             h_3\;\;\;\;\; &\mapsfrom& b_3
          \end{array}
        }
      \ar@<-40pt>@{->>}[dddd]
      \ar@{<-^{)}}[dddr]
      &&&&
      \mathbb{R}
      \left[
        \!\!\!
        \def\arraystretch{1}
        {\begin{array}{ll}
          c_5, & b_4,
          \\
          c_4, & b_3
        \end{array}}
        \!\!\!
      \right]
      \!\big/\!
      \left(
        \arraycolsep=1.4pt
        {\begin{aligned}
          d\, c_5 & = 0\;\,,\; d\, b_4 \; = c_5
          \\[-4pt]
          d\, c_4 & = c_5\,, \; d\, b_3 \; = b_4 - c_4
        \end{aligned}}
      \right)
      \ar@{<-^{)}}[dddr]
      \ar@<-56pt>@{->>}[dddd]|>>>>>>>>>>>{ \phantom{AA} \atop \phantom{AA} }
      \\
      \\
      \\
      &
      \mathbb{R}
      [
        f_2
      ]
      \big/\!
      (
        d\, f_2 \,=\, 0
      )
      \ar@<-22pt>[dddd]^>>>>>>>>>>>{
            \mathrm{id}
      }
      \ar@{<-}[rrrr]^-{
        \scalebox{.7}{$
          \arraycolsep=1.4pt
          \def\arraystretch{1}
          {\begin{array}{ccc}
            0 &\mapsfrom& c_5
            \\
            f_2 \wedge f_2 &\mapsfrom& c_4
          \end{array}}
        $}
      }
      &&&&
      \mathbb{R}
      \!
      \left[
        \!\!\!
        \def\arraystretch{1}
        {\begin{array}{c}
          c_{5}
          \\
          c_{4}
        \end{array}}
        \!\!\!
      \right]
      \!\big/\!
      \left(
        {\begin{aligned}
          d\, c_5 & = 0
          \\[-4pt]
          d\, c_4 & = c_5
        \end{aligned}}
      \right)
      \ar@<-26pt>@{->>}[dddd]
      \\
      \mathbb{R}
      \left[
        \!\!\!
        \def\arraystretch{1}
        {\begin{array}{l}
          h_3,
          \\
          f_2
        \end{array}}
        \!\!\!
      \right]
      \!\big/\!
      \left(
        {\begin{aligned}
          d\, h_3 & = - f_2 \wedge f_2
          \\[-4pt]
          d\, f_2 & = 0
        \end{aligned}}
      \right)
      \ar@{<-}[rrrr]
        |<<<<<<<<{ \phantom{AAA} }
        ^>>>>>>>>>>>{
          \tiny
          \arraycolsep=1.4pt
          \def\arraystretch{1}
          \begin{array}{ccc}
            f_2 \wedge f_2 & \mapsfrom&   c_4
          \end{array}
      }
      _>>>>>>>>>>>{
           \tiny
           \;\;\;\;\;
          \arraycolsep=1.4pt
          \def\arraystretch{1}
          \begin{array}{ccc}
             h_3 \;\;\; & \mapsfrom & b_3
          \end{array}
      }
      \ar@{<-^{)}}[dddr]|->>>{\;\; \phantom{AA}}
      &&&&
      \mathbb{R}[ c_4, b_3]
         \!\big/\!
      (
                d\, c_4  = c_5\,, \, d\, b_3 = - c_4
)
      \ar@{<-^{)}}[dddr]|->>>{\;\;\;\;\;\phantom{A}}
      \\
      \\
      \\
      &
      \mathbb{R}
      [
        f_2
      ]
      \!\big/\!
      (
        d\, f_2 \,=\, 0
      )
      \ar@{<-}[rrrr]^-{
        \tiny
        \arraycolsep=1.4pt
        \def\arraystretch{1}
        \begin{array}{ccc}
          f_2 \wedge f_2
          &
          \mapsfrom
          &
          c_4
        \end{array}
      }
      &&&&
      \mathbb{R}
      [
        c_{4}
      ]
      \big/
      (
        d\, c_4 \,=\, 0
      )\;.
    }
    $}
  $$

   \noindent {\bf (iii)}
  Finally, consider the following further attaching map \eqref{AttachingMap}
  to the previous stage, in degree 7:
  \begin{equation}
    \label{AttachingMapInDegree7ForZ2EquivariantTwistorSpace}
    \hspace{-3mm} 
    \phi_7
    \;
    :
    \!\!\!\!\!
    \raisebox{30pt}{
    \xymatrix@R=1.5em@C=15pt{
      \ZTwo/1
      \ar@(ul,ur)|-{\; \ZTwo}
      \ar[d]
      &\longmapsto \;\;
      \mathbb{R}
      \left[
        \!\!
        \def\arraystretch{1}
        {\begin{array}{l}
          h_3,
          \\
          \omega_4,
          \\
          f_2
        \end{array}}
        \!\!
      \right]\!\!
      \big/\!\!
      \left(
        {\begin{aligned}
          d\, h_3 & = \omega_4 - f_2 \wedge f_2
          \\[-4pt]
          d\, \omega_4 & = 0
          \\[-4pt]
          d\, f_2 & = 0
        \end{aligned}}
      \right)
      \ar@{->>}@<-52pt>[d]
      \ar@{<-}[rr]^-{
        \hspace{1pt}
        \adjustbox{
          scale=.5,
          raise=2pt
        }{$
          - \omega_4 \wedge \omega_4
          \;\mapsfrom\;
          c_8
        $}
      }
      &&
      \mathbb{R}
      \langle
        c_{8}
      \rangle
      \ar[d]
      \\
      \ZTwo/\ZTwo
      &\longmapsto \;\;
      \mathbb{R}
      \left[
        \!\!
        \def\arraystretch{1}
        {\begin{array}{l}
          h_3,
          \\
          f_2
        \end{array}}
        \!\!
      \right]\!\!
      \big/\!\!
      \left(
        {\begin{aligned}
          d\, h_3 & = \phantom{\omega_4} - f_2 \wedge f_2
          \\[-4pt]
          d\, f_2 & = 0
        \end{aligned}}
      \right)
      \ar@{<-}[rr]
      &&
      0
      \mathrlap{\,.}
    }
    }
  \end{equation}

\noindent  Here the equivariant dual vector space on the right is
  again injective, by \eqref{TheInjectiveZ2EquivariantDualVectorSpaceInBulk}
  in Example \ref{InjectiveZ2EquivariantDualVectorSpaces}.
  Therefore, the corresponding elementary extension
  \eqref{ElementaryExtensionPushout}
  is by pushout along the following morphism of
  dgc-algebras
  $$
  \hspace{1mm} 
    \widetilde \phi^\bullet_7
    \;
    :
    \!
    \raisebox{30pt}{
    \xymatrix@R=1.5em@C=5pt{
      \ZTwo/1
      \ar@(ul,ur)|-{\; \ZTwo}
      \ar[d]
      &\longmapsto&
      \mathbb{R}
      \left[
        \!\!
        \def\arraystretch{1}
        {\begin{array}{l}
          h_3,
          \\
          \omega_4,
          \\
          f_2
        \end{array}}
        \!\!
      \right]
      \big/\!
      \left(
        {\begin{aligned}
          d\, h_3 & = \omega_4 - f_2 \wedge f_2
          \\[-4pt]
          d\, \omega_4 & = 0
          \\[-4pt]
          d\, f_2 & = 0
        \end{aligned}}
      \right)
      \ar@{->>}@<-52pt>[d]
      \ar@{<-}[rr]^-{
        \hspace{30pt}
        \adjustbox{
          scale=.7,
          raise=7pt
         }{$
          - \omega_4 \wedge \omega_4
          \;\mapsfrom\;
          c_8
        $}
      }
      &&
      \mathbb{R}
      [
        c_{8}
      ]
      \!\big/\!
      (
        d\, c_8 \,=\, 0
      )
      \ar[d]
      \\
      \ZTwo/\ZTwo
      &\longmapsto&
      \mathbb{R}
      \left[
        \!\!
        \def\arraystretch{1}
        {\begin{array}{l}
          h_3,
          \\
          f_2
        \end{array}}
        \!\!
      \right]
      \big/\!
      \left(
        {\begin{aligned}
          d\, h_3 & = \phantom{\omega_4} - f_2 \wedge f_2
          \\[-4pt]
          d\, f_2 & = 0
        \end{aligned}}
      \right)
      \ar@{<-}[rr]
      &&
      0
      \,.
    }
    }
  $$

\noindent  This pushout is the identity on $\ZTwo/\ZTwo$,
  and is an ordinary cell attachment of plain dgc-algebras
  on $\ZTwo/1$, hence yields the following equivariant
  dgc-algebra, which is thereby seen to be minimal
  (Def. \ref{MinimalEquivariantdgcAlgebras}):
  \vspace{-1mm}
  \begin{equation}
    \label{CP3ModZ2MinimalEquvariantdgcAlgebra}
    A
    \;\;\;
     :=
    \;\;\;
    \raisebox{30pt}{
    \xymatrix@R=1.5em@C=15pt{
      \ZTwo/1
      \ar@(ul,ur)|-{\; \ZTwo}
      \ar[d]
      &\longmapsto&
      \mathbb{R}
      \left[
        \!\!
        \def\arraystretch{1}
        {\begin{array}{l}
          \omega_7,
          \\
          h_3,
          \\
          \omega_4,
          \\
          f_2
        \end{array}}
        \!\!
      \right]
      \big/\!
      \left(
        {\begin{aligned}
          d\, \omega_7 & = - \omega_4 \wedge \omega_4
          \\[-4pt]
          d\, h_3 & = \omega_4 - f_2 \wedge f_2
          \\[-4pt]
          d\, \omega_4 & = 0
          \\[-4pt]
          d\, f_2 & = 0
        \end{aligned}}
      \right)
      \ar@{->>}@<-52pt>[d]
      \\
      \ZTwo/\ZTwo
      &\longmapsto&
      \mathbb{R}
      \left[
        \!\!
        \def\arraystretch{1}
        {\begin{array}{l}
          h_3,
          \\
          f_2
        \end{array}}
        \!\!
      \right]
      \big/\!
      \left(
        {\begin{aligned}
          d\, h_3 & = \phantom{\omega_4} - f_2 \wedge f_2
          \\[-4pt]
          d\, f_2 & = 0
        \end{aligned}}
      \right).
    }
    }
  \end{equation}

  \noindent  
  In summary, the graded equivariant dual vector space of generators
  (Def. \ref{MinimalEquivariantLInfinityAlgebra})
  of this minimal
  equivariant dgc-algebra is the following:
  \begin{align}
    \label{GeneratorsForMinimalEquivariantModelOFTwistorSpaceModZ2}
    \underline{\mathfrak{g}}^A_\bullet
     =
  &
  \mbox{
  \def\tabcolsep{3pt}
  \begin{tabular}{|c||c|c|c|c|c|c|c|c|c|}
    \hline
    $\ZTwo/H$
    &
    $\underline{\mathfrak{g}}^A_2$
    &
    $\underline{\mathfrak{g}}^A_3$
    &
    $\underline{\mathfrak{g}}^A_4$
    &
    $\underline{\mathfrak{g}}^A_5$
    &
    $\underline{\mathfrak{g}}^A_6$
    &
    $\underline{\mathfrak{g}}^A_7$
    &
    $\underline{\mathfrak{g}}^A_8$
    &
    $\underline{\mathfrak{g}}^A_9$
    &
    $\cdots$
    \\[4pt]
    \hline
    \hline
    $\ZTwo/1$
    &
    $\mathbf{1}$
    &
    $0$
    &
    $0$
    &
    $0$
    &
    $0$
    &
    $\mathbf{1}$
    &
    $0$
    &
    $0$
    &
    $\cdots$
    \\
    \hline
    $\ZTwo/\ZTwo$
    &
    $\mathbf{1}$
    &
    $\mathbf{1}$
    &
    $0$
    &
    $0$
    &
    $0$
    &
    $0$
    &
    $0$
    &
    $0$
    &
    $\cdots$
    \\
    \hline
    \multicolumn{1}{c}{}
    &
    \multicolumn{1}{c}{
      \tiny
      \eqref{AttachingMapInDegree2ForZ2EquivariantTwistorSpace}
    }
    &
    \multicolumn{1}{c}{
      \tiny
      \eqref{AttachingMapInDegree3ForZ2EquivariantTwistorSpace}
    }
    &
    \multicolumn{1}{c}{}
    &
    \multicolumn{1}{c}{}
    &
    \multicolumn{1}{c}{}
    &
    \multicolumn{1}{c}{
      \tiny
      \eqref{AttachingMapInDegree7ForZ2EquivariantTwistorSpace}
    }
    &
    \multicolumn{1}{c}{}
    &
    \multicolumn{1}{c}{}
    &
    \multicolumn{1}{c}{}
  \end{tabular}
  }
\\
 & \in
  \;
  \ZTwoEquivariantGradedVectorSpaces\;.
  \nonumber 
  \end{align}

\end{example}

\begin{lemma}[Minimal equivariant dgc-algebras are projectively cofibrant {\cite[Thm. 4.2]{Scull08}}]
  \label{MinimalEquivariantdgcAlgebrasAreCofibrant}
  All elementary extensions \eqref{ElementaryExtensionPushout}
  are cofibrations
  $$
    \xymatrix{
      A
      \ar[rr]^-{ \in \; \mathrm{Cof} }
      &&
      A\big[\mathfrak{b}^n V_n\big]_{\phi_n}
    }
    \;\;\;\;
    \in
    \EquivariantdgcAlgebrasProj
    \,.
  $$

  \noindent
Hence all
  relative minimal equivariant dgc-algebra
  inclusions \eqref{RelativeMinimalEquivariantdgcAlgebra}
  are cofibrations
  and, in particular, all
  minimal equivariant dgc-algebras \eqref{GroundFieldInclusionIntoEquivariantMinimaldgcAlgebra}
  are cofibrant objects in
  the model category
  $\EquivariantdgcAlgebrasProj$
  (Prop. \ref{ProjectiveModelStructureOnConnectiveEquivariantdgcAlgebras}).
\end{lemma}

\begin{prop}[Existence of equivariant minimal models {\cite[Thm. 3.11, Cor. 3.9]{Scull02}}]
  \label{ExistenceOfEquivariantMinimalModels}
  $\,$

  \noindent
  Let $A \in \EquivariantdgcAlgebras$ (Def. \ref{EquivariantdgcAlgebras})
  be cohomologically 1-connected, in that the
  equivariant cochain cohomology groups (Def. \ref{EquivariantCochainCohomologyGroups})
  are trivial in degrees $\leq 1$:
  \begin{equation}
    \label{Cohomologically1Connected}
    \underline{H}^0(A)
    \;\simeq\;
    \underline{\mathbb{R}}
    \phantom{AAAA}
    \mbox{and}
    \phantom{AAAA}
    \underline{H}^1(A)
    \;\simeq\;
    0
    \,.
  \end{equation}
  \vspace{-.6cm}

  \noindent
  {\bf (i)} There exists a minimal equivariant dgc-algebra
  (Def. \ref{MinimalEquivariantdgcAlgebras})
  equipped with a quasi-isomorphism
  \begin{equation}
    \label{EquivariantMinimalModelQuasiIsomorphism}
    \xymatrix{
      A_{\mathrm{min}}
      \ar[rr]^-{ p^{\mathrm{min}}_A }_-{ \in \; \mathrm{W} }
      &&
      A
    }.
  \end{equation}

  \noindent
  {\bf (ii)} This is unique up to isomorphism, in that
  for $A'_{\mathrm{min}} \overset{\in \, \mathrm{W}}{\longrightarrow} A$
  any other such, there is a commuting diagram of the form
  \vspace{-2mm}
  $$
    \xymatrix@R=-1pt@C=4em{
      A_{\mathrm{min}}
      \ar[dr]_{\in \, \mathrm{W}}
      \ar[rr]^-{ \simeq }
      &&
      A'_{\mathrm{min}}
      \ar[dl]^-{ \in \, \mathrm{W} }
      \\
      & A
    }
  $$

  \noindent
  with the top morphism an isomorphism of equivariant dgc-algebras.
\end{prop}

\begin{remark}[Existence of equivariant relative minimal models]
  \label{ExistenceOfEquivariantRelativeMinimalModels}
  By analogy with the theory of
  (relative) minimal models in
  non-equivariant dgc-algebraic rational homotopy theory
  (e.g., \cite[\S 7]{BousfieldGugenheim76}\cite{Halperin83}\cite[Thm. 14.12]{FHT00}\cite[Prop. 3.50]{FSS23-Char}),
  it is
  to be expected that Prop. \ref{ExistenceOfEquivariantMinimalModels}
  holds in greater generality:
  \begin{itemize}[
  leftmargin=8mm,
  topsep=1pt,
  itemsep=-1pt
]
  
  \item[{\bf (a)}] The existence of equivariant minimal models
  should hold more generally for fixed locus-wise
  nilpotent $G$-spaces (not necessarily fixed-locus wise
  simply-connected).
  
  \item[{\bf (b)}] There should exist also
  equivariant {\it relative} minimal models,
  unique up to relative isomorphism,
  of any morphism between fixed locus-wise nilpotent
  spaces of $\mathbb{R}$-finite homotopy type.
\end{itemize}

 \noindent  While a proof of these more general statements should be
  a fairly straightforward generalization of the
  proofs of the existing results, it does not seem to
  be available in the literature.
  Nonetheless,
  for our main example of interest (Example \ref{GhetEquivariantParametrizedTwistorSpace})
  we explicitly find the equivariant relative minimal model
  (in Prop. \ref{Z2EquivariantRelativeMinimalModelOfSpin3ParametrizedTwistorSpace}
  below).

\end{remark}

\subsection{Equivariant rational homotopy theory}
\label{EquivariantRationalHomotopyTheory}

We review the fundamentals of equivariant rational homotopy theory
\cite{Tri82}\cite{Tri96} \cite{Golasinski97b}\cite{Scull02}\cite{Scull08}
and prove our main technical result
(Prop. \ref{Z2EquivariantRelativeMinimalModelOfSpin3ParametrizedTwistorSpace} below).
Throughout we make free use of plain (non-equivariant)
dgc-algebraic rational homotopy theory
\cite{BousfieldGugenheim76} (review in
\cite{FHT00}\cite{Hess07}\cite{GriffithMorgan13}\cite[\S 3.2]{FSS23-Char}).

\medskip

\noindent {\bf Equivariant rationalization.} Equivariant rational homotopy theory
is concerned with the following concept:

\begin{defn}[Equivariant rationalization {\cite[\S II.3]{May96}\cite[\S 2.6]{Tri82}}]
  \label{EquivariantRationalization}
  $\,$

  \noindent
  Let $\mathscr{X} \,\in\, \EquivariantHomotopyTypesSimplyConnected$
  (Def. \ref{SubcategoryOfEquivariantSimplyConnectedRFiniteHomotopyTypes}).

  \noindent
  {\bf (i)} $\mathscr{X}$  is called
  \emph{rational} (here: over the real numbers, see \cite[Rem. 3.51]{FSS23-Char})
  if all its
  equivariant homotopy groups (Def. \ref{EquivariantHomotopyGroups})
  carry the structure of equivariant vector spaces
  (here: over the real numbers, Def. \ref{EquivariantVectorSpaces}):
  \begin{equation}
    \label{RationalEquivariantHomotopyType}
    \hspace{-4mm} 
    \mbox{
      $\mathscr{X}$
      is rational over the reals
    }
   \; \Leftrightarrow \;
    \underline{\pi}_{\; \bullet + 1}(X)
      \in
    \xymatrix{
      \EquivariantVectorSpaces
      \ar[r]
      &
      \EquivariantGroups
      \,.
    }
  \end{equation}

  \noindent
  {\bf (ii)} A \emph{rationalization} of $\mathscr{X}$
  (here: over the real numbers)
  is a morphism
  \begin{equation}
    \label{RationalizationUnitOnEquivariantHomotopyTypes}
    \xymatrix{
      \mathscr{X}
      \ar[rr]^-{ \eta^{\mathbb{R}}_{\scalebox{.7}{$\mathscr{X}$}} }
      &&
      L_{\mathbb{R}}\mathscr{X}
    }
    \;\;\;\;\;\;\;\;\;
    \in
    \;
    \EquivariantHomotopyTypes
  \end{equation}

  \noindent
  to a rational equivariant homotopy type
  \eqref{RationalEquivariantHomotopyType} which induces
  isomorphisms on all equivariant rational cohomology groups
  (Example \ref{EquivariantDualVectorSpacesOfRealCohomologyGroups}):
  \vspace{-2mm}
  $$
    \xymatrix{
      \underline{H}^\bullet
      \big(
        L_{\mathbb{R}}
        \mathscr{X}
        ;
        \,
        \mathbb{R}
      \big)
      \ar[rr]^-{
        \big(
          \eta^{\mathbb{R}}_{\scalebox{.7}{$\mathscr{X}$}}
        \big)^{\ast}
      }_-{ \simeq }
      &&
      \underline{H}^\bullet
      \big(
        \mathscr{X}
        ;
        \,
        \mathbb{R}
      \big).
    }
  $$
  
\end{defn}

\vspace{0mm}
\noindent
In other words: equivariant rationalization is plain
rationalization
(e.g. \cite[Def. 3.55]{FSS23-Char})
at each stage $G/H \,\in\, G\mathrm{Orb}$.
\begin{prop}[Uniqueness of equivariant rationalization {\cite[\S II, Thm. 3.2]{May96}}]
  Equivariant rationalization
  (Def. \ref{EquivariantRationalization})
  of equivariantly simply-connected
  equivariant homotopy types
  exists essentially uniquely.
\end{prop}

\medskip

\noindent {\bf Equivariant PL de Rham theory.}

\begin{defn}[Equivariant PL de Rham complex]
  \label{EquivariantPLdeRhamComplexFunctor}
  Write
  $$
    \xymatrix@R=-5pt{
      \EquivariantSSet
      \ar[rr]^-{
        \Omega^\bullet_{\mathrm{PLdR}}
      }
      &&
      \EquivariantdgcAlgebrasOp
      \\
      \mathscr{X}
      \ar@{}[rr]|-{\longmapsto }
      &&
      \Big(
        G/H
        \;\mapsto\;
        \Omega^\bullet_{\mathrm{PLdR}}
        \big(
          \mathscr{X}(G/H)
        \big)
      \! \Big)
    }
  $$

  \noindent
  for the functor
  from equivariant simplicial sets
  (Def. \ref{EquivariantSSetCategory})
  to the opposite of equivariant dgc-algebras
  (Def. \ref{EquivariantdgcAlgebras}).
  This applies the plain PL de Rham functor
  \cite{Sullivan77}\cite[p. 1.-7]{BousfieldGugenheim76}\cite[Def. 3.56]{FSS23-Char}
  (assigning dgc-algebras of piecewise polynomial differential forms)
  to diagrams of simplicial sets parametrized over the
  orbit category.
\end{defn}

\begin{prop}[Equivariant PL de Rham theorem {\cite[Thm. 4.9]{Tri82}}]
  For any $\mathscr{X} \in \EquivariantSSet$
  (Def. \ref{EquivariantSSetCategory})
  and
  $\mathscr{A}_{\mathbb{R}} \in \EquivariantVectorSpaces$
  (Def. \ref{EquivariantVectorSpaces}),
  we have a
  natural isomorphism
  \vspace{-1mm}
  $$
    H^\bullet
    \big(
      \mathscr{X};
      \,
      \mathscr{A}_{\mathbb{R}}
    \big)
    \;\;
    \simeq
    \;\;
    H^\bullet
    \big(
      \Omega^\bullet_{\mathrm{PLdR}}
      (
        \mathscr{X};
        \,
        \mathscr{A}_{\mathbb{R}}
      )
    \big)
  $$

  \noindent
  between the Bredon cohomology of $\mathscr{X}$
  (Example \ref{BredonCohomology})
  with coefficients
  in $\mathscr{A}_{\mathbb{R}}$,
  and the cochain cohomology of the
  equivariant PL de Rham complex
  of $X$ (Def. \ref{EquivariantPLdeRhamComplexFunctor})
  with coefficients in $\mathscr{A}_{\mathbb{R}}$.
\end{prop}

\begin{prop}[Quillen adjunction between equivariant simplicial sets
and equivariant dgc-algebras {\cite[Prop. 5.1]{Scull08}}]
  \label{QuillenAdjunctionBetweenEquivariantSSetAndEquivariantdgcAlgebras}
  The equivariant PL de Rham complex construction
  (Def. \ref{EquivariantPLdeRhamComplexFunctor})
  is the left adjoint in a Quillen adjunction
  \vspace{0mm}
  $$
    \xymatrix{
      \EquivariantdgcAlgebrasProjOp \;\;
      \ar@{<-}@<+7pt>[rr]^-{ \Omega^\bullet_{\mathrm{PLdR}} }
      \ar@<-7pt>[rr]_-{ \exp }^-{ \bot_{ \mathrlap{\mathrm{Qu}} } }
      &&
   \;\;   G \mathrm{SSet}_{\mathrm{proj}}
    }
  $$

\noindent
  between the projective model structure on
  equivariant simplicial sets
  (Prop. \ref{ModelCategoryOnEquivariantSSet})
  and the opposite of the projective model structure on
  connective equivariant dgc-algebras (Prop. \ref{ProjectiveModelStructureOnConnectiveEquivariantdgcAlgebras}).
\end{prop}

\newpage

\noindent {\bf The fundamental theorem of dgc-algebraic equivariant rational homotopy theory.}

\begin{prop}[Fundamental theorem of dgc-algebraic equivariant rational homotopy theory {\cite[Thm. 5.6]{Scull08}}]
  \label{FundamentalTheoremOfdgcAlgebraicEquivariantRationalHomotopyTheory}
  On equivariant 1-connected $\mathbb{R}$-finite homotopy
  types
  (Def. \ref{SubcategoryOfEquivariantSimplyConnectedRFiniteHomotopyTypes}):

  \noindent
  {\bf (i)}
  The derived PL de Rham adjunction
  (Prop. \ref{QuillenAdjunctionBetweenEquivariantSSetAndEquivariantdgcAlgebras})
  restricts to an equivalence of homotopy categories
  $$
    \xymatrix{
      \big(
      \EquivariantSimplyConnectedRFiniteHomotopyTypes
      \big)^{\mathbb{R}}
      \;
      \ar@{<-}@<+6pt>[rrr]^-{
        \mathbb{L} \Omega^\bullet_{\mathrm{PLdR}}
      }
      \ar@<-6pt>[rrr]_-{ \mathbb{R} \exp }^-{ \simeq }
      &&&
   \;   \mathrm{Ho}
      \Big(
        \EquivariantdgcAlgebrasProjOp
      \Big)^{\geq 2}_{\mathrm{fin}}
    }
  $$

  \noindent
  between those simply-connected $\mathbb{R}$-finite
  equivariant homotopy types
  (Def. \ref{SubcategoryOfEquivariantSimplyConnectedRFiniteHomotopyTypes})
  which are rational (Def. \ref{EquivariantRationalization})
  over the real numbers and
  formal duals of cohomologically connected 1-connected
  \eqref{Cohomologically1Connected} equivariant dgc-algebras.

  \noindent
  {\bf (ii)} The derived adjunction unit
  is equivariant rationalization (Def. \ref{EquivariantRationalization}):
  \begin{equation}
    \label{DerivedPLdRUnitEquivalentToRationalization}
    \mathscr{X}
    \,\in\,
    \EquivariantSimplyConnectedRFiniteHomotopyTypes
    \;\;
    \Rightarrow
    \;\;
    \raisebox{20pt}{
    \xymatrix@R=18pt@C=3em{
      \mathscr{X}
      \ar[rr]^-{
        \mathbb{D}
        \eta^{\mathrm{PLdR}}_{\scalebox{0.7}{$\mathscr{X}$}}
      }
      \ar@{=}[d]
      &&
      \mathbb{R}\exp
      \,\circ\,
      \mathbb{L} \Omega^\bullet_{\mathrm{PLdR}}
      \big(
        \mathscr{X}
      \big)\;.
      \ar[d]^-{ \simeq }
      \\
      \mathscr{X}
      \ar[rr]^-{
        \eta^{\mathbb{R}}_{\scalebox{0.7}{$\mathscr{X}$}}
      }
      &&
      L_{\mathbb{R}} \mathscr{X}
    }
    }
  \end{equation}
\end{prop}

\begin{remark}
  That the equivariant derived PLdR-unit \eqref{DerivedPLdRUnitEquivalentToRationalization}
  models equivariant rationalization is not made explicit in \cite{Scull08},
  but it follows immediately from the fact that:

\noindent   {\bf (a)} by definition, the equivariant PLdR adjunction is stage-wise over
  $G/H \,\in\, G \mathrm{Orb}$ the plain PLdR adjunction;

\noindent   {\bf (b)} the derived unit of the plain PLdR-adjunction
  models plain rationalization
  by the non-equivariant fundamental theorem
  (e.g. \cite[Prop. 3.60]{FSS23-Char});
  and

\noindent   {\bf (c)}
  that equivariant rationalization
  (Def. \ref{EquivariantRationalization}) is stage-wise plain rationalization.
\end{remark}

\medskip

\noindent {\bf Equivariant rational Whitehead $L_\infty$-algebras}

\begin{defn}[Equivariant Whitehead $L_\infty$-algebra]
  \label{EquivariantWhiteheadLInfinityAlgebra}
  For
  $\raisebox{1pt}{\textesh}\orbisingular \big(X\!\sslash\!G\big) \,\in\, \EquivariantSimplyConnectedRFiniteHomotopyTypes$
  (Def. \ref{SubcategoryOfEquivariantSimplyConnectedRFiniteHomotopyTypes}),
  we say that its
  \emph{equivariant Whitehead $L_\infty$-algebra}
  $$
    \mathfrak{l}
    \orbisingular \big(X\!\sslash\!G\big)
    \;\in\;
    \EquivariantLInfinityAlgebras
  $$

\noindent
  is the equivariant $L_\infty$-algebra (Def. \ref{EquivariantLInfinityAlgebras})
  whose equivariant Chevalley-Eilenberg algebra
  \eqref{DefiningEmbeddingOfEquivariantLInfinityAlgebrasInEquivariantdgcAlgebras}
  is the minimal model
  (well-defined by Prop. \ref{ExistenceOfEquivariantMinimalModels})
  of the equivariant PL de Rham complex (Def. \ref{EquivariantPLdeRhamComplexFunctor})
  of $\raisebox{1pt}{\textesh}\orbisingular \big(X\!\sslash\!G\big)$:
\vspace{-1mm}
  \begin{equation}
    \label{WhiteheadLInfinityAlgebraMinimalModelQuasiIso}
    \hspace{-4mm} 
    \mathrm{CE}
    \big(
      \mathfrak{l}
      \orbisingular \big(X\!\sslash\!G\big)
   \! \big)
   \! :=\!\!
    \xymatrix@C=1em{
      \Omega^\bullet_{\mathrm{PLdR}}(X)_{\mathrm{min}}
      \ar[rr]^-{ p^{\mathrm{min}} }_-{ \in \; \mathrm{W} }
      &&
      \Omega^\bullet_{\mathrm{PLdR}}(X)
    }
    \in
    \EquivariantdgcAlgebras
    .
  \end{equation}

\end{defn}

\begin{prop}[Equivariant rational homotopy groups in the equivariant
Whitehead $L_\infty$-algeba {\cite[Thm. 6.2 (2)]{Tri82}}]
  \label{EquivariantRationalHomotopyGroupsInEquivariantWhiteheadLInfinityAlgebra}
  For
  $\raisebox{1pt}{\rm\textesh}\orbisingular \big(X\!\sslash\!G\big) \,\in\, \EquivariantSimplyConnectedRFiniteHomotopyTypes$
  (Def. \ref{SubcategoryOfEquivariantSimplyConnectedRFiniteHomotopyTypes}),
  the equivariant rational homotopy groups of $\Omega X$
  (Example \ref{EquivariantRationalHomotopyGroups})
  are equivalent to the
  underlying equivariant graded vector space (Def. \ref{MinimalEquivariantLInfinityAlgebra})
  of the equivariant Whitehead $L_\infty$-algebra
  (Def. \ref{EquivariantWhiteheadLInfinityAlgebra})
  of $\orbisingular \big(X\!\sslash\!G\big)$:
  \begin{equation}
    \overset{
      \mathclap{
      \raisebox{3pt}{
        \tiny
        \color{darkblue}
        \bf
        \begin{tabular}{c}
          equivariant
          \\[-3pt]
          Whitehead $L_\infty$-algebra
        \end{tabular}
      }
      }
    }{
    \big(
      \mathfrak{l}
      \orbisingular
      \big(
        X \!\sslash\! G
      \big)
    \big)_\bullet
    }
    \;\;\;\;
      \simeq
    \;\;\;\;
    \overset{
      \mathclap{
      \raisebox{3pt}{
        \tiny
        \color{darkblue}
        \bf
        \begin{tabular}{c}
          equivariant rational
          \\[-3pt]
          homotopy groups of
          \\[-3pt]
          equivariant loop space
        \end{tabular}
      }
      }
    }{
      \underline{\pi}_{\, \bullet}
      (
        \Omega X
      )
      \otimes_{\scalebox{.5}{$\mathbb{Z}$}}
      \mathbb{R}
    }
    \,.
  \end{equation}

\end{prop}

\noindent {\bf Examples of equivariant Whitehead $L_\infty$-algebras.}

\begin{prop}[$\ZTwo$-Equivariant minimal model of twistor space]
  \label{Z2EquivariantMinimalModelOfTwistorSpace}
  The equivariant minimal model
  (Def. \ref{MinimalEquivariantdgcAlgebras})
  of
  the $\Grefl$-equivariant twistor space
  (Example \ref{GhetEquivariantTwistorSpace})
  is the following $\ZTwo$-equivariant dgc-algebra
  (Def. \ref{EquivariantdgcAlgebras}):

  \vspace{-5mm}
  \begin{equation}
    \label{MinimalModelOfTwistorSpaceModZ2}
    \mathrm{CE}
    \big(
      \mathfrak{l}
      \orbisingular
      \big(
        \mathbb{C}P^3 \!\sslash\! \ZTwo
      \big)
    \big)
    \;
    :
    \!\!\!
  \raisebox{40pt}{
  \xymatrix@C=1pt{
    \ZTwo/\ZTwo
    \ar@{<-}[d]
    &
    \mapsto
    &
    \mathbb{R}
    \!
    \left[
      \!\!
      \def\arraystretch{1}
      {\begin{array}{c}
        h_3,
        \\
        f_2
      \end{array}}
      \!\!
    \right]
    \!\big/\!
    \left(
      {\begin{aligned}
        d\, h_3  & = \phantom{\omega_4}\; - f_2 \wedge f_2
        \\[-4pt]
        d\, f_2 & = 0
      \end{aligned}}
    \right)
    \ar@<-50pt>@{<<-}[d]
    \\
    \ZTwo/1
    \ar@(dl,dr)|-{\;\ZTwo}
    &\mapsto&
    \mathbb{R}
    \!
    \left[
      \!\!
      \def\arraystretch{1}
      {\begin{array}{c}
        h_3,
        \\
        f_2
        \\
        \omega_7,
        \\
        \omega_4
      \end{array}}
      \!\!
    \right]
    \!\big/\!
    \left(
      {\begin{aligned}
        d\, h_3  & = \omega_4 - f_2 \wedge f_2
        \\[-4pt]
        d\, f_2 & = 0
        \\[-4pt]
        d\, \omega_7 & = - \omega_4 \wedge \omega_4
        \\[-4pt]
        d\, \omega_4 & = 0
      \end{aligned}}
    \right)
  }
  }
\end{equation}
\end{prop}

\begin{proof}
  {\bf (i)}
  Checking that \eqref{MinimalModelOfTwistorSpaceModZ2}
  is indeed a minimal equivariant dgc-algebra is the content of
  Example \ref{CheckingTwistorSpaceModZ2MinimalModel},
  where this minimal algebra is obtained in
  \eqref{CP3ModZ2MinimalEquvariantdgcAlgebra}.

  \noindent
  {\bf (ii)}
  It remains to see that
  \eqref{MinimalModelOfTwistorSpaceModZ2}
  has indeed
  the algebraic homotopy type
  of the rationalized equivariant twistor space,
  under the fundamental theorem (Prop. \ref{FundamentalTheoremOfdgcAlgebraicEquivariantRationalHomotopyTheory}).
  By \eqref{Z2EquivariantTwistorSpaceAsPresheafOnOrbitCategory},
  this amounts to showing that the right vertical morphism
  of ordinary dgc-algebras in \eqref{MinimalModelOfTwistorSpaceModZ2}
  is a dgc-algebraic model
  (under the non-equivariant fundamental theorem of rational homotopy theory,
  \cite[\S 8]{BousfieldGugenheim76} reviewed as \cite[Prop. 3.59]{FSS23-Char})
  of the inclusion of the fiber of the twistor fibration \eqref{TwistorFibration}.
  But, by \cite[Lem. 3.71]{FSS23-Char}),
  the dgc-algebra model for this fiber
  is the cofiber of the minimal relative model of the
  twistor fibration. The latter is given in \cite[Lem. 2.13]{FSS20c},
  and its cofiber manifestly coincides with
  \eqref{MinimalModelOfTwistorSpaceModZ2}.

  \noindent
  {\bf (iii)} As a consistency check, notice that the equivariant rational homotopy groups
  of twistor space \eqref{Z2EquivariantHomotopyGroupsOfTwistorSpace}
  do match the generators \eqref{GeneratorsForMinimalEquivariantModelOFTwistorSpaceModZ2}
  of this minimal model;
  as it must be, by Prop. \ref{EquivariantRationalHomotopyGroupsInEquivariantWhiteheadLInfinityAlgebra}.
\end{proof}

\begin{prop}[$\ZTwo$-Equivariant relative minimal model of $\SpLR$-parametrized twistor space]
  \label{Z2EquivariantRelativeMinimalModelOfSpin3ParametrizedTwistorSpace}
  The equivariant relative minimal model
  (Def. \ref{MinimalEquivariantdgcAlgebras})
  of
  the $\Grefl$-equivariant $\SpLR$-parametrized twistor space
  (Ex. \ref{GhetEquivariantParametrizedTwistorSpace})
  is the following $\ZTwo$-equivariant dgc-algebra
  (Def. \ref{EquivariantdgcAlgebras}) under
  $
    \mathrm{CE}
    \big(
      \mathfrak{l} B \SpLR
    \big)
    \!=\!
    \mathbb{R}
      \big[
        \!
        \scalebox{.76}{$\tfrac{1}{4}$}p_1
        \!
      \big]
    \!\big/\!
    \big(
      d\, \scalebox{.76}{$\tfrac{1}{4}$}p_1
      = 0
    \big)
  $:
  \vspace{0mm}
  \begin{equation}
    \label{MinimalModelOfSpinParametrizedTwistorSpaceModZ2}
    \hspace{-3mm}
    \scalebox{0.69}{$
    \mathrm{CE}
    \bigg(\!\!
      \Big(
      \mathfrak{l}_{\scalebox{.6}{$B \SpLR$}}
      \big(
        \orbisingular
        (
          \overset{
            \mathclap{
            \;\;\;\;\;\;\;
            \raisebox{4pt}{
            \rotatebox[origin=l]{43}{
            $
            \mathclap{
            \raisebox{3pt}{
              \tiny
              \color{darkblue}
              \bf
              \def\arraystretch{.7}
              \begin{tabular}{c}
                twistor
                \\
                space
              \end{tabular}
            }
            }
            $
            }
            }
            }
          }{
            \mathbb{C}P^3
          }
          \overset{
            \mathclap{
            \;\;\;\;\;\;\;\;\;\;\;
            \raisebox{4pt}{
            \rotatebox[origin=l]{40}{
            $
            \mathclap{
            \raisebox{3pt}{
              \tiny
              \color{darkblue}
              \bf
              \def\arraystretch{.7}
              \begin{tabular}{c}
                orbifolded
                \\
                wrt $\Grefl$
              \end{tabular}
            }
            }
            $
            }
            }
            }
          }{
            \!\sslash
            \Grefl
          }
        )
      \big)
      \overset{
            \mathclap{
            \;\;\;\;\;\;\;\;\;
            \raisebox{6pt}{
            \rotatebox[origin=l]{40}{
            $
            \mathclap{
            \raisebox{3pt}{
              \tiny
              \color{darkblue}
              \bf
              \def\arraystretch{.7}
              \begin{tabular}{c}
                parametrized
                \\
                wrt $\SpLR$
              \end{tabular}
            }
            }
            $
            }
            }
            }
      }{
        \!\sslash
        \SpLR
      }
    \!\!  \Big)
   \!\! \bigg)
    \;
    :
    \!\!
  \raisebox{40pt}{
  \xymatrix@C=-2pt@R=1.5em{
    \ZTwo/1
    \ar@{->}[d]
    \ar@(ul,ur)|-{\;\ZTwo}
    &
    \longmapsto
    &
    \mathrm{CE}
    \big(
      \mathfrak{l} B \SpLR
    \big)
    \!\!
    \left[
      \!\!
      \def\arraystretch{1}
      {\begin{array}{c}
        h_3,
        \\
        f_2
        \\
        \omega_7,
        \\
        \widetilde \omega_4
      \end{array}}
      \!\!
    \right]
    \!\big/\!
    \left(
      \def\arraystretch{1}
      {\begin{aligned}
        d\, h_3  & = \widetilde \omega_4 - \tfrac{1}{2}p_1 -  f_2 \wedge f_2
        \\[-4pt]
        d\, f_2 & = 0
        \\[-4pt]
        d\, \omega_7
          & =
          -
          \widetilde \omega_4 \wedge
          \big(
            \widetilde \omega_4 - \tfrac{1}{2}p_1
          \big)
        \\[-4pt]
        d\, \widetilde \omega_4 & = 0
      \end{aligned}}
    \right)
    \ar@<-42pt>@{->>}[d]
    \\
    \ZTwo/\ZTwo
    &\longmapsto&
    \mathrm{CE}
    \big(
      \mathfrak{l} B \SpLR
    \big)
    \!\!
    \left[
      \!\!
      \def\arraystretch{1}
      {\begin{array}{c}
        h_3,
        \\
        f_2
      \end{array}}
      \!\!
    \right]
    \!\big/\!
    \left(
      \def\arraystretch{1}
      {\begin{aligned}
        d\, h_3  & = \phantom{\omega_4}\; - \tfrac{1}{2}p_1 - f_2 \wedge f_2
        \\[-4pt]
        d\, f_2 & = 0
      \end{aligned}}
    \right)
    \,,
  }
  }
$}
\end{equation}
\vspace{-.5cm}

\noindent
where

\noindent
{\bf (a)}  all closed generators are normalized such as
to be rational images of integral and integrally in-divisible classes;

\noindent
{\bf (b)} $\omega := \widetilde \omega - \tfrac{1}{4}p_1$
is fiberwise the pullback along
$\mathbb{C}P^3 \overset{t_{\mathbb{H}}}{\longrightarrow} S^4$
\eqref{TwistorFibration}
of the volume element on $S^4$;
\\
{\bf (c)} $f_2$ is fiberwise the volume element on
$S^2 \xrightarrow{\mathrm{fib}(t_{\mathbb{H}}) } \mathbb{C}P^3$.
\end{prop}
\begin{proof}
  {\bf (i)}
  To see that \eqref{MinimalModelOfSpinParametrizedTwistorSpaceModZ2}
  is relative minimal,
  observe that it is
  obtained from the equivariant base dgc-algebra
  \vspace{1mm} 
  $$
    \xymatrix@R=1em{
      \ZTwo/1
      \ar@(ul,ur)|-{\; \ZTwo \,}
      \ar[d]
      &
      \longmapsto
      &
      \mathrm{CE}
      \big(
        \mathfrak{l}
        B \SpLR
      \big)
      \ar[d]_-{ \mathrm{id} }
      \ar@{=}[r]
      &
      \mathbb{R}
      \big[
        \tfrac{1}{4}p_1
      \big]
      \!\big/\!
      \big(
        d\, \tfrac{1}{4}p_1
        \,=\,
        0
      \big)
      \\
      \ZTwo/\ZTwo
      &\longmapsto&
      \mathrm{CE}
      \big(
        \mathfrak{l}
        B \SpLR
      \big)
    }
  $$
  \vspace{-2mm}

  \noindent
  by the same three cell attachments
  as in the construction of the absolute minimal model of Example, \ref{CheckingTwistorSpaceModZ2MinimalModel}
  for the plain equivariant twistor space
  (Prop. \ref{Z2EquivariantMinimalModelOfTwistorSpace}),
  subject only to these replacements:
  $$
    \begin{aligned}
      f_2 \wedge f_2 & \;\longmapsto\; f_2 \wedge f_2 + \tfrac{1}{2}p_1
      \\[-3pt]
      \omega_4 \wedge \omega_4
      & \;\longmapsto\;
      \widetilde \omega_4
      \wedge
      \big(
        \widetilde \omega_4
        -
        \tfrac{1}{2} p_1
      \big)
    \end{aligned}
  $$

\noindent
  in the attaching maps $\phi_3$ \eqref{AttachingMapInDegree3ForZ2EquivariantTwistorSpace}
  and $\phi_7$ \eqref{AttachingMapInDegree7ForZ2EquivariantTwistorSpace}, respectively.

\vspace{1mm} 
  \noindent
  {\bf (ii)}
  By the fundamental theorem (Prop. \ref{FundamentalTheoremOfdgcAlgebraicEquivariantRationalHomotopyTheory}),
  it remains to see that \eqref{MinimalModelOfSpinParametrizedTwistorSpaceModZ2}
  is weakly equivalent to the
  relative equivariant PL de Rham complex
  of equivariant parametrized twistor space:

\vspace{1mm} 
  \noindent
  {\bf (ii.1)}
    First observe that the relative minimal model
    $
      \mathrm{CE}
      \big(
        \mathfrak{l}
        \big(
          t_{\mathbb{H}} \sslash \SpLR
        \big)
      \big)
    $
    for the
    {\it non-}equivariant
    $\SpLR$-parametrized
    twistor fibration $t_{\mathbb{H}}$,
    relative to the minimal model
    of $S^4 \!\sslash\! \SpLR$
    relative to $B \SpLR$, is as follows,
    with generators normalized
    as stated in the claim above:
  \begin{equation}
    \label{RelativeMinimalModelForParametrizedTwistorFibration}
    \hspace{-.8cm}
    \adjustbox{scale=0.83, raise=83pt}{
    \xymatrix@C=-10pt@R=1.8em{
      &
      S^2 \!\sslash\! \SpLR
      \ar[dd]^-{
        \scalebox{.7}{$
          \def\arraystretch{1.2}
          \begin{array}{c}
            \mathrm{hofib}_{\scalebox{.7}{$B \SpLR$}}
            (
              t_{\mathbb{H}}
              \sslash
              \SpLR
            )
            \\
            \simeq
            \;
            \mathrm{hofib}(t_{\mathbb{H}})
            \sslash
            \SpLR
            \\
            \mbox{
              (by Lemma \ref{HomotopyFibersOfHomotopyQuotientedMorphisms})
            }
          \end{array}
        $}
      }
      \ar[dddl]_-{
              \rho_{S^2}
      }
      &&&
      \mathbb{R}
      \big[
        \!
        \scalebox{.76}{$\tfrac{1}{4}$}p_1
        \!
      \big]
      \!\!
      \left[
        \!\!\!
        \def\arraystretch{1}
        {\begin{array}{l}
          h_3,
          \\
          f_2,
        \end{array}}
        \!\!\!
      \right]
      \!\big/\!
      \left(
       {\begin{aligned}
        d\, h_3 & =
          {\phantom{\widetilde \omega_4}}\;\;\;
          - \tfrac{1}{2}p_1 - f_2 \wedge f_2
        \\[-4pt]
        d\, f_2 & = 0
      \end{aligned}}
      \!
      \right)
      \ar@<-56pt>@{<<-}[dd]^-{
        \;
        \scalebox{.7}{$
        \mathrm{cof}_{\scalebox{.7}{$B \SpLR$}}
        \Big(
          \mathrm{CE}
          \big(
            \mathfrak{l}
            (t_{\mathbb{H}} \sslash \SpLR)
          \big)
        \Big)
        $}
      }
      \\
      {\phantom{A}}
      \\
      &
      \mathbb{C}P^3 \!\sslash\! \SpLR
      \ar[dd]^-{
        \mathclap{\phantom{\vert^{\vert}}}
        t_{\mathbb{H}} \sslash \SpLR
        \mathclap{\phantom{\vert_{\vert}}}
      }
      \ar[dl]|-{
        \;
        \mathclap{\phantom{\vert^{\vert}}}
        \rho_{\mathbb{C}P^3}
        \mathclap{\phantom{\vert_{\vert}}}
        \;
      }
      &
      {\phantom{.}}
      &&
      \mathbb{R}
      \big[
        \!
        \scalebox{.76}{$\tfrac{1}{4}$}p_1
        \!
      \big]
      \!\!
      \left[
        \!\!\!
        \def\arraystretch{1}
        {\begin{array}{l}
          h_3,
          \\
          f_2,
          \\
          \omega_7,
          \\
          \widetilde \omega_4
        \end{array}}
        \!\!\!
      \right]
      \!\big/\!
      \left(
       {\begin{aligned}
        d\, h_3 & = \widetilde \omega_4 - \tfrac{1}{2}p_1 - f_2 \wedge f_2
        \\[-4pt]
        d\, f_2 & = 0
        \\[-4pt]
        d\, \omega_7
          & =
          -
          \widetilde \omega_4 \wedge
          \big(
            \widetilde \omega_4
            -
            \tfrac{1}{2}p_1
          \big)
        \\[-4pt]
        d\, \widetilde \omega_4 & = 0
      \end{aligned}}
      \!
      \right)
      \ar@<-42pt>@{<-^{)}}[dd]^-{
        \;
        \scalebox{.7}{$
        \mathrm{CE}
        \big(
          \mathfrak{l}
          (t_{\mathbb{H}} \sslash \SpLR)
        \big)
        $}
        \mathrlap{
          \mbox{
            \tiny
            \def\arraystretch{1}
            {\begin{tabular}{l}
              \color{greenii}
              \bf
              relative minimal model
              \\
              {\color{greenii}
              \bf
              for $t_{\mathbb{H}} \sslash \SpLR$}
              \\
              (by \cite[Thm. 2.14]{FSS20c})
            \end{tabular}}
          }
        }
      }
      \\
      B \SpLR
      &
      {\phantom{AAAAAAAAAAAAAAA}}
      &&
      \mathbb{R}
      \!
      \big[
        \!
        \scalebox{.76}{$\tfrac{1}{4}$}p_1
        \!
      \big]
      \ar@{_{(}->}[ur]
      \ar@{^{(}->}[dr]
      \\
      {\phantom{AAAAAAAA}}
      &
      S^4 \!\sslash\! \SpLR
      \ar[ul]^-{
              \rho_{S^4}
      }
      & &&
      \mathbb{R}
      \!
      \big[
        \!
        \scalebox{.76}{$\tfrac{1}{4}$}p_1
        \!
      \big]
      \!\!
      \overset{
      }{
      \left[
        \!\!\!
        \def\arraystretch{1}
        {\begin{array}{l}
          \omega_7,
          \\
          \widetilde \omega_4
        \end{array}}
        \!\!\!
      \right]
      }
      \!\big/\!
      \left(
        \!
      {\begin{aligned}
        d\, \omega_7
          & =
          -
          \widetilde \omega_4 \wedge
          \big(
            \widetilde \omega_4
            -
            \tfrac{1}{2}p_1
          \big)
        \\[-4pt]
        d\, \widetilde \omega_4 & = 0
      \end{aligned}}
      \!
      \right)
    }
    }
  \end{equation}

  \noindent
  This is the statement of
  \cite[Thm. 2.14]{FSS20c},
  using the following notational simplifications in the present case:

  \noindent
  {\bf (a)}
    the Euler 8-class
    $\rchi_8$ appearing in \cite[(39)]{FSS20c}
    vanishes here under restriction along $B \SpLR \to B \mathrm{Sp}(2)$;

  \noindent
  {\bf (b)} we have applied to \cite[(49)]{FSS20c} the dgc-algebra isomorphism given by
  \begin{equation}
    \label{IsomorphismForTildeomega4}
       h_3 \;\leftrightarrow \; h_3\,,
      \qquad
       f_2 \; \leftrightarrow \; f_2\,,
      \qquad
       \omega_7 \; \leftrightarrow \; \omega_7\,,
      \qquad
       \omega_4 \; \leftrightarrow \; \widetilde \omega_4 - \tfrac{1}{4}p_1\;.
  \end{equation}

  \vspace{1mm} 
  \noindent
  {\bf (ii.2)}
  This being a non-equivariant relative minimal model,
  it comes with horizontal weak equivalences
  of non-equivariant dgc-algebras
  as shown
  in the bottom square of the following commuting diagram
  (by, e.g., \cite[Thm. 14.12]{FHT00}),
  which induces
  (by the {\it fiber lemma} \cite[\S II]{BousfieldKan72}
  in the form \cite[Prop. 15.5]{FHT00}\cite[Thm. 5.1]{FHT15})
  a weak equivalence on
  plain cofibers (which is forms on $S^2$, by Lemma \ref{HomotopyFibersOfHomotopyQuotientedMorphisms}),
  as shown in the following top square:
  \newpage
  \begin{equation}
    \label{FiberLemmaForSpParametrizedTwistorSpace}
    \hspace{-.5cm} 
    \xymatrix@R=15pt@C=7pt{
      \Omega^\bullet_{\mathrm{PLdR}}
      \big(
        S^2 {\phantom{\sslash \SpLR}}
      \big)
      \ar@{<-}[dd]^-{
       \scalebox{.7}{$   \Omega^\bullet_{\mathrm{PLdR}}
        \big(
          \mathrm{fib}
          \big(
            t_{\mathbb{H}}
            \sslash
            \SpLR
          \big)
        \big)
        $}
      }^-{
      }
      \ar@{<--}[rr]^-{ \in \; \mathrm{W} }
      &&
      {\phantom{
      \big[
        \!
        \scalebox{.76}{$\tfrac{1}{4}$}p_1
        \!
      \big]
      }}
      \mathbb{R}
      \!\!
      \left[
        \!\!\!
        \def\arraystretch{1}
        {\begin{array}{l}
          h_3,
          \\
          f_2,
        \end{array}}
        \!\!\!
      \right]
      \!\big/\!
      \left(
       {\begin{aligned}
        d\, h_3 & =
          {\phantom{\widetilde \omega_4}}\;\;\;
          {\phantom{- \tfrac{1}{2}p_1}}
          - f_2 \wedge f_2
        \\[-4pt]
        d\, f_2 & = 0
      \end{aligned}}
      \!
      \right)
      \ar@<-56pt>@{<<-}[dd]^-{
        \;
        \scalebox{.7}{$
        \mathrm{cof}
        \Big(
          \mathrm{CE}
          \big(
            \mathfrak{l}
            (t_{\mathbb{H}} \sslash \SpLR)
          \big)
        \Big)
        $}
      }
      \\
      \\
      \Omega^\bullet_{\mathrm{PLdR}}
      \big(
        \mathbb{C}P^3 \!\sslash\! \SpLR
      \big)
      \ar@{<-}[dd]^-{
  \scalebox{.7}{$        \Omega^\bullet_{\mathrm{PLdR}}
        \big(
          t_{\mathbb{H}} \sslash \SpLR
        \big)
        $}
      }
      \ar@{<-}[rr]^-{ \in \; \mathrm{W} }
      &&
      \mathbb{R}
      \big[
        \!
        \scalebox{.76}{$\tfrac{1}{4}$}p_1
        \!
      \big]
      \!\!
      \left[
        \!\!\!
        \def\arraystretch{1}
        {\begin{array}{l}
          h_3,
          \\
          f_2,
          \\
          \omega_7,
          \\
          \widetilde \omega_4
        \end{array}}
        \!\!\!
      \right]
      \!\big/\!
      \left(
       {\begin{aligned}
        d\, h_3 & = \widetilde \omega_4 - \tfrac{1}{2}p_1 - f_2 \wedge f_2
        \\[-4pt]
        d\, f_2 & = 0
        \\[-4pt]
        d\, \omega_7
          & =
          -
          \widetilde \omega_4 \wedge
          \big(
            \widetilde \omega_4
            -
            \tfrac{1}{2}p_1
          \big)
        \\[-4pt]
        d\, \widetilde \omega_4 & = 0
      \end{aligned}}
      \!
      \right)
      \ar@<-56pt>@{<-^{)}}[dd]^-{
        \;
        \scalebox{.7}{$
        \mathrm{CE}
        \big(
          \mathfrak{l}
          (t_{\mathbb{H}} \sslash \SpLR)
        \big)
        $}
      }
      \\
      \\
      \Omega^\bullet_{\mathrm{PLdR}}
      \big(
        S^4 \!\sslash\! \SpLR
      \big)
      \ar@{<-}[rr]^-{ \in \; \mathrm{W} }
      &&
      \mathbb{R}
      \!
      \big[
        \!
        \scalebox{.76}{$\tfrac{1}{4}$}p_1
        \!
      \big]
      \!\!
      \overset{
      }{
      \left[
        \!\!\!
        \def\arraystretch{1}
        {\begin{array}{l}
          \omega_7,
          \\
          \widetilde \omega_4
        \end{array}}
        \!\!\!
      \right]
      }
      \!\big/\!
      \left(
        \!
      {\begin{aligned}
        d\, \omega_7
          & =
          -
          \widetilde \omega_4 \wedge
          \big(
            \widetilde \omega_4
            -
            \tfrac{1}{2}p_1
          \big)
        \\[-4pt]
        d\, \widetilde \omega_4 & = 0
      \end{aligned}}
      \!
      \right)
    }
  \end{equation}

  \noindent
  (Here we are using that with $t_{\mathbb{H}}$ also
  $t_{\mathbb{H}}\sslash \SpLR
    \;:=\;
  \frac{t_{\mathbb{H}}\times W \SpLR}{\SpLR} $ is a fibration,
  by the right Quillen functor \eqref{SimplicialBorelConstructionEquivalence}
  in Prop. \ref{InfinityActionsEquivalentToFibrationsOverClassifyingSpace},
  and that all spaces involved are simply-connected, so that
  all the technical assumptions in \cite[(5.1)]{FHT15} are indeed met.)

\vspace{1mm}
  \noindent
  {\bf (ii.3)}
  Then observe that
  \begin{align}
    \label{RealCohomologyOfS2HomotopyQuotientedBySpLR}
    H^\bullet
    \big(
      S^2 \!\sslash\! \SpLR
      ;
      \mathbb{R}
    \big)
   & \simeq
    \mathbb{R}
    \big[
      \omega_2,\, \tfrac{1}{4}p_1
    \big]/
    \big(
      (\omega_2)^2
    \big)
    \nonumber 
\\
&    
    \simeq
    H^\bullet
    \big(
      B \SpLR
      ;
      \,
      \mathbb{R}
    \big)
    \,\otimes_{\scalebox{.5}{$\mathbb{R}$}}\,
    H^\bullet
    \big(
      S^2
      ;
      \,
      \mathbb{R}
    \big).
  \end{align}

\noindent   This follows readily from the Gysin exact sequence
  (e.g. \cite[\S 15.30]{Switzer75})
  \begin{equation}
    \label{GysinSequenceForSpLRActionOnS2}
    \hspace{-.9cm}
\adjustbox{scale=0.9}{$
    \begin{tikzcd}[column sep=small]
 \cdots \rar & H^\bullet
      \big(
        B \SpLR
        ;
        \,
        \mathbb{R}
      \big)
       \ar[r, "{\rho^\ast_{S^2}}" ]
      &  H^\bullet
      \big(
        S^2 \!\sslash\! \SpLR
        ;
        \,
        \mathbb{R}
      \big)
        \ar[r, "{\int_{S^2} }"]
             \ar[draw=none]{d}[name=X, anchor=center]{}
    &  H^{\bullet - 2}
      \big(
        B \SpLR
        ;
        \,
        \mathbb{R}
      \big)
      \ar[rounded corners,
            to path={ -- ([xshift=2ex]\tikztostart.east)
                      |- (X.center) \tikztonodes
                      -| ([xshift=-2ex]\tikztotarget.west)
                      -- (\tikztotarget)}]{dll}[at end]{{ c \cup (-) = 0 }} \\      
   & H^{\bullet + 1}
      \big(
        B \SpLR
        ;
        \,
        \mathbb{R}
      \big) \ar[r] & \cdots 
\end{tikzcd}
$}
  \end{equation}

  \noindent
  for the $S^2$-fiber sequence
  $
         S^2 \xrightarrow{ \mathrm{hofib}(\rho_{S^2}) }
       S^2 \sslash \SpLR
      \xrightarrow{ \rho_{S^2} }
       B \SpLR
      $
    that corresponds to the $\SpLR$-action on $S^2$,
  by Prop. \ref{InfinityActionsEquivalentToFibrationsOverClassifyingSpace};
  and using that
  $
    H^\bullet
    \big(
      B \SpLR;
      \,
      \mathbb{R}
    \big)
    $
    $
    \simeq
       \mathbb{R}\big[ \scalebox{.8}{$\tfrac{1}{4}$}p_1 \big]
  $
  (e.g. \cite[Lemma 4.24]{FSS23-Char})
  is concentrated in degrees divisible by 4
  (so that, in particular, the Euler class
  $c \,\in\, H^3\big( B \SpLR;\, \mathbb{R} \big) \,\simeq\, 0$
  in \eqref{GysinSequenceForSpLRActionOnS2} vanishes).

  \noindent
  But using \eqref{RealCohomologyOfS2HomotopyQuotientedBySpLR}
  in \eqref{FiberLemmaForSpParametrizedTwistorSpace}
  implies that also the induced map on {\it relative} fibers
  \eqref{SquareForBGParametrizedHomotopyFiber}
  over $B \SpLR$ is a weak equivalence:

  \newpage
  
  \begin{equation}
    \label{FiberLemmaForSpParametrizedTwistorSpaceRelative}
    \hspace{-4mm}
    \adjustbox{scale=0.9}{$
    \xymatrix@R=12pt@C=7pt{
      \ZTwo/\ZTwo
      \ar@{<-}[dd]
      &
      \Omega^\bullet_{\mathrm{PLdR}}
      \big(
        S^2 \sslash \SpLR
      \big)
      \ar@{<-}[dd]^-{\adjustbox{raise=1.5cm, scale=0.7}{$
        \Omega^\bullet_{\mathrm{PLdR}}
        \big(
          \mathrm{fib}_{B \SpLR}
          \big(
            t_{\mathbb{H}}
            \sslash
            \SpLR
          \big)
        \big)
        $}
      }^-{
        \;\simeq
        \;\;
        \scalebox{0.7}{$
        \Omega^\bullet_{\mathrm{PLdR}}
        \big(
          \mathrm{fib}
          (
            t_{\mathbb{H}}
          )
          \sslash
          \SpLR
        \big)
        $}
      }
      \ar@{<--}[rr]^-{ \in \; \mathrm{W} }
      &&
      \mathbb{R}
      \big[
        \!
        \scalebox{.76}{$\tfrac{1}{4}$}p_1
        \!
      \big]
      \!\!
      \left[
        \!\!\!
        \def\arraystretch{1}
        {\begin{array}{l}
          h_3,
          \\
          f_2,
        \end{array}}
        \!\!\!
      \right]
      \!\big/\!
      \left(
       {\begin{aligned}
        d\, h_3 & =
          {\phantom{\widetilde \omega_4}}\;\;\;
          - \tfrac{1}{2}p_1 - f_2 \wedge f_2
        \\[-4pt]
        d\, f_2 & = 0
      \end{aligned}}
      \!
      \right)
      \ar@<-56pt>@{<<-}[dd]^-{
        \;
        \scalebox{.7}{$
        \mathrm{cof}_{\scalebox{.7}{$B \SpLR$}}
        \Big(
          \mathrm{CE}
          \big(
            \mathfrak{l}
            (t_{\mathbb{H}} \sslash \SpLR)
          \big)
        \Big)
        $}
      }
      \\
      \\
      \ZTwo/1
      \ar@(ld,dr)|-{ \;\ZTwo\; }
      &
      \Omega^\bullet_{\mathrm{PLdR}}
      \big(
        \mathbb{C}P^3 \!\sslash\! \SpLR
      \big)
      \ar@{<-}[rr]^-{ \in \; \mathrm{W} }
      &&
      \mathbb{R}
      \big[
        \!
        \scalebox{.76}{$\tfrac{1}{4}$}p_1
        \!
      \big]
      \!\!
      \left[
        \!\!\!
        \def\arraystretch{1}
        {\begin{array}{l}
          h_3,
          \\
          f_2,
          \\
          \omega_7,
          \\
          \widetilde \omega_4
        \end{array}}
        \!\!\!
      \right]
      \!\big/\!
      \left(
       {\begin{aligned}
        d\, h_3 & = \widetilde \omega_4 - \tfrac{1}{2}p_1 - f_2 \wedge f_2
        \\[-4pt]
        d\, f_2 & = 0
        \\[-4pt]
        d\, \omega_7
          & =
          -
          \widetilde \omega_4 \wedge
          \big(
            \widetilde \omega_4
            -
            \tfrac{1}{2}p_1
          \big)
        \\[-4pt]
        d\, \widetilde \omega_4 & = 0
      \end{aligned}}
      \!
      \right).
    }
    $}
  \end{equation}

\vspace{1mm}
\noindent
{\bf (ii.4)}
By Lemma \ref{HomotopyFibersOfHomotopyQuotientedMorphisms}
applied to \eqref{GHetEquivariantSpLRParametrizedTwistorSpace},
we see that the left morphism
in \eqref{FiberLemmaForSpParametrizedTwistorSpaceRelative}
is equivalently the inclusion of the fixed-locus
in the $\Grefl$-equivariant $\SpLR$-parametrized twistor space
(Example \ref{GhetEquivariantParametrizedTwistorSpace}).
Thus, by the stage-wise definition of the equivariant
PL de Rham complex (Def. \ref{EquivariantPLdeRhamComplexFunctor}),
it follows that the left morphism in
\eqref{FiberLemmaForSpParametrizedTwistorSpaceRelative}
is the PL de Rham complex of
$\Grefl$-equivariant $\SpLR$-parametrized twistor space
(as indicated by alignment with the $\Grefl$-orbit category on the
far left of \eqref{FiberLemmaForSpParametrizedTwistorSpace}).
Finally this means, by the fundamental theorem
(Prop. \ref{FundamentalTheoremOfdgcAlgebraicEquivariantRationalHomotopyTheory}),
that the commuting square in \eqref{FiberLemmaForSpParametrizedTwistorSpace}
exhibits the claimed equivariant dgc-algebra
\eqref{MinimalModelOfSpinParametrizedTwistorSpaceModZ2InIntroduction}
as indeed modeling the equivariant rational homotopy type of the
$\Grefl$-equivariant $\SpLR$-parametrized twistor space.
(The images on the left of the generators on the right of
\eqref{FiberLemmaForSpParametrizedTwistorSpace}
are indeed all invariant under the $\Grefl \subset \mathrm{Sp}(2)$-action,
by \cite[Lemma 5.5]{BMSS19}).
\end{proof}

\subsection{Equivariant non-abelian de Rham theorem}
\label{EquivariantNonAbelianDeRhamTheorem}

We introduce properly equivariant non-abelian de Rham cohomology
with coefficients in equivariant $L_\infty$-algebras,
in direct generalization of the non-equivariant
discussion in \cite[\S 3.3]{FSS23-Char}.
Our key example here
is the non-abelian cohomology of equivariant twistorial
differential forms (Example \ref{FlatEquivariantTwistorialDifferentialForms} below).
The main result is the proper equivariant non-abelian de Rham
theorem (Prop. \ref{EquivariantNonabelianDeRhamTheorem})
and its twisted version (Prop. \ref{EquivariantTwistedNonabelianDeRhamTheorem}).
The specialization to
traditional Borel-equivariant abelian de Rham cohomology is
the content of Prop. \ref{ReproducingTraditionalBorelEquivariantdeRhamCohomology}
below.

\medskip

\noindent {\bf Flat equivariant $L_\infty$-algebra valued differential forms.}

\vspace{1mm}
\noindent In equivariant generalization of \cite[Def. 3.77]{FSS23-Char}, we set:

\begin{defn}[Flat equivariant $L_\infty$-algebra valued differential forms]
  \label{FlatEquivariantLInfinityAlgebraValuedForms}
  Let
  $
    \underline{\mathfrak{g}}
    \,\in\,
    \EquivariantLInfinityAlgebras
  $
  (Def. \ref{EquivariantLInfinityAlgebras})
  and $G \acts \, X \,\in\, \GActionsOnSmoothManifolds$ (Def. \ref{ProperGActionsOnSmoothManifolds}).
  Then the set of {\it flat equivariant $\underline{\mathfrak{g}}$-valued differential forms} on $X$
  is the hom-set \eqref{HomSets}
  $$
    \Omega_{\mathrm{dR}}
    \big(\!
      \orbisingular
      \big(
        X \!\sslash\! G
      \big)
      ;
      \,
      \underline{\mathfrak{g}}
    \, \big)_{\mathrm{flat}}
    \;\;
    :=
    \;\;
    \EquivariantdgcAlgebras
    \Big(
      \mathrm{CE}
      \big(
        \underline{\mathfrak{g}}
      \big)
      \,,\,
      \Omega^\bullet_{\mathrm{dR}}
      \big(
        \orbisingular
        \big(
          X \!\sslash\! G
        \big)
      \big)
    \! \Big)
  $$

  \noindent
  of equivariant dgc-algebras
  (Def. \ref{EquivariantdgcAlgebras})
  from the equivariant Chevalley-Eilenberg algebra
  \eqref{DefiningEmbeddingOfEquivariantLInfinityAlgebrasInEquivariantdgcAlgebras}
  of $\underline{\mathfrak{g}}$
  to the equivariant smooth de Rham complex
  (Def. \ref{EquivariantSmoothdeRhamComplex})
  of $X$.
\end{defn}

In equivariant generalization of \cite[Def. 3.92]{FSS23-Char}, we set:

\begin{defn}[Flat twisted equivariant $L_\infty$-algebra valued differential forms on $G$-orbifold]
 \label{FlatTwistedEquivariantLInfinityAlgebraValuedForms}
Consider an {\it equivariant $L_\infty$-algebraic local coefficient
bundle} in the form of a fibration of equivariant $L_\infty$-algebras
(Def. \ref{EquivariantLInfinityAlgebras})
whose equivariant Chevalley-Eilenberg algebras \eqref{DefiningEmbeddingOfEquivariantLInfinityAlgebrasInEquivariantdgcAlgebras},
are relative minimal (Def. \ref{MinimalEquivariantdgcAlgebras})
\begin{equation}
  \label{EquivariantLInfinityAlgebraicCoefficientBundle}
  \raisebox{18pt}{
  \xymatrix@R=12pt@C=4em{
    \underline{\mathfrak{g}}
    \ar[rr]^-{
      \mathrm{fib}
      (\, \underline{\mathfrak{p}} \,)
    }
    &
    \ar@{}[d]|-{
    \mathclap{
      \mbox{
        \tiny
        \color{darkblue}
        \bf
        \def\arraystretch{1}
        \begin{tabular}{c}
          equivariant
          $L_\infty$-algebraic
          \\
          local coefficient bundle
        \end{tabular}
      }
    }
    }
    &
    \underline{\widehat {\mathfrak{b}}}
    \ar@{->>}[d]^-{
      \underline{\mathfrak{p}}
    }
    \\
    &&
    \underline{\mathfrak{b}}
  }
  }
  \phantom{AAAAA}
  \in
  \;
  \EquivariantLInfinityAlgebras
  \,.
\end{equation}

\noindent
Then, for $G \acts \; X \,\in\, \GActionsOnSmoothManifolds$
(Def. \ref{ProperGActionsOnSmoothManifolds})
equipped with an \emph{equivariant non-abelian de Rham twist}
\begin{equation}
  \label{TwistForFlatTwistedEquivariantNonabelianForms}
  \tau_{\mathrm{dR}}
  \;\in\;
  \Omega_{\mathrm{dR}}
  \big(\!
    \orbisingular
    \big(
      X \!\sslash\! G
    \big)
    ;
    \,
    \underline{\mathfrak{b}}
  \big)_{\mathrm{flat}}
\end{equation}

\noindent
given by a flat equivariant $\underline{\mathfrak{b}}$-valued differential form
(Def. \ref{FlatEquivariantLInfinityAlgebraValuedForms}) on $X$,
the set of
{\it flat $\tau_{\mathrm{dR}}$-twisted equivariant
$\underline{\mathfrak{g}}$-valued}
differential forms on $X$ is the hom-set \eqref{HomSets} in the
co-slice category of
$\EquivariantdgcAlgebras$ (Def. \ref{EquivariantdgcAlgebras})
under $\mathrm{CE}(\underline{\mathfrak{g}})$ from
$\mathrm{CE}(\underline{\mathfrak{p}})$ to $\tau_{\mathrm{dR}}$:
\vspace{-1mm} 
\begin{equation}
  \label{SetOfFlatTwistedEquivariantLIninfityAlgebraValuedDifferentialForms}
  \hspace{-3mm}
  \adjustbox{scale=0.9}{$
  \begin{aligned}
  \Omega^{\tau_{\mathrm{dR}}}_{\mathrm{dR}}
  \big(\!
    \orbisingular
    \big(
      X \!\sslash\! G
    \big)
    ,
    \underline{\mathfrak{g}}
  \, \big)_{\mathrm{flat}}
  &
    :=
  \big(
    \EquivariantdgcAlgebras
  \big)^{\mathrm{CE}(\underline{\mathfrak{b}})/}
  \big(
    \mathrm{CE}
    \big(
      \underline{\mathfrak{p}}
    \,\big)
    ,
    \tau_{\mathrm{dR}}
  \big)
  \\[-2pt]
  &
  \; =
  \;\;
  \left\{\!\!\!\!
  \raisebox{14pt}{
  \xymatrix@C=28pt@R=.3em{
    \Omega^\bullet_{\mathrm{dR}}
    \big(
      \orbisingular
      \big(
        X \!\sslash\! G
      \big) \!
    \big)
    \ar@{<-}[dr]_-{
      \mathllap{
        \mbox{
          \tiny
          \color{greenii}
          \bf
          twist
        }
      }
      \tau_{\mathrm{dR}}
    }
    \ar@{<--}[rr]^-{
      \mbox{
        \tiny
        \color{greenii}
        \bf
        \def\arraystretch{1}
        \begin{tabular}{c}
      flat twisted equivariant
          \\
          $\mathfrak{g}$-valued differential form
        \end{tabular}
      }
    }
    &&
    \mathrm{CE}
    \big(\,
      \underline{\widehat{\mathfrak{b}}}
    \,\big)
    \ar@{<-^{)}}[dl]^-{
       \underset{
        \mbox{
          \tiny
          \color{greenii}
          \bf
          \hspace{-8pt}
          \def\arraystretch{1}
          \begin{tabular}{c}
            local 
            \\
            coefficients
          \end{tabular}
        }       
       }{
       \mathrm{CE}(\underline{\mathfrak{p}})
       }
    }
    \\
    &
   \mathrm{CE}
    \big(
      \underline{\mathfrak{b}}
    \big)^{\phantom{A}}
  }
  }
\!\!\! \right\}.
  \end{aligned}
  $}
\end{equation}

\end{defn}

\smallskip

\noindent {\bf Equivariant non-abelian de Rham cohomology.}

\vspace{-1mm} 
\begin{notation}[Cylinder orbifold]
  \label{CylinderOrbifold}
For $G \acts \; X \;\in\; \GActionsOnSmoothManifolds$
(Def. \ref{ProperGActionsOnSmoothManifolds}),
let the product manifold $X \times \mathbb{R}$
be equipped with the $G$-action given by

  \vspace{-6mm}
   \begin{align*}
      G \times ( X \times \mathbb{R})
&\longrightarrow
      X \times \mathbb{R}
      \\[-3pt]
      (
        g, (x, \;\;t)
      )
& \longmapsto
      (
        g \cdot x
        \,,\,
        t
      )\;.
\end{align*}
  \vspace{-.3cm}

\noindent We say that the resulting $G$-orbifold (Def. \ref{GOrbifolds})
$
  \orbisingular
  \big(
    (X \times \mathbb{R})
    \sslash G
  \big)
   \in
   \GOrbifolds
$
is the {\it cylinder orbifold} of
$\orbisingular \big( X \!\sslash\! G\big)$, and we write

  \vspace{-.3cm}
  \begin{equation}
    \label{CanonicalInclusionIntoCylinderOrbifold}
    \xymatrix@C=7pt@R=-3pt{
      \orbisingular
      \big(
        (
          X \times \{0\}
        )
        \!\sslash\!
        G
      \big)
      \ar@{^{(}->}[rr]^-{ i_0 }
      &&
      \orbisingular
      \big(
        (
          X \times \mathbb{R}
        )
        \!\sslash\!
        G
      \big)
      \ar@{<-^{)}}[rr]^-{ i_1 }
      &&
      \orbisingular
      \big(
        (
          X \times \{1\}
        )
        \!\sslash\!
        G
      \big)
      \\
      \rotatebox[origin=c]{-90}{$\simeq$}
      &&
      &&
      \rotatebox[origin=c]{-90}{$\simeq$}
      \\
      \orbisingular
      \big(
        X
        \!\sslash\!
        G
      \big)
      &&
      &&
      \orbisingular
      \big(
        X
        \!\sslash\!
        G
      \big)
    }
  \end{equation}

\noindent
for the canonical inclusion maps and
\vspace{-2mm} 
\begin{equation}
  \label{CylinderOrbifoldProjectionMap}
  \xymatrix{
    \orbisingular
    \big(
      (X \times \mathbb{R})
      \!\sslash\!
      G
    \big)
    \ar[r]^-{ p_X }
    &
    \orbisingular
    \big(
      X
      \!\sslash\!
      G
    \big)
  }
\end{equation}

\vspace{-1mm} 
\noindent
for the canonical projection map.

\end{notation}

\noindent In equivariant generalization of \cite[Def. 3.83]{FSS23-Char}, we set:

\begin{defn}[Coboundaries between flat equivariant $L_\infty$-algebra valued differential forms]
  \label{CoboundariesBetweenFlatEquivariantLInfinityAlgebraValuedForms}
 Let
  $
    \underline{\mathfrak{g}}
    \,\in\, \EquivariantLInfinityAlgebras$
  (Def. \ref{DefiningEmbeddingOfEquivariantLInfinityAlgebrasInEquivariantdgcAlgebras}) and
  $G \acts \; X \,\in\, \GActionsOnSmoothManifolds$ (Def. \ref{ProperGActionsOnSmoothManifolds}).

 \noindent {\bf (i)}  Then, given flat differential forms
  $
    A_0, A_1
    \;\in\;
    \Omega_{\mathrm{dR}}
    \big(\!
      \orbisingular
      (
        X \!\sslash\! G
      )
      ;
      \,
      \underline{\mathfrak{g}}
    \, \big)_{\mathrm{flat}}
  $
  (Def. \ref{FlatEquivariantLInfinityAlgebraValuedForms}),
  a {\it coboundary} between them
  $
    \xymatrix{
      A_0
      \ar@{=>}[r]^{ \widetilde A }
      &
      A_1
    }
  $
  is a flat equivariant $\underline{\mathfrak{g}}$-valued differential form
  (Def. \ref{FlatEquivariantLInfinityAlgebraValuedForms})
  on the cylinder orbifold (Notation \ref{CylinderOrbifold})
  \begin{equation}
    \label{FlatFormOnCylinderOrbifold}
    \widetilde A
    \;\;
    \in
    \;
    \Omega_{\mathrm{dR}}
    \Big( \!
      \overset{
        \mathclap{
        \raisebox{3pt}{
          \tiny
          \color{darkblue}
          \bf
          cylinder orbifold
        }
        }
      }{
      \orbisingular
      \big(
        (
          X \times \mathbb{R}
        )
        \!\sslash\! G
      \big)
      }
      ;
      \,
      \underline{\mathfrak{g}}
    \Big)_{\mathrm{flat}}
  \end{equation}

  \noindent
  such that this restricts to the given pair of forms
  \begin{equation}
    \label{RestrictionOfCoboundary}
    i_0^\ast
    \big(
      \,\widetilde A \,
    \big)
    =
    A_0
    \phantom{AA}
    \mbox{and}
    \phantom{AA}
    i_1^\ast
    \big(\,
      \widetilde A \,
    \big)
    =
    A_1
  \end{equation}

  \noindent
  along the canonical inclusions \eqref{CanonicalInclusionIntoCylinderOrbifold}.

  \noindent
 {\bf (ii)} We denote the relation given by existence of a coboundary
  by
  $\;A_1 \;\sim\; A_2$.
\end{defn}

\begin{lemma}[Equivalence of equivariant smooth and PL de Rham complex of smooth orbifold]
  \label{EquivalenceOfEquivariantSmoothAndPLdeRhamComplexOfSmoothOrbifold}
  Let
  $ G \acts \; X \,\in\, \GActionsOnSmoothManifolds$
  (Def. \ref{ProperGActionsOnSmoothManifolds}).
  Then the
  corresponding equivariant PL de Rham complex
  (Def. \ref{EquivariantPLdeRhamComplexFunctor})
  is isomorphic to the
  equivariant smooth de Rham complex (Example \ref{EquivariantSmoothdeRhamComplex})
  in the homotopy category of equivariant dgc-algebras
  (Prop. \ref{ProjectiveModelStructureOnConnectiveEquivariantdgcAlgebras}):
  \begin{equation}
    \label{IsomorphismBetweenEquivariantPLAndSmoothdeRhamComplexes}
    \hspace{-3mm} 
    \Omega^\bullet_{\mathrm{dR}}
    \big(\!
      \orbisingular
      (
        X \!\sslash\! G
      )
    \big)
    \simeq
    \Omega^\bullet_{\mathrm{PLdR}}
    \big(
      \orbisingular
      (
        X \!\sslash\! G
      )
    \big)
    \;
    \in
    \;
    \mathrm{Ho}
    \Big(\!\!
      \EquivariantdgcAlgebrasProj
    \Big).
  \end{equation}

\end{lemma}
\begin{proof}
  Observe that the analogous non-equivariant statement
  holds by
  \cite[Lem. 3.90]{FSS23-Char}, using \cite[Cor. 9.9]{GriffithMorgan13},
  and that its proof proceeds by analyzing natural constructions
  applied
  to a choice of smooth triangulation of the given smooth manifold $X$.

  Now, for a smooth manifold equipped with a smooth
  $G$-action $G \acts \; X$, we may
  choose a $G$-equivariant smooth triangulation,
  by the equivariant triangulation theorem
  \cite{Illman78}\cite{Illman83}.
  Given this, the remainder of the non-equivariant proof
  applies stage-wise over the orbit category. Since
  the weak equivalences of equivariant dgc-algebras
  are the stage-wise weak equivalences of non-equivariant
  dgc-algebras (Prop. \ref{ProjectiveModelStructureOnConnectiveEquivariantdgcAlgebras}),
  the claim follows.
\end{proof}

In an equivariant generalization of \cite[Def. 3.84]{FSS23-Char}, we set:

\begin{defn}[Equivariant non-abelian de Rham cohomology]
  \label{EquivariantNonabelianDeRhamCohomology}
  Let
  $ G \acts \; X \,\in\, \GActionsOnSmoothManifolds$
  (Def. \ref{ProperGActionsOnSmoothManifolds})
  and
  $
    \underline{\mathfrak{g}}
    \;\in\;
    \EquivariantLInfinityAlgebras
  $
  (Def. \ref{EquivariantLInfinityAlgebras}).
  The {\it equivariant non-abelian de Rham cohomology}
  of $G \acts \; X$ with coefficients in
  $ \underline{\mathfrak{g}} $
  is the quotient of the set of flat equivariant
  differential forms (Def. \ref{FlatEquivariantLInfinityAlgebraValuedForms})
  by the coboundary relation (Def. \ref{CoboundariesBetweenFlatEquivariantLInfinityAlgebraValuedForms}):
  $$
    H_{\mathrm{dR}}
    \big(\!
      \orbisingular
      \big(
        X \!\sslash\! G
      \big)
      ;
      \,
      \mathfrak{\mathfrak{g}}
     \big)
    \;\;
      :=
    \;\;
    \Big(
    \Omega_{\mathrm{dR}}
    \big(
      \orbisingular
      (
        X \!\sslash\! G
      )
      ;
      \,
      \mathfrak{\mathfrak{g}}
    \big)_{\mathrm{flat}}
    \Big)_{\!\!\big/_{\sim}}.
  $$
\end{defn}

In equivariant generalization of \cite[Thm. 3.87]{FSS23-Char}, we have:

\begin{prop}[Equivariant non-abelian de Rham theorem]
  \label{EquivariantNonabelianDeRhamTheorem}
  Let $\mathscr{A} \,\in\, \EquivariantSimplyConnectedRFiniteHomotopyTypes$
  (Def. \ref{SubcategoryOfEquivariantSimplyConnectedRFiniteHomotopyTypes})
  and $ G \acts \; X \,\in\, \GActionsOnSmoothManifolds$
  (Def. \ref{ProperGActionsOnSmoothManifolds}),
  such that its equivariant shape (Def. \ref{EquivariantShape})
  is also equivariantly simply-connected and of $\mathbb{R}$-finite type:
  $
    \raisebox{1pt}{\rm\textesh}
    \,\orbisingular\,
    \big(
      X \!\sslash\! G
    \big)
    \;\in\;
    \EquivariantSimplyConnectedRFiniteHomotopyTypes
    \,.
  $
    Then, at least if $G$ has order 4 or is cyclic of prime order
  (Remark \ref{EquivariantSmoothDifferentialFormsAreInjective}),
  there is an equivalence between:

  \noindent
  {\bf (a)}  real equivariant non-abelian cohomology
  (Def. \ref{EquivariantNonAbelianCohomology})
  with coefficients in the equivariant rationalization
  $L_{\mathbb{R}}\mathscr{A}$
  (Def. \ref{EquivariantRationalization})
  and

  \noindent {\bf (b)}
  equivariant non-abelian de Rham cohomology
  (Def. \ref{EquivariantNonabelianDeRhamCohomology})
  of the $G$-orbifold $\orbisingular \big( X \!\sslash\! G\big)$
  (Def. \ref{GOrbifolds})
  with coefficients in the equivariant
  Whitehead $L_\infty$-algebra
  $\mathfrak{l} \mathscr{A}$ (Def. \ref{EquivariantWhiteheadLInfinityAlgebra}):
  \begin{equation}
    \label{deRhamTheoremEquivalence}
    \xymatrix{
      \overset{
        \mathclap{
        \raisebox{3pt}{
          \tiny
          \color{darkblue}
          \bf
          \def\arraystretch{1}
          \begin{tabular}{c}
            equivariant non-abelian
            \\
            real cohomology
          \end{tabular}
        }
        }
      }{
      H
      \Big(
        \raisebox{1pt}{\rm\textesh}
        \,
        \orbisingular
        \big(
          X \!\sslash\! G
        \big)
        \,;\,
        L_{\mathbb{R}}
        \mathscr{A}
      \Big)
      }
      \;\;
        \simeq
      \;\;
      \overset{
        \mathclap{
        \raisebox{3pt}{
          \tiny
          \color{darkblue}
          \bf
          \def\arraystretch{1}
          \begin{tabular}{c}
            equivariant non-abelian
            \\
            de Rham cohomology
          \end{tabular}
        }
        }
      }{
      H_{\mathrm{dR}}
      \Big(
        \!
        \orbisingular
        \big(
          X \!\sslash\! G
        \big)
        \,;\,
        \mathfrak{l}\mathscr{A}
      \Big)
      }
      \,.
    }
  \end{equation}

\end{prop}
\begin{proof}
  Consider the following sequence of bijections:
  $$
    \begin{array}{l}
      H
      \big(
        \!
        \orbisingular
        (
          X \!\sslash\! G
        )
        \,;\,
        L_{\mathbb{R}}
        \mathscr{A}
      \big)
      \\
      :=
      \;
      \EquivariantHomotopyTypes
      \big(
        \!
        \orbisingular
        (
          X \!\sslash\! G
        )
        \,,\,
        L_{\mathbb{R}}
        \mathscr{A}
      \big)
      \\
      \simeq
      \,
      \mathrm{Ho}
      \Big(
        \EquivariantdgcAlgebrasProj
      \Big)
      \Big(
        \Omega^\bullet_{\mathrm{PLdR}}
        (
          \mathscr{A}
        )
        \,,\,
        \Omega^\bullet_{\mathrm{PLdR}}
        \big(
          \orbisingular
          (
            X \!\sslash\! G
          )
        \big)
      \!\Big)
      \\
      \simeq
      \mathrm{Ho}
      \Big(
        \EquivariantdgcAlgebrasProj
      \Big)
      \Big(
        \mathrm{CE}
        \big(
          \mathfrak{l}\mathscr{A}
        \big)
        \,,\,
        \Omega^\bullet_{\mathrm{dR}}
        \big(
          \orbisingular
          (
            X \!\sslash\! G
          )
       \big)
      \!\Big)
      \\
      \simeq
      \EquivariantdgcAlgebrasProj
      \Big(
        \mathrm{CE}
        \big(
          \mathfrak{l}\mathscr{A}
        \big)
        \,,\,
        \Omega^\bullet_{\mathrm{dR}}
        \big(
          \orbisingular
          (
            X \!\sslash\! G
          )
        \big)
      \!\Big)_{\!\!\!\big/\sim_{\scalebox{.5}{right homotopy}}}
      \\
      \simeq
      \Big(
      \Omega_{\mathrm{dR}}
      \big(
        \!
        \orbisingular
        \big(
          X \!\sslash\! G
        \big)
        \,;\,
        \mathfrak{l}\mathscr{A}
      \big)_{\mathrm{flat}}
      \Big)_{\!\!\!\big/\sim}
      \\
      =:
      H_{\mathrm{dR}}
      \big(
        \!
        \orbisingular
        (
          X \!\sslash\! G
        )
        \,;\,
        \mathfrak{l}\mathscr{A}
      \big)
      \,.
    \end{array}
  $$

  \noindent
The first step is Def. \ref{EquivariantNonAbelianCohomology},
  while the second step is the fundamental theorem
  (Prop. \ref{FundamentalTheoremOfdgcAlgebraicEquivariantRationalHomotopyTheory}).
  In the third step we are:

  \noindent
  {\bf (a)} post-composing
  in the homotopy category
  with the isomorphism
  $
    \Omega^\bullet_{\mathrm{PLdR}}(-)$
    
    $
    \simeq
    \Omega^\bullet_{\mathrm{dR}}(-)
  $
  \eqref{IsomorphismBetweenEquivariantPLAndSmoothdeRhamComplexes};

  \noindent
  {\bf (b)} pre-composing with the isomorphism
  $
    \mathrm{CE}
    \big(
      \mathfrak{l} \mathscr{A}
    \big)
    \;\simeq\;
    \Omega^\bullet_{\mathrm{PLdR}}\big( \mathscr{A} \big)
  $
  exhibiting the minimal model \eqref{WhiteheadLInfinityAlgebraMinimalModelQuasiIso}.

 Now the domain object
  $\mathrm{CE}( \mathfrak{l}\mathscr{A})$ is cofibrant
  (by Lemma \ref{MinimalEquivariantdgcAlgebrasAreCofibrant})
  and the codomain object
  $
    \Omega^\bullet_{\mathrm{dR}}
    \big(\!
      \orbisingular
      (
        X \!\sslash\! G
      )
    \big)
  $
  is fibrant
  (by Prop. \ref{EquivariantSmoothDeRhamComplexIsProjectivelyFibrant}).
Consequently,  the hom-set in the homotopy category is equivalently given
  (\cite[\S I.1 Cor. 7]{Quillen67}, see \cite[Prop. A.16]{FSS23-Char})
  by
  right-homotopy classes of
  equivariant dgc-algebra homomorphisms between these
  objects, shown in the fourth step.

  To exhibit these right homotopies, we may choose as path-space object
  (\cite[Def. I.4]{Quillen67}, see \cite[A.11]{FSS23-Char})
  the equivariant de Rham complex on the cylinder orbifold
  (Notation \ref{CylinderOrbifold}):
  this qualifies as a path space object
  by stage-wise application of \cite[Lem. 3.88]{FSS23-Char}
  and using again the argument of
  Lemmas
  \ref{ZpEquivariantSmoothDifferentialFormsFormInjectiveDualVectorSpace},
  \ref{Z4EquivariantSmoothDifferentialFormsAreInjective},
  \ref{Z2TimesZ2EquivariantDifferentialFormsAreInjective}
  for equivariant fibrancy.
  But with this choice of path space object, the
  right homotopy relation manifestly coincides
  (by stage-wise application of \cite[Lem. 3.89]{FSS23-Char})
  with the
  coboundary relation
  on equivariant non-abelian forms (Def. \ref{CoboundariesBetweenFlatEquivariantLInfinityAlgebraValuedForms}).
  which is the fifth step above. With this, the last step is
  Def. \ref{EquivariantNonabelianDeRhamCohomology}.

  In conclusion, the composite of this chain of bijections
  gives the claimed bijection \eqref{deRhamTheoremEquivalence}.
\end{proof}

\smallskip

\noindent {\bf Twisted equivariant non-abelian de Rham cohomology.}

\noindent In equivariant generalization of \cite[Def. 3.97]{FSS23-Char}, we set:

\begin{defn}[Coboundaries between flat twisted equivariant $L_\infty$-algebra valued differential forms]
  \label{CoboundariesBetweenFlatTwistedEquivariantNonAbelianDifferentialForms}
  Given an equivariant $L_\infty$-algebraic
  local coefficient bundle \eqref{EquivariantLInfinityAlgebraicCoefficientBundle}
  \vspace{-1mm} 
\begin{equation}
  \raisebox{25pt}{
  \xymatrix@R=18pt@C=4em{
    \underline{\mathfrak{g}}
    \ar[rr]^-{
      \mathrm{fib}
      (
        \underline{\mathfrak{p}}
      )
    }
    &
    \ar@{}[d]|-{
    \mathclap{
      \mbox{
        \tiny
        \color{darkblue}
        \bf
        \def\arraystretch{1}
        \begin{tabular}{c}
          equivariant
          $L_\infty$-algebraic
          \\
          local coefficient bundle
        \end{tabular}
      }
    }
    }
    &
    \underline{\widehat {\mathfrak{b}}}
    \ar@{->>}[d]^-{
      \underline{\mathfrak{p}}
    }
    \\
    &&
    \underline{\mathfrak{b}}
  }
  }
  \phantom{AAAAA}
  \in
  \;
  \EquivariantLInfinityAlgebras
  \,,
\end{equation}

\vspace{-1mm} 
\noindent
and given $G \acts \; X \,\in\, \GActionsOnSmoothManifolds$
(Def. \ref{ProperGActionsOnSmoothManifolds})
equipped with an equivariant non-abelian de Rham twist
\eqref{TwistForFlatTwistedEquivariantNonabelianForms}
$$
  \tau_{\mathrm{dR}}
  \;\in\;
  \Omega_{\mathrm{dR}}
  \Big(\!
    \orbisingular
    \big(
      X \!\sslash\! G
    \big)
    ;\,
    \underline{\mathfrak{b}}
  \Big)
  \,,
$$

\noindent {\bf (i)} we say that a {\it coboundary} between
a pair
$$
  A_0,\, A_1
  \;\in\;
  \Omega^{\tau_{\mathrm{dR}}}_{\mathrm{dR}}
  \Big(\!
    \orbisingular
    \big(
      X \!\sslash\! G
    \big)
    ;\,
    \underline{\mathfrak{g}}
  \Big)
$$

\noindent
of flat equivariant $\tau_{\mathrm{dR}}$-twisted
$\underline{\mathfrak{g}}$-valued differential forms
(Def. \ref{FlatEquivariantLInfinityAlgebraValuedForms})
is such a form on the cylinder orbifold (Notation \ref{CylinderOrbifold})
\vspace{-1mm} 
$$
  \widetilde
  A
  \;\in\;
  \Omega^{ p_X^\ast(\tau_{\mathrm{dR}}}_{\mathrm{dR}) }
  \Big(\!
    \overset{
      \mathclap{
      \raisebox{3pt}{
        \tiny
        \color{darkblue}
        \bf
        \begin{tabular}{c}
          cylinder orbifold
        \end{tabular}
      }
      }
    }{
    \orbisingular
    \big(
      (X \times \mathbb{R}) \!\sslash\! G
    \big)
    }
    ;\,
    \underline{\mathfrak{g}}
  \Big)
$$

\noindent
twisted by the pullback of the given twist to the cylinder orbifold
(along the canonical projection \eqref{CylinderOrbifoldProjectionMap}),
such that  this restricts to the given pair of forms
\begin{equation}
  \label{TwistedCoboundaryRestrictions}
  i_0^\ast
  \big(\,
    \widetilde A \,
  \big)
  \;=\;
  A_0
  \phantom{AA}
  \mbox{and}
  \phantom{AA}
  i_1^\ast
  \big(\,
    \widetilde A \,
  \big)
  \;=\;
  A_1
\end{equation}

\noindent
along the canonical inclusions \eqref{CanonicalInclusionIntoCylinderOrbifold}.

\noindent {\bf (ii)} We denote the relation that there exists such a coboundary
by $\; A_0 \sim A_1$.
\end{defn}

In equivariant generalization of \cite[Def. 3.98]{FSS23-Char}, we set:

\begin{defn}[Twisted equivariant non-abelian de Rham cohomology]
 \label{TwistedEquivariantNonAbelianDeRhamCohomology}
 Let $G \acts \; X \in \GActionsOnSmoothManifolds$
 (Def. \ref{ProperGActionsOnSmoothManifolds})
 and let
 $
   \underline{\mathfrak{g}}
   \to
   \underline{\widehat{\mathfrak{b}}}
   \to
   \underline{\mathfrak{b}}
 $
 be an equivariant $L_\infty$-algebraic local coefficient bundle
 \eqref{EquivariantLInfinityAlgebraicCoefficientBundle},
 and let
 \begin{equation}
   \label{ClassOfEquivariantNonAbelianTwist}
   \big[
     \tau_{\mathrm{dR}}
   \big]
   \;\in\;
   H_{\mathrm{dR}}
   \big(
     \orbisingular
     (
       X \!\sslash\! G
     )
     ;\,
     \underline{\mathfrak{b}}
   \big)_{\mathrm{flat}}
 \end{equation}

 \noindent
 be the equivariant non-abelian de Rham cohomology class
 (Def. \ref{EquivariantNonabelianDeRhamCohomology})
 of an equivariant twist \eqref{TwistForFlatTwistedEquivariantNonabelianForms}.
  Then we say that the
 {\it equivariant $\tau_{\mathrm{dR}}$-twisted
 de Rham cohomology} of the $G$-orbifold
 $\orbisingular \big( X \!\sslash\! G \big)$
 (Def. \ref{GOrbifolds}) with coefficients in $\underline{\mathfrak{g}}$
 is the quotient of the set of
 {\it equivariant $\tau_{\mathrm{dR}}$-twisted
 $\underline{\mathfrak{g}}$-valued
 differential forms} (Def. \ref{FlatTwistedEquivariantLInfinityAlgebraValuedForms})
 by the coboundary relation from Def. \ref{CoboundariesBetweenFlatTwistedEquivariantNonAbelianDifferentialForms}:
 $$
   H_{\mathrm{dR}}^{\tau_{\mathrm{dR}}}
   \Big(\!
     \orbisingular
     \big(
       X \!\sslash\! G
     \big);
     \,
     \underline{\mathfrak{g}}
   \Big)
   \;\;
     :=
   \;\;
   \Omega_{\mathrm{dR}}^{\tau_{\mathrm{dR}}}
   \Big(\!
     \orbisingular
     \big(
       X \!\sslash\! G
     \big);
     \,
     \underline{\mathfrak{g}}
   \Big)_{\!\!\!\big/\sim}
   \,.
 $$
\end{defn}

\begin{notation}[Equivariant local coefficient bundle with relative minimal model]
  \label{EquivariantLocalCoefficientBundleWithRelativeMinimalModel}
  Given
  an equivariant local coefficient bundle
  \eqref{EquivariantLocalCoefficientBundle}

  \vspace{-.2cm}
  \begin{equation}
    \label{EquivariantLocalCoefficientBundleInSimplyConnectedRFiniteHomotopyTypes}
    \raisebox{20pt}{
    \xymatrix@R=1.5em{
      \mathscr{A}
      \ar[rr]^-{
        \scalebox{.7}{$
        \mathrm{hofib}
        (
          \rho_{\scalebox{.7}{$\mathscr{A}$}}
        )
        $}
      }
      &
      \ar@{}[d]|-{
        \mbox{
          \tiny
          \color{darkblue}
          \bf
          \begin{tabular}{c}
            equivariant
            \\[-2pt]
            local coefficient
            \\[-2pt]
            bundle
          \end{tabular}
        }
      }
      &
      \mathscr{A}\!\sslash\!\!\mathscr{G}
      \ar[d]^-{
        \scalebox{.7}{$
        \rho_{\scalebox{.7}{$\mathscr{A}$}}
        $}
      }
      \\
      &&
      B \mathscr{G}
    }
    }
    \;\;\;\;\;\;\;
    \in
    \EquivariantSimplyConnectedRFiniteHomotopyTypes
  \end{equation}

\noindent
all of whose objects
are equivariantly 1-connected and of $\mathbb{R}$-finite type
(Def. \ref{SubcategoryOfEquivariantSimplyConnectedRFiniteHomotopyTypes}),
assume (Remark \ref{ExistenceOfEquivariantRelativeMinimalModels}) that
$\rho_{\scalebox{.6}{$A$}}$
admits an equivariant relative minimal model
(Def. \ref{MinimalEquivariantdgcAlgebras}). This is to be denoted as follows:

\begin{equation}
  \label{EquivariantRelativeMinimalModelOfLocalCoefficientBundle}
  \hspace{-5mm}
  \adjustbox{scale=0.89}{$
  \raisebox{35pt}{
  \xymatrix@R=3pt@C=7pt{
    & & &
    \mathrm{CE}
    \big(
      \mathfrak{l} \mathscr{A}
    \big)
    \ar@{<-}[rrrr]^-{
      \scalebox{.7}{$
        \mathrm{cofib}
        \big(
          \mathrm{CE}
          (
            \mathfrak{l} \rho_{\scalebox{.6}{$\mathscr{A}$}}
          )
        \big)
      $}
    }
    &&&&
    \mathrm{CE}
    \big(
      \mathfrak{l}_{\scalebox{.6}{$B \mathscr{G}$}}
      (\mathscr{A} \!\sslash\! \mathscr{G} )
    \big)
    \ar@{<-}[dddd]^-{\hspace{-1.2cm}
      \scalebox{.7}{$
        \mathrm{CE}
        \big(
          \mathfrak{l} \rho_{\scalebox{.6}{$\mathscr{A}$}}
        \big)
      $}
      {
        \mbox{
          \tiny
          \color{greenii}
          \bf
          \def\arraystretch{1}
          \begin{tabular}{c}
            equivariant relative
            \\
            minimal model
          \end{tabular}
        }
      }
    }
    \\
    \Omega^\bullet_{\mathrm{PLdR}}
    \big(
      \mathscr{A}
    \big)
    \ar@{<-}[rrrr]|-{\;
      \scalebox{.7}{$
      \Omega^\bullet_{\mathrm{PLdR}}
      \big(
        \mathrm{hofib}(\rho_{\scalebox{.6}{$\mathscr{A}$}})
      \big)
      $}
   \; }
    &&&&
    \Omega^\bullet_{\mathrm{PLdR}}
    \big(
      \mathscr{A}
      \!\sslash\!
      \mathscr{G}
    \big)
    \ar@{<-}[dddd]_-{
      \mathllap{
        \mbox{
          \tiny
          \color{greenii}
          \bf
          \def\arraystretch{1}
          \begin{tabular}{c}
            equivariant dgc-algebra model
            \\
            of local coefficient bundle
          \end{tabular}
        }
      }
      \Omega^\bullet_{\mathrm{PLdR}}
      (
        \rho_{\scalebox{.6}{$\mathscr{A}$}}
      )
    }
    \ar@{<-}[urrr]_-{
       \scalebox{.7}{$
         p^{
           \mathrm{min}_{\scalebox{.4}{$B \mathscr{G}$}}
         }_{
           \scalebox{.7}{$
             \mathscr{A} \!\sslash\! \mathscr{G}
           $}
         }
       $}
      \, \in \, \mathrm{W}
    }
    \\
    \\
    \\
    &&&& &&&
    \mathrm{CE}
    \big(
      \mathfrak{l}
      B \mathscr{G}
    \big)
    \mathrlap{\hspace{-4mm}
      \mbox{
        \tiny
        \color{darkblue}
        \bf
        \def\arraystretch{1}
        \begin{tabular}{c}
          equivariant
          \\
          minimal model
        \end{tabular}
      }
    }
    \\
    &&&&
    \Omega^\bullet_{\mathrm{dR}}
    \big(
      B \mathscr{G}
    \big)
    \ar@{<-}[urrr]_-{
       \scalebox{.7}{$
         p^{
           \mathrm{min}
         }_{
           \scalebox{.7}{$
             B \mathscr{G}
           $}
         }
       $}
       \, \in \, \mathrm{W}
    }
  }
  }
  $}
\end{equation}

\noindent Notice that the corresponding fibration of equivariant
$L_\infty$-algebras (Def. \ref{EquivariantLInfinityAlgebras})
serves as an equivariant $L_\infty$-algebraic
local coefficient bundle \eqref{EquivariantLInfinityAlgebraicCoefficientBundle}.
\end{notation}

In equivariant generalization of \cite[Thm. 3.104]{FSS23-Char},
we have:

\vspace{-1mm} 
\begin{prop}[Twisted equivariant non-abelian de Rham theorem]
  \label{EquivariantTwistedNonabelianDeRhamTheorem}
Consider the following
 \begin{itemize}
 
\item Let $\rho_{\scalebox{.7}{$\mathscr{A}$}}$
be an equivariant local coefficient bundle
of equivariantly 1-connected $G$-spaces of
finite $\mathbb{R}$-homotopy type, which admits
an equivariant relative minimal model; all as in
Notation \ref{EquivariantLocalCoefficientBundleWithRelativeMinimalModel}.

\item Moreover, let $G \acts \; X \,\in\, \GActionsOnSmoothManifolds$
(Def. \ref{ProperGActionsOnSmoothManifolds}) be such that also
its equivariant shape (Def. \ref{EquivariantShape})
is equivariantly 1-connected and of $\mathbb{R}$-finite type,
$
  \raisebox{1pt}{\rm\textesh}\,\orbisingular\,
  \big(
    X \!\sslash\! G
  \big)
  \in
  \EquivariantSimplyConnectedRFiniteHomotopyTypes
$
and let this be eq- uipped with
an equivariant twist $\tau$ \eqref{ATwist}
with coefficients in the equivariant rationalization
(Def. \ref{EquivariantRationalization}) of $B \mathscr{G}$.

\item Write $\tau_{\mathrm{dR}}$ for a representative of
the image
under the equivariant non-abelian de Rham theorem
(Prop. \ref{EquivariantNonabelianDeRhamTheorem})
of the class of this twist in
equivariant $\mathfrak{l} B \mathscr{A}$-valued
de Rham cohomology (Def. \ref{EquivariantNonabelianDeRhamCohomology})
that the equivariant local coefficient bundle
\eqref{EquivariantLocalCoefficientBundleInSimplyConnectedRFiniteHomotopyTypes}
admits an equivariant relative minimal model (Def. \ref{MinimalEquivariantdgcAlgebras})
\vspace{-1mm} 
\begin{equation}
  \label{TwistsForTwistedEquivariantNonAbelianDeRhamTheorem}
  \xymatrix@R=-4pt{
    H
    \Big(
      \raisebox{1pt}{\rm\textesh}
      \,
      \orbisingular
      \,
      \big(
        X \!\sslash\! G
      \big)
      ;\,
      L_{\mathbb{R}} B \mathscr{G}
    \Big)
    &
    \simeq
    &
    H_{\mathrm{dR}}
    \Big(\!
      \orbisingular
      \big(
        X \!\sslash\! G
      \big);
      \,
      \mathfrak{l} B \mathscr{G}
    \Big).
    \\
    \overset{
      \mathclap{
      \raisebox{3pt}{
        \tiny
        \color{darkblue}
        \bf
        rational twist
      }
      }
    }{
      [\tau]
    }
    &
    \overset{
        \mathclap{
        \raisebox{10pt}{
          \tiny
          \color{greenii}
          \rm
          \begin{tabular}{c}
            equivariant non-abelian
            \\[-2pt]
             de Rham theorem
          \end{tabular}
        }
        }
    }{
      \longmapsto
    }
    &
    \overset{
      \mathclap{
      \raisebox{3pt}{
        \tiny
        \color{darkblue}
        \bf
        de Rham twist
      }
      }
    }{
      [\tau_{\mathrm{dR}}]
    }
  }
\end{equation}
\end{itemize}

\vspace{-1mm} 
\noindent
Then there is an equivalence between:

\noindent
{\bf (a)} the $\tau$-twisted equivariant real non-abelian
cohomology (Def. \ref{EquivariantTwistedNonAbelianCohomology})
with local coefficients in $\rho_{\scalebox{.7}{$\mathscr{A}$}}$, and

\noindent
{\bf (b)} the $\tau_{\mathrm{dR}}$-twisted equivariant
de Rham cohomology
(Def. \ref{TwistedEquivariantNonAbelianDeRhamCohomology})
with local coefficients
in $\mathfrak{l}_{B  \scalebox{.65}{$\mathscr{G}$}} \rho_{\scalebox{.7}{$\mathscr{A}$}}$
\eqref{EquivariantRelativeMinimalModelOfLocalCoefficientBundle}:
\vspace{-1mm} 
\begin{equation}
  \label{TwistedEquivariantNonabeliandeRhamEquivalence}
  \overset{
    \mathclap{
    \raisebox{3pt}{
      \tiny
      \color{darkblue}
      \bf
      \def\arraystretch{1}
      \begin{tabular}{c}
        twisted equivariant
        \\
        non-abelian real cohomology
      \end{tabular}
    }
    }
  }{
  H^\tau
  \Big(
    \raisebox{1pt}{\rm\textesh}
    \,\orbisingular\,
    \big(
      X \!\sslash\! G
    \big)
    ;\,
    L_{\mathbb{R}}\mathscr{A}
  \Big)
  }
  \;\;
    \simeq
  \;\;
  \overset{
    \mathclap{
    \raisebox{3pt}{
      \tiny
      \color{darkblue}
      \bf
      \def\arraystretch{1}
      \begin{tabular}{c}
        twisted equivariant
        \\
        non-abelian de Rham cohomology
      \end{tabular}
    }
    }
  }{
  H^{\tau_{\mathrm{dR}}}
  \Big(\!
    \orbisingular\,
    \big(
      X \!\sslash\! G
    \big)
    ;\,
    \mathfrak{l}\mathscr{A}
  \Big)
  }
  \,.
\end{equation}
\end{prop}
\begin{proof}
  The proof proceeds in direct joint generalization
  of the proofs of Prop. \ref{EquivariantNonabelianDeRhamTheorem}
  (equivariant case) and \cite[Thm. 3.104]{FSS23-Char} (twisted case).

 \noindent First, by the fundamental theorem (Prop. \ref{FundamentalTheoremOfdgcAlgebraicEquivariantRationalHomotopyTheory}),
  the twisted real cohomology is given by morphisms in the
  homotopy category of the co-slice model category of this form:

  \vspace{-4mm}
  \begin{align}
    \label{CocycleInEquivariantTwistedRealCohomologyUnderPLdR}
    \raisebox{10pt}{
    \xymatrix@R=7pt@C=25pt{
      \Omega^\bullet_{\mathrm{PLdR}}
      \big(
        \raisebox{1pt}{\textesh}
        \,\orbisingular\,
        (
          X \!\sslash\! G
        )
      \big)
      \ar@{<-}[dr]_-{
        \Omega^\bullet_{\mathrm{PLdR}}(\tau)\;\;\;\;\;
      }
      \ar@{<--}[rr]
      &&
      \Omega^\bullet_{\mathrm{PLdR}}
      \big(
        \mathscr{A}
        \!\sslash\!
        \mathscr{G}
      \big)
      \ar@{<-}[dl]^-{
        \;\;\;\;\;\;
        \Omega^\bullet_{\mathrm{PLdR}}
        (
          \rho_{\scalebox{.7}{$\mathscr{A}$}}
        )
      }
      \\
      &
      \Omega^\bullet_{\mathrm{PLdR}}
      \big(
        B \mathscr{G}
      \big)
       }
    }
    \\
    \in
    \;
    \mathrm{Ho}
    \Big(\!
      \EquivariantdgcAlgebrasProj
    \Big).
    \nonumber 
  \end{align}

  \noindent
  Second, by

  {\bf (a)} post-composition with the isomorphism
  $
    \Omega^\bullet_{\mathrm{PLdR}}(-)
      \;\simeq\;
    \Omega^\bullet_{\mathrm{dR}}(-)
  $
  \eqref{IsomorphismBetweenEquivariantPLAndSmoothdeRhamComplexes},

  {\bf (b)} pre-composition with the equivalence from the
  equivariant relative minimal model
  \eqref{EquivariantRelativeMinimalModelOfLocalCoefficientBundle},

  \noindent
  this becomes equivalent to morphisms of this form:
  \begin{align}
    \label{MorphismInCoslideFromEquivariantMinimalModelTodeRham}
      \raisebox{10pt}{
    \xymatrix@R=5pt@C=30pt{
      \Omega^\bullet_{\mathrm{dR}}
      \big(
        \raisebox{1pt}{\textesh}
        \,\orbisingular\,
        (
          X \!\sslash\! G
        )
     \! \big)
      \ar@{<-}[dr]_-{
        \tau_{\mathrm{dR}} \;
      }
      \ar@{<--}[rr]
      &&
      \Omega^\bullet_{\mathrm{PLdR}}
      \big(
        \mathscr{A}
        \!\sslash\!
        \mathscr{G}
      \big)
      \ar@{<-}[dl]^-{
        \scalebox{.7}{\;\;\; $
        \mathrm{CE}
        \big(
          \mathfrak{l} \rho_{\scalebox{.7}{$\mathscr{A}$}}
        \big)
        $}
      }
      \\
      &
      \mathrm{CE}
      \big(
        \mathfrak{l}
        B \mathscr{G}
      \big)
    }
    }
    \\
    \in
    \;
    \mathrm{Ho}
    \Big(\!
      \EquivariantdgcAlgebrasProj
    \Big)
    \,.
    \nonumber
  \end{align}

  \noindent
  But, in this form,
  \begin{itemize}[
    leftmargin=.8cm
  ]
  \item[{\bf (a)}] the codomain $\tau_{\mathrm{dR}}$ is a fibrant
  object in the coslice model category, 
  
  since $\Omega^\bullet_{\mathrm{dR}}\big( - \big)$ is fibrant in
  the un-sliced model structure (Prop. \ref{EquivariantSmoothDeRhamComplexIsProjectivelyFibrant});

  \item[{\bf (b)}] the relative minimal model
  domain $\mathrm{CE}\big( \mathfrak{l} \rho_{\scalebox{.7}{$\mathscr{A}$}} \big)$
  is cofibrant, by Lemma \ref{MinimalEquivariantdgcAlgebrasAreCofibrant}.

  \end{itemize}

  \noindent
It follows
  (\cite[\S I.1 Cor. 7]{Quillen67}, see \cite[Prop. A.16]{FSS23-Char})
  that a morphism of the form
  \eqref{MorphismInCoslideFromEquivariantMinimalModelTodeRham} in the
  homotopy category is equivalently the right homotopy class
  of an actual homomorphism of equivariant dgc-algebras in the coslice,
  hence is equivalently
  the right homotopy class of
  a flat equivariant twisted
  $\mathfrak{l}\mathscr{A}$-valued differential form,
  by Def. \ref{FlatTwistedEquivariantLInfinityAlgebraValuedForms}.

  Finally, in joint generalization of the proof of
  Prop. \ref{EquivariantNonabelianDeRhamTheorem} (equivariant case)
  and \cite[Lem. 3.105]{FSS23-Char} (twisted case),
  we see that a path space object
  (\cite[Def. I.4]{Quillen67}, see \cite[A.11]{FSS23-Char})
  exhibiting these right
  homotopies in the coslice is given by pullback to the
  equivariant smooth de Rham complex of the cylinder orbifold
  \eqref{FlatFormOnCylinderOrbifold}. But with that choice,
  right homotopies are manifestly the same as
  coboundaries
  of flat equivariant twisted
  $\mathfrak{l}\mathscr{A}$-valued differential forms
  (Def. \ref{CoboundariesBetweenFlatTwistedEquivariantNonAbelianDifferentialForms}),
  and hence the claim follows.
\end{proof}

\smallskip

\noindent {\bf Twisted non-abelian Borel-Weil-Cartan equivariant de Rham cohomology.}
 Finally, we combine traditional
Borel(-Weil-Cartan) $T$-equivariant de Rham cohomology
(\cite{AtiyahBott84}\cite[\S 5]{MathaiQuillen86}\cite{Kalkman93}\cite{GuilleminSternberg99},
  review in \cite{Meinrenken06} \cite{KubelThom15}\cite{Pestun17}),
with proper $G$-equivariance and generalize it to
non-abelian $L_\infty$-algebra coefficients.

\medskip

By Prop. \ref{InfinityActionsEquivalentToFibrationsOverClassifyingSpace}
and Remark \ref{EquivariantInfinityAction}, any
Borel $T$-equivariantized $G$-orbifold carries a canonical
twist in equivariant non-abelian cohomology
$H^1(- , T) \,\simeq\, H(-, B T)$.
The following is the de Rham image of that twist:

\begin{defn}[Canonical de Rham twist on Borel $T$-equivariant $G$-orbifolds]
\label{ChernWeildeRhamTwistOnTEquivariantGOrbifolds}
Let $( T \!\times G) \acts \; X \in \TGActionsOnSmoothManifolds$
(Def \ref{ProperGActionsOnSmoothManifolds})
for $T \in \mathrm{CompactLieGroups}$ finite-dimensional
and simply-connected,
with Lie algebra $\mathfrak{t}$ \eqref{LieAlgebraOfCompactEquivarianceGroup},
regarded as a smooth $G$-equivariant $L_\infty$-algebra
(Def. \ref{EquivariantLInfinityAlgebras}).
 We say that the {\it canonical de Rham twist}
 on the corresponding $T$-parametrized $G$-orbifold
 is the canonical inclusion of equivariant dgc-algebras
 (Def. \ref{EquivariantdgcAlgebras}) from the
 minimal model for the classifying space of $T$
 (regarded as a smooth $G$-equivariant homotopy type, Example \ref{SmoothSingularHomotopyTypes})
 into the proper $G$-equivariant \& Borel $T$-equivariant
 smooth de Rham complex (Example \ref{ProperGEquivariantAndBorelTEquivariantSmoothDeRhamComplex}):
 $$
 \adjustbox{scale=0.92}{$
   \xymatrix@R=-1pt@C=-5pt{
     \Omega^\bullet_{\mathrm{dR}}
     \Big(\!
       \big(
         \orbisingular
         (
           X \!\sslash\! G
         )
       \big)
       \!\sslash\! T
     \Big)
     \ar@{<-}[dd]^-{ \tau^{\mathrm{can}}_{\mathrm{dR}} }
     &&
    \overset{
      \mathclap{
      \raisebox{-15pt}{
        \tiny
        \color{darkblue}
        \bf
        \begin{tabular}{c}
          Cartan model for
          $T$-equivariant
          \\[-2pt]
          Borel cohomology
          of $H$-fixed locus $X^H$
        \end{tabular}
      }
      }
    }{
    \bigg(\!
    \Omega^\bullet_{\mathrm{dR}}
    \big(
      X^H
    \big)
    \otimes
    \mathbb{R}
    \Big[
      \{r^{\, a}_2\}_{a = 1}^{\mathrm{dim}(T)}
    \Big]
    \,,\,
      d_{\mathrm{dR}}
        +
      r_2^{\, a} \wedge  \iota_{t_a}
   \!\! \bigg)^{\!\!T}
    }
    \ar@{<-^{)}}[dd]
     \\
     &
     \!\!\!\!\!
     :
     \;\;\;\;\;\;\;
     G/H
     \;\;\;
     \longmapsto
     \!\!
     &
     \\
     \mathrm{CE}
     \big(
       \mathfrak{l} B T
     \big)
     &&
     \left(
     \mathbb{R}
     \Big[
       \{
         r^{\, a}_2
       \}_{a = 1}^{\mathrm{dim}(T)}
     \Big]
     \right)^T
   }
   $}
 $$
where on the bottom we used the
abstract Chern-Weil isomorphism \eqref{AbstractChernWeilHomomorphism}
in the form discussed in \cite[\S 4.2]{FSS23-Char}.
\end{defn}

\begin{example}[Equivariant Cartan map]
  \label{EquvariantCartanMap}
  In the situation of Def. \ref{ChernWeildeRhamTwistOnTEquivariantGOrbifolds},
  consider the case when the $T$-action is free, hence that
  $X := P$ is the total space of
  a $G$-equivariant $T$-principal bundle $P \to B := P/T$
  (e.g. \cite[p .2]{KubelThom15}).
  Then, for any choice of $G$-invariant $N$-principal connection
  $
    \nabla
    \,\in\,
    N\mathrm{Connections}(P)^G
  $,
  we have the following weak equivalence
  (in the sense of Prop. \ref{ProjectiveModelStructureOnConnectiveEquivariantdgcAlgebras})
  of
  $G$-equivariant dgc-algebras (Def. \ref{EquivariantdgcAlgebras})
  in the co-slice under the minimal model dgc-algebra of the
  classifying space \eqref{AbstractChernWeilHomomorphism}:
  \vspace{-.4cm}
  $$
    \hspace{-3mm}
    \scalebox{0.79}{$
    \begin{array}{c}
    \xymatrix@C=-34pt@R=7pt{
     \Omega^\bullet_{\mathrm{dR}}
     \Big(\!\!
       \big(
         \orbisingular
         (
           X \!\sslash\! G
         )
       \big)
       \!\sslash\! T
     \!\Big)
     \ar[rr]^-{ \in \; \mathrm{W} }
     \ar@{<-}[ddr]_-{ \tau^{\mathrm{can}}_{\mathrm{dR}} }
     &
     \phantom{--------}
     &
     \Omega^\bullet_{\mathrm{dR}}
     \big(\!
         \orbisingular
         (
           B \!\sslash\! G
         )
    \! \big)
     \ar@{<-}[ddl]^-{ \mathrm{cw}_T }
     &
     {\phantom{AAAAAAAAAA}}
     &
    \Big(
    \Omega^\bullet_{\mathrm{dR}}
    \big(
      X^H
    \big)
    \Big[
      \{r^{\, a}_2\}_{a = 1}^{\mathrm{dim}(T)}
    \Big]
    \Big)^T
    \ar[rr]^-{
        \scalebox{0.8}{$
         \arraycolsep=1.4pt
         \def\arraystretch{1}
        {\begin{array}{lll}
          \omega & \mapsto & \omega_{\; \mathrm{hor}}
          \\
          \color{orangeii}
          r^{\, a}_2
          &
          \color{orangeii}
          \mapsto
          &
          \color{orangeii}
          F_\nabla^a
        \end{array}}
        $}
    }
    \ar@{<-^{)}}[ddr]
    &
    &
    \Omega^\bullet_{\mathrm{dR}}
    \big(
      B^H
    \big)
    \ar@{<-}[ddl]^<<<{
      \underset{
        \mathclap{
        \raisebox{4pt}{
          \tiny
          \color{greenii}
          \bf
          {\begin{tabular}{c}
            Chern-Weil hom.
          \end{tabular}}
        }
        }
      }{
        {
          \color{orangeii}
          c \, \mapsto c(F_\nabla)
        }
      }
    }
    \\
    & & \;\;\;
    \ar@{}[rr]|-{
      \qquad : \quad 
     \scalebox{1}{$G/H$} 
     \;\; \longmapsto \qquad \qquad \qquad 
    }
    & &
    \\
     &
     \mathrm{CE}
     \big(
       \mathfrak{l}
       B T
     \big)
    &&&&
    \Big(
    \mathbb{R}
    \Big[
      \{r^{\, a}_2\}_{a = 1}^{\mathrm{dim}(T)}
    \Big]
    \Big)^T
    \;\;\;\;\;\;
    }
    \end{array}
    $}
  $$

  \noindent
 This is from the proper $G$-equivariant Borel $T$-equivariant
  smooth de Rham complex
  of $X$ (Example \ref{ProperGEquivariantAndBorelTEquivariantSmoothDeRhamComplex})
  to the proper $G$-equivariant smooth de Rham complex
  over $X/T$ (Example \ref{EquivariantSmoothdeRhamComplex}),
  which is stage-wise over $G/H$
  the {\it Cartan map} quasi-isomorphism
  \cite[\S 5]{GuilleminSternberg99}
  (review in \cite[(20), (30)]{Meinrenken06})
  from the Cartan model of $X^H$
  \eqref{CartanModelOnFixedLocus}
  to the ordinary
  smooth de Rham complex of $B^H = (X/N)^H$.
  This sends the Cartan model generators $r^{\, a}_2$
  to the curvature form component $F^a_\nabla$ of
  the given connection,
  and hence
  restricts on universal real characteristic classes,
  represented by invariant polynomials $c$,
  to the
  Chern-Weil homomorphism
  assigning characteristic forms: $c \,\mapsto\, c(F_\nabla)$.
\end{example}

\begin{example}[Tangential de Rham twists on $G$-orbifolds with $T$-structure]
  \label{TangentialDeRhamTwistsOnGOrbifoldsWithTStructure}
In further specialization of Example \ref{EquvariantCartanMap},
let $X \acts \; B \,\in\, \GActionsOnSmoothManifolds$ (Def. \ref{ProperGActionsOnSmoothManifolds}) be equipped
with $G$-equivariant $T \subset \mathrm{GL}(\mathrm{dim}(X))$-structure
(see \cite[p. 9]{SS20b} for pointers),
namely with a
$G$-equivariant reduction of its $\mathrm{GL}(\mathrm{dim}(X))$-frame bundle
to a $T$-principal $T$-frame bundle $T\mathrm{Fr}(X)$:
\vspace{-1mm}
$$
\hspace{1cm}
  \xymatrix@R=-2pt@C=3.5em{
    \mathllap{
      \mbox{
        \tiny
        \color{darkblue}
        \bf
        $T$-frame bundle
      }
      \;
    }
    T\mathrm{Fr}(X)
    \ar[dr]
    \ar@(ul,ur)^-{ \;T \times G\; }
    \;
    \ar@{^{(}->}[rr]^-{
      \mbox{
        \tiny
        \color{greenii}
        \bf
        \begin{tabular}{c}
           $G$-equivariant
           \\[-2pt]
           $T$-structure
        \end{tabular}
      }
    }
    &&
    \;
    \mathrm{Fr}(X)
    \ar@(ul,ur)^-{ \; T \times G\; }
    \ar[dl]
    \mathrlap{
      \;
      \mbox{
        \tiny
        \color{darkblue}
        \bf
        frame bundle
      }
    }
    \\
    &
    X
    \ar@(dl,dr)|-{ \;G\; }
  }
$$

\noindent Then Example \ref{EquvariantCartanMap}
induces on the $G$-orbifold
$\orbisingular \big( X \!\sslash\! G \big)$
(Def. \ref{GOrbifolds})
an equivariant
non-abelian de Rham twist \eqref{ClassOfEquivariantNonAbelianTwist}
encoding all the real characteristic forms of the
given $G$-equivariant $T$-structure on $X$
(the {\it tangential twist}):
$$
  \xymatrix@R=4pt@C=45pt{
    \Omega^\bullet_{\mathrm{dR}}
    \Big(\!
      \big(
        \orbisingular
        (
          T\mathrm{Fr}(X)  \!\sslash\! G
        )
      \big)
      \!\sslash\! T
    \Big)
    \;\;
    \ar[rr]_-{ \in \; \mathrm{W} }^-{
      \mathclap{
      \mbox{
        \tiny
        \color{greenii}
        \bf
        \def\arraystretch{1}
        \begin{tabular}{c}
          Cartan map
          equivalence
        \end{tabular}
      }
      }
    }
    \ar@{<-}[dr]^>>>>>>>>>{ \tau^{\mathrm{can}}_{\mathrm{dR}} }_-{  \!\!\!\!\!\!\!\!\!\!
        \mbox{ \hspace{-5mm}
          \tiny
          \color{greenii}
          \bf
          \def\arraystretch{1}
          \begin{tabular}{c}
            canonical de Rham 
            \\
            twist 
            on
            orbifold's 
            \\
            $T$-frame bundle
          \end{tabular}
          \;\;
        }
     \;\;\;   
     }
    &&
    \Omega^\bullet_{\mathrm{dR}}
    \big(\!
      \orbisingular
      (
        X \!\sslash\! G
      )
    \big).
    \ar@{<-}[dl]_<<<<<<<<<{  \mathrm{cw}_T}^-{\hspace{4mm}
        \mbox{
          \tiny
          \color{greenii}
          \bf
          \def\arraystretch{1}
          \begin{tabular}{c}
            tangential 
            \\
            de Rham twist
            \\
            on $G$-orbifold
          \end{tabular}
        }
      }
    \\
    &
    \mathrm{CE}
    \big(
      \mathfrak{l}
      B T
    \big)
  }
$$

\end{example}

In further generalization of Def. \ref{TwistedEquivariantNonAbelianDeRhamCohomology}, we set:

\begin{defn}[Proper $G$-equivariant \& Borel $T$-equivariant
twisted non-abelian de Rham cohomology]
 \label{ProperGEquivariantAndBorelTEquivariantTwistedNonabelianDeRhamCohomology}

Let $(T \times G) \acts  \; X \,\in\, \TGActionsOnSmoothManifolds$
(Def. \ref{ProperGActionsOnSmoothManifolds})
for $T$ finite-dimensional, compact, and simply-connected, and let
\begin{equation}  \label{EquivariantLInfinityAlgebraicCoefficientBundleForProperAndBoreldeRham}
  \xymatrix@R=10pt{
   \underline{\mathfrak{g}}
   \ar[rr]^-{
     \mathrm{hofib}(\,\underline{\mathfrak{p}}\,)
   }
   &&
   \underline{\widehat{\mathfrak{b}}}
   \ar[d]^-{ \underline{\mathfrak{p}} }
   \\
   &&
   \mathfrak{l} B T
   }
 \end{equation}

 \noindent
 be an equivariant $L_\infty$-algebraic local coefficient bundle
 \eqref{EquivariantLInfinityAlgebraicCoefficientBundle}
 over the Whitehead $L_\infty$-algebra of $B T$
 (i.e., whose Chevalley-Eilenberg algebra is \eqref{AbstractChernWeilHomomorphism}).

 \noindent {\bf (i)}  We say that the set of
 {\it flat, canonically twisted, proper $G$-equivariant
 \&  Borel $T$-equivariant, $\underline{\mathfrak{g}}$-valued
 differential forms} on $X$ is
 the hom-set \eqref{HomSets} in the
 co-slice of
 $G$-equivariant dgc-algebras (Def. \ref{EquivariantdgcAlgebras})
 from $\mathrm{CE}\big( \underline{\mathfrak{p}}\big)$
 \eqref{DefiningEmbeddingOfEquivariantLInfinityAlgebrasInEquivariantdgcAlgebras}
 to the canonical de Rham twist (Def. \ref{ChernWeildeRhamTwistOnTEquivariantGOrbifolds})
 on the corresponding $T$-parametrized $G$-orbifold:
 \begin{equation}
   \label{SetOfCanonicallyTwistedForms}
   \hspace{-4mm} 
   \def\arraystretch{1.8}
   \begin{array}{l}
   \Omega^{\tau^{\mathrm{can}}_{\mathrm{dR}}}_{\mathrm{dR}}
   \Big(\!
     \big(
       \orbisingular
       (
         X \!\sslash\! G
       )
     \big)
     \!\sslash\!
     T
     ;
     \,
     \underline{\mathfrak{g}}
   \Big)
   \\
   \;
   :=
   \Big(
     \EquivariantdgcAlgebrasProj
   \Big)^{ \mathrm{CE}( \mathfrak{l} B T )/ }
   \Big(
     \mathrm{CE}\big(\, \underline{\mathfrak{p}} \, \big)
     \;
     ,
     \tau^{\mathrm{can}}_{\mathrm{dR}}
   \Big)
   \\[5pt]
   \; =
   \left\{ \!\!\!\!
   \raisebox{2pt}{
   \xymatrix@R=-4pt@C=49pt{
     \Omega^\bullet_{\mathrm{dR}}
     \Big(\!
       \big(
         \orbisingular
         (X \!\sslash\! G)
       \big)
       \!\sslash\! T
     \Big)
     \ar@{<--}[rr]^-{
       \mbox{
         \tiny
         \color{greenii}
         \bf
         \def\arraystretch{1}
         \begin{tabular}{c}
           flat canonically-twisted
           \\
           proper $G$-equivariant \& Borel $T$-equivariant
           \\
           $\underline{\mathfrak{g}}$-valued differential form
         \end{tabular}
       }
     }
     \ar@{<-}[dr]_-{ \tau^{\mathrm{can}}_{\mathrm{dR}} }
     &&
     \mathrm{CE}\big(\, \underline{\widehat{\mathfrak{b}}} \, \big)
     \ar@{<-}[dl]^-{
       \mathrm{CE}(\, \underline{\mathfrak{p}}\,)
     }
     \\
     &
     \mathrm{CE}
     \big(
       \mathfrak{l}
       B T
     \big)
   }
   }
   \!\!\! \right\}\!.
   \end{array}
 \end{equation}

 \noindent {\bf (ii)} A {\it coboundary} between two such elements is
 defined, as in Def. \ref{CoboundariesBetweenFlatEquivariantLInfinityAlgebraValuedForms},
 by a concordance form on the cylinder orbifold:
 \begin{equation}
   \label{CoboundaryBetweenCanonicallyTwistedForms}
   \widetilde A
   \;\;
     \in
   \;\;
   \Omega^{
     p_X^\ast(
       \tau^{\mathrm{can}}_{\mathrm{dR}}
     )
   }_{\mathrm{dR}}
   \Big(\!
     \big(
       \orbisingular
       \big(
         (X \times \mathbb{R}) \!\sslash\! G
       \big)
     \big)
     \!\sslash\!
     T
     ;
     \,
     \underline{\mathfrak{g}}
   \Big).
 \end{equation}

 \noindent
 The corresponding
 twisted equivariant non-abelian
 de Rham cohomology is defined, as in Def. \ref{TwistedEquivariantNonAbelianDeRhamCohomology},
 to be the set of coboundary-classes of
 the elements in the set \eqref{SetOfCanonicallyTwistedForms}:
 $$
   H^{\tau^{\mathrm{can}}_{\mathrm{dR}}}_{\mathrm{dR}}
   \Big(\!
     \big(
       \orbisingular
       (
         X \!\sslash\! G
       )
     \big)
     \!\sslash\!
     T
     ;
     \,
     \underline{\mathfrak{g}}
   \Big)
   \;\;
     :=
   \;\;
   \Omega^{\tau^{\mathrm{can}}_{\mathrm{dR}}}_{\mathrm{dR}}
   \Big(\!
     \big(
       \orbisingular
       (
         X \!\sslash\! G
       )
     \big)
     \!\sslash\!
     T
     ;
     \,
     \underline{\mathfrak{g}}
   \Big)_{\big/ \sim}.
 $$

\end{defn}

In Borel-equivariant generalization of \cite[Prop. 3.86]{FSS23-Char},
we have:

\begin{prop}[Reproducing traditional Borel-Weil-Cartan equivariant de Rham cohomology]
  \label{ReproducingTraditionalBorelEquivariantdeRhamCohomology}
  For the case of trivial proper equivariance, $G = 1$,
  consider
  $T \acts \; X  \,\in\, \TActionsOnSmoothManifolds$
  (Def. \ref{ProperGActionsOnSmoothManifolds})
  and let the equivariant $L_\infty$-algebraic
  coefficient bundle \eqref{EquivariantLInfinityAlgebraicCoefficientBundleForProperAndBoreldeRham}
  be the trivial bundle with fiber the
  line Lie $n$-algebra $\mathfrak{b}^{n+1}\mathbb{R}$ (\cite[Ex. 3.27]{FSS23-Char}).
    Then the
  canonically twisted proper $G$-equivariant \& Borel $T$-equivariant
  non-abelian de Rham cohomology of $X$
  (Def. \ref{ProperGEquivariantAndBorelTEquivariantTwistedNonabelianDeRhamCohomology})
  reduces to the traditional Borel-Weil-Cartan equivariant
  de Rham cohomology
  (the cochain cohomology of the Cartan model complex \eqref{CartanModelOnFixedLocus})
  in degree $n$:
  $$
    \xymatrix{
      \overset{
        \mathclap{
        \raisebox{3pt}{
          \tiny
          \color{darkblue}
          \bf
          \def\arraystretch{1}
          \begin{tabular}{c}
            Borel-Weil-Cartan equivariant
            \\
            de Rham cohomology
          \end{tabular}
        }
        }
      }{
      H_{\mathrm{dR},T}^n
      \big(
        X
      \big)
      }
      \;\;
        \simeq
      \;\;
      H_{\mathrm{dR}}
      \big(
        X \!\sslash\! T
        ;
        \,
        \mathfrak{b}^n \mathbb{R}
      \big).
    }
  $$

\end{prop}
\begin{proof}
  From unravelling the definitions
  it is clear that, under the given assumptions, the
  defining set of cochains \eqref{SetOfCanonicallyTwistedForms}
  reduces to the set of closed degree $n$ elements in the
  Cartan model complex \eqref{CartanModelOnFixedLocus}
  on $X = X^1$.
    Hence, given any pair of such, it is sufficient to see
  that the coboundaries according to \eqref{CoboundaryBetweenCanonicallyTwistedForms}
  exist precisely if a coboundary with respect to the
  Cartan model differential $d_{\mathrm{dR}} + r^{\, a}_2 \wedge \iota_{\, t_a}$
  exists.

  In the case when the second summand $r^{\, a}_2 \wedge \iota_{\, t_a}$
  vanishes, this is shown by the proof in \cite[Prop. 3.86]{FSS23-Char},
  using the fiberwise Stokes theorem for fiber integration
  over $[0,1] \subset \mathbb{R}$.
  Inspection shows that
  this proof generalizes verbatim in the presence of the
  second summand in the Cartan differential, using that this second summand
  evidently anti-commutes with the
  fiber integration operation:
  $$
    r^{\, a} \wedge \iota_{\, t_a} \int_{[0,1]} \widetilde C
    \;=\;
    -
    \int_{[0,1]}
    r^{\, a} \wedge \iota_{\, t_a}
    \widetilde C
    \,.
  $$

\vspace*{-1.3\baselineskip}
\end{proof}

\begin{remark}[Localization in gauge theory]
Prop. \ref{ReproducingTraditionalBorelEquivariantdeRhamCohomology}
means that the equivariant de Rham cohomology considered
here subsumes the traditional Borel-equivariant de Rham
cohomology that is used, for instance, in localization
of gauge theories (see \cite{Pestun12}\cite{PZ+17}), and generalizes
it to finite proper equivariance groups and to non-abelian
coefficients.
\end{remark}

In equivariant generalization of \cite[Ex. 3.96]{FSS23-Char}, we have:

\begin{example}[Flat equivariant twistorial differential forms]
  \label{FlatEquivariantTwistorialDifferentialForms}
Consider the equivariant relative Whitehead $L_\infty$-algebra
\eqref{MinimalModelOfSpinParametrizedTwistorSpaceModZ2}
of
$\Grefl$-equivariant \& $\SpLR$-parametrized
twistor space \eqref{GHetEquivariantSpLRParametrizedTwistorSpace}
(from Thm. \ref{Z2EquivariantRelativeMinimalModelOfSpin3ParametrizedTwistorSpace})
as an equivariant $L_\infty$-algebraic local coefficient
bundle \eqref{EquivariantLInfinityAlgebraicCoefficientBundle}
\vspace{-1mm} 
\begin{equation}
  \label{WhiteheadLInfinityOfEquivariantParametrizedTwistorSpaceAsLocalCoefficients}
  \raisebox{30pt}{
  \xymatrix@R=1.6em{
    \mathfrak{l}
    \orbisingular
    \big(
      \mathbb{C}P^3
      \!\sslash\!
      \Grefl
    \big)
    \ar[r]
    &
    \mathfrak{l}_{
      B \SpLR
    }
    \big(
    \orbisingular
    \big(
      \mathbb{C}P^3
      \!\sslash\!
      \Grefl
    \big)
      \!\sslash\!
      \SpLR
    \big)
    \ar[d]^-{
      \rho_{\scalebox{0.6}{$
        \orbisingular
        \left(
          \mathbb{C}P^3
          \!\sslash\!
          \Grefl
        \right)
        $}
      }
    }
    \\
    &
    \mathfrak{l} B \SpLR
  }
  }
\end{equation}

\noindent
Let $X \in \ZTwoActionsOnSmoothManifolds$
(Def. \ref{ProperGActionsOnSmoothManifolds})
be a spin 8-manifold
with fixed locus \eqref{FixedLoci} denoted

\vspace{-1mm} 
\begin{equation}
  \label{TheGeneralSpacetimeOrbifold}
  \orbisingular
  \big(
    X^{} \!\sslash\! \ZTwo
  \big)
  \;\;\;\;
  :
  \;\;\;\;
  \raisebox{20pt}{
  \xymatrix@R=.8em{
    \ZTwo/1
    \ar@(ul,ur)|-{\; \ZTwo \, }
    \ar[d]
    &\longmapsto&
    X^{11}
    \ar@(ul,ur)|-{\; \ZTwo \,}
    \ar@{<-^{)}}[d]
    \\
    \ZTwo/\ZTwo
    &\longmapsto&
    \mathclap{\phantom{\vert^{\vert^{\vert}}}}
    X^{\Grefl}
  }
  }
\end{equation}

\noindent
and equipped with $\ZTwo$-invariant $\SpLR$-structure $\tau$,
compatible $\Grefl$-invariant $\SpLR$-connection
$\nabla \in \SpLR \mathrm{Connections}(X)$,
and corresponding tangential de Rham twist
(Example \ref{TangentialDeRhamTwistsOnGOrbifoldsWithTStructure})
$$
  \xymatrix@R=-3pt{
    \Omega^\bullet_{\mathrm{dR}}
    \big(
      \orbisingular
      (
        X \!\sslash\! \ZTwo
      )
    \big)
    \ar@{<-}[rr]^-{ \tau_{\mathrm{dR}} }
    &&
    \mathrm{CE}
    \big(
      \mathfrak{l} B \SpLR
    \big).
    \\
    \tfrac{1}{4}p_1(\nabla)
    \ar@{}[rr]|-{ \longmapsfrom }
    &&
    \tfrac{1}{4}p_1
      }
$$

\noindent
Then the set
of flat $\tau_{\mathrm{dR}}$-twisted equivariant differential  forms
(Def. \ref{SetOfFlatTwistedEquivariantLIninfityAlgebraValuedDifferentialForms})
with local coefficients in \eqref{WhiteheadLInfinityOfEquivariantParametrizedTwistorSpaceAsLocalCoefficients}
is of the following form:
  \begin{equation}
    \label{TheBianchiIdentities}
    \hspace{-2mm}
    \adjustbox{scale=0.85}{$
    \begin{aligned}
    &
    \overset{
      \mathclap{
      \raisebox{3pt}{
        \tiny
        \color{darkblue}
        \bf
        {\begin{tabular}{c}
          flat equivariant
          twistorial 
          \\[-2pt]
          differential forms
          on $\ZTwo$-orbifold $X$
        \end{tabular}}
      }
      }
    }{
      \Omega^{\tau_{\mathrm{dR}}}_{\mathrm{dR}}
      \Big(
      \!
      \orbisingular
      \big(
        X \!\sslash\! \ZTwo
      \big)
      ;
      \;
      \mathfrak{l}
      \orbisingular
      \big(
        \mathbb{C}P^3 \!\sslash\! \Grefl
      \big)
      \!
    \Big)_{\mathrm{flat}}
    }
    \\
    &
    =
    {\small
    \left\{\!
      \!\!\!
      {\begin{array}{c}
        \phantom{2} H_3,
        \\
        \phantom{2} F_2,
        \\
        2 G_7,
        \\
        \phantom{2} \widetilde G_4
        \\
        \in \Omega^\bullet_{\mathrm{dR}} \big(
        X^{11}
      \big)
      \end{array}}
      \!\!
    \left\vert 
          \overset{
        \raisebox{3pt}{
          \tiny
          \color{orangeii}
          \bf
          \begin{tabular}{c}
            twisted Bianchi identities
            \\[-2pt]
            in bulk $\Grefl$-orientifold
          \end{tabular}
        }
      }{
      {\begin{aligned}
        d\, \phantom{2} H_3
          & =
          \widetilde G_4 - \tfrac{1}{2} p_1(\nabla) -  F_2 \wedge F_2,
        \\
        d\, \phantom{2} F_2 & = 0,
        \\
        d\, 2G_7 &
          = - \widetilde G_4 \wedge
          \big(
            \widetilde G_4 -
            \tfrac{1}{2}p_1(\nabla)
          \big)
        \\
        d\, \widetilde G_4 & = 0,
      \end{aligned}}
      }
    \right.
    \;\;
    \overset{
      \raisebox{4pt}{
      \tiny
      \color{orangeii}
      \bf
      {\begin{tabular}{c}
        restriction to
        $\Grefl$-fixed locus
      \end{tabular}}
      }
    }
    {\begin{aligned}
      d H_3\vert_{X^{\scalebox{.5}{$\Grefl$}}}
        & =
        - \tfrac{1}{2}p_1
        \big(
          \nabla\vert_{X^{\scalebox{.5}{$\Grefl$}}}
        \big)
        - F_2 \wedge F_2\vert_{X^{\scalebox{.5}{$\Grefl$}}}
      \\
      \phantom{F_2}
      \\
      G_7\vert_{X^{\scalebox{.5}{$\Grefl$}}}
      &
      = 0,
      \\
      \widetilde G_4\vert_{X^{\scalebox{.5}{$\Grefl$}}}
      &
      = 0
    \end{aligned}}
   \! \right\}.
    }
  \end{aligned}
  $}
\end{equation}

\noindent
This follows as an immediate consequence of Prop. \ref{Z2EquivariantRelativeMinimalModelOfSpin3ParametrizedTwistorSpace},
according to which an element $\mathcal{F}$ of this set of forms
is a morphism of equivariant dgc-algebras
of the following form 
  \begin{equation}
    \label{EquivariantDgcHomomorphismForTwistorialDifferentialForms}
    \hspace{-2mm}
    \adjustbox{scale=0.78}{$
    \mathcal{F}
    :
    \!\!
  \raisebox{40pt}{
  \xymatrix@C=-20pt@R=1.8em{
    \ZTwo/1
    \ar@{->}[d]
    \ar@(ul,ur)|-{\;\ZTwo}
    &
    \phantom{--}
    \longmapsto
    \phantom{--}
    &
    \Omega^\bullet_{\mathrm{dR}}(X)
    \ar@{<-}[rr]^-{
      \scalebox{.7}{$
        \arraycolsep=1.6pt\def\arraystretch{1}
        \begin{array}{lcl}
          H_3 &\mapsfrom& h_3
          \\
          F_2 &\mapsfrom& f_2
                       \end{array}
                       $}
                       }_-{
        \scalebox{.7}{$
        \arraycolsep=1.6pt
        \def\arraystretch{1}
        \begin{array}{lcl}
                    2 G_7 &\mapsfrom& \omega_7
          \\
          \widetilde G_4 &\mapsfrom& \widetilde \omega_4
        \end{array}
      $}
    }
    \ar[d]^-{
      \scalebox{.8}{$
        \arraycolsep=1.6pt
        \def\arraystretch{.7}
        \begin{array}{c}
          \alpha
          \\
          \mapsdown
          \\
          \alpha \vert_{X^{\ZTwo}}
        \end{array}
      $}
    }
    &
    {\phantom{AAAAAAAAA}}
    &
    \mathrm{CE}
    \big(
      \mathfrak{l} B \SpLR
    \big)
    \!\!
    \left[
      \!\!
      {\begin{array}{c}
        h_3,
        \\
        f_2
        \\
        \omega_7,
        \\
        \widetilde \omega_4
      \end{array}}
      \!\!
    \right]
    \!\big/\!
    \left(
      {\begin{aligned}
        d\, h_3  & = \widetilde \omega_4 - \tfrac{1}{2}p_1 -  f_2 \wedge f_2
        \\[-4pt]
        d\, f_2 & = 0
        \\[-4pt]
        d\, \omega_7
          & =
          -
          \widetilde \omega_4 \wedge
          \big(
            \widetilde \omega_4 - \tfrac{1}{2}p_1
          \big)
        \\[-4pt]
        d\, \widetilde \omega_4 & = 0
      \end{aligned}}
    \right)
    \ar@<-42pt>@{->>}[d]
    \\
    \ZTwo/\ZTwo
    &\longmapsto&
    \Omega^\bullet_{\mathrm{dR}}
    \big(
      X^{\ZTwo}
    \big)
    \ar@{<-}[rr]
    &&
    \mathrm{CE}
    \big(
      \mathfrak{l} B \SpLR
    \big)
    \!\!
    \left[
      \!\!
      {\begin{array}{c}
        h_3,
        \\
        f_2
      \end{array}}
      \!\!
    \right]
    \!\big/\!
    \left(
      {\begin{aligned}
        d\, h_3  & = \phantom{\omega_4}\; - \tfrac{1}{2}p_1 - f_2 \wedge f_2
        \\[-4pt]
        d\, f_2 & = 0
      \end{aligned}}
    \right)
    .
  }
  }
  $}
\end{equation}

\end{example}

\subsection{Equivariant non-abelian character map}
  \label{TheEquivariantTwistedNonAbelianCharacterMap}

The Chern character in K-theory is just one
special case of a plethora of character maps in a variety of
flavors of generalized cohomology theories.
As highlighted in \cite{FSS23-Char}\cite{SS24-Flux}, from the
point of view of homotopy-theoretic non-abelian
cohomology theory -- where all cohomology classes
are represented by (relative, parametrized)
homotopy classes of maps into a classifying space
(fibered, parametrized $\infty$-stack) --
character maps are naturally realized as the non-abelian
cohomology operations induced by {\it rationalization}
of the classifying space
(followed by a de Rham-Dold-type equivalence
bringing the resulting rational cohomology theory
into canonical shape).

\medskip
Seen through the lens of Elmendorf's theorem
(Prop. \ref{ElmendorfTheorem}),
rationalization in proper equivariant homotopy theory
(Def. \ref{EquivariantRationalization})
is stage-wise, on fixed loci, given by rationalization in
non-equivariant homotopy theory. Consequently,  the equivariant character
maps are fixed loci-wise given by non-equivariant characters,
hence are fixed loci-wise given by rationalization
(followed by a de Rham equivalence).

\medskip
For this reason, we will be brief here and refer to \cite{FSS23-Char} for background and further detail.
We just make explicit
now the concrete model of the equivariant
non-abelian character map by means of
the equivariant PL de Rham Quillen adjunction
from Prop. \ref{QuillenAdjunctionBetweenEquivariantSSetAndEquivariantdgcAlgebras}.
and then we discuss one example:
the character map in equivariant twistorial Cohomotopy theory.

\medskip

\noindent
{\bf The character map in equivariant non-abelian cohomology.}
$\,$

\smallskip 
\noindent In equivariant generalization of \cite[Def. 4.1]{FSS23-Char}, we set:

\begin{defn}[Rationalization in equivariant non-abelian cohomology]
 \label{RationalizationInEquivariantNonabelianCohomology}
Let $\mathscr{A} \in \EquivariantSimplyConnectedRFiniteHomotopyTypes$
(Def. \ref{SubcategoryOfEquivariantSimplyConnectedRFiniteHomotopyTypes}).
Then we say that
{\it rationalization in $\mathscr{A}$-cohomology}
is the equivariant non-abelian cohomology operation
(Def. \ref{EquivariantNonabelianCohomologyOperation})
from $\mathscr{A}$-cohomology to real $L_{\mathbb{R}}\mathscr{A}$-cohomology
which is induced \eqref{MorphismInducingCohomologyOperation}
by the rationalization unit \eqref{RationalizationUnitOnEquivariantHomotopyTypes}
on $\mathscr{A}$:
$$
  \xymatrix{
    H
    (
      -;
      \mathscr{A}
    )
    \ar[rr]^-{\scalebox{.8}{$
      (
              \eta^{\mathbb{R}}_{\scalebox{.7}{$\mathscr{A}$}}
      )_\ast
      $}
    }
    &&
    H
    (
      -;
      L_{\mathbb{R}}
      \mathscr{A}
    )
    \,.
  }
$$
\end{defn}

In an equivariant generalization of \cite[Def. 4.2]{FSS23-Char}, we set:

\begin{defn}[Equivariant non-abelian character map]
  \label{EquivariantNonAbelianCharacterMap}
  Let
  $ G \acts \; X \,\in\, \GActionsOnSmoothManifolds$
  (Def. \ref{ProperGActionsOnSmoothManifolds})
  and
  $
    \underline{\mathfrak{g}}
  $
  (Def. \ref{EquivariantLInfinityAlgebras}).
  Then the \emph{equivariant non-abelian character map}
  on
  equivariant non-abelian $\mathscr{A}$-cohomology
  (Def. \ref{EquivariantNonAbelianCohomology})
  over the orbifold $\orbisingular \big( X \sslash G\big)$
  (Def. \ref{GOrbifolds})
  is the composite of the
  rationalization cohomology operation (Def. \ref{RationalizationInEquivariantNonabelianCohomology})
  with the equivariant non-abelian de Rham theorem
  (Prop. \ref{EquivariantNonabelianDeRhamTheorem})
  over the orbifold $\orbisingular(X \!\sslash\! G)$
  (Def. \ref{GOrbifolds})
  \begin{equation}
    \label{TheEquivariantNonabelianCharacterMap}
    \hspace{5mm}
    \adjustbox{scale=0.9}{$
    \overset{
      \mathclap{
      \raisebox{3pt}{
        \tiny
        \color{greenii}
        \bf
        \def\arraystretch{1}
        \begin{tabular}{c}
          Equivariant non-abelian
          \\
          character map
        \end{tabular}
      }
      }
    }{
      \mathrm{ch}_{\scalebox{.7}{$\mathscr{A}$}}(X)
    }
    \;\;:
    \xymatrix@C=23pt{
      \underset{
        \mathclap{
        \raisebox{-3pt}{
          \tiny
          \color{darkblue}
          \bf
          \def\arraystretch{1}
          \begin{tabular}{c}
            equivariant non-abelian
            \\
            $\mathscr{A}$-cohomology
          \end{tabular}
        }
        }
      }{
      H
      \big(\!
        \orbisingular
        (
           X \!\sslash\! G
        )
        ;
        \,
        \mathscr{A}
      \big)
      }
      \ar[rr]_-{\scalebox{.8}{$
        \big( \eta^{\mathbb{R}}_{\scalebox{.7}{$\mathscr{A}$}}\big)_\ast
     $}
      }^-{\!\!\!\!\!
        \adjustbox{
          raise=5pt,
          scale=.65
        }{
          \color{greenii}
          \bf
          rationalization
        }
      }
      \;
      &&
      H
      \big(\!
        \orbisingular
        (
           X \!\sslash\! G
        )
        ;
        \,
        L_{\mathbb{R}}\mathscr{A}
      \big)
      \ar[d]_-{ 
        \rotatebox[origin=c]{90}{$\sim$}
      }^-{\!\!\!\!
        \mbox{
          \hspace{-11pt}
          \tiny
          \color{greenii}
          \bf
          \def\arraystretch{1}
          \begin{tabular}{c}
            equivariant non-abelian
            \\
            de Rham theorem
          \end{tabular}
        }
      }
      \\
      &&
      \underset{
        \mathclap{
        \raisebox{-3pt}{
          \tiny
          \color{darkblue}
          \bf
          \def\arraystretch{1}
          \begin{tabular}{c}
            equivariant non-abelian
            de Rham cohomology
            \\
            with coefficient in
            equivariant Whitehead $L_\infty$-algebra
          \end{tabular}
        }
        }
      }{
      H_{\mathrm{dR}}
      \big(\!
        \orbisingular
        (
           X \!\sslash\! G
        )
        ;
        \,
        \mathfrak{l}\mathscr{A}
      \big).
      }
    }
    $}
  \end{equation}

\end{defn}

\medskip

\noindent {\bf The character map in twisted equivariant non-abelian cohomology.}

  \noindent  In equivariant generalization of \cite[Def. 5.2]{FSS23-Char}, we set:

\begin{defn}[Rationalization in twisted equivariant non-abelian cohomology]
 Let $\rho_{\scalebox{.7}{$\mathscr{A}$}}$
 be an equivariant local coefficient bundle
 of equivariantly 1-connected $G$-spaces of
 finite $\mathbb{R}$-homotopy type, which admits
 an equivariant relative minimal model; all as in
 Notation \ref{EquivariantLocalCoefficientBundleWithRelativeMinimalModel}.
 Then {\it rationalization} in twisted equivariant non-abelian cohomology
 with local coefficients in $\rho_{\scalebox{.7}{$\mathscr{A}$}}$
 (Def. \ref{EquivariantTwistedNonAbelianCohomology})
 is the equivariant non-abelian cohomology operation
 $$
   \big(
     \eta^{\mathbb{R}}_{
       \scalebox{.67}{$
         \rho_{
           \scalebox{.7}{$
             \mathscr{A}
           $}
         }
       $}
     }
   \big)_\ast
   \;\;:\;\;
   \xymatrix{
     H^\tau
     \big(
       \mathscr{X}
       ;
       \,
       \mathscr{A}
     \big)
     \ar[rrrr]^-{
       \scalebox{.75}{$
       \big(
         \mathbb{D}\eta^{\mathrm{PLdR}}_{\rho_{\scalebox{.7}{$\mathscr{A}$}}}
         \,\circ\,
         (-)
       \big)
       \,\circ\,
       \mathbb{L}
       \big(
         \eta^{\mathbb{R}}_{\scalebox{.7}{$B \mathscr{G}$}}
       \big)_!
       $}
     }
     &&&&
     H^{L_{\mathbb{R}}\tau}
     \big(
       \mathscr{X}
       ;
       \,
       L_{\mathbb{R}}\mathscr{A}
     \big)
   }
 $$

\noindent  which is induced
 (as shown in \cite[(264)]{FSS23-Char})
 by the pasting composite with the
 naturality square on $\rho_{\mathscr{A}}$
 of the rationalization unit
 (Def. \ref{EquivariantRationalization}).
 By the fundamental theorem
 (Prop. \ref{FundamentalTheoremOfdgcAlgebraicEquivariantRationalHomotopyTheory}),
 this means explicitly:
 the left derived base change
 (e.g. \cite[Ex. A.18]{FSS23-Char})
 along the PLdR-adjunction unit
 (Prop. \ref{QuillenAdjunctionBetweenEquivariantSSetAndEquivariantdgcAlgebras})
 on $B \mathscr{G}$ followed by composition with the
 following commuting square, regarded as a morphism in the
 slice over its bottom right object:
 \vspace{2mm} 
 $$
   \mathbb{D} \eta^{\mathbb{R}}_{
     \scalebox{.67}{$
       \rho_{
         \scalebox{.67}{$
           \mathscr{A}
         $}
       }
     $}
   }
     := \!
  \left(\!\!\!\!
   \raisebox{28pt}{
   \xymatrix@C=15pt@R=3em{
     \mathscr{A}
       \!\sslash\!
     \mathscr{G}
     \ar[d]_-{
       \scalebox{.7}{$
         \rho_{\scalebox{.7}{$\mathscr{A}$}}
       $}
     }
     \ar[rr]_-{
       \;
       \scalebox{.7}{$
       \eta^{\mathrm{PLdR}}_{\scalebox{.7}{$\mathscr{A} \!\sslash\! \mathscr{G}$}}
       $}
       \;
     }
     \ar@/^1.6pc/[rrrr]|-{
       \;
       \scalebox{.7}{$
         \mathbb{D}
         \eta^{\mathrm{PLdR}}_{
           \scalebox{.7}{$\mathscr{A}\!\sslash\! \mathscr{G} $}
         }
         \;\simeq\;
         \eta^{\mathbb{R}}_{
           \scalebox{.7}{$\mathscr{A}\!\sslash\! \mathscr{G} $}
         }
       $}
       \;
     }
     &&
     \exp
       \,\circ\,
     \Omega_{\mathrm{PLdR}}
     \big(
       \mathscr{A}
         \!\sslash\!
       \mathscr{G}
     \big)
     \ar[rr]_-{
       \;
       \scalebox{.7}{$
         p^{
           \mathrm{min}_{\scalebox{.6}{$B \mathscr{G}$}}
         }_{
           \scalebox{.7}{$
             \mathscr{A} \!\sslash\! \mathscr{G}
           $}
         }
       $}
       \;
     }
     \ar[d]_-{
       \scalebox{.7}{$
         \exp \,\circ\, \Omega^\bullet_{\mathrm{PLdR}}
         \big(
           \rho_{\scalebox{.7}{$\mathscr{A}$}}
         \big)
       $}
     }
     &&
     \exp \,\circ\,
     \mathrm{CE}
     \big(
       \mathfrak{l}_{\scalebox{.6}{$B \mathscr{G}$}}
       (\mathscr{A} \!\sslash\! \mathscr{G} )
     \big)
     \ar[d]_-{
       \exp
       \,\circ\,
       \mathrm{CE}
       (
         \mathfrak{l}
         \rho_{\scalebox{.7}{$\mathscr{A}$}}
       )
     }
     \\
     B \mathscr{G}
     \ar[rr]^-{
       \;
       \scalebox{.7}{$
       \eta^{\mathrm{PLdR}}_{\scalebox{.7}{$B \mathscr{G}$}}
       $}
       \;
     }
     \ar@/_1.6pc/[rrrr]|-{
       \;
       \scalebox{.7}{$
         \mathbb{D}
         \eta^{\mathrm{PLdR}}_{
           \scalebox{.7}{$B \mathscr{G} $}
         }
         \;\simeq\;
         \eta^{\mathbb{R}}_{
           \scalebox{.7}{$B \mathscr{G} $}
         }
       $}
       \;
     }
     &&
     \exp
     \,\circ\,
     \Omega^\bullet_{\mathrm{PLdR}}
     \big(
       B \mathscr{G}
     \big)
     \ar[rr]^-{
       \;
       \scalebox{.7}{$
         p^{
           \mathrm{min}
         }_{
           \scalebox{.7}{$
             B \mathscr{G}
           $}
         }
       $}
       \;
     }
     &&
     \exp \,\circ\,
     \mathrm{CE}
     \big(
       \mathfrak{l}
       ( B \mathscr{G} )
     \big)
   }
   }
\!\!\!\!   \right).
 $$
 
\vspace{2mm} 
 \noindent
 Here the left-hand side is the naturality square of the
 equivariant PL de Rham adjunction (Prop. \ref{QuillenAdjunctionBetweenEquivariantSSetAndEquivariantdgcAlgebras}),
 while the right-hand side is the
 image under $\exp$ of the relative minimal model
  \eqref{EquivariantRelativeMinimalModelOfLocalCoefficientBundle}.
 (Hence the composite represents the naturality square of the
 derived PL de Rham adjunction unit, see e.g. \cite[Ex. A.21]{FSS23-Char}).

\end{defn}

In equivariant generalization of \cite[Def. 5.4]{FSS23-Char}, we set:
\begin{defn}[Twisted equivariant non-abelian character map]
  \label{TwistedEquivariantNonabelianCharacterMap}
 Let $G \acts \; X$  $\in\, \GActionsOnSmoothManifolds$
 (Def. \ref{ProperGActionsOnSmoothManifolds}),
 and let $\rho_{\scalebox{.7}{$\mathscr{A}$}}$
 be an equivariant local coefficient bundle
 of equivariantly 1-connected $G$-spaces of
 finite $\mathbb{R}$-homotopy type, which admits
 an equivariant relative minimal model; all as in
 Notation \ref{EquivariantLocalCoefficientBundleWithRelativeMinimalModel}.
  Then the {\it twisted equivariant non-abelian character map}
 is the twisted equivariant cohomology operation
 \begin{equation}
\label{TheTwistedEquivariantNonabelianCharacterMap}
   \hspace{.8cm}
     \overset{
     \mathclap{
     \raisebox{3pt}{
       \tiny
       \color{greenii}
       \bf
       \def\arraystretch{1}
       \begin{tabular}{c}
         twisted equivariant
         \\
         non-abelian character
       \end{tabular}
     }
     }
   }{
     \mathrm{ch}^\tau_{\scalebox{.67}{$\mathscr{A}$}}
   }
   \;\;
     :
   \xymatrix@C=25pt{
     \underset{
       \mathclap{
       \raisebox{-3pt}{
         \tiny
         \color{darkblue}
         \bf
         \begin{tabular}{c}
           twisted equivariant
           \\[-2pt]
           non-abelian 
           $\mathscr{A}$-cohomology
         \end{tabular}
       }
       }
     }{
     H^\tau
     \big(
       \orbisingular
       (
         X \!\sslash\! G
       )
       ;
       \,
       \mathscr{A}
     \big)
     }
     \ar[rr]_-{
       \big(
         \eta^{\mathbb{R}}_{
           \scalebox{.67}{$
             \rho_{\scalebox{.67}{$\mathscr{A}$}}
           $}
         }
       \big)_\ast
     }^-{\!\!\!\!
       \mathclap{
       \mbox{
         \tiny
         \color{greenii}
         \bf
         rationalization
       }
       }
     }
     &&
     H^{L_{\mathbb{R}}\tau}
     \big(
       \orbisingular
       (
         X \!\sslash\! G
    )
       ;
       \,
       L_{\mathbb{R}}\mathscr{A}
     \big)
     \ar[d]_-{ 
       \rotatebox[origin=c]{-90}{$\sim$} 
     }^-{\!\!\!\!
       \mathrlap{
       \mbox{
         \hspace{-9pt}
         \tiny
         \color{greenii}
         \bf
         \def\arraystretch{1}
         \begin{tabular}{c}
           equivariant twisted 
           \\
           non-abelian
           \\
           de Rham theorem
         \end{tabular}
       }
       }
     }
     \\
     &&
     \underset{
       \mathclap{
       \raisebox{-3pt}{
         \tiny
         \color{darkblue}
         \bf
         \def\arraystretch{1}
         \begin{tabular}{c}
           twisted equivariant
           \\
           non-abelian de Rham cohomology
         \end{tabular}
       }
       }
     }{
     H^{\tau_{\mathrm{dR}}}
     \big(
       \orbisingular
       (
         X \!\sslash\! G
       )
       ;
       \,
       \mathfrak{l}\mathscr{A}
     \big)
     }
   }
 \end{equation}

\noindent
from twisted equivariant non-abelian cohomology
(Def. \ref{EquivariantTwistedNonAbelianCohomology})
with local coefficients in $\rho_{\scalebox{.7}{$\mathscr{A}$}}$
to twisted equivariant non-abelian de Rham cohomology
(Def. \ref{TwistedEquivariantNonAbelianDeRhamCohomology})
with coefficients in $\mathfrak{l}\rho_{\scalebox{.7}{$\mathscr{A}$}}$
(as in Notation \ref{EquivariantLocalCoefficientBundleWithRelativeMinimalModel}).

\end{defn}

Finally, we have:

\begin{remark}[Proof of Theorem \ref{FluxQuantizationInEquivariantTwistorialCohomotopy}]
\label{ProofOfMainTheorem}
We collect our results:

\noindent
{\bf (i)} That the Bianchi identities
in the
twistorial character map are as shown
on p. \pageref{FluxQuantizationInEquivariantTwistorialCohomotopy}
follows by Prop. \ref{Z2EquivariantRelativeMinimalModelOfSpin3ParametrizedTwistorSpace},
as discussed in Example \ref{FlatEquivariantTwistorialDifferentialForms}.

\noindent {\bf (ii)} That the quantization conditions in the twistorial
character are as shown in \eqref{IntegralityConditions}
follows by observing that the twisted equivariant character map
(Def. \ref{TwistedEquivariantNonabelianCharacterMap}) is
fixed-locus wise equivalent to the corresponding non-equivariant
twisted character map \cite[Def. 5.4]{FSS23-Char}
(for instance by the fundamental theorem, Prop. \ref{FundamentalTheoremOfdgcAlgebraicEquivariantRationalHomotopyTheory},
using that the equivariant PL de Rham adjunction is stage-wise
given by the non-equivariant PL de Rham adjunction,
Prop. \ref{QuillenAdjunctionBetweenEquivariantSSetAndEquivariantdgcAlgebras}).

\noindent {\bf (iii)}  In particular, at global stage $\Grefl/1 \,\in\, \ZTwo \mathrm{Orb}$
on the bulk $X^1 = X$, the equivariant
twistorial character restricts to the
non-equivariant twistorial character map for which the claimed flux
quantization conditions have been proven in
\cite[Prop. 3.13]{FSS20b}\cite[Thm. 4.8]{FSS20c}\cite[Cor. 3.11]{FSS20c}, see also \cite[\S 5.3]{FSS23-Char}.
\end{remark}
This establishes Thm. \ref{FluxQuantizationInEquivariantTwistorialCohomotopy}.

\section{Application to flux-quantization}
\label{ApplicationToFLuxQuantization}

Here we briefly indicate the meaning and significance of the above algebro-topological result in and to theoretical physics, specifically concerning the problem of ``flux quantization'' 
\cite{SS24-Flux}
in a candidate theory of strongly-coupled quantum systems going by the working title ``M-theory'' \cite{Duff99}.

\smallskip

\noindent
{\bf Cohomology and Gauge fields.}
Beyond all the details, a remarkable general fact --- that the applied algebraic topologists may find entertaining --- is the fundamental role that cohomology (generalized, twisted, equivariant, differential, non-abelian, ...) has come to play in the fine-grained description of gauge fields (``force fields'') in fundamental physics, especially of ``higher gauge fields'' -- whose flux-densities are higher-degree differential forms on spacetime satisfying differential ``Bianchi'' or ``Gau{\ss} law'' equations -- that appear in attempts to fill certain gaps in the contemporary understanding of fundamental physics.

\smallskip 
In short, such flux densities are to be regarded as but the character images \eqref{CharacterInIntroduction} of classes in some (generalized non-abelian) cohomology theory, the choice of which is a {\it flux-quantization law} that controls global (brane-) {\it charges} imprinted on the gauge field, and the further refinement of these to cocycles in {\it differential} cohomology encodes the ``gauge potentials'' typically discussed in the physics literature, 
on which the eponymous gauge transformations are given by the corresponding coboundaries.

\medskip

\def\arraystretch{.5}
\begin{tabular}{c}
\hypertarget{TableCG}{}
\def\tabcolsep{1.3}
\begin{minipage}{12cm}
  \footnotesize
  {\bf Table CG.}
  While cohomology has of course many and diverse applications, in physics no less than in other fields, the role of cohomology  specifically in the {\it global} description of (higher) gauge fields (``force fields'') is profound: In a generalization of the seminal historical observation (``Dirac charge quantization'') that electromagnetic field configurations are globally to be identified with 2-cocycles in ordinary differential cohomology of spacetime, higher gauge field species are similarly to be identified with generalized cohomology theories whose further properties and attributes closely reflect the field's physical nature, as indicated on the right.
\end{minipage}
\\
\\
{\small 
\def\arraystretch{1.3}
\def\tabcolsep{10pt}
\begin{tabular}{|c|c|}
  \hline
  {\bf Cohomology}
  &
  {\bf Gauge fields}
  \\
  \hline
  \rowcolor{lightgray}
  -Theory
  &
  Flux quantization law
  \\
  \rowcolor{white}
  Cocycle 
  &
  Field configuration
  \\
  \rowcolor{lightgray}
  Coboundary
  &
  Gauge transformation
  \\
  \rowcolor{white}
  Character 
    & 
  Flux densities    
  \\
  \rowcolor{lightgray}
  Ordinary-
  &
  Electromagnetic
  \\
  \rowcolor{white}
  Differential-
  &
  Gauge potentials
  \\
  \rowcolor{lightgray}
  Twisted-
  &
  Background fields
  \\
  \rowcolor{white}
  Equivariant-
  &
  on orbifolds
  \\
  \rowcolor{lightgray}
  Real- & on orientifolds
  \\
  \rowcolor{white}
  Nonabelian-
  &
  Nonlinear Gau{\ss} law
  \\
  \hline
\end{tabular}
}
\end{tabular}

\medskip

Conversely, this means that a fair amount of algebro-topological sophistication may be needed to propose or construct a cohomology theory suitable for flux quantization of a given higher gauge theory, and then to deduce its implications to be compared with physical expectations and, ultimately, with experiment. Much room is left here for applied algebraic topologists to get involved.

We briefly indicate how the equivariant twistorial Cohomotopy from the main text is motivated as a flux-quantization law, and what some of its implications are:

\medskip
\noindent
{\bf Characters arising in supergravity.}
With the character map 
\eqref{CharacterInIntroduction}
describing which generalized cohomology theories $\mathscr{A}$ may serve as flux-quantization laws for given generalized Gau{\ss} laws $\mathfrak{l}\mathscr{A}$ on flux densities, we have to ask: 
{\it What are natural generalized such Gau{\ss} laws $\mathfrak{l}\mathscr{A}$?}
Remarkably, a profound source is {\it super-gravity}, in the following way (pointers in \cite{GSS24-SuGra}):

\smallskip

It is a century-old observation due to {\'E}. Cartan that a field configuration of gravity is most usefully understood as a torsion-free {\it coframe field} $E$ (Cartan's ``moving frame'') on spacetime with coefficients in the typical tangent space $\mathbb{R}^{1,d}$ (Minkowski spacetime), subject to a corresponding ``1st order''-formulation of Einstein's equations. A miracle happens as this situation is generalized from ordinary tangent spaces to tangent super-spaces $\mathbb{R}^{1,d\,\vert\, \mathbf{N}}$, meaning to super-vector spaces (namely: $\ZTwo$-graded vector spaces regarded with the unique non-trivial symmetric braided monoidal category structure) whose odd component carries the structure of a real spinor representation $\mathbf{N} \,\in\, \mathrm{Rep}_{\mathbb{R}}\big(\mathrm{Spin}(1,d)\big)$:

\smallskip

\noindent
{\bf -- 11D Super-gravity.}
Namely a field configuration of 11D super-gravity is a {\it super}torsion-free super-coframe field $(E, \Psi)$
on super-spacetime with coefficients in $\mathbb{R}^{1,10\,\vert\, \mathbf{32}}$, where -- remarkably -- the corresponding Einstein-Rarita-Schwinger equations of motion are now {\it equivalent} 
\cite[Thm. 3.1]{GSS24-SuGra}
simply to the statement that flux super-densities of the following form (meaning: super-differential forms whose local expansion in the co-frame field is of this prescribed form): \footnote{
  In \eqref{11DSuperFluxDensities} we include a conventional sign in the definition of $G_7^s$ to comply with the sign convention used in the main text. 
}
\vspace{-2mm} 
\begin{equation}
  \label{11DSuperFluxDensities}
  \def\arraystretch{1.7}
  \begin{array}{ccl}
  G^s_4 
  &\defneq&
  \phantom{+}
  (G_4)_{a_1 \cdots a_4}
  E^{a_1}\cdots E^{a_4}
  \,+\,
  \tfrac{1}{2}\big(
    \overline{\Psi}
    \,\Gamma_{a_1 a_2}\,
    \Psi
  \big)
  E^{a_1} E^{a_2}
  \,,
  \\
  G^s_7 
  &\defneq&
  -
  (G_7)_{a_1 \cdots a_7}
  E^{a_1}\cdots E^{a_7}
  \,-\,
  \tfrac{1}{5!}\big(
    \overline{\Psi}
    \,\Gamma_{a_1 \cdots a_5}\,
    \Psi
  \big)
  E^{a_1} \cdots E^{a_5}
  \end{array}
\end{equation}
satisfy the non-linear Bianchi/Gau{\ss} law encoded by the Whitehead $L_\infty$-algebra of the 4-sphere:
\begin{equation}
  \label{SuperBianchiAndEOMOf11DSuGra}
  \def\arraystrech{2}
  \begin{array}{rcl}
  \begin{tikzcd}[column sep=large]
  X
  \ar[
    r,
    dashed,
    "{
      (
        G_4^s
        ,\,
        G_7^s
      )    
    }"
  ]
  &
  \Omega^1(-;\, \mathfrak{l}
  S^4)_{\mathrm{flat}}
  \end{tikzcd}
  &
  \Leftrightarrow
  &
  \left\{\!\!\!
  \adjustbox{raise=1pt}{
  \begin{tikzcd}[
    row sep=-2pt,
    column sep=0pt
  ]
  \mathrm{d}\, G_4^s
  &=&
  0
  \\
  \mathrm{d}\, G_7^s
  &=&
  -\tfrac{1}{2}G^s_4\, G_4^s
  \end{tikzcd}
  }
  \right\}
  \\
  &
  \Leftrightarrow
  &
  \left\{\!\!\!\!\!
  \scalebox{.8}{
    \def\arraystretch{.9}
    \begin{tabular}{c}
      Equations of Motion
      \\
      of 11D Supergravity
      \\
      on supertorsion-free
      \\
      super-coframe $(E,\Psi)$
    \end{tabular}
  }
  \right.
  \end{array}
\end{equation}
(In particular, the equations of motion include the Hodge duality relation $G_7 \;=\; - \star \, G_4$ over the underlying ordinary spacetime.)

Hence the non-linear Gau{\ss} law $\mathfrak{l}S^4$ not only arises in but effectively {\it constitutes} 11D supergravity.
But the miracle does not end here:

\smallskip

\noindent
{\bf -- M5-brane probes.}
Given the above Cartan-geometric formulation of 11D super-gravity, all based on consideration of the Kleinian local model space $\mathbb{R}^{1,d\,\vert\, \mathbf{N}}$, it is natural to consider Kleinian sub-spaces and ask for their globalization to sub-supermanifolds of 11D spacetime. These are the ``worldvolumes'' of ``probe super-branes''.
Concretely, any Clifford algebra basis element $\Gamma_{p+1 \cdots \ten} \,\in\, \mathrm{Pin}^+(1,d)$ which squares to $+ 1$ (a ``$p$-brane involution'' \cite[\S 4.1]{HSS18}) corresponds to a projection operator on the Kleinian model space
\begin{equation}
  \label{pBraneProjection}
  P \;:=\;
  \tfrac{1}{2}\big(
    \mathrm{id}
    \,+\,
    \Gamma_{p+1 \cdots \ten}
  \big)
  \;:\;
  \begin{tikzcd}
    \mathbb{R}^{1,d\,\vert\,\mathbf{N}}
    \ar[
      r,
      ->>
    ]
    &
    \mathbb{R}^{1,p\,\vert\, \mathbf{N}/2}
    \ar[
      r,
      hook
    ]
    &
    \mathbb{R}^{1,d\,\vert\,\mathbf{N}}
  \end{tikzcd}
\end{equation}
which projects out a sub-space $\mathbb{R}^{1,p\,\vert\,\mathbf{N}/2}$ of half the odd dimensionl (jargon: ``$\sfrac{1}{2}$BPS''). Thus we may ask for super-manifolds $\Sigma^{1,p\,\vert\,\mathbf{N}/2}$ carrying such a $\sfrac{1}{2}$BPS-valued coframe field $(e,\psi)$ and immersed into an ambient $X^{1,d\,\vert\, \mathbf{N}}$ with coframe field $(E,\Psi)$ such the inclusion relation \eqref{pBraneProjection} is suitably exhibited tangentspace-wise.

Such {\it $\sfrac{1}{2}$BPS super-immersions} (\cite[Def. 2.19]{GSS24-FluxOnM5} essentially known in the literature as ``super-embeddings'') exist in 11D supergravity in particular for $p = 5$, known as immersions of {\it probe M5-branes} into spacetime. Remarkably, the $\sfrac{1}{2}$BPS-immersion condition entails and is essentially implied by the existence of a 3-flux density super-form 
$$
  H^s_3 
  \;\defneq\;
  (H_3)_{a_1 a_2 a_3}
  e^{a_1} \, e^{a_2}\, e^{a_3}
$$
on the M5's worldvolume $\Sigma^{1,p\,\vert\,2\mathbf{8}_+}$, such that it is a coboundary for the pullback of the 4-flux to the worldvolume, and hence constitutes a lift to the Gau{\ss} law encoded by the quaternionic Hopf fibration $\mathfrak{l}_{S^4} S^7$:
\begin{equation}
  \label{GaussLawOnM5Immersion}
  \hspace{-6mm}
  \def\arraystretch{2}
  \begin{array}{rcl}
  \adjustbox{raise=3pt}{
  \begin{tikzcd}[column sep=huge, 
    row sep=15pt
  ]
    \Sigma^{1,5\,\vert\,2\cdot \mathbf{8}_+}
    \ar[
      r,
      dashed,
      "{
        H^s_3
      }"
    ]
    \ar[
      dd,
      "{ \phi }",
      "{
        \scalebox{.6}{
          \def\arraystretch{.8}
          \begin{tabular}{c}
          $\sfrac{1}{2}$BPS 
          \\
          immersion
          \end{tabular}
        }
      }"{sloped,  rotate=180}
    ]
    &
    \Omega^1_{\mathrm{dR}}\big(
      -;
      \mathfrak{l}_{S^4}
      \mathcolor{purple}{S^7}
    \big)_{\!\mathrlap{\mathrm{flat}}}
    \ar[
      dd,
      ->>,
      "{
        \mathfrak{l}(
          \mathbb{H}\scalebox{.7}{-Hopf fib.}
        )
      }"{description, pos=.4}
    ]
    \\
    \\
    X^{1,10\,\vert\,\mathbf{32}}
    \ar[
      r,
      "{
        (G_4^s,\, G_7^s)
      }"
    ]
    &
    \Omega^1_{\mathrm{dR}}\big(
      -;
      \mathfrak{l}S^4
    \big)_{\!\mathrlap{\mathrm{flat}}}
  \end{tikzcd}
  }
  &
  \Leftrightarrow
  &
  \left\{\!\!\!\!\!
  \adjustbox{raise=2pt}{
  \begin{tikzcd}[
    row sep=-3pt,
    column sep=0pt
  ]
    \mathrm{d}\,
    H_3 
      &=& 
    \phi^\ast G^s_4
    \\[+22pt]
    \mathrm{d}\, G_4
      &=&
    0
    \\
    \mathrm{d}\, G_7
    &=&
    - \tfrac{1}{2}G_4^s \, G_4^s
  \end{tikzcd}
  }
  \!\!\!\!\!\right\}
  \\
  &
  \Leftrightarrow
  &
  \left\{\!\!\!\!\!
  \scalebox{.8}{
    \begin{tabular}{c}
      $\sfrac{1}{2}$BPS immersion
      \\[-12pt]
      of M5-worldvolume
      \\[-12pt]
      in 11D SuGra solution.
    \end{tabular}
  }
  \right.
  \end{array}
\end{equation}

\noindent This means that at this point, a valid flux quantization law for these fields is given by Cohmotopy: 4-Cohomotopy for the bulk C-field (as such proposed in \cite[\S 2.5]{Sati13} and developed in \cite{FSS19b}\cite{GradySati21}\cite{GSS24-SuGra}), twisting 3-Cohomotopy (classified by the $S^3$-fiber of the quaternionic Hopf fibration) on the brane's worldvolume (discussed in \cite{FSS19c}\cite{FSS21-TwistedString}\cite{GSS24-FluxOnM5}).

\medskip 
But here we take into account one more field:

\medskip

\noindent
{\bf -- Chern-Simons gauge field.} In view of this effective re-definition -- of on-shell 11D supergravity with probe branes -- in terms of (non-linear) Gau{\ss} laws for super-flux densities on super-space, we may go ahead and consider a further super-flux density
$$
  F_2^s
  \;\;
  \defneq
  \;\;
  (F_2)_{a_1 a_2}
  e^{a_1}\, e^{a_2}
$$
on the M5-worldvolume, subjected to the Gau{\ss} law for an ordinary gauge field, but again imposed on super-space:
$$
  \mathrm{d}
  \,
  F_2^s
  \;\;
  =
  \;\;
  0
  \,.
$$
Analysis of the super-components immediately shows that this is equivalent to the further equation of motion 
\begin{equation}
  \label{CSEquationOfMotion}
  F_2 \;=\; 0
\end{equation}
as befits a(n abelian) Chern-Simons gauge field.

While this equation of motion \eqref{CSEquationOfMotion} means that such super-flux $F^s_2$ is actually ``rationally invisible'', under flux-quantization it may still contribute pure torsion-effects to the other higher gauge fields: Namely if we add --- without changing the above equations of motion!, due to \eqref{CSEquationOfMotion} --- a summand of $F_2^s \, F_2^s$ to the Gau{\ss} law for $H^s_3$, then (according to Thm. \ref{FluxQuantizationInEquivariantTwistorialCohomotopy}) it is no longer controlled by the quaternionic Hopf fibration but by the twistor fibration:
\begin{equation}
  \label{GaussLawOnM5ImmersionWithCS}
  \hspace{-6mm}
  \def\arraystretch{2}
  \begin{array}{rcl}
  \adjustbox{raise=3pt}{
  \begin{tikzcd}[column sep=large, 
    row sep=15pt
  ]
    \Sigma^{1,5\,\vert\,2\cdot \mathbf{8}_+}
    \ar[
      r,
      dashed,
      "{
        (F_2^s,\, H^s_3)
      }"
    ]
    \ar[
      dd,
      "{ \phi }"
    ]
    &
    \Omega^1_{\mathrm{dR}}\big(
      -;
      \mathfrak{l}_{S^4}
      \mathcolor{purple}{
        \mathbb{C}P^3
      }
    \big)_{\!\mathrlap{\mathrm{flat}}}
    \ar[
      dd,
      ->>,
      "{
        \mathfrak{l}(
        \scalebox{.7}{twistor. fib.}
        )_\ast
      }"{description, pos=.4}
    ]
    \\
    \\
    X^{1,10\,\vert\,\mathbf{32}}
    \ar[
      r,
      "{
        (G_4^s,\, G_7^s)
      }"
    ]
    &
    \Omega^1_{\mathrm{dR}}\big(
      -;
      \mathfrak{l}S^4
    \big)_{\!\mathrlap{\mathrm{flat}}}
  \end{tikzcd}
  }
  &
  \Leftrightarrow
  &
  \left\{\!\!\!\!\!
  \adjustbox{raise=2pt}{
  \begin{tikzcd}[
    row sep=-4pt,
    column sep=0pt
  ]
    \mathrm{d}\,
    F_2^s &=& 0
    \\
    \mathrm{d}\,
    H^s_3 
      &=& 
    \phi^\ast G^s_4
    \\
    && - F_2^s\, F_2^s
    \\[+4pt]
    \mathrm{d}\, G^s_4
      &=&
    0
    \\
    \mathrm{d}\, G^s_7
    &=&
    - \tfrac{1}{2}G_4^s \, G_4^s
  \end{tikzcd}
  }
  \!\!\!\!\! \right\}
  \\[35pt]
  &
  \Leftrightarrow
  &
  \left\{\!\!\!\!\!
  \scalebox{.8}{
    \def\arraystretch{1}
    \begin{tabular}{c}
      $\sfrac{1}{2}$BPS immersion
      \\
      of M5-worldvolume
      \\
      \color{purple}
      with CS gauge field
      \\
      in 11D SuGra solution
    \end{tabular}
  }
  \right.
  \end{array}
\end{equation}

As the notation already suggests, a flux-quantization law admissible for this system of non-linear Gau{\ss} laws is the non-abelian cohomology theory whose classifying space is $\mathbb{C}P^3$ (over $S^4$), hence the ``twistorial Cohomotopy'' of \cite{FSS20c}\cite{SS23-Mf}.
The character map on this cohomology theory is just what we develop in the main text, in twisted equivariant generalization, and we close by commenting on the consequences:

\medskip

\noindent
{\bf -- Tangential twisting and shifted integrality.} 
The higher gauge fields (flux densities) considered above are all defined on given (immersions of) super-spacetimes, and as such with respect to the {\it background field} of (super-gravity).
According to the dictionary of \hyperlink{TableCG}{\it Table CG}, background fields manifest as {\it twisting} of the flux-quantizing cohomology theory.
Since the topological charges of gravity are encoded in the frame bundle (or its associated tangent bundle) of spacetime, classified by a map $X \xrightarrow{\mathrm{Fr}_X} B \mathrm{Spin}(1,10)$, we are looking for a corresponding ``tangential'' twisting of twistorial Cohomotopy. The subgroup of $\mathrm{Spin}(1,10)$ that preserves the quaternionic Hopf fibration is $\mathrm{Sp}(2)\cdot \mathrm{Sp}(1)$ \cite[Prop. 2.20]{FSS19b}, and the subgroup that preserves the twistor fibration is still $\mathrm{Sp}(2)$ \cite[Prop. 2.2]{FSS20c}, of which finally in the main text we consider the further subgroup $\mathrm{Sp}(1)$, for definiteness.

The shifted flux quantization \eqref{IntegralityConditions} of $G_4$, which is implied 
\cite[Prop. 3.13]{FSS19b}
by this tangential twisting is thought to be \cite{Witten97} a key aspect of the completion of 11D SuGra to ``M-theory''.

\medskip

\noindent
{\bf -- Equivariance and anyonic solitons.}
As indicated in \hyperlink{TableCG}{\it Table CG}, passage to equivariant cohomology on $G$-spaces corresponds to considering higher gauge fields on (super) $G$-orbifolds (cf. \cite{SS19a}\cite{BSS19}).
In the present context, an interesting situation are M5-branes wrapped on $S^1$-bundles over 2-dimensional orbifolds locally of the following form (cf. \cite[p. 7]{GukovPei17}):
$$
  \Sigma^{1,5}
  \;\defneq\;
  \begin{tikzcd}[
    column sep=-8pt
  ]
  \mathbb{R}^{1,0}
  &\times&
  \mathbb{R}^{2}_{\cup\{\infty\}}
  &\times&
  S^1
  &\times&
  \mathbb{R}^2
  \ar[
    in=60,
    out=180-60,
    looseness=3.5,
    shift right=3pt,
    "{
      \;\;\,\mathclap{\ZTwo}\;\;\,
    }"{description}
  ]
  \mathrlap{\,,}
  \end{tikzcd}
$$
where $\ZTwo$ acts on $\mathbb{R}^2$ by point reflection, and where $(-)_{\cup\{\infty\}}$ denotes one-point compactification by adding a ``point at infinity'', as suitable for measuring solitonic charges (\cite[\S 2.2]{SS24-Flux}).

On worldvolume domains of this form, flux-quantization  in equivariant twistorial Cohomotopy restricts on the orbi-singularity 
to flux-quantization in 2-Cohomotopy 
\eqref{GHetEquivariantSpLRParametrizedTwistorSpace}.
$$
  \hspace{8mm}
  \underset{
    \mathclap{
      \raisebox{-4pt}{
        \scalebox{.7}{
          \color{darkblue}
          \bf
          \def\arraystretch{.9}
          \begin{tabular}{c}
            equivariant twistorial
            Cohomotopy
            \\
            of M5-brane worldvolume
            \\
            wrapped on Seifert-orbisingularity
          \end{tabular}
        }
      }
    }
  }{
  \mathcal{T}_{\ZTwo}
  \big(\!
    \orbisingular
    (
      \Sigma^{1,5}
      \!\sslash\!
      \ZTwo
    )
  \big)
  }
  \quad \;\; 
  =
  \left\{\!\!\!
  \begin{tikzcd}[
    column sep=12pt
  ]
    \ZTwo/1
  \ar[
    in=60,
    out=180-60,
    looseness=3.5,
    shift right=3pt,
    "{
      \;\;\mathclap{\ZTwo}\;\;\,
    }"{description}
  ]
    \mathrlap{\;\;:}
    \ar[d]
    &[-17pt]
    \mathbb{R}^{1,0}
    \times
    \mathbb{R}^2_{\cup \{\infty\}}
    \times
    S^1
    \times 
    \mathbb{R}^2
    \ar[r, dashed]
    &
    \mathbb{C}P^3
    \\
    \ZTwo/\ZTwo
    \mathrlap{\;\;:}
    &
    \mathbb{R}^{1,0}
    \times
    \mathbb{R}^2_{\cup \{\infty\}}
    \times
    S^1
    \ar[
      u,
      hook
    ]
    \ar[
      r,
      dashed
    ]
    &
    S^2
    \ar[
      u,
      hook
    ]
  \end{tikzcd}
 \!\!\! \right\}_{\!\!\big/\sim}
$$

By a recent result \cite{SS24-AbAnyons}, this has the interesting consequence of implying that the corresponding solitonic field configurations have {\it anyonic quantum states}  described by abelian Chern-Simons quantum observables. Such a derivation is of considerable interest in application to quantum materials and to quantum computation (cf. \cite{SS23-ToplOrder}); several authors have argued for a similar conclusion on more informal grounds, following \cite{ChoGangKim20}. 

We discuss this application in more detail in the companion article \cite{SS24-AbAnyonsOnSeifert}.



\end{document}